\documentclass[12pt,a4paper,english]{article}
\usepackage[latin9]{inputenc}
\usepackage{amsmath}
\usepackage{amssymb}
\usepackage{esint}
\usepackage{authblk}
\usepackage[authoryear]{natbib}
\usepackage{longtable}
\usepackage{multirow}
\usepackage{hyperref}
\usepackage{booktabs}
\usepackage{rotating,capt-of}
\usepackage{float}

\usepackage{bm}
\usepackage{appendix}
\usepackage{soul}
\usepackage{enumitem}
\usepackage{titletoc}

\usepackage{amsfonts}
\usepackage[english]{babel}
\usepackage{graphicx}
\usepackage{color}
\usepackage{colordvi}
\usepackage{multirow}
\usepackage{multicol}
\usepackage{babel}
\usepackage{amsthm}
\newdimen\defpicwidth
\newdimen\defepswidth
\def\SKparam#1#2{}
\newcommand{\R}{\mathbb{R}}
\newcommand{\IF}{\boldsymbol{1}}    % indicator function
\long\def\@makecaption#1#2{%
  \vskip\abovecaptionskip
  \sbox\@tempboxa{#1. #2}%
  \ifdim \wd\@tempboxa >\hsize
    {\centerline{\renewcommand{\baselinestretch}{0.8}%
\small\normalsize\parbox{0.8\textwidth}{#1. #2}}}\par
  \else
    \global \@minipagefalse
    \hbox to\hsize{\hfil\box\@tempboxa\hfil}%
  \fi
  \vskip\belowcaptionskip}
{\begin{longtable}{|p{1mm}p{3mm}p{0.8\textwidth}p{1mm}|}
\hline\hline
&&&\\[-4mm]
&& \hskip-20pt \centerline{\bf \summaryname}&\\[1mm]
\hline\hline\endfirsthead
\hline\hline
&&&\\[-4mm]
&& \hskip-20pt \centerline{\bf \summarycontname }&\\[1mm]
\hline\hline\endhead
\hline\hline\endfoot
}
{\end{longtable}}
\newcommand{\rank}{\mathop{\rm{rank}}}
\providecommand{\pdf}[1]{}

\newcommand{\dirxplore}[1]{../#1}
\newcommand{\dircodes}[1]{\wwwebook/codes/#1}
\def\wwwebook{http://www.quantlet.de}
\def\definitionname{DEFINITION}
\def\remarkname{REMARK}
\def\summarycontname{Summary (continued)}
\def\summaryname{Summary}
\def\bookname{}
\ifx\sfb\undefined
  \newcommand{\sfb}{}
  \fi

  \def\bookname{sfs}
\SKparam{BEGIN}{xpl ps pdf}
\SKparam{INPUT}{titlesv}
\SKparam{END}{}
\SKparam{BEGIN}{html java}
\SKparam{INPUT}{matitcs}
\SKparam{END}{}
\SKparam{LIB}{sfs}
\def\wwwebook{http://www.quantlet.com/mdstat}

\newcommand{\E}{\mathop{\mbox{\sf E}}}     % ISE
\newcommand{\N}{\mathop{\mbox{\sf N}}}

\SKparam{BEGIN}{ps pdf}
\SKparam{INPUT}{dedic}
\SKparam{INPUT}{pref}
\SKparam{END}{}

\renewcommand{\P}{\mathrm{P}}            % ISE

\def\defeq{\stackrel{\mathrm{def}}{=}}  % for definitions

\def\remarkname{REMARK}

\newtheorem{theorem}{THEOREM}[section]
\newtheorem{corollary}{COROLLARY}[section]
\newtheorem{lemma}{LEMMA}[section]

\theoremstyle{definition}
\newtheorem{remark}{\remarkname}[section]
\theoremstyle{plain}
\newtheorem{definition}{\definitionname}[section]

\newcommand{\vps}{{\varepsilon}}
\newtheorem{assumption}{Assumption}[section]
\numberwithin{figure}{section}
\numberwithin{table}{section}
\def\defeq{\stackrel{\mathrm{def}}{=}}  % for definitions
\newcommand{\bigO}{\mathcal{O}}
\newcommand{\smallO}{\mbox{\tiny $\mathcal{O}$}}

\def\mG{\mathcal{G}}
\def\mV{\mathcal{V}}
\def\mF{\mathcal{F}}

\def\ohH{\overline{\mathcal{H}}}
\def\tg{\tilde{g}}

\def\oH{\overline{H}}

\def\beq{\begin{equation}}
\def\eeq{\end{equation}}

\def\th{\tilde{h}}
\def\vec{\operatorname{vec}}
\def\exp{\operatorname{exp}}
\makeatother

\theoremstyle{definition}

\usepackage[margin=1in]{geometry}

\usepackage{titlesec}
\titlespacing*{\section}
{0pt}{1.5ex plus 1ex minus .2ex}{1.3ex plus .2ex}
\titlespacing*{\subsection}
{0pt}{1.5ex plus 1ex minus .2ex}{1.3ex plus .2ex}

\title{\Large{\textbf{Uniform Inference on High-dimensional Spatial Panel Networks}\thanks{Corresponding author: Chen Huang, chen.huang@econ.au.dk.}}}
\author[a]{Victor Chernozhukov}
\affil[a]{\small{Department of Economics and Operations Research Center, MIT}}
\author[b]{Chen Huang}
\affil[b]{\small{Department of Economics and Business Economics, Aarhus University}}
\author[c]{Weining Wang}
\affil[c]{\small{School of Economics, University of Bristol}}

\begin{document}
	\maketitle
	\begin{abstract}
	%Social network analysis has gained significant attention recently, identification, estimation and inference issues  are intrinsically important in understanding the underlying network structure. We try to uncover network effect with a predetermined adjacency matrix, and in addition we allow a flexible network specification by incorporating an unspecified network structure. The unspecified network structure can be regarded as latent or misspecified network linkages. To achieve high quality estimation for parameters in both component, we propose to estimate via a double regularized GMM framework. We allow explicitly a factor structure in the instrument variables to address the illposeness of the inverse matrix involved. Moreover this framework also allow us to conduct the inference. Consistency and asymptotic normality is provided with accounting for general spatial temporal dependency of the underlying data generating processes. Simulations and applications demonstrate good performance of our proposed method.
	%The identification, estimation and inference issues are intrinsically important in understanding the underlying network structure.
	%and in addition we allow for a flexible sparse deviation by incorporating an unobserved structure
	\noindent %We propose employing a debiased-regularized, high-dimensional generalized method of moments (GMM) estimator to perform inference on large-scale spatial panel networks. 
    We propose employing a high-dimensional generalized method of moments (GMM) estimator, regularized for dimension reduction and subsequently debiased to correct for shrinkage bias (referred to as a debiased-regularized estimator), for inference on large-scale spatial panel networks. In particular, the network structure, which incorporates a flexible sparse deviation that can be regarded either as a latent component or as a misspecification of a predetermined adjacency matrix, is estimated using a debiased machine learning approach. The theoretical analysis establishes the consistency and asymptotic normality of our proposed estimator, taking into account general temporal and spatial dependencies inherent in the data-generating processes. %The dimensionality allowance in presence of dependency is discussed. 
    A primary contribution of our study is the development of a uniform inference theory, which enables hypothesis testing on the parameters of interest, including zero or non-zero elements in the network structure. Additionally, the asymptotic properties of the estimator are derived for both linear and nonlinear moments. Simulations demonstrate the superior performance of our proposed approach. Finally, we apply our methodology to investigate the spatial network effects of stock returns.
	\par
	\vspace{0.5cm}
	\noindent {\em Keywords}: debiased machine learning, GMM, high-dimensional time series, network analysis, spatial panel data
\end{abstract}

\section{Introduction}

Network analysis has gained significant interest in recent years. In particular, measuring connectedness within a complex system has become a central task in learning networks. Various forms of regression, where the dependent variables are affected by the outcomes and characteristics of network members, have been formulated for that purpose. The established literature on social network analysis favors using a predetermined network structure, which is fully characterized by a specified adjacency matrix, to study peer effects in social networks; see, for example, \citet{lee2007identification,bramoulle2009identification,lee2010specification,yang2017identification,zhu2020multivariate}. As for spatial panel networks, \cite{kuersteiner2020dynamic} consider a class of GMM estimators for general dynamic panel models that allow for potential endogeneity and cross-sectional dependence. %arising from spatial lags and unspecified common shocks. %, as well as time-varying interactive effects. 
An alternative to imposing a known network structure is to estimate the adjacency matrix, provided that the structural parameters are already identified. Examples of related studies include \citet{blume2015linear,de2018recovering,lewbel2019social}. %Additionally, corrupted network structure subject to missing or mismeasured links is considered in some recent works by \citet{lewbel2023estimating,lewbel2023ignoring}.

With the rise of big data availability, many applications are concerned with 
%systems of high-dimensional individuals. \cite{zhu2020multivariate} study a multivariate spatial autoregressive model for large-scale social networks with a prespecified adjacency matrix. Furthermore, as more flexibility is allowed, estimating the interaction matrix would pose the challenge of dealing with too many unknown parameters. 
large-scale networks consisting of a large number of individuals. In particular, spatial panel data involving high-dimensional time series are observed in many financial and economic network analyses. This poses the challenge of estimating too many unknown parameters. To reduce the dimensionality, various machine learning methods based on sparsity and penalization are employed to shrink the parameters. \citet{manresa2013estimating} uses LASSO (Least Absolute Shrinkage and Selection Operator) to quantify the spillover effects in social networks, where the endogenous interactions are not taken into consideration. \citet{de2018recovering} apply Adaptive Elastic Net GMM to estimate the interaction model with important contributions to the identification of the structural parameters. 
\cite{ata2018latent} consider a reduced-form estimation with the innovative discovery of the algebraic results on how the sparsity of the structural parameters relates to that of the parameters in the reduced form. 
%i.e. a sparsity assumption on the adjacency matrix $W$ would lead to approximate sparsity for the matrix $(\mathbf I-\rho W)^{-1}$.
\citet{lam2014regularization} study the penalized estimation of the spatial weight matrix in a spatial lag model through adaptive LASSO and show the oracle properties of the sparse estimator. \citet{wang2024panel} develop a high-dimensional interactive fixed effects estimator that allows for a growing number of latent factors and apply it to peer-effects analysis in networks with sparse links. They demonstrate the consistency of the new estimator and the asymptotic normality of the post-selection estimator of the slope coefficients. In this paper, we also aim to conduct inference on the network structure.

Machine learning methods are notably effective in improving prediction performance. However, statistical inference may suffer from substantial bias due to omitted variables. Debiasing is necessary to construct high-quality point and interval estimates. Taking LASSO-type methodologies as example, \citet{lam2014regularization} establish the asymptotic normality of non-zero elements in the network structure. However, in practice, we often lack prior information about whether parameters are truly non-zero, necessitating a uniform inference theory that allows testing any parameters of interest. For independent and identically distributed (i.i.d.) data, extensive research explores uniform inference in high-dimensional regression settings under exogeneity conditions (e.g., \cite{BCH2014,zhang2014debiased,BCK15Bio,DML}) and, more generally, considers GMM frameworks that allow for endogeneity (e.g., \cite{belloni2018high,belloni2017simultaneous,caner2018high}), through various de-biasing and orthogonalization techniques. 
Building on the idea of orthogonality, \cite{ata2018latent} present an algorithm incorporating bias-corrected Dantzig selector estimator to investigate large networks with latent agents, though without accounting for temporal dependence. Addressing data-generating processes exhibiting dependency, \cite{lasso2018} study LASSO-based inference for exogenous regression under general temporal and cross-sectional dependence. 

In this paper, we are motivated by the need to understand the connectedness within a complex spatial panel network. Our focus is on exploring network structures, which need not be sparse, while allowing for flexible sparse deviations. These deviations can be viewed as either latent or misspecified relative to a predetermined adjacency matrix (e.g., credit chains or common ownership information in a financial system). Specifically, we examine network formation by framing the problem as a general system of dynamic regression equations, considering both temporal and spatial dependencies inherent in the data-generating processes. Methodologically, we extend the model setting in \citet{lasso2018} by allowing for endogeneity in the covariates, which is a natural concern when the regression system is featured with simultaneity by incorporating contemporaneous lags. As a result, sufficiently many moment conditions involving instrumental variables (IV) are needed and we build a debiased-regularized, high-dimensional GMM estimator to facilitate valid inference. Notably, the double LASSO estimation steps used in \citet{lasso2018} for debiasing are unsuitable in our case due to the endogeneity issue. This necessitates the identification of an appropriate moment selection matrix to achieve the desired orthogonality for valid inference. Given the high-dimensional nature of the covariance matrix and its inverse, a unified regularized estimation framework is required to ensure the consistency of both the preliminary estimator and the matrices involved in the debiasing step.  
% {\red
% Since we are in a more general setup compared to \citet{lasso2018}, the analytical formulas are more complex, and the estimators cannot be explained through straightforward regression steps. Methodologically, we also apply the Dantzig selector instead of LASSO, with a particular emphasis on network estimation using a predefined adjacency matrix.
% }

For implementation, we propose employing a Generalized Dantzig Selector (GDS) as an initial step, followed by a debiasing step. Theoretically, we establish the consistency of the GDS estimator and derive the linearization of the debiased estimator to enable the application of the central limit theorem for uniform inference on the parameters of interest (whether of fixed or growing dimension). %The dimensionality allowance in presence of dependency is addressed. 
In particular, we show the asymptotic properties of the debiased-regularized GMM (DRGMM) estimator for both linear and nonlinear moments cases. Moreover, we discuss the connection to the semiparametric efficiency literature, particularly in relation to the construction of our estimator when the dimension of the parameters of interest is fixed. %In addition, to deliver better computational performance, we provide a simple sample-splitting procedure and prove the asymptotics of the estimators formally.

We contribute to the literature in four respects. First, we develop a method for estimating parameters in a high-dimensional endogenous equation system that incorporates both spatial and temporal dynamics. Our theoretical framework accords with general dynamic panel models, capturing heterogeneity through individual-specific parameters. Second, we propose a latent model that shrinks toward a pre-specified network structure. In particular, we provide theoretical insights into how the restricted eigenvalue conditions on the design matrix adapt to the transformation of the covariates. Third, we employ a debiased machine learning approach to conduct simultaneous hypothesis testing on high-dimensional parameters. Finally, we demonstrate the practical utility of our method through an empirical application in a financial network context.

Compared to the high-dimensional GMM estimator developed in \citet{belloni2018high}, this study involves a spatial panel model setup, rather than i.i.d. data, introducing several technical challenges. First, to prove consistency, the verification of certain high-level assumptions requires significantly different steps. We demonstrate the validity of concentration %and identification 
under spatial-temporal dependent processes, %in Lemma \ref{cont} and \ref{id}, 
ensuring that panel data with a network structure can be properly handled. Furthermore, to extend the framework to nonlinear and even non-smooth moments, we employ different techniques for proving tail probabilities and concentration inequalities, as detailed in Appendix \ref{nonlinearmo}. 

The following notations are adopted throughout the paper. For a vector $v = (v_1, \ldots, v_p)^\top$, let $|v |_k = (\sum_{i=1}^p |v_i|^k)^{1/k}$ with $k\geq1$, $|v|_\infty = \max\limits_{1\leq i\leq p} |v_i|$, and $|v|_0$ denote the number of nonzero components of the vector. For a random variable $X$, let $\|X\|_r\defeq(\E|X|^r)^{1/r}$, with $r>0$. For a matrix $A = (a_{ij})\in\R^{p\times q}$, we define $|A|_1 = \max\limits_{1\leq j\leq q} \sum_{i=1}^p |a_{ij}|$, $|A|_{\infty} = \max\limits_{1\leq i\leq p} \sum_{j=1}^q|a_{ij}|$, $|A|_{\max} = \max\limits_{1\leq i \leq p,1\leq j\leq q}|a_{ij}|$, and the spectral norm $|A|_2 = \sup_{|v|_2\leq1} |Av|_2$. 
Moreover, let $\lambda_i(A)$ denote the $i$-th largest eigenvalue of a square matrix $A$, and let $\lambda_{\min}(A)$ and $\lambda_{\max}(A)$ denote the minimal and maximal eigenvalues of $A$, respectively. Similarly, let 
$\sigma_i(A)$ denote the $i$-th largest singular value of $A$, with $\sigma_{\min}(A)$ and $\sigma_{\max}(A)$ representing the minimal and maximal singular values of $A$, respectively. Let $\mathbf I_{p}$ denote the identity matrix of size $p\times p$. For any measurable function on a measurable space $g:\mathcal{W}\rightarrow\R$, define the sample average over the indices $t=1,\ldots,n$ as $\E_n(g(\omega_t))\defeq n^{-1}\sum_{t=1}^n g(\omega_t)$. Given two sequences of positive numbers $a_n$ and $b_n$, write $a_n\lesssim b_n$ (resp. $a_n\asymp b_n$) if there exists constant $C>0$ (independent of $n$) such that $a_n/b_n\leq C$ (resp. $1/C \le a_n / b_n\leq C$) for all large $n$. For a sequence of random variables $x_n$, we use the notation $x_n\lesssim_{\P} b_n$ to denote $x_n=\bigO_{\P}(b_n)$.

The rest of the article is organized as follows: Section \ref{model} outlines the model specification %with a simple example, as well as the general system model 
and estimation steps.
%Section \ref{simpleexample} introduces a simple example. Section \ref{system} shows the general system model with a few more examples, and the estimation steps.
Section \ref{theoretical} presents the main theoretical results for the case of linear moments. %Section \ref{nonlinearmo} provides concentration inequalities for nonlinear moments.
%In Sections \ref{efficiency} and \ref{splitting}, we discuss the connection to the semiparametric efficiency and a sample splitting procedure.
Sections \ref{sim} and \ref{app} provide simulation studies and an empirical application on financial network analysis with potential misspecification.
%Section \ref{app} offers an empirical application to financial network analysis with potential misspecification.
%We have seen that our method can help to uncover the unobserved link in the financial network.
The technical proofs and additional details--including %concentration inequalities for 
extension to nonlinear moments, connection to semiparametric efficiency, and %a sample splitting procedure 
supplementary %examples and remarks 
discussions--are provided in the Online Appendix. The codes to implement the algorithms are publicly accessible via the GitHub repository: \href{https://github.com/huangche/Uniform-Inference-on-High-dimensional-Spatial-Panel-Networks}{Uniform-Inference-on-High-dimensional-Spatial-Panel-Networks}.

	\section{Model and Estimation}\label{model}

	\subsection{Model Specification}\label{system}
	For time points $t=1,\ldots,n$ and individual entities $j=1,\ldots,p$ (both $n,p$ tend to infinity), we consider a spatial panel network model for the nodal response $y_{j,t}$: %$y_t=(y_{1,t}, \ldots, y_{p,t})^\top$:   
	\begin{equation}\label{spn}
	y_{j,t} = \rho^0 w_j^\top y_t+  \delta_j^{0\top} y_t+ \vps_{j,t}, 
	\end{equation}
	where we have an observed network structure $w_j=(w_{j1},\ldots,w_{jp})^\top$ for all $j=1,\ldots,p$, and $\rho^0$ is the spatial autoregressive parameter. In particular, $w_j^\top y_t$ is an observed weighted variable, and vectors $\delta^0_j=(\delta^0_{j1},\ldots,\delta^0_{jp})$, $j=1,\ldots,p$, denote approximately sparse misspecification errors of the network structure. Estimation and inference of $\delta_j^0$ and $\rho^0$ are of interest in analyzing both the actual connectedness among individuals and the joint network effect. 
	
	We let $w_{jj}=0$ and assume $\delta^0_{jj}=0$ for all $j$. It is worth noting that endogeneity is a concern, since the inclusion of $y_{k,t}$ ($k\neq j$) induces simultaneity in the structural equation system. To handle the simultaneity bias, instrumental variables (denoted by $z_{j,t}$) are needed. For example, lags $y_{j,t-1},y_{j,t-2}\ldots$ are commonly used in practice.  We shall further assume that $\vps_{j,t}$ are martingale difference sequences with respect to a suitable filtration, as defined below, and allow for temporal and spatial dependencies in the observed data sample (see \hyperref[A_dgp1]{(A1)(i)}, \hyperref[A_dan]{(A2)} and \hyperref[A_error]{(A3)}).
	
	As a practical example, in \cite{de2018recovering}, $y_{j,t}$ refers to the state tax liabilities for state $j$ in year $t$, $w_{jk}$ is observed as some known geographic measurement of neighborhood, and $\delta^0_{jk}$ contributes to the measurement deviations. In this case, the overall network effect, i.e., $\rho^0 w_j + \delta^0_j$ is interpreted as an overall economic measurement of the connections. On this basis, the social network effect of tax competition is analyzed. 
	
	In addition, we can expand the model by including equation-specific covariates $u_{j,t}\in\R^{d_j}$ whose dimension may grow with the sample size:
	\begin{equation}\label{spn2}
	y_{j,t} = \rho^0 w_j^\top y_t+  \delta_j^{0\top} y_t+ \beta^{0\top} u_{j,t} + \vps_{j,t},\quad %\E(u_{j,t}\vps_{j,t})=0,\quad 
	j=1,\ldots,p.  
	\end{equation}
	The compact form of the model is given by:
	\begin{equation*}
	y_{t} = \rho^0 W y_t+  \Delta^0 y_t+ u_t\beta^0 + \vps_{t},
	\end{equation*}
	where $y_t=(y_{1,t}, \ldots, y_{p,t})^\top$, $u_t=(u_{1,t}^\top,\ldots,u_{p,t}^\top)^\top$, and $\vps_t=(\vps_{1,t}, \ldots, \vps_{p,t})^\top$. In this expression, $W$ and $\Delta^0$ are $p\times p$ matrices, with the $j$-th row of $W$ being $w_j^\top$ and the $j$-th row of $\Delta^0$ being $\delta_j^{0\top}$. If the covariates $u_t$ are exogenous, the transformed covariates $Wu_t$ and $W^2u_t$ are also commonly used as instrumental variables. 
	% {\color{red}
	% The key distinction is that the orthogonalization issue brought by $W y_t$, a straightforward two-step regression approach, as used in \cite{lasso2018}, cannot be directly applied for debiasing. Consequently, it becomes necessary to identify an appropriate rotation matrix for the moment equations to ensure orthogonality. Since these rotation matrices operate in high-dimensional settings, high-dimensional matrix estimation techniques must be employed. To ensure consistency in the estimation framework for both the original estimator and the rotation matrix during the debiasing step, we utilize a Dantzig type estimator for estimating the correlation and covariance matrices.
	% }
	
	Following the spatial econometrics literature, we assume that $|(\rho^0W+\Delta^0)^t|_\infty\leq|c|^t$ with some $|c|<1$ to ensure the stationarity of the model. Without loss of generality, for identification purposes, suppose it is known that there exist $j^*,k^*$ ($k^*\neq j^*$) such that $W_{j^*k^*}\neq0$ and $\Delta^0_{j^*k^*}=0$. This implies that at least one of the non-zero actual links can be correctly specified by the observed linkage. This assumption ensures that the regression does not suffer from multicollinearity. 
	
	In reality, $W$ might be either sparse or dense. On the other hand, it is noted in the literature that the classical spatial estimator for $\rho^0$, such as the IV estimator, would not be consistent if the misspecification error $\Delta^0$ is too dense; see recent works by \cite{lewbel2023estimating,lewbel2023ignoring}. We therefore posit that $\Delta^0$ is approximately sparse, though the observed or actual network structure may not necessarily be sparse. 
	
	When multiple options for the pre-specified matrix $W$ are available, a linear combination of the potential matrices $W_i$, $i=1,\ldots,M$, can be incorporated into the model. Such a generalization has been considered in articles such as \citet{lam2014regularization,higgins2023shrinkage}, with $M$ increasing as $n$ grows. In this case, a regularized estimation can be performed on the weights associated with $W_i$'s, and the sparse weights would be included as part of the unknown parameters %$\bm\beta^0$ 
	in our framework.
	
	In our empirical section \ref{app}, we attempt to quantify the spillover effect among individual stocks, where $y_t$ denotes a vector of stock returns, $W$ is a network matrix corresponding to the common shareholder information, and $\rho^0$ measures the joint network effect. The purpose of this application is to understand the overall network effect among firms and to uncover the latent links. 
	
	% Denote by $e_j$ the $p\times1$ unit vector with the $j$-th element is equal to $1$. Define $\bm X_t =[e_j^{\top}\otimes  y_{t}^{\top}]_{j=1}^p$ ($p\times p^2$), $\widetilde{\bm B}_{p^2\times(p^2+1)}=([e_j^{\top}\otimes\mathbf I_p]_{j=1}^p, [w_j]_{j=1}^p)$, and $\widetilde{\bm\beta}^0=(\delta_1^{0\top},\ldots,\delta_p^{0\top},\rho^0)^\top$, where the notation $[A_j]_{j=1}^p$ indicates we stack $A_j$ by rows over $j=1,\ldots,p$. The model can be expressed as:
	% $$y_t = \bm X_t\bm B \bm\beta^0 + u_t\vartheta^0 + \vps_{t},$$
	% where $\bm B$ is $\widetilde{\bm B}$ with the ($pj^*+k^*$)-th column eliminated, and $\bm\beta^0$ is $\widetilde{\bm\beta}^0$ with the ($pj^*+k^*$)-th element removed. In this model, $\bm X_{t}$ and $u_t$ are the original covariates, and $\bm X_{t}\bm B$ are the transformed covariates. %Particularly, we might be interested in testing $\bm\beta^0$, including the spatial autoregressive parameter and the network structure. The coefficients on the control variables $\vartheta$ would be classified into nuisance parameters. 
	
	%\begin{remark}\label{general}
	It is worth noting that the spatial panel network model we have discussed fits within the framework of high-dimensional regression equations, potentially involving endogeneity. In Appendix \ref{general}, we present a general model framework that encompasses many examples in panel or longitudinal data analysis. For instance, the general model can be dynamic, allowing for the inclusion of lagged values of $y_{j,t}$ in the covariates. The primary theorems presented in Section \ref{theoretical} and Appendix \ref{nonlinearmo} are applicable to the estimator of the general model when using linear or nonlinear moments.

	\subsection{Estimation}\label{est}
	In this subsection, we outline the estimation steps for the DRGMM estimator, which include obtaining a preliminary estimator using the Dantzig selector and the subsequent debiasing procedure, allowing us to perform inference on the parameters of interest. 
	
	For each equation $j=1,\ldots,p$, let $x_{j,t}$ and $\vartheta_j^0$ collect the regressors and the corresponding coefficients respectively. Recall the existence of indices $(j^*,k^*)$, where $j^*\neq k^*$. Specifically, for $j\neq j^*$, we have:
	$$x_{j,t}=(w_j^\top y_t, y_{t}^\top, u_{j,t}^\top)^\top, \quad \vartheta_j^0=(\rho^0,\delta_j^{0\top},\beta^{0\top})^\top;$$
	for $j=j^*$, we have:
	$$x_{j,t}=(w_j^\top y_t, y_{t,-k^*}^\top, u_{j,t}^\top)^\top,\quad \vartheta^0_j=(\rho^0,\delta_{j,-k^*}^{0\top},\beta^{0\top})^\top,$$
	where $y_{t,-k^*}$ denotes the subvector of $y_t$ obtained by excluding the $k^*$th element $y_{t,k^*}$, and similarly for $\delta^0_{j,-k^*}$. With these notations, we can rewrite the model in \eqref{spn2} in the form of $y_{j,t}=x_{j,t}^\top\vartheta_j^{0}+\vps_{j,t}$. 
	Let $K_j$ denote the dimension of $x_{j,t}$.  
	Define $\theta^0 = (\rho^0, \delta_1^{0\top},\ldots,\delta_p^{0\top},\beta^{0\top})^\top\in\R^K$ to collect all the parameters in the model. Note that, given $\delta_{j^*k^*}^0=0$, the parameter sets $(\vartheta_1^0,\ldots,\vartheta_p^0)$ and $\theta^0$ contain the same unknown parameters. We shall estimate $\theta^0$ under the assumption that it is sparse.
	
	Due to the endogeneity in the structural model, we introduce the instrumental variables $z_t=[z_{j,t}]_{j=1}^p\in\R^q$, where $q=\sum_{j=1}^p q_j\geq K$, to construct the moments. Specifically, $z_{j,t}\in\R^{q_j}$ contains the instrumental variables for the $j$-th equation, ensuring that $\E(\vps_{j,t}|z_{j,t})=0$. Here, the notation $[A_j]_{j=1}^p$ indicates that we stack $A_j$ by rows over $j=1,\ldots,p$. %{\color{red} check the MDS condition....}
	
	For each $j=1,\ldots,p$, we define a vector-valued score function $g_j(D_{j,t},\theta)$ that maps $\R^{K_j+q_j}\times\R^K$ to $\R^{q_j}$, where $D_{j,t}\defeq(x_{j,t}^\top, z_{j,t}^\top)^\top$. %(we assume that $D_{j,t}$ is stationary over $t$ in \hyperref[A_dgp1]{(A1)(i)}). 
	For the case with linear moments, the score function is given by 
	$g_j(D_{j,t},\theta)=z_{j,t}\vps_j(D_{j,t},\theta)$, where $ \vps_j(D_{j,t},\theta)=y_{j,t}-x_{j,t}^\top\vartheta_j$. 
	%$\vps_j(D_{j,t},\theta)=y_{j,t}- \rho w_j^\top y_t - \delta_j^\top y_t - \vartheta^\top u_{j,t}$.
	Thus, the moment functions mapping $\Theta\subseteq\R^K$ to $\R^{q_j}$ are:
	$$g_j(\theta)=\E g_j(D_{j,t},\theta)= \E [z_{j,t}(y_{j,t}-x_{j,t}^\top\vartheta_j)],$$
	and we have $g_j(\theta^0)=0$. 
	By stacking the moment functions across equations, %by rows,
	we get $g(\theta)=[g_j(\theta)]_{j=1}^p$. The empirical counterpart is computed as: %$\hat g(\theta) =[\E{_n}z_{j,t}\vps_j(D_{j,t},\theta)]_{j=1}^p$.
	$$\hat g(\theta) =[\E{_n}g_j(D_{j,t},\theta)]_{j=1}^p=[\E{_n}\{z_{j,t}(y_{j,t}-x_{j,t}^\top\vartheta_j)\}]_{j=1}^p.$$
	Additionally, the covariance matrix of the score functions is defined as  
	$$\Omega_{q\times q} \defeq \frac{1}{n}\E\Big[\Big\{\sum_{t=1}^n g(D_t,\theta^0)\Big\}\Big\{\sum_{t=1}^n g(D_t,\theta^0)\Big\}^\top\Big],$$
	%$\Omega_{q\times q} \defeq \E [g(D_t,\theta^0)g(D_t,\theta^0)^\top]$, 
	where $D_t=[D_{j,t}]_{j=1}^p$ and $g(D_t,\theta)=[g_j(D_{j,t},\theta)]_{j=1}^p\in\R^{q}$. In our case, this simplifies to $\Omega=\E\big[[z_{j,t}\vps_{j,t}]_{j=1}^p([z_{j,t}\vps_{j,t}]_{j=1}^p)^\top\big]$.
	
	Suppose the parameter vector $\theta^0\in\R^K$ is partitioned into two parts: the parameters of interest $\theta_1^0\in\R^{K^{(1)}}$ and the nuisance parameters $\theta_2^0\in\R^{K^{(2)}}$, where $K^{(1)}+K^{(2)} =K$. In this context, we are primarily interested in $\theta_1^0=(\rho^0, \delta_1^{0\top},\ldots,\delta_p^{0\top})^\top$, which includes the spatial autoregressive parameter and the misspecification errors of the network structure. Meanwhile, the coefficients on the control variables, denoted by $\theta_2^0=\beta^0$, are treated as nuisance parameters. Let $G_1$ and $G_2$ denote the Jacobian matrices of the moment function $g(\theta)$ with respect to $\theta_1$ and $\theta_2$, respectively. Specifically, since $\theta_1$ contains both common and equation-specific parameters, $G_1\defeq\partial_{\theta_1^\top}g(\theta_1,\theta_2^0)|_{\theta_1=\theta_1^0}$ can be decomposed as $G_1=(G_{11},G_{12})$, where $G_{11}$ is a $q\times1$ vector given by $-[\E(z_{j,t}w_j^\top y_t)]_{j=1}^p$, and $G_{12}$ is a $q\times (K^{(1)}-1)$ block diagonal matrix whose $j$th block is $-\E(z_{j,t}y_t^\top)$.{\linespread{1}\footnote{Under the assumption that $\delta^0_{j^*k^*}$ is known to be zero, we can simplify the parameter vector $\theta^0_1$ by excluding $\delta^0_{j^*k^*}$. Consequently, the corresponding column in the Jacobian matrix $G_{12}$ associated with $\delta^0_{j^*k^*}$ (i.e., the ($pj^*+k^*$)-th colum) should also be removed.}} %Given the assumption that $\delta^0_{j^*k^*}$ is known to be zero, the ($pj^*+k^*$)-th column in $G_{12}$ is set to be a vector of zeros. 
	Additionally, $G_2\defeq\partial_{\theta_2^\top}g(\theta_1^0,\theta_2)|_{\theta_2=\theta_2^0}$ is a $q\times K^{(2)}$ matrix given by $-[\E(z_{j,t}u_{j,t}^\top)]_{j=1}^p$. Other ways to partition the parameters are also possible, with the expressions for $G_1$ and $G_2$ adjusted accordingly.

	The DRGMM estimator procedure will be carried out in two steps:%\footnote{
	% {\color{red}It is worth noting that in this setup, we assume that $n,p,K_j,q_j$ can be all divergent. Moreover, the above setup includes a special case of many IV problems with $n,q_j\to\infty$ and $p=1$. If so, we shall make $\lambda_n = 0$ and directly solve $|\hat{g}(\theta)|_{\infty} =0$.}}
	\begin{itemize}
		\item[1.][Estimation] %Define $\Theta$ as an $s$-sparse parameter space with $\theta \neq 0$. %, i.e. $\Theta=\{\theta\neq0:|\theta|_0\leq s,s>0\}$. 
		Following \cite{belloni2018high}, we consider a Dantzig type of regularization to estimate $\theta^0$, which is an extension of the estimator proposed by \citet{lounici2008high}. Let $\lambda_n>0$. The %{\color{red} the Regularized Minimum Distance (RMD)}
		Generalized Dantzig Selector (GDS) estimator $\hat\theta=(\hat{\theta}_1^\top, \hat{\theta}_2^\top)^\top $ is defined as:
		\begin{equation} \label{danzig}
		%\hat\theta=  \arg\min_{\theta \in \Theta} |\theta|_{1} \quad \text{subject to}\quad \max_{1\leq j\leq p}\big|\E{_n}\{z_{j,t}(y_{j,t}-\tilde x_{j,t}^\top\vartheta_j)\}\big|_\infty\leq\lambda_n.
		\hat\theta=  \arg\min_{\theta \in \Theta} |\theta|_{1} \quad \text{subject to}\quad |\hat g(\theta)|_{\infty} \leq \lambda_n.
		\end{equation}
		Specifically, in the case of linear moments, $|\hat g(\theta)|_{\infty}=\max\limits_{1\leq j\leq p}\big|\E{_n}\{z_{j,t}(y_{j,t}-x_{j,t}^\top\vartheta_j)\}\big|_\infty$.
		
		%the estimator reduces to the ordinary Dantzig Selector proposed by \citet{dantzig2007}. 
		%where $\hat g(\theta) = \hat g(\theta_1,\theta_2)=[\E{_n}g_j(D_{j,t},\theta_1,\theta_2)]_{j=1}^p$.
		
		\item[2.][Debiasing] In order to partial out the effect of the nuisance parameters $\theta_2$, we first consider the moment functions: $M(\theta_1,\theta_2) = \{\mathbf I_q - G_2P(\Omega, G_2)\}g(\theta_1,\theta_2)$, where $P(\Omega, G_2) = (G_2^{\top}\Omega^{-1} G_2)^{-1}G_2^{\top}\Omega^{-1}$. It follows that $M(\theta^0_1,\theta^0_2)=0$ and the Neyman orthogonality property $\partial_{\theta_2^\top}M(\theta_1^0,\theta_2)|_{\theta_2=\theta_2^0}=0$ is satisfied.  Moreover, to construct the approximate mean estimator, we further consider the moment functions: % given by
		\begin{align*}
		\widetilde M(\theta_1,\theta_2;\gamma)&=G_1^\top\Omega^{-1}\{\mathbf I_q - G_2P(\Omega, G_2)\}G_1(\theta_1-\gamma) + G_1^\top\Omega^{-1}M(\gamma,\theta_2)\\
		&=G_1^\top\Omega^{-1}\{\mathbf I_q - G_2P(\Omega, G_2)\}\{G_1(\theta_1-\gamma) + g(\gamma,\theta_2)\},
		\end{align*}
		satisfying $\widetilde M(\theta_1^0,\theta_2^0;\theta_1^0)=0$, $\partial_{\gamma^\top}\widetilde M(\theta_1^0,\theta_2^0;\gamma)|_{\gamma=\theta_1^0}=0$, and $\partial_{\theta_2^\top}\widetilde M(\theta_1^0,\theta_2;\theta_1^0)|_{\theta_2=\theta_2^0}=0$.{\linespread{1}\footnote{These Neyman orthogonality properties ensure that the first-order asymptotic distribution of the debiased estimator is independent of the specific construction of the preliminary estimator in the first step. Essentially, any prediction-based machine learning estimator with a sufficiently fast convergence rate can be utilized.}} 
		%{\color{red}$\check{\theta}_1$'s first-order asymptotic distribution does not depend on how we construct $\hat\theta$ in the first step. It is because that $\widetilde M(\theta_1,\theta_2;\gamma)$ is orthogonal to $\theta_2$ and $\gamma$. Therefore we can use any estimator as long as it has a sufficiently fast convergence rate.}} 
		
		This motivates updating the estimator for the parameters of interest by solving $\widetilde M(\theta_1,\hat\theta_2;\hat\theta_1)=0$ with respect to $\theta_1$. Specifically, the solution, denoted as $\check\theta_1$, is given by %, namely in the form of
		\begin{equation}\label{est.eq}		
		\check{\theta}_1=\hat{\theta}_1 - [\hat{G}_1^{\top}\hat{\Omega}^{-1}\{\mathbf I_q - \hat{G}_2P(\hat{\Omega},\hat{G}_2)\}\hat{G}_1]^{-1}\hat{G}_1^{\top}\hat{\Omega}^{-1}\{\mathbf I_q -\hat{G}_2 P(\hat{\Omega},\hat{G}_2)\}\hat g(\hat\theta_1, \hat\theta_2),
		\end{equation}
		where $\hat\Omega = %\E_n\{g(D_t,\hat\theta_1,\hat\theta_2)g(D_t,\hat\theta_1,\hat\theta_2)^\top\}=
		\E_n\big[[z_{j,t}\vps_{j,t}]_{j=1}^p([z_{j,t}\vps_{j,t}]_{j=1}^p)^\top\big]$, and $\hat G_1$ and $\hat G_2$ are estimators for $G_1$ and $G_2$, respectively. Specifically, let $T_1$ be a nonnegative threshold parameter. The $i$-th row and $j$-th column element of $\hat G_1$ is defined as: 
		$$\hat{G}_{1,ij} = \hat G^1_{1,ij }\IF\{|\hat G^1_{1,ij}|>T_1\},$$
		where the matrix $\hat G_1^1$ is given by %\partial_{\theta_1^\top}\hat g(\theta_1,\hat\theta_2)|_{\theta_1=\hat\theta_1}=
		$\hat G_1^1=(\hat G_{11}, \hat G_{12})$, with $\hat G_{11}=-[\E_n(z_{j,t}w_j^\top y_t)]_{j=1}^p$ and $\hat G_{12}$ being a block diagonal matrix whose $j$th block is $-\E_n(z_{j,t}y_t^\top)$. Similarly, for $\hat G_2$, thresholding is applied to %\partial_{\theta_2^\top}\hat g(\hat\theta_1,\theta_2)|_{\theta_2=\hat\theta_2}=
		$\hat G_2^1=-[\E_n(z_{j,t}u_{j,t}^\top)]_{j=1}^p$. The selection of the threshold will be discussed in the proof of Lemma \ref{sparse} in Appendix \ref{a9.3}.
		\item[3.][Inference]
		Simultaneous inference on the parameters of interest, $\theta_1^0$, can be performed by using either the asymptotic confidence intervals in \eqref{CI_asy} or the bootstrap confidence intervals in \eqref{CI_boot}, as outlined in Section \ref{inference}.
	\end{itemize}
	
	Our estimation procedure is designed for settings where $n,p,K_j$, and $q_j$ (and thus $K$ and $q$) can all diverge. It includes a special case of many IV problems with $n,q\to\infty$ while the number of unknown parameters is fixed. In this scenario, regularization on the parameters is not required in the first estimation step. For instance, in our supplementary simulation study in Appendix \ref{ab}, we consider the Arellano-Bond (AB) estimator for dynamic panel models, where an excessive number of instruments is used to estimate two parameters. We use the conventional AB estimator as the preliminary estimator, which is then refined in a subsequent debiasing step to address overidentification with %regularization in the 
	optimal moment selection. 
	
	It is worth noting that in the high-dimensional setting ($q>n$), $\hat\Omega$ is singular due to the rank deficiency, necessitating the use of a regularized estimator for the precision matrix. Specifically, we consider the constrained $\ell_1$-minimization for inverse matrix estimation (CLIME; see \citealp{cai2011constrained}). In Appendix \ref{a9.2}, we will present a feasible debiased estimator $\check\theta_1$ that incorporates approximate inverse matrices. The convergence rates of the estimators involved in addressing the rank deficiency are analyzed in several auxiliary lemmas in the same appendix.
	
	To provide additional clarity on the debiasing step, in Appendix \ref{2sls}, we establish a link between our debiased estimator and the Two-Stage Least Squares (2SLS) estimator in a low-dimensional framework, where the number of unknown parameters and moment conditions remain fixed. 
	
	In Section \ref{check}, we will demonstrate that the debiased estimator $\check\theta_1$ is asymptotically unbiased and Gaussian. This allows us to perform simultaneous inference on the parameters of interest.
	
	%{\color{red}Shall we put it in the appendix???}
	\begin{remark}[Tuning Parameters]
		The estimation procedure involves tuning parameters. %$\lambda_n$ in step 1 and $\ell_n^\Upsilon$ in step 2. 
		Theoretically, $\lambda_n$ in step 1 must be large enough to satisfy \hyperref[A_tune]{(A5)}, with its order depending on data's dimensionality and degree of dependency (see the discussion under Theorem \ref{theorem.cons}). %Remark \ref{consrate}). 
		Empirically, $\lambda_n$ can be selected based on quantiles from standard normal distribution or through multiplier block bootstrap, as discussed in \citet{lasso2018}. 
		
		For the CLIME tuning parameter %$\ell_n^\Upsilon$ 
	    in step 2, the admissible rate in theory is shown in Lemma \ref{ln_ups} and Remark \ref{rateUps} in the appendix. In practice, the problem in \eqref{clime} can be decomposed into $q$ vector minimizations. For each vector, we use the tuning parameter $1.2\times\inf_{a\in\R^q}|a\hat\Omega-e_j^\top|_\infty$, where $a$ is a row vector, and $e_j$ is the $q\times1$ unit vector with the $j$-th element equal to 1, for $j=1,\ldots,q$. This choice is inspired by \citet{gold2020inference}.
	\end{remark}

	\section{Main Results}  \label{theoretical}
	In this section, we present the theoretical foundation of the proposed estimator for the case of linear moments. %properties of the consistency and the inference debiasing estimator. 
	Specifically, Section \ref{const} focuses on the consistency of the preliminary GDS estimator, $\hat\theta$, in step 1, 
	while Section \ref{check} examines the inference procedure for the final DRGMM estimator, $\check\theta_1$, %(obtained using \eqref{maineq}), 
	for the parameters of interest. Extensions of the main theory to the case of nonlinear moments are discussed in Section \ref{nonlinearmo} of the Appendix.
	
	Throughout this section, we impose the following assumptions and definitions: %conditions on the data generating processes.
	
	\begin{itemize}
		\item[(A1)]\label{A_sta}
		(Stationarity) 
		\begin{itemize}
			\item[(i)]\label{A_dgp1} Given any $j=1,\ldots,p$, and for all $k=1,\ldots d_j,m=1,\ldots,q_j$, let $u_{jk,t}$, $z_{jm,t}$, and $\vps_{j,t}$ be stationary processes over $t$, admitting the representation forms $u_{jk,t} = f^u_{jk}(\ldots, \zeta_{jk,t-1}, \zeta_{jk,t})$, $z_{jm,t} = f^z_{jm}(\ldots, \xi_{jm,t-1},\xi_{jm,t})$, and $\vps_{j,t}=f^\vps_j(\ldots,\eta_{j,t-1},\eta_{j,t})$, where $\zeta_{jk,t}$, $\xi_{jm,t}$, and $\eta_{j,t}$ for $t\in\mathbb Z$ are i.i.d. random elements across $t$, and $f^u_{jk}(\cdot), f^z_{jm}(\cdot), f^\vps_j(\cdot)$ are measurable functions. 
			\item[(ii)]\label{A_dgp2} The network structure satisfies $|(\rho^0W+\Delta^0)^t|_\infty\leq|c|^t$ with some $|c|<1$.
		\end{itemize}
	\end{itemize}
	
	\begin{definition}[Dependence Adjusted Norm]\label{dep}
		Let $\zeta_{jk,0}$ be replaced by an i.i.d. copy $\zeta_{jk,0}^\ast$, and define $u_{jk,t}^{\ast}=f^u_{jk}(\ldots,\zeta^\ast_{jk,0},\ldots,\zeta_{jk,t-1}, \zeta_{jk,t})$. For $r\geq1$, define the functional dependence measure $\delta_{r,j,k,t}= \|u_{jk,t}- u_{jk,t}^{\ast}\|_r$, which measures the dependency of $\zeta_{jk,0}$ on $u_{jk,t}$. Also, define $\Delta_{d,r,j,k}=\sum_{t=d}^\infty\delta_{r,j,k,t}$, which accumulates the effects of $\zeta_{jk,0}$ on $u_{jk,t\geq d}$. Moreover, the dependence adjusted norm of $u_{jk,t}$ is denoted by $\|u_{jk,\cdot}\|_{r,\varsigma}=\sup_{d\geq0}(d+1)^{\varsigma}\Delta_{d,r,j,k}$, where $\varsigma>0$. Similarly, we can define $\|z_{jm,\cdot}\|_{r,\varsigma}$ and $\|\vps_{j,\cdot}\|_{r,\varsigma}$ in the same fashion.
	\end{definition}
	
	\begin{itemize}
		\item[(A2)]\label{A_dan}
		(Dependency) For each $j=1,\ldots,p$, $k=1,\ldots d_j$, and $m=1,\ldots,q_j$, assume that $\|u_{jk,\cdot}\|_{r,\varsigma}<\infty$, $\|z_{jm,\cdot}\|_{r,\varsigma}<\infty$, and $\|\vps_{j,\cdot}\|_{r,\varsigma}<\infty$ for some $r\geq8$ and $\varsigma>0$.
		\item[(A3)]\label{A_error}
		(Error Terms) For all $j=1,\ldots,p$, assume that $\vps_{j,t}$ are martingale difference sequences with $\E(\vps_{j,t}|\mathcal F_{t-1})=0$, $\E(\vps^2_{j,t}|\mathcal F_{t-1})=\sigma_{jj}$, $\E(\vps_{j,t}\vps_{j',t}|\mathcal F_{t-1})=\sigma_{jj'}$, and satisfy $\E(z_{jm,t}\vps_{j,t})=0$ for any $j,j'=1,\ldots,p$ and $m=1,\ldots,q_j$. The filtration is defined as $\mathcal F_{t}\defeq\{(\zeta_{jk,s})_{s\leq t},(\xi_{jm,s})_{s\leq t},(\eta_{j,s})_{s\leq t}\mid k=1,\ldots d_j,m=1,\ldots,q_j,j=1,\ldots,p\}$.
		%$\mathcal F_t=(\ldots,\zeta_{t-1}, \xi_{t-1},\eta_{t-1},\zeta_t,\xi_t,\eta_t)$. 
		%{\color{red} Analogous assumption to \hyperref[A_dan]{(A2)} also holds for $\vps_{j,t}$.}
		
		%{\color{red} answering referer 2 here we should add a comment concerning the moment does not involve long run variance because of the MDS nature of the error term.}
		\item[(A4)]\label{ES}(Exact Sparsity) 
		{There exists a subset $\mathcal{I}\subset\{1,\ldots,K\}$ with cardinality $|\mathcal{I}|=s=\smallO(n)$ such that $\theta^0_k\neq0$ only for $k\in \mathcal{I}$.}
		% \item[(A4')]\label{AS}(Approximate Sparsity) For some constant $C>0$ and $c>1/2$, the absolute values of the parameters $(|\theta^0_k|)_{k=1}^K$ can be rearranged in non-increasing order to $(|\theta_k^{0\ast}|)_{k=1}^K$ such that $|\theta^{0\ast}_k|\leq Cj^{-c}$ for $k=1,\ldots,K$.
		\item[(A5)]\label{A_tune} (Regularization Parameter) The regularization parameter $\lambda_n>0$ is selected such that 
		$$|\hat g(\theta^0)|_{\infty}=\max\limits_{1\leq j\leq p}\big|\E{_n}(z_{j,t}\vps_{j,t})\big|_\infty\leq \lambda_n$$
		holds with probability at least $1-\alpha$, where $0<\alpha<1$.
		
		\item[(A6)]\label{A_id} (Identification) 
		Let $G\defeq\partial_{\theta^\top}g(\theta)|_{\theta=\theta^0}$ and let $\mathcal{I}$ be a subset of $\{1,\ldots,K\}$. For $a\geq1$, define $$\kappa_a^G(s, u) \defeq\min\limits_{\mathcal I:|\mathcal{I}|\leq s} \min\limits_{\theta \in \mathcal C_{\mathcal{I}}(u):|\theta|_a=1} |G\theta|_\infty,$$
		where $\mathcal C_{\mathcal{I}}(u) = \{\theta\in\R^K: |\theta_{\mathcal{I}^C}|_1\leq u|\theta_\mathcal{I}|_1\}$, with $u>0$, $\mathcal{I}^C=\{1,\ldots,K\}\setminus \mathcal{I}$, and $\theta_{\mathcal{I}},\theta_{\mathcal{I}^C}$ are sub-vectors of $\theta$ corresponding to $\mathcal{I},\mathcal{I}^C$. Assume that 
		$$\kappa_a^{G}(s,u)\geq s^{-1/a} C(u),\, a\in\{1,2\},$$
		%holds with probability approaching 1, 
		where $C(u)$ is a decreasing function of $u$, mapping from $(0,\infty)$ to $(0,\infty)$.
	\end{itemize}
	
	In \hyperref[A_dgp1]{(A1)(i)}, we allow for overlap in the innovations $\zeta_{jk,t},\xi_{jm,t},\eta_{j,t}$ as long as the exogeneity condition %$\E(u_{j,t}\vps_{j,t})=0$ and 
	$\E(z_{j,t}\vps_{j,t})=0$ is satisfied. 
	% For the definition of dependence adjusted norm, note that there is a more general way to define it. Let $u_{jk,t}^{\ast}(\ell)=f^u_{jk}(\ldots,\zeta^\ast_{jk,t-\ell}, \ldots,\zeta_{jk,t})$ where $\zeta_{jk,t-\ell}$ is replaced by an i.i.d. copy $\zeta^\ast_{jk,t-\ell}$. The functional dependence measure is denoted by $\delta_{r,j,k,t}(\ell) = \|u_{jk,t}- u_{jk,t}^{\ast}(\ell)\|_r$. Define $\Delta_{d,r,j,k}= \sum_{\ell=d}^\infty\max\limits_{1\leq t\leq n}\delta_{r,j,k,t}(\ell)$, which measures the cumulative effects. Some non-stationary time series cases can also be covered under the assumption that $\|u_{jk,\cdot}\|_{r,\varsigma}=\sup\limits_{d\geq0}(d+1)^{\varsigma}\Delta_{d,r,j,k}<\infty$.
	\hyperref[A_dan]{(A2)} assumes a sufficient decay rate of dependency. In the main text of this paper, we focus on the weak dependence case with $\varsigma>1/2-1/r$. In the detailed proofs in the appendix, we will discuss how the rates adapt to the case of strong dependence. It is worth noting that \hyperref[A_dan]{(A2)}, together with the stationary condition \hyperref[A_dgp2]{(A1)(ii)}, implies that the dependence adjusted norm for the transformed covariates $\|x_{jk,\cdot}\|_{r,\varsigma}$ is also finite. 
	
	% Assumption \hyperref[A_error]{(A3)} imposes restrictions on the dependence structure of the error term. %For simplicity, 
	% Specifically, we assume that the error term follows a martingale difference sequence (m.d.s.) with respect to the filtration $\mathcal F_{t-1}$. This assumption prevents serial correlation in the errors, but it remains reasonable given that our general modeling framework accommodates dynamics through the inclusion of sufficiently many lags.  
	% Due to the m.d.s. nature of the error term, the long-run variance of the score functions does not need to be considered in the formation of $\Omega$, which is involved in the debiasing step of our estimation procedure. 
	Assumption \hyperref[A_error]{(A3)} restricts the dependence structure of the error term by assuming it follows a martingale difference sequence (m.d.s.) with respect to the filtration $\mathcal F_{t-1}$. While this rules out serial correlation, it remains reasonable as our general modeling framework accommodates dynamics through the inclusion of sufficiently many lags.  
	Due to the m.d.s. nature of the error term, the long-run variance of the score functions need not be considered in forming $\Omega$ for debiasing. Additionally, we impose some structure on the conditional variance-covariance matrix to simplify the derivation. However, this setting could be extended to handle more complex structures, such as serial correlations, unobserved heterogeneity, and factor structures. See Appendix \ref{error} for further discussion.
	
	\hyperref[ES]{(A4)} focuses on the sparsity of the true parameter $\theta^0$, which is the assumption we primarily rely on in demonstrating the main theorems. This condition can be extended to the case of approximate sparsity, a more general assumption in the literature on high-dimensional data analysis. In Appendix \ref{approx}, we will %restate the required identification condition and provide estimation bounds 
	derive the estimation error bounds under the approximate sparsity assumption, taking into account the approximation error.
	% \hyperref[ES]{(A4)} and \hyperref[AS]{(A4')} are two different assumptions regarding the sparsity of the true parameter $\theta^0$. We note that \hyperref[AS]{(A4.ii)} can be reformulated to \hyperref[ES]{(A4.i)}. Suppose $\theta^0$ is approximately sparse. %and denote by $\theta^0_{[j]}$ the value of the true parameter corresponding to $|\theta^0|_j^\ast$, as defined in \hyperref[AS]{(A4.ii)}. 
	% We sparsify $\theta^0$ to $\theta^0(\tau)=(\theta_1^0(\tau),\ldots,\theta_K^0(\tau))^\top$, where
	% $$\theta_j^0(\tau)=\operatorname{sign}(\theta^{0\ast}_j)\tilde\theta_j(\tau),\quad\tilde\theta_j(\tau)=\begin{cases}
	% 	|\theta^{0\ast}_j|+\delta/(s-1) &\mbox{ if } Aj^{-\bar a}>\tau, \\
	% 	0&\mbox{ otherwise },
	% 	\end{cases}$$
	% with $\tau$ chosen such that $s=\lfloor(A/\tau)^{1/\bar a}\rfloor=\smallO(n)$ and $s>1$, and $\delta=\sum_{j=1}^K|\theta^{0\ast}_j|\IF(Aj^{-\bar a}\leq\tau)$. Then, we have 
	%         \begin{align*}
	% 	|\theta^0(\tau)|_1&=\sum_{j=1}^s|\theta^{0\ast}_j| +\frac{\delta s}{s-1}=\sum_{j=1}^s|\theta^{0\ast}_j| +\frac{s}{s-1}\sum_{j=s+1}^K|\theta^{0\ast}_j|\geq|\theta^0|_1.
	% 	\end{align*}
	% 	It follows that $\{\theta \in \Theta: |\theta|_1 \leq |\theta^0|_1\}\subseteq\{\theta \in \Theta: |\theta|_1 \leq |\theta^0(\tau)|_1\}$.
	
	\hyperref[A_tune]{(A5)} ensures that $\theta^0$ is feasible for the problem in \eqref{danzig} with probability at least $1-\alpha$. 
	\hyperref[A_id]{(A6)} is an identification assumption that is crucial for ensuring consistency. In Appendix \ref{sec.id}, we discuss the conditions required to validate this assumption. 
	
	\subsection{Consistency of the GDS Estimator $\hat\theta$} \label{const}
	In order to establish the consistency of the GDS estimator $\hat\theta$, we need to derive the error bound for $|\hat{\theta}- \theta^0|_a$ for $a=1$ or $2$, and analyze the convergence rate. Under the identification condition \hyperref[A_id]{(A6)}, the error bound for $|\hat{\theta}- \theta^0|_a$ follows from the error bound for $| g(\hat\theta) - g(\theta^0)|_{\infty}$ (we will elaborate on this argument in the proof of Theorem \ref{theorem.cons}). Using the identity $g(\theta^0) = 0$, we can bound $|g(\hat{\theta})- g(\theta^0)|_{\infty}$ as follows:
	\begin{equation*}
	|g(\hat{\theta})- g(\theta^0)|_{\infty} = |g(\hat{\theta})|_{\infty} \leq |g(\hat{\theta})-\hat g(\hat{\theta})|_{\infty} + |\hat g(\hat{\theta})|_{\infty}.
	\end{equation*}
	Recalling the definition of the GDS estimator, we have $|\hat g(\hat{\theta})|_{\infty}\leq\lambda_n$. Let $\mathcal{R}(\theta^0) \defeq \{\theta \in \Theta: |\theta|_1 \leq |\theta^0|_1\}$ denote the restricted set. As a consequence of \hyperref[A_tune]{(A5)}, we could have $\hat\theta\in\mathcal R(\theta^0)$ with probability at least $1-\alpha$. The remaining task is to demonstrate the concentration result, i.e., to show that:
	$$\sup\limits_{\theta \in \mathcal{R}(\theta^0)}|\hat g (\theta)- g(\theta)|_{\infty}\leq\epsilon_n$$
	holds with probability approaching 1, for a sequence of positive constants $\epsilon_n\downarrow0$ as $n\to\infty$. 
	
	%Assume $\epsilon_n\downarrow0,\delta_n\downarrow0$ are sequences of positive constants. 
	We focus on cases with linear moments, where $g(\theta)=G\theta+g(0)$ and $\hat g(\theta)=\hat G\theta +\hat g(0)$, with $G=\partial_{\theta^\top}g(\theta)$ %|_{\theta=\theta^0}$ 
	and $\hat G=\partial_{\theta^\top}\hat g(\theta)$ %|_{\theta=\theta^0}$. %Specifically, we have...
	being independent of $\theta$.
	It follows that 
	\begin{eqnarray*}
		\sup_{\theta \in \mathcal{R}(\theta^0)}|\hat g (\theta)- g(\theta)|_{\infty}
		&=&  \sup_{\theta \in \mathcal{R}(\theta^0)}|(\hat G- G)\theta|_\infty + |\hat g(0) - g(0)|_{\infty}\\
		&\leq & \sup_{\theta \in \mathcal{R}(\theta^0)}|\theta|_1 |\hat G-G|_{\max}+ |\hat g(0)- g(0)|_{\infty}\\
		&\leq & |\theta^0|_1 |\hat G-G|_{\max}+ |\hat g(0)- g(0)|_{\infty},
	\end{eqnarray*}
	where $|\cdot|_{\max}$ denotes the element-wise max norm of a matrix.
	
	To derive the convergence rate, we need to analyze the rates of $|\hat G-G|_{\max}$ and $|\hat g(0)- g(0)|_{\infty}$ by applying the concentration inequality in Lemma \ref{tail}. %{\color{red}Recall where $r$ is defined}
	For this purpose, we define the following quantities: 
	\begin{align*}
	\Phi_{r,\varsigma}^x = &\max_{1\leq j\leq p,1\leq k\leq d_j} \|x_{jk,\cdot}\|_{r,\varsigma}, \,
	\Phi_{r,\varsigma}^{\vps z}=\max_{1\leq j\leq p,1\leq m\leq q_j}\|\vps_{j,\cdot}z_{jm,\cdot}\|_{r,\varsigma}, \\
	&\quad\quad\Phi_{r,\varsigma}^{xz} = \max_{1\leq j\leq p,1\leq k\leq d_j,1\leq m\leq q_j} \|x_{jk,\cdot} z_{jm,\cdot}\|_{r,\varsigma},
	\end{align*}
	which are all bounded by constants for some $r\geq4$ and $\varsigma>0$ according to \hyperref[A_dan]{(A2)}. Additionally, we define $\Phi_{r,\varsigma}^{yz} = \max\limits_{1\leq j\leq p,1\leq m\leq q_j} \|y_{j,\cdot} z_{jm,\cdot}\|_{r,\varsigma}$. 
	
	For each equation $j$, we aggregate the dependence adjusted norm of the vector of processes $x_{j,t}$ as follows:  
	$$\||x_{j,\cdot}|_{\infty}\|_{r,\varsigma} = \sup_{d\geq 0}(d+1)^{\varsigma}\Delta_{d,r,j},\quad\Delta_{d,r,j}=\sum_{t=d}^\infty\||x_{j,t}- x_{j,t}^{\ast}|_\infty\|_r.$$
	This is in comparison to the dependence adjusted norm for a univariate process
	as in Definition \ref{dep}. Similarly, we define $\||x_{j,\cdot}z_{jm,\cdot}|_{\infty}\|_{r,\varsigma}$. Additionally, we aggregate over $j=1,\ldots,p$ by:
	$$\Big\|\max_{1\leq j\leq p}|x_{j,\cdot}|_{\infty}\Big\|_{r,\varsigma} = \sup_{d\geq 0}(d+1)^{\varsigma}\Delta_{d,r},\quad \Delta_{d,r}=\sum_{t=d}^\infty\Big\|\max_{1\leq j\leq p}|x_{j,t}- x_{j,t}^{\ast}|_\infty\Big\|_r.$$
	The definition for $\Big\|\max\limits_{1\leq j\leq p,1\leq m\leq q_j}|x_{j,\cdot}z_{jm,\cdot}|_{\infty}\Big\|_{r,\varsigma}$ follows similarly.
	
	%{\color{red}Shall we also delete the strong dependence case here?}
	\begin{lemma}[Concentration]\label{cont}
		Assuming that conditions \hyperref[A_sta]{(A1)}-\hyperref[ES]{(A4)} hold, %and either \hyperref[ES]{(A4.i)} or \hyperref[AS]{(A4.ii)} hold, 
		we have
		\begin{equation*}
		\sup_{\theta \in \mathcal{R}(\theta^0)}|\hat{g}(\theta)- g(\theta)|_{\infty} \lesssim_{\P} b_n s+ b'_n=:\epsilon_n,
		\end{equation*}
		where
		\begin{align*}
		b_n&=cn^{-1/2}(\log P_n)^{1/2} + cn^{-1}n^{1/r}(\log P_n)^{3/2}\Big\|\max_{1\leq j\leq p,1\leq m\leq q_j}|x_{j,\cdot}z_{jm,\cdot}|_\infty\Big\|_{r,\varsigma},\\
		b'_n&=cn^{-1/2}(\log P_n)^{1/2}\Phi_{2,\varsigma}^{yz} + cn^{-1}n^{1/r}(\log P_n)^{3/2}\Big\|\max_{1\leq j\leq p,1\leq m\leq q_j}|y_{j,\cdot}z_{jm,\cdot}|_{\infty}\Big\|_{r,\varsigma},
		\end{align*}
		with $r$ and $\varsigma$ satisfying \hyperref[A_dan]{(A2)}, and $P_n =(q\vee n\vee e)$.  %$P=\sum_{j=1}^pq_jK_j$,
		%and {\red $c_{n,\varsigma} = n^{1/r}$} with the $r$ that makes \hyperref[A_dan]{(A2)} hold. %for $\varsigma> 1/2 - 1/r$, and $c_{n,\varsigma} = n^{1/2-\varsigma}$ for $0<\varsigma<1/2 - 1/r$. 
	\end{lemma}
	
	%\begin{remark}[Discussion of the concentration rate]\label{cont_rate}
	In the case where the dependence adjusted norms involved in $b_n$ and $b'_n$ %$\Phi_{2,\varsigma}^{xz}, \|\max_{j,m}|\tilde x_{j,\cdot}z_{jm,\cdot}|_\infty\|_{r,\varsigma}, \Phi_{2,\varsigma}^{yz}, \|\max_{j,m}|\tilde y_{j,\cdot}z_{jm,\cdot}|_\infty\|_{r,\varsigma}$ 
	are bounded by constants, %and for $\varsigma> 1/2 - 1/r$ (weak dependence case), if 
	and assuming that $n^{-1/2+1/r}(\log P_n) =\bigO(1)$ for sufficiently large $r$, we have the concentration rate 
	$$\epsilon_n \lesssim (s+1)n^{-1/2}(\log P_n)^{1/2},$$
	which matches the rate shown in Lemma 3.3 of \cite{belloni2018high} for i.i.d. data.
	%\end{remark}
	
	% On the other hand, as a consequence of \hyperref[A_tune]{(A5)}, we could have $\hat\theta\in\mathcal R(\theta^0)$ (i.e., $|\hat\theta|_1 \leq |\theta^0|_1$) with probability at least $1-\alpha$, if a solution $\hat\theta$ to the problem in \eqref{danzig} exists. 
	% Consider the event 
	% $$\{|\hat g(\theta^0)|_{\infty}\leq \lambda_n, \hat\theta\in\mathcal R(\theta^0), |\hat g(\hat\theta) - g(\hat\theta)|_{\infty}\leq \epsilon_n\}.$$
	% By the concentration inequality in Lemma \ref{cont}, assumption \hyperref[A_tune]{(A5)}, and the union bound, this event holds with probability at least $1-\alpha-\smallO(1)$. 
	Combining the results from Lemma \ref{cont} with the identification condition \hyperref[A_id]{(A6)}, we obtain the bound on the estimation error of the GDS estimator. The rate of consistency is stated in the following theorem.

	\begin{theorem}[Consistency of the GDS Estimator]\label{theorem.cons}
		%Under the conditions of Lemma \ref{cont}, and 
		Assuming that conditions \hyperref[A_sta]{(A1)}-\hyperref[A_id]{(A6)} hold, and recalling the definitions of $b_n$ and $b'_n$ from Lemma \ref{cont}, we obtain the following error bound: %the consistency of the GDS estimator defined in \eqref{danzig}: %in the linear moments case
		\begin{equation} \label{ratetheta}
		|\hat{\theta}- \theta^0|_a \lesssim (b_ns+ b'_n+\lambda_n)s^{1/a} C(u)^{-1}=:d_{n,a},\, a\in\{1,2\},
		\end{equation}
		which holds with probability at least $1-\alpha-\smallO(1)$. 
		% Here, $b_n(1+u)s\lesssim C(u)$ %(recall the condition ii) in Lemma \ref{id})
		% holds for sufficiently large $n$ and some $u>0$.
		%	{According to Corollary 5.1 of \citet{lasso2018}, the order of $\lambda_n$ is given by
		% $$n^{-1}\max_{j,m}\Big(\|z_{jm,\cdot}\vps_{j,\cdot}\|_{2,\varsigma}(n\log q)^{1/2}\vee\|z_{jm,\cdot}\vps_{j,\cdot}\|_{r,\varsigma}(n\varpi_{n,\varsigma}q)^{1/r}\Big),$$
		%	where for $\varsigma > 1/2-1/r$, $\varpi_{n,\varsigma} = 1$; for $\varsigma< 1/2-1/r$, $\varpi_{n,\varsigma} = n^{r/2-1- \varsigma r}$.}
	\end{theorem}
	
	According to Corollary 5.1 of \citet{lasso2018}, the order of $\lambda_n$ is given by
	$$n^{-1}\max_{1\leq j\leq p,1\leq m\leq q_j}\Big(\|z_{jm,\cdot}\vps_{j,\cdot}\|_{2,\varsigma}(n\log q)^{1/2}\vee\|z_{jm,\cdot}\vps_{j,\cdot}\|_{r,\varsigma}(nq)^{1/r}\Big).$$
	%where for $\varsigma > 1/2-1/r$, $\varpi_{n,\varsigma} = 1$ {\color{red}; and for $\varsigma< 1/2-1/r$, $\varpi_{n,\varsigma} = n^{r/2-1- \varsigma r}$}. 
	%\begin{remark}[Discussion of the consistency rate]\label{consrate}
	%As a continuation of Remark \ref{cont_rate}, 
	%We additionally assume $\max_{j,m}\|z_{jm,\cdot}\vps_{j,\cdot}\|_{r,\varsigma}$ is bounded by constant.
	%involved in $b_n+ b'_n$  are bounded, and statements therein holds true.
	%In particular, for the weak dependence case ($\varsigma > 1/2-1/r$), we have $\lambda_n \lesssim n^{-1/2}(\log q)^{1/2}$, given that $(nq)^{1/r} \lesssim (n\log q)^{1/2}$. 
	In the case where$(nq)^{1/r} \lesssim (n\log q)^{1/2}$, we have $\lambda_n \lesssim n^{-1/2}(\log q)^{1/2}$. This implies that if $r$ is sufficiently large, $q$ can diverge as a polynomial rate of $n$ (a better dimension allowance for $q$ is possible under stronger exponential moment conditions; see Comment 5.5 in \citet{lasso2018}). Consequently, assuming $\max\limits_{1\leq j\leq p,1\leq m\leq q_j}\|z_{jm,\cdot}\vps_{j,\cdot}\|_{r,\varsigma}$ is bounded by a constant, we have:
	$$d_{n,a} \lesssim (s+2)s^{1/a}n^{-1/2} (\log P_n)^{1/2},$$
	which is of the same order as the rate for the i.i.d. case studied in Theorem 3.1 of \cite{belloni2018high}.

\subsection{Inference Theory for the Debiased Estimator $\check\theta_1$}\label{check}
In this subsection we show the asymptotic properties of the debiased estimator $\check\theta_1$ obtained in the second step. %as in \eqref{maineq}. 
In particular, we provide a key representation that linearizes the estimator, facilitating the application of a high-dimensional Gaussian approximation theorem for inference.

\subsubsection{Linearization}\label{linear.rate}

Define ${A}\defeq{G}_1^{\top}{\Omega}^{-1}(\mathbf I_q - {G}_2P({\Omega},{G}_2))$ and $B\defeq(AG_1)^{-1}$, where $P(\Omega, G_2) \defeq (G_2^{\top}\Omega^{-1} G_2)^{-1}G_2^{\top}\Omega^{-1}$. Consider estimators of $A$ and $B$, denoted by $\hat{A}$ and $\hat{B}$. More details about the construction of these estimators are discussed in Section \ref{est} and Appendix \ref{a9.2}.
%As discussed in Section \ref{est}, we consider an estimator of $A$ given by $\hat{A}=\hat{G}_1^{\top}\hat\Upsilon(\mathbf I_q - \hat{G}_2\hat{\Xi}\hat{G}_2\hat\Upsilon)$ and an approximation of $B$ by $\hat B=\hat\Pi-\hat\Pi(\mathbf I_q+\hat F\hat\Pi)^{-1}\hat F\hat\Pi$.

%We denote by {\color{red}$\tilde{G}_1\defeq\partial_{\theta_1^\top}\hat g(\theta_1,\tilde\theta_2)|_{\theta_1=\tilde\theta_1}$THis one does not need for a linear moment} the partial derivative of $\hat{g}(\theta_1,\theta_2)$ with respect to $\theta_1$ valued at $\tilde{\theta}=(\tilde\theta_1,\tilde\theta_2)$, which is the corresponding point lying in the line segment between $\hat{\theta}=(\hat\theta_1,\hat\theta_2)$ and $\theta^{0}=(\theta_1^0,\theta_2^0)$. In the case of linear moments models $\tilde{G}_1 = \hat G_1$.
%Otherwise when the nuisance and the parameter of interest is clearly separable, for example we can suppose $\gamma, \lambda$ be the nuisance and $\theta_1 = \mbox{vec}{(\Delta)}$ is the parameter of interest
%{\red Note that $(\hat{A}\hat{G}_2)^{-1}\hat{A}\hat{G}_2$ is not exactly $\mathbf I_q$ when $(\hat{A}\hat{G}_2)^{-1}$ is estimated using constrained $\ell_1$-minimization for inverse matrix estimation (CLIME, see \cite{cai2011constrained}).}

We shall analyze the accuracy of estimator $\check\theta_1$. % in \eqref{maineq}. 
Observe that
\begin{equation}\label{linearization}
\check{\theta}_1  -\theta_1^{0} =\hat{\theta}_{1} -\theta_{1}^0 - \hat B\hat{A}\hat{g}({\hat{\theta}})=-BA\hat{{g}}({{\theta}^0})+ r_n,
\end{equation}
where $r_n=r_{n,1}+r_{n,2}$, and
$$
r_{n,1}=(\mathbf I - \hat B\hat A\hat G_1)(\hat\theta_1-\theta_1^0),\, %{\red r_{n,2}=\hat B\hat A(\hat G_1-\tilde G_1)(\hat\theta_1-\theta_1^0)},\,
r_{n,2}=(BA - \hat B\hat A)\hat g(\theta^0).
$$
Note that, due to the Neyman orthogonality property, the term $r_{n,1}$ is expected to be small. Under mild conditions, the last term $r_{n,2}$ is also expected to vanish. 
% \begin{eqnarray*}
%&&\hat{\theta}^2  -\theta_2^{0} \\&=&
%\check{\theta}_{2} -\theta_{2}^0 - (\hat{A}\hat{G}_2)^{-1}\hat{A}\hat{g}({\check{\theta}})\\
% &=& -(A{G}_2)^{-1}A\hat{{g}}({{\theta}^0})+ (I - (\hat{A}\hat{G}_2)^{-1}\hat{A}\hat{G}_2)(\check{\theta}_2- \theta_{2}^0)+  (\hat{A}\hat{G}_2)^{-1}(\hat{A}\hat{G}_2)(\check{\theta}_2- \theta_{2}^0)\\&&+ ((A G_2)^{-1}A -(\hat{A} \hat{G}_2)^{-1}\hat{A} )\hat{g}(\theta^0) + (\hat{A} \hat{G}_2)^{-1}\hat{A} (\hat{g}(\theta^0) - \hat{g}(\check{\theta})) \\
% &=& -(A{G}_2)^{-1}A\hat{{g}}({{\theta}^0})+ r_n,
%  \end{eqnarray*}
%  where $r_n = (I - (\hat{A}\hat{G}_2)^{-1}\hat{A}\hat{G}_2)(\check{\theta}_2- \theta_{2}^0)+ (\hat{A}\hat{G}_2)^{-1}\hat{A} (\hat{G}_2- \tilde{G}_{2})(\check{\theta}_2-\theta_2^0)+ ((A G_2)^{-1}A -(\hat{A} \hat{G}_2)^{-1}\hat{A} )\hat{g}(\theta^0) \defeq r_{n1}+ r_{n2}+ r_{n3}.$
%\\
By applying the triangle inequality and H\"{o}lder's inequality, we have the following bounds for the two terms $r_{n,1}$ and $r_{n,2}$, respectively:
\begin{eqnarray*}
	|r_{n,1}|_{\infty} &\leq&  |\mathbf I - \hat B\hat{A}\hat{G}_1|_{\max}|\hat{\theta}_1- \theta_1^0|_{1}\\
	&\leq&  |B|_{\infty}|AG_1 - \hat{A}\hat{G}_1|_{\max}|\hat{\theta}_1- \theta_1^0|_{1}+ |\hat B- B|_{\max}|\hat{A} \hat{G}_1|_{1}|\hat{\theta}_1- \theta_1^0|_{1},\\
% |r_{n,2}|_{\infty}&\leq&  |\hat B|_{\infty}|\hat{A}(\hat{G}_1- \tilde{G}_{1})|_{\max}|\hat{\theta}_1-\theta_1^0|_1, {\red \text {this does exists anymore}}
% \\
	|r_{n,2}|_{\infty}&\leq &  |\hat B- B|_{\infty}  |A|_{\infty} |\hat{g}(\theta^0)|_{\infty} +
	|\hat B|_{\infty}|\hat A- A|_{\infty} |\hat{g}(\theta^0)|_{\infty}.
\end{eqnarray*}

The linearized representation in \eqref{linearization} shows that the debiased estimator $\check\theta_1$ can be expressed by the true parameter $\theta_0$ plus a weighted empirical moment function evaluated at $\theta_0$, along with an approximation error. Consequently, relying on a high-dimensional Gaussian approximation of the leading term $BA\hat{g}({{\theta}^0})=(A{G}_1)^{-1}A \hat{g}({{\theta}^0})$, as will be discussed in Section \ref{inference}, valid inference can be conducted, provided that the linearization errors are asymptotically negligible in the sense that $|r_n|_{\infty}$ is of small order. We now present a theorem for the linearization of the debiased estimator. 

%under the linear case.

%{\color{red}Just make sure that the sparsity of $G$ has not so much to do with the sparsity of $\theta^0$.}

\begin{theorem}[Linearization]\label{linear}
	Under assumptions \hyperref[A_sta]{(A1)}-\hyperref[A_id]{(A6)}, along with \hyperref[A_clime]{(A8)} in Appendix \ref{a9.2}, %the conditions in Lemma %\ref{cont}, \ref{id}, 
 %\ref{df}, \ref{ratef}}, 
 and the Gaussian approximation assumption for $g(D_t,\theta^0)$ (as in \hyperref[A_Srate]{(A7)}, with the dimensionality $|\mathcal S|$ replaced by $q$), suppose that $|A|_{\max}\leq C$ for some constant $C>0$, and there exist upper bounds such that $|A|_\infty\leq\iota$, $|AG_1|_\infty\leq\omega$, and $|(AG_1)^{-1}|_\infty\leq \kappa$. 
 % suppose that $|A|_{\max}\leq C$ and $|A|_\infty\leq\iota$ for some positive constants $C$ and $\iota$. Furthermore, assume that for some $\omega_1=\smallO(n)$, we have $|AG_1|_\infty\leq\omega_1^{1/2}|AG_1|_2\asymp\omega_1^{3/2}$, and $|(AG_1)^{-1}|_\infty\leq \vartheta\asymp\omega_1^{-1}$ if $K^{(1)}$ is fixed, $|(AG_1)^{-1}|_\infty\leq\omega_1^{1/2}|(AG_1)^{-1}|_2\leq\vartheta\asymp \omega_1^{3/2}$ if $K^{(1)}$ is diverging. 
 Then, we have
	\begin{equation*}
	\check{\theta}_1  -\theta_1^{0} =-(A{G}_1)^{-1}A \hat{g}({{\theta}^0})+ r_n,
	\end{equation*}
	where  
$|r_n|_\infty\lesssim_\P \varrho_{n,1}+\varrho_{n,2}$, with $\varrho_{n,1}$ and $ \varrho_{n,2}$ defined in \eqref{rater} in the detailed proof. 
\end{theorem}

The proof of this theorem and the detailed rate of $|r_n|_\infty$ are deferred to Appendix \ref{a9.3}. In particular, %continuing to Remarks \ref{cont_rate} and \ref{consrate}, 
we will discuss the rate specifically under the special case where all the dependence adjusted norms involved are bounded by constants in Remark \ref{rnfinal}.
To enable valid inference through the Gaussian approximation on the leading term, we require that the linearization errors be sufficiently small, ensuring that $\sqrt{n}|r_n|_\infty=\smallO_\P(1)$. This condition imposes restrictions on the allowed dimensionality and sparsity relative to $n$, under mild assumptions.

	\subsubsection{Simultaneous Inference}\label{inference}
	In this subsection, we cite a high-dimensional Gaussian approximation theorem  to facilitate the simultaneous inference of the parameters.
	The theorem is adapted from \cite{ZW15gaussian}. Specifically, we focus on testing the hypothesis $H_0:\theta_{1,k}^0=0,\forall k\in \mathcal S$, where $\mathcal S\subseteq \{1,\ldots,K^{(1)}\}$, and $\theta_{1,k}^0$ denotes the $k$-th element of the vector $\theta_1^0$. To proceed with this inference, we first revisit some key definitions from Section \ref{est}. 
	
	For the case of linear moments, the score functions evaluated at the true parameters are given by $g_j(D_{j,t},\theta^0)= z_{j,t}\vps_{j,t}$, where $D_{j,t}\defeq(x_{j,t}^\top, z_{j,t}^\top)^\top$. Let $D_t=[D_{j,t}]_{j=1}^p$ and $g(D_t,\theta)=[g_j(D_{j,t},\theta)]_{j=1}^p$.
	Define the vector $\mG_t=(\mG_{k,t})_{k\in\mathcal S}$, where $\mG_{k,t}=-\zeta_kg(D_t,\theta^0)$, and $\zeta_k$ is the $k$-th row of the matrix $(AG_1)^{-1}A$. 
	Assuming $|(AG_1)^{-1}A|_\infty$ is bounded by a constant, for any $k\in\mathcal S$, the dependence adjusted norm of $\mG_{k,t}$ is bounded by 
	$$\|\mG_{k,\cdot}\|_{r,\varsigma}\lesssim\max_{1\leq j\leq p,1\leq m\leq q_j}\|z_{jm,\cdot}\|_{2r,\varsigma}\|\vps_{j,\cdot}\|_{2r,\varsigma}.$$
	
	For simultaneous inference, we allow the number of parameters being tested, i.e., the cardinality $|\mathcal S|$, to increase as $n\to\infty$. Specifically, we consider a polynomial growth rate, $|\mathcal S|=n^c$ for some $c>0$. 
	The admissible growth rate is specified in the following assumption:
	\begin{itemize}
		\item[(A7)]\label{A_Srate}
		(Gaussian Approximation) With the same $r$ and $\varsigma$ that satisfy \hyperref[A_dan]{(A2)}, and assuming $\varsigma>1/2-1/r$ (weak dependence case), let $|\mathcal S|^{2/r}n^{2/r-1/2}\{\log(|\mathcal S|n)\}^{3/2}\to0$ as $n\to\infty$, where $|\mathcal S|=n^c$ for some $c>0$.
	\end{itemize}

	We now state the Gaussian approximation results as follows. Denote by $c_\alpha$ the $(1-\alpha)$ quantile of the $\max_{k\in S}|\mathcal Z_k|$, where $\mathcal Z_k$ are the standard normal random variables. Let $\sigma_k^2$ be the $k$-th diagonal element of the covariance matrix $(AG_1)^{-1}A\Omega A^\top\{(AG_1)^{-1}\}^\top=(AG_1)^{-1}$. Under \hyperref[A_Srate]{(A7)} and the same conditions as in Theorem \ref{linear}, assume %for each $k\in \mathcal S$ 
	that there exists a constant $C>0$ such that {$\min\limits_{k\in \mathcal S}\operatorname{Var}\big(n^{-1/2}\sum_{t=1}^n\mG_{k,t}\big)\geq C$}.  %where "$\operatorname{avar}$" denote the long-run variance. 
	Then, we have
	\begin{equation*}
	\lim_{n\to\infty}\big|\P(\sqrt{n}|\check\theta_{1,k}-\theta_{1,k}^0|\leq c_{\alpha}\sigma_k,\forall k\in \mathcal S) - (1-\alpha)\big|=0.
	\end{equation*}
	The conclusion also holds when $\sigma_k$ is replaced by a consistent estimator $\hat\sigma_k$. Consequently, for each $k\in\mathcal S$, we can construct the two-sided $(1-\alpha)$ confidence interval using asymptotic normality as:
	\begin{equation}\label{CI_asy}
	[\check\theta_{1,k}-\hat\sigma_kn^{-1/2}c_\alpha, \,\check\theta_{1,k}+\hat\sigma_kn^{-1/2}c_\alpha].
	\end{equation}
	
	Based on the Gaussian approximation results, the multiplier bootstrap can be employed to determine the critical value. To account for temporal dependence, we adopt a block multiplier bootstrap procedure using non-overlapping blocks. Define the vector $\widehat{\mathcal T}=(\widehat{\mathcal T})_{k\in\mathcal S}$, where 
	\begin{equation*}
	\widehat{\mathcal T}_k=-\frac{1}{\sqrt{n}}\sum_{i=1}^{l_n}e_{i}\sum_{l=(i-1)b_n+1}^{ib_n}\hat\zeta_kg(D_l,\hat\theta),\quad k\in\mathcal S,
	\end{equation*}
	$\hat\zeta_k$ is the $k$-th row of the matrix $(\hat A\hat G_1)^{-1}\hat A$ and $e_i$ are independently drawn from $\N(0, 1)$. Here, $l_n$ and $b_n$ denote the numbers of blocks and block size, respectively, with $b_n=\lfloor n/l_n\rfloor$. To ensure the validity of the multiplier bootstrap, as shown in the following theorem, we assume that the block size grows at a polynomial rate such that $b_n=\bigO(n^\eta)$ for some $0<\eta<1$. Intuitively, a larger block size is needed to effectively capture the dependency structure, while sufficiently many blocks are required for robust approximation of the bootstrapped statistics. To address this trade-off, a set of accompanying conditions is imposed, further narrowing the admissible range of $\eta$ to determine the optimal $b_n$ in \eqref{bn}, as detailed in the proof. This range is influenced by its interplay with $r$ and $\varsigma$ (satisfying \hyperref[A_dan]{(A2)}) and the size of $|\mathcal S|$.
	
	\begin{theorem}[Multiplier Bootstrap]\label{theorem.inference}
		Let $c^\ast_{\alpha}$ denote the $(1-\alpha)$ conditional quantile of $\max_{k\in S}|\widehat{\mathcal T}_k|$. Under \hyperref[A_Srate]{(A7)} and the same conditions as in Theorem \ref{linear}, assuming that $|(AG_1)^{-1}A|_\infty<\infty$, $\sqrt{n}|r_n|_\infty =\smallO_\P(1)$, and %$\Phi^\mG_{r,\varsigma}<\infty$ with $r>4$, 
		$b_n = \bigO(n^{\eta})$ for some $0 <\eta< 1$ (the specific rate is provided in \eqref{bn} in the detailed proof), 
		we have: 
		\begin{equation*}
		%\lim_{n\to\infty}\big|\P\big(\check\theta_{1,k}-n^{-1/2}c^\ast_{\alpha}\hat\sigma_k\leq\theta^0_{1,k}\leq\check\theta_{1,k}+n^{-1/2}c^\ast_{\alpha}\hat\sigma_k,\,\forall k\in\mathcal S\big) - (1-\alpha)\big|=0.
		\lim_{n\to\infty}\big|\P(\sqrt{n}|\check\theta_{1,k}-\theta_{1,k}^0|\leq c^\ast_{\alpha}\hat\sigma_k,\forall k\in \mathcal S) - (1-\alpha)\big|=0.
		\end{equation*}
	\end{theorem}
	As a result of Theorem \ref{theorem.inference}, for each $k\in\mathcal S$, we can construct the two-sided $(1-\alpha)$ bootstrap confidence interval as: 
	\begin{equation}\label{CI_boot}
	[\check\theta_{1,k}-\hat\sigma_kn^{-1/2}c^\ast_\alpha, \, \check\theta_{1,k}+\hat\sigma_kn^{-1/2}c^\ast_\alpha].
	\end{equation}

	\section{Simulation Study}\label{sim}
	In this section, we illustrate the finite sample properties of our proposed methodology across different simulation scenarios. %Section \ref{sing} concerns the results in a single equation setting, and Section \ref{multi} addresses some multiple equation cases.
 %Section \ref{spn_sim} 
 We first present results for the primary example of spatial panel networks discussed in Section \ref{system}, while Appendix \ref{ab} focuses on dynamic linear panel models.
 %\subsection{Spatial Panel Networks}\label{spn_sim}
 
%Recall 
Consider the spatial panel network model %defined in \eqref{spn2} 
with covariates:
	\begin{equation*}%\label{net_YX}
	y_{j,t}=\rho^0 h_j^{0\top} y_{t} + \beta^{0\top} u_{j,t} + \vps_{j,t},\quad j=1,\ldots,p,t=1,\ldots,n,\, u_{j,t}\in\R^{d},
	\end{equation*}
	where $h^0_j=(h^0_{j1},\ldots,h^0_{jp})^\top$ and $h^0_{jk}$ ($k\neq j$) represents the actual, unobserved peer effect of unit $k$ on unit $j$. Our objective is to estimate the joint network effect $\rho^0$, recognizing that $h_j^0$ may be misspecified as an observed network structure $w_j=(w_{j1},\ldots,w_{jp})^\top$ for all $j=1,\ldots,p$. The model can then be rewritten as:
 $$y_{j,t}=\rho^0 w_j^\top y_{t} + \rho^0 \delta_j^{0\top}y_t + \beta^{0\top} u_{j,t}+\vps_{j,t},$$
 where the vectors $\delta^0_j=(h_j^{0\top} - w_j^\top)$, $j=1,\ldots,p$, capture the misspecification errors of the network structure.

 We randomly generate the actual links by using independent Bernoulli random variables, each with a probability of 0.5 of equaling one. Additionally, we set $h^0_{jj}=0$ and apply normalization to $h^0_j$ for each $j=1,\ldots,p$. We assume that misspecification occurs randomly with a probability of 0.2 when an actual link is non-zero; that is, $h^0_{jk}\neq0$ but $w_{jk}=0$.

 To incorporate the dependency, we let the instrumental variables $Z_{j,t}\in\R^{q_j}$ for $j=1,\ldots,p$, follow a linear process such that $Z_{j,t}=\sum_{\ell=0}^{\infty}A^j_\ell\xi_{j,t-\ell}$, where $A^j_\ell=(\ell+1)^{-\tau-1}M^j_\ell$, and $M^j_\ell$ are independently drawn from Ginibre matrices (i.e., all entries of $M^j_\ell$ are i.i.d. $\operatorname{N}(0,1)$). In practice, the sum is truncated to $\sum_{\ell=0}^{500}A^j_\ell\xi_{j,t-\ell}$. We set $\tau=1.0$ for weaker dependence and $\tau=0.1$ for stronger dependence. For the $q_j$-dimensional vector $\xi_{j,t}$, we define each element as $\xi_{jk,t}=e_{jk,t}(0.8e_{jk,t-1}^2+0.2)^{1/2}$ for $k=1,\ldots,q_j$, where $e_{jk,t}$ are i.i.d. and follow a scaled $t(8)$-distribution: $t(8)/\sqrt{8/(8-2)}$, with $t(8)$ being the Student's $t$-distribution with 8 degrees of freedom. 

 Then, for each $j=1,\ldots,p$, we generate the $d$-dimensional covariates $u_{j,t}$ as follows:
 $$u_{j,t}=\pi^\top Z_{j,t} + v_{j,t},$$
	where the $q_j\times d$ matrix $\pi$ is defined as $\pi=(3+3\kappa^{q_j/3})^{-1/3}(\iota_3\otimes \mathbf{I}_{q_j/3})$, with $d=q_j/3$, $\iota_3$ being a $3\times1$ vector of ones, and $\kappa=0.5$. We let $\beta^0=(10,10,10,10,10,5,5,5,1,1,0_{d-10}^\top)^\top$. The errors $\vps_{j,t}$ and $v_{j,t}$ are generated independently from standard normal distribution. 

 We consider two cases: $p=30,d=30,n=100$ and $p=50,d=50,n=200$, where the total number of parameters, $p^2+d+1$, amounts to 931 and 2,551, respectively. The total number of moment conditions, $q=\sum_{j=1}^pq_j$, is 2,700 for the first case and 7,500 for the second. Specifically, we focus on $\rho^0$, $\beta^0$, and $\tilde\delta^0$, which includes the first 50 elements of the stacked vector $[\tilde\delta^0_j]_{j=1}^p$ as parameters of interest. Here, $\tilde\delta^0_j$ is defined as a subvector of $\delta^0_j$ with elements known to be zero removed. These removed zero elements correspond to non-zero $w_{jk}$ values, which are assumed to be correctly specified in this setting. 

 To assess the estimation accuracy of our proposed two-step method, we compute the absolute deviation for estimating $\rho^0$ and the estimation error for the vectors $\beta^0$ and $\tilde\delta^0$, measured by the Euclidean norm. These calculations are performed on estimators both with and without applying the debiasing step, namely, the DRGMM and GDS estimators. In the first step, we use penalty that is independent of the design matrix. Specifically, we set $\lambda_n=\Phi^{-1}(1-0.1/(2q))\max\limits_{1\leq j\leq q}\hat\sigma_j^2/\sqrt{n}$, given that $\hat g_j(\theta)$ asymptotically follows $\operatorname{N}(0,\sigma^2_j/n)$ for $j=1,\ldots,q$. This choice is intentionally conservative to mitigate the risk of overfitting, following \citet{belloni2018high}. When the debiasing step is applied, we treat $\rho^0$, $\beta^0$, and $\tilde\delta^0$ as the parameters of interest respectively. For the convenience of comparison, we present the estimation errors as ratios, which measure the relative difference between the results obtained using the DRGMM and GDS estimator. In particular, a ratio smaller than 1 indicates better performance when the debiasing step is applied. The results, summarized in Tables \ref{simest1} and \ref{simest2}, are aggregated over 500 replications using both the mean and the median.

    \begin{table}[H]
		\begin{center}
        \renewcommand{\arraystretch}{0.5}
			\begin{tabular}{p{1.8cm} ccccc}
				\hline\hline
				& \multicolumn{2}{c}{\small{$p=d=30,n=100$}} && \multicolumn{2}{c}{\small{$p=d=50,n=200$}}\\
				\cline{2-3}\cline{5-6}
				& \footnotesize{$\rho^0=0.7$} &  \footnotesize{$\rho^0=0.5$} && \footnotesize{$\rho^0=0.7$} &  \footnotesize{$\rho^0=0.5$}  \\
				%\cline{2-6}
                    \hline
				& \multicolumn{5}{c}{\small{Absolute deviation for $\rho^0$}}\\
				%\hline
				\footnotesize{Mean} &  0.2015  & 0.3253  &&   0.2466  & 0.3051 \\
				\footnotesize{Median} &  0.1985 &  0.2430 &&   0.2454 &  0.3064 \\
                    % \hline
                     & \multicolumn{5}{c}{\small{Euclidean norm for $\beta^0$}}\\
                    % \hline
				\footnotesize{Mean} &  0.2407 & 0.2725 &&   0.2751 & 0.3192 \\
    			\footnotesize{Median} &  0.2383 & 0.2709 &&   0.2757 & 0.3196 \\
                    %\hline
                    & \multicolumn{5}{c}{\small{Euclidean norm for $\tilde\delta^0$}}\\
                     %\hline
       			\footnotesize{Mean} &  0.4371 & 0.5034 &&  0.3690 & 0.4677 \\
    			\footnotesize{Median} &  0.4282 & 0.4863 &&  0.3666 & 0.4576 \\
				\hline\hline
			\end{tabular}
		\end{center}
        \vspace{-0.5cm}
		\linespread{1}\caption{ Estimation errors for $\rho^0$, $\beta^0$ and $\tilde\delta^0$ as ratios (DRGMM relative to GDS) in the weaker dependence case ($\tau=1.0$).}\label{simest1}
	\end{table}

        \begin{table}[H]
		\begin{center}
        \renewcommand{\arraystretch}{0.5}
			\begin{tabular}{p{1.8cm} ccccc}
				\hline\hline
				& \multicolumn{2}{c}{\small{$p=d=30,n=100$}} && \multicolumn{2}{c}{\small{$p=d=50,n=200$}}\\
				\cline{2-3}\cline{5-6}
				& \footnotesize{$\rho^0=0.7$} &  \footnotesize{$\rho^0=0.5$} && \footnotesize{$\rho^0=0.7$} &  \footnotesize{$\rho^0=0.5$}  \\
				%\cline{2-6}
                    \hline
				& \multicolumn{5}{c}{\small{Absolute deviation for $\rho^0$}}\\
				%\hline
				\footnotesize{Mean} &  0.1466  & 0.2587  &&   0.2339  & 0.2809\\
				\footnotesize{Median} &  0.1470 &  0.1633 &&    0.2325 &  0.2822\\
                    % \hline
                     & \multicolumn{5}{c}{\small{Euclidean norm for $\beta^0$}}\\
                    % \hline
				\footnotesize{Mean} &  0.1714 & 0.1766 &&   0.2539 & 0.2871 \\
    			\footnotesize{Median} &  0.1706 & 0.1752 &&   0.2542 & 0.2852 \\
                    %\hline
                    & \multicolumn{5}{c}{\small{Euclidean norm for $\tilde\delta^0$}}\\
                     %\hline
       			\footnotesize{Mean} &  0.3470 & 0.3792 &&   0.4839 & 0.4413 \\
    			\footnotesize{Median} &  0.3392 & 0.3701 &&   0.4808 & 0.4389 \\
				\hline\hline
			\end{tabular}
		\end{center}
            \vspace{-0.5cm}
		\linespread{1}\caption{Estimation errors for $\rho^0$, $\beta^0$ and $\tilde\delta^0$ as ratios (DRGMM relative to GDS) in the stronger dependence case ($\tau=0.1$).}\label{simest2}
	\end{table}

 Additionally, we evaluate the inference performance by examining the empirical power and size of the confidence intervals (with a nominal confidence level of 95\%) constructed using the limiting distribution theory outlined in Section \ref{check}. Specifically, the average rejection rate of the null hypotheses for the truly zero components reflects size performance, while the testing power is evaluated for the truly non-zero components. Inference results are reported separately for the structural parameters $(\rho^0,\beta^0)$ and for the network structure $\tilde\delta^0$. For comparison, the average false positive rate for truly zero parameters and the average true positive rate for truly non-zero parameters under the GDS estimator are also reported to assess the necessity of uniform inference via debiasing. The results, based on 500 simulations, are presented in Tables \ref{siminf1} and \ref{siminf2}.

     \begin{table}[H]
		\begin{center}
        \renewcommand{\arraystretch}{0.5}
			\begin{tabular}{p{1.8cm} ccccc}
				\hline\hline
				& \multicolumn{2}{c}{\small{$p=d=30,n=100$}} && \multicolumn{2}{c}{\small{$p=d=50,n=200$}}\\
				\cline{2-3}\cline{5-6}
				& \footnotesize{$\rho^0=0.7$} &  \footnotesize{$\rho^0=0.5$} && \footnotesize{$\rho^0=0.7$} &  \footnotesize{$\rho^0=0.5$}  \\
				%\cline{2-6}
                    \hline
				& \multicolumn{5}{c}{\small{Size/False positive rate for $(\rho^0,\beta^0)$}}\\
				%\hline
				\footnotesize{DRGMM} &  0.04  & 0.06  &&   0.03  & 0.04 \\
				\footnotesize{GDS} &  0.00 &  0.00 &&   0.00 &  0.00 \\
                    %\hline
                    & \multicolumn{5}{c}{\small{Power/True positive rate for $(\rho^0,\beta^0)$}}\\
                    \footnotesize{DRGMM} &  1.00  & 1.00  &&   1.00  & 1.00 \\
    			\footnotesize{GDS} &  0.83 &  0.83 &&   1.00 &  1.00 \\
                    %\hline
                     & \multicolumn{5}{c}{\small{Size/False positive rate for $\tilde\delta^0$}}\\
				\footnotesize{DRGMM} &  0.04  & 0.07  &&   0.03  & 0.06 \\
				\footnotesize{GDS} &  0.30 &  0.30 &&   0.29 &  0.15\\
                    %\hline
                    & \multicolumn{5}{c}{\small{Power/True positive rate for $\tilde\delta^0$}}\\
                    \footnotesize{DRGMM} &  0.94  & 0.70  &&   1.00  & 0.96 \\
    			\footnotesize{GDS} &  0.33 &  0.31 &&   0.44 &  0.28 \\
				\hline\hline
			\end{tabular}
		\end{center}
                \vspace{-0.5cm}
		\linespread{1}\caption{Empirical size and power for testing the structural parameters $(\rho^0,\beta^0)$ and for the network structure $\tilde\delta^0$ under the DRGMM estimator, along with the average false positive and true positive rates for truly zero and non-zero parameters under the GDS estimators, in the weaker dependence case ($\tau=1.0$).}\label{siminf1}
	\end{table}

    \begin{table}[H]
		\begin{center}
        \renewcommand{\arraystretch}{0.5}
			\begin{tabular}{p{1.8cm} ccccc}
				\hline\hline
				& \multicolumn{2}{c}{\small{$p=d=30,n=100$}} && \multicolumn{2}{c}{\small{$p=d=50,n=200$}}\\
				\cline{2-3}\cline{5-6}
				& \footnotesize{$\rho^0=0.7$} &  \footnotesize{$\rho^0=0.5$} && \footnotesize{$\rho^0=0.7$} &  \footnotesize{$\rho^0=0.5$}  \\
				%\cline{2-6}
                    \hline
				& \multicolumn{5}{c}{\small{Size/False positive rate for $(\rho^0,\beta^0)$}}\\
				%\hline
				\footnotesize{DRGMM} &  0.02  & 0.02  &&   0.02  & 0.04 \\
				\footnotesize{GDS} &  0.00 &  0.00 &&   0.00 &  0.00\\
                    %\hline
                    & \multicolumn{5}{c}{\small{Power/True positive rate for $(\rho^0,\beta^0)$}}\\
                    \footnotesize{DRGMM} &  1.00  & 1.00  &&   1.00  & 1.00 \\
    			\footnotesize{GDS} &  0.83 &  0.83 &&   1.00 &  1.00 \\
                    %\hline
                     & \multicolumn{5}{c}{\small{Size/False positive rate for $\tilde\delta^0$}}\\
				\footnotesize{DRGMM} &  0.02  & 0.03 &&   0.05  & 0.04 \\
				\footnotesize{GDS} &  0.34 &  0.33 &&   0.20 &  0.15 \\
                    %\hline
                    & \multicolumn{5}{c}{\small{Power/True positive rate for $\tilde\delta^0$}}\\
                    \footnotesize{DRGMM} &  0.93  & 0.82  &&   1.00  & 1.00 \\
    			\footnotesize{GDS} &  0.41 &  0.35 &&   0.29 &  0.30 \\
				\hline\hline
			\end{tabular}
		\end{center}
            \vspace{-0.5cm}
		\linespread{1}\caption{Empirical size and power for testing the structural parameters $(\rho^0,\beta^0)$ and for the network structure $\tilde\delta^0$ under the DRGMM estimator, along with the average false positive and true positive rates for truly zero and non-zero parameters under the GDS estimators, in the stronger dependence case ($\tau=0.1$).}\label{siminf2}
	\end{table}
    
    From Tables \ref{simest1} and \ref{simest2}, it is evident that debiased regularization significantly outperforms the one-step GDS estimator in estimating the structural parameters, particularly when a stronger network effect is observed in $\rho^0$. Our proposed DRGMM estimator performs well in recovering the true network structure, with its superiority becoming more pronounced for larger-scale networks and under stronger dependency. Overall, we observe that the estimation errors are robust across simulations. 

    From Tables \ref{siminf1} and \ref{siminf2}, we find that inference after applying the debiasing step provides size control close to the nominal level and high empirical power in most cases. While the Dantizig selection successfully detects the truly non-zero structural parameters, it is not reliable for recovering the latent network structure. Notably, our proposed method effectively avoids an excess of false positives, which can occur with the one-step regularized selection. These results confirm the necessity of uniform inference on parameters, including both truly zero and non-zero elements in the network structure.

	\section{Empirical Analysis: Spatial Network of Stock Returns}\label{app}
	In this section, our proposed methodology is employed to study the spatial network effect of stock returns. We use the public cross-ownership information as the pre-specified social network structure \citep{zhu2019network}; however, there might be misspecification in the network since some of the cross-shareholder information is not published. Our purpose is to analyze the network effect and simultaneously recover the unobserved linkages.
	
	%\subsection{Data and Model Setting}
	Our empirical illustration is carried out on a dataset consisting of 100 individual stocks traded on the Chinese A-share market (Shanghai Stock Exchange and Shenzhen Stock Exchange), spanning 14 sectors as defined by the Industry Classification Guidelines of the China Securities Regulatory Commission. The data covers the period from January 2, 2019 to December 31, 2019 (i.e., $244$ trading days). Daily stock returns and annual cross-ownership data were sourced from the \href{https://www.wind.com.cn}{Wind Data Service}.
	
	The spatial network model is specified as follows:
	\begin{align*}
	%r_{j,t} &= \alpha_j + \rho h_j^\top r_t + \gamma X_{j,t} + \vps_{j,t} = \alpha_j + \rho w_j^\top r_t + \rho(h_j^\top - w_j^\top) r_t + \gamma X_{j,t} + \vps_{j,t},\\
	y_{j,t}&=\alpha_j + \rho^0 w_j^\top y_{t} + \rho^0 (h_j^{0\top} - w_j^\top)y_t + \beta^{0} u_{j,t}+\vps_{j,t},
	\end{align*}
	where $j=1,\ldots,p$ indexes the individual stocks, $y_t=(y_{1,t},\ldots,y_{p,t})^\top$ represents the daily log returns, and $u_{j,t}$ denotes the daily turnover ratio (trading volume divided by shares outstanding), which is used as a firm-specific control variable. We assume that $\vps_{j,t}$ satisfies assumption \hyperref[A_error]{(A3)}, with $\E(\vps_{j,t}\vps_{j',t}|\mathcal F_{t-1})=0$ for $j\neq j'$.
	%$X_t$ contains the control variables, including the log returns of Wind entire-A index and Chicago Board Options Exchange China ETF volatility index.
	An unobserved individual effect, $\alpha_j$, is included to account for potential serial correlation in the error term. Within our estimation framework, these fixed effects can be treated as equation-specific intercepts during estimation.
	
	The term $w_{jk}$ represents the public cross-ownership between stock $k$ and $j$, defined as $w_{jk}=1$ if company $j$ holds shares in company $k$ based on available information, and $w_{jk}=0$ otherwise. The resulting network structure for $w_{jk}$ ($j,k=1,\ldots,p$) is illustrated in Figure \ref{network_W}, where the stocks are grouped by sector. Notably, cross-ownership relationships are observed across sectors.
	
	\begin{figure}[H]\begin{centering}
			\vspace{-0.5cm}
			\includegraphics[scale=0.8]{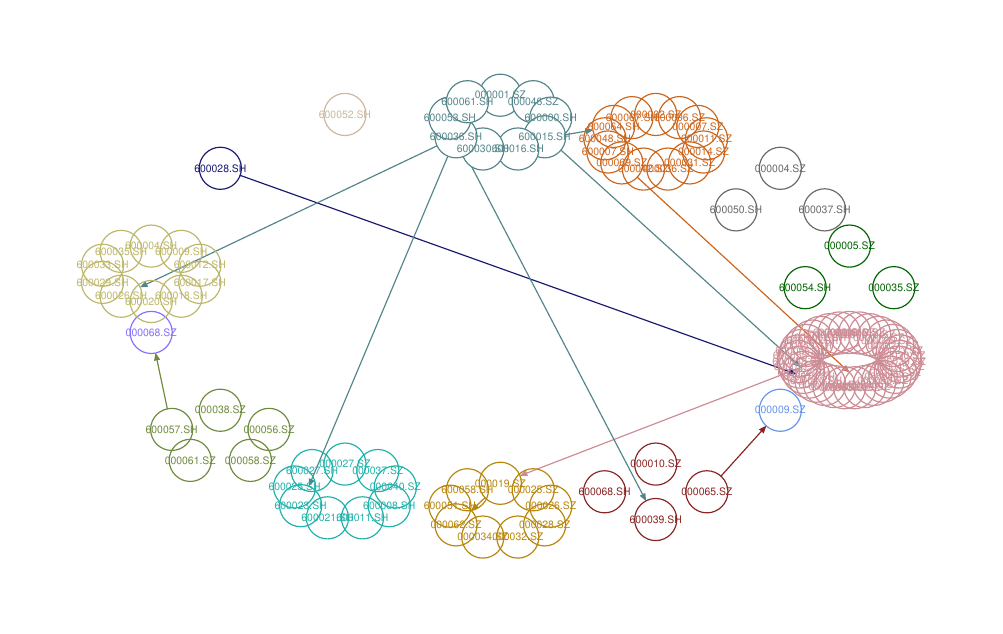}
			\vspace{-1cm}
			\linespread{1}\caption{Visualization of the network structure defined by the known $w_{jk}$ ($j,k=1,\ldots,p$) observed in 2019. Nodes represent companies, and directed edges indicate cross-ownership links. Companies are grouped and color-coded based on their sector classification.}
			\label{network_W}\end{centering}
	\end{figure}
	
	It is possible for $w_{jk}=0$ while $h^0_{jk}\neq0$ if some shareholders of company $j$ are not publicly disclosed. %By convention, 
	We set $w_{jj}=h^0_{jj}=0$. Without loss of generality, we assume that misspecification errors occur only when the actual link is non-zero; specifically, cases where $w_{jk}\neq0$ while $h^0_{jk}=0$ are excluded. Our goal is to estimate the network effect $\rho^0$ and the misspecification errors $\delta^0_{jk}=h^0_{jk}-w_{jk}$ %by GDS 
	using our proposed approach.
	%, where the lags $r_{t-1},r_{t-2}$ are chosen as the instruments variables.
	Ultimately, we aim to recover the latent linkages $h^0_{jk}$ based on inference results for the deviation $\delta^0_{jk}$, particularly in cases where $w_{jk}$ is observed to be zero.
	
	In particular, the two-step DRGMM estimation procedure described in Section \ref{est} is applied, with $y_{t-1},y_{t-2}$ chosen as the instrumental variables. The resulting debiased estimators are $\check\rho=0.2214$ and $\check\beta =0.0012$, with standard errors of 0.0061 and 0.0001. Both $\rho^0$ and $\beta^0$ are found to be statistically significant. 
	% To further justify the prediction performance of our proposed method, we define the prediction error by the root-mean-square deviation: $\sqrt{\sum_{j=1}^p\sum_{t=1}^n(\hat y_{j,t}-y_{j,t})^2}$, where the predicted returns are given by $\hat y_{j,t}=\hat\rho w_j^\top y_t + \hat\rho\hat\delta_j^\top y_t + \hat\beta u_{j,t}$. In particular, we compare the prediction errors with two alternatives: spatial autoregressive (SAR) model solely based on the observed network structure $w$, and one-step GDS without debiasing. We consider the same moment conditions (i.e., same IVs) for all these competitors.
	For comparison, we also fit the spatial autoregressive (SAR) model  based solely on the observed network structure, using the same moment conditions (i.e., same instruments) for estimation. The estimated structural parameters $(\rho^0,\beta^0)$ are 0.3188 and 0.0014, with standard errors of 0.2181 and 0.0001, respectively. Notably, our proposed approach identifies a significantly stronger network effect. 
	%	\begin{table}[H]
	%		\begin{center}
	%			\begin{tabular}{p{2.4cm} ccc}
	%				\hline\hline
	%				& \footnotesize{SAR with $w$} &  \footnotesize{RMD} & \footnotesize{DRGMM}  \\
	%				\hline
	%				\footnotesize{prediction error} & 0.02339  & 0.02271 & 0.02268 \\
	%				\hline\hline
	%			\end{tabular}
	%		\end{center}
	%		\caption{The prediction error measured by RMSE under different approaches.}\label{pred}
	%	\end{table}
	% We find that taking the possible misspecification into account and estimating the error via regularization would improve the cross-sectional prediction accuracy in the spatial network of stock returns by 3\%. The two regularized approaches give comparable prediction performance.
	
	Furthermore, it is of interest to test the latent network structure, and the inference theory based on DRGMM provides a formal framework for doing so. Following the discussion in Section \ref{check}, we perform individual hypothesis tests on $H_0^{(j,k)}:\delta_{jk}^0=0$ if the preliminary estimator in step 1, $\hat\delta_{jk}$, is found to be non-zero and $w_{jk}$ is observed to be zero. A total of 60 edges are considered, and debiasing is applied to the entire vector. The recovered network structure after hypothesis testings is shown in Figure \ref{network_H}, where a directed link from $k$ to $j$ indicates that $\delta^0_{jk}$ is significantly non-zero, implying that $h^0_{jk}$ should also be non-zero. 
	
	\begin{figure}[H]\begin{centering}
			\vspace{-0.5cm}
			\includegraphics[scale=0.8]{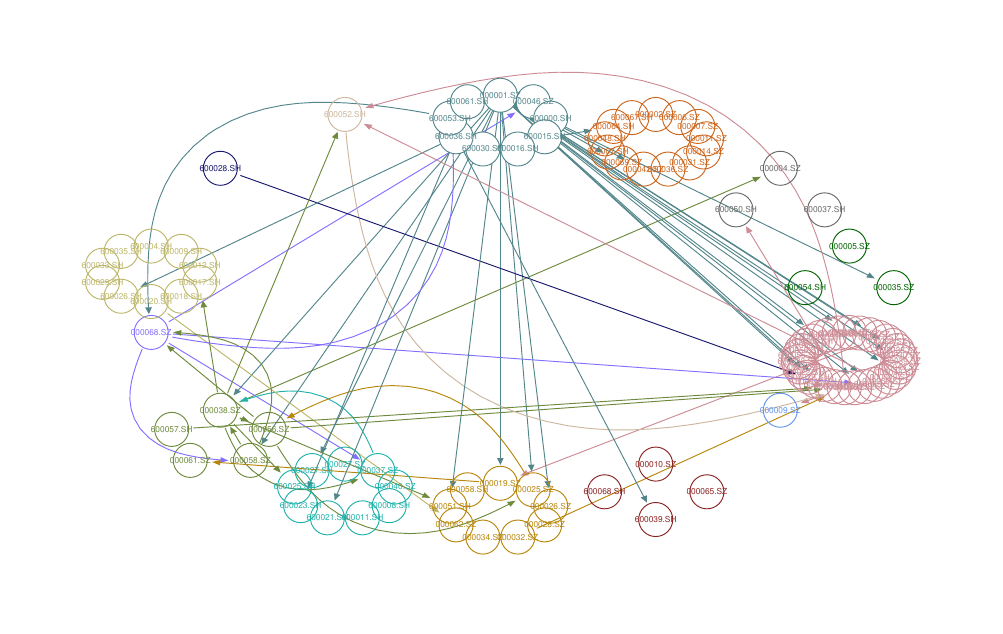}
			\vspace{-1cm}
			\linespread{1}\caption{Visualization of the recovered network structure using DRGMM, followed by individual testing based on the 2019 data sample.}%, where the nodes represent the companies and the edges with direction indicate the links of cross ownership. The companies are clustered with different colors by sector classification.}
			\label{network_H}\end{centering}
	\end{figure}
	
	We find that the recovered network, which accounts for latent link structures, differs substantially from the pre-specified network. Notably, the finance and insurance sector emerges as the one with the highest outbound degree centrality, while the most intensive connections are directed towards the manufacturing sector. At the individual stock level, Ping An Bank Co., Ltd. (000001.SZ) from the finance and insurance sector has the highest outbound degree centrality, with a value of 23, while ZTE Corporation (000063.SZ) from the manufacturing sector exhibits the highest inbound degree centrality, with a value of 5. These results highlight the importance of addressing misspecified network links when analyzing risk channels and financial stability within a financial system.
	
	\section*{Acknowledgments}
	We thank Lung-Fei Lee for prompting the impetus to explore this topic. We are also grateful to Tim Christensen, Aureo de Paula, Wolfgang H\"ardle, Elena Manresa, Gerard van den Berg, and Jeffrey Wooldridge for helpful discussions. In addition, we thank the Editor, one Associate Editor, and two referees for their valuable comments, which have significantly improved the paper. We remain responsible for any errors or omissions. Chen Huang acknowledges financial support from the Independent Research Fund Denmark through the Inge Lehmann Grant (1132-00019B). Weining Wang is supported through the project ``IDA Institute of Digital Assets'', CF166/15.11.2022, financed under the Romania's National Recovery and Resilience Plan; and the Marie Sk\l{}odowska-Curie Actions under the European Union's Horizon Europe research and innovation program for the Industrial Doctoral Network on Digital Finance, Project No. 101119635. Lastly, we thank GPT-4 for proofreading assistance; all content was reviewed and edited by the authors, who take full responsibility for the final version of the manuscript. 
	
	\bibliography{biball}
	\bibliographystyle{apalike}

	\newpage
	\vskip 2em \centerline{\Large \bf Appendix} \vskip -1em
	\setcounter{subsection}{0}
	\vskip 2em
	
	%\section{Appendix}
	\begin{appendices}
		%\tableofcontents
%		\startcontents[appendices]
%		\printcontents[appendices]{l}{1}%{\section*{Appendices}\setcounter{tocdepth}{2}}
		
		\renewcommand{\thesubsection}{A.\arabic{subsection}}
		\setcounter{equation}{0}
		\renewcommand{\theequation}{A.\arabic{equation}}
		\setcounter{theorem}{0}
		\renewcommand{\thetheorem}{A.\arabic{theorem}}
		\setcounter{lemma}{0}
		\renewcommand{\thelemma}{A.\arabic{lemma}}
		\setcounter{figure}{0}
		\renewcommand{\thefigure}{A.\arabic{figure}}
		\setcounter{table}{0}
		\renewcommand{\thetable}{A.\arabic{table}}
		\setcounter{remark}{0}
		\renewcommand{\theremark}{A.\arabic{remark}}
		\setcounter{corollary}{0}
		\renewcommand{\thecorollary}{A.\arabic{corollary}}
		\setcounter{assumption}{0}
		\renewcommand{\theassumption}{A.\arabic{assumption}}
		\setcounter{example}{0}
		\renewcommand{\theexample}{A.\arabic{example}}
		
		\section{Detailed Proofs}
		\subsection{Some Useful Lemmas}
		
		%\begin{lemma}[Weyls' inequality for Hermitian matrix]
		%	We let $H$ be the exact matrix and $P$ be a perturbation matrix that represents the uncertainty. Consider the matrix $M = H + P$. If any two of $M$, $H$ and $P$ are $n$ by $n$ Hermitian matrices, where $M$ has eigenvalues
		%	$$\mu_1 \geq \cdots \geq \mu_n,$$
		%	and $H$ has eigenvalues
		%	$$\nu_1 \geq \cdots \geq \nu_n,$$
		%	and $P$ has eigenvalues
		%	$$\rho_1 \geq \cdots \geq \rho_n.$$
		%	Then the following inequalities hold for $i= 1,\ldots ,n,$
		%	$$\nu_i + \rho_n \le \mu_i \le \nu_i + \rho_1.$$
		%	More generally, if  $j+k-n\geq i \geq r+s-1$, we have
		%	$$\nu_j + \rho_k \leq \mu_i \leq \nu_r + \rho_s.$$
		%\end{lemma}
		
		\begin{lemma}[Weyls' inequality for singular values]  \label{svd}
			Let $H$ ($m\times n$) be the exact matrix and $P$($m\times n$) be a perturbation matrix that represents the uncertainty.  Consider the matrix $M = H + P$. If any two of $M$, $H$ and $P$ are $m$ by $n$ real matrices, where $M$ has singular values
			$$\mu_1 \geq \cdots \geq \mu_{\min{(m, n)}},$$
			$H$ has singular values
			$$\nu_1 \geq \cdots \geq \nu_{\min{(m, n)}},$$		
			and $P$ has singular values
			$$\rho_1 \geq \cdots \geq \rho_{\min{(m, n)}}.$$
			Then the following inequalities hold for $k =1,\ldots,\min(m,n)$, %$ 1\leq k\leq \min(m,n)$,
			$$ \max_{0 \leq i \leq \min(m,n)-k }\{\nu_{k+i}- \rho_{i+1}, -\nu_{i+1}+ \rho_{i+k}, 0 \}\leq \mu_{k} \leq \min_{1 \leq i\leq k}(\nu_i+ \rho_{k-i+1}).$$
		\end{lemma}
		\begin{proof}
			% The result is a direct consequence of Theorem 2 in \cite{queiro1995singular} with completion of the $m\times n$ matrix to square matrix by letting the zero entries and the non-zero singular values stay the same.
			The result follows from Theorem 2 in \cite{queiro1995singular}, which applies to square matrices. To extend this to the $m\times n$ case, the matrix is completed into a square matrix by adding zero entries, ensuring that the non-zero singular values remain unchanged.
		\end{proof}
		
		\begin{lemma}[Corollary 3.3 of \cite{lu2000some}] \label{prod}
			Suppose that $B$ and $A$ are $m \times l$ and $l\times n$ matrices respectively, and let $p = \max\{m,n,l\}$ and $q = \min\{m,n,l\}$. Then for each $k = 1, \cdots q$,
			$$\sigma_{k}(BA) \leq \min_{1\leq i\leq k} \sigma_i(B) \sigma_{k+1-i}(A).$$
			If $p<2q$, then for each $k=1,\ldots,2q-p$,
			$$\max_{k+p-q\leq i\leq q}\sigma_i(B)\sigma_{p+k-i} (A)\leq \sigma_{k}(BA).$$
		\end{lemma}
		
		\begin{lemma}[Theorem 6.2 of \cite{ZW15gaussian}, Tail probabilities for high dimensional partial sums]\label{tail}
			For a mean zero $p$-dimensional random variable $X_t\in \R^p$ ($p>1$), let $S_n = \sum^{n}_{t=1}X_t$ and assume that $\||X_\cdot|_{\infty}\|_{q,\varsigma}< \infty, $ where $q>2$ and $\varsigma\geq0$, and $\Phi_{2, \varsigma} = \max_{1\leq j \leq p} \|X_{j,\cdot}\|_{2, \varsigma}<\infty$. i) If $\varsigma> 1/2 -1/q$, then for $x \gtrsim \sqrt{n \log p} \Phi_{2,\varsigma}+ n^{1/q}(\log p)^{3/2}\||X_\cdot|_{\infty}\|_{q,\varsigma}$,
			$$\P(|S_{n}|_{\infty}\geq x )\leq \frac{C_{q,\varsigma}n (\log p)^{q/2} \||X_{\cdot}|_{\infty}\|^q_{q,\varsigma}}{x^q}+ C_{q,\varsigma}\exp\bigg(\frac{-C_{q,\varsigma}x^2}{n\Phi^2_{2,\varsigma}}\bigg).$$
			ii) If $0<\varsigma< 1/2 -1/q$, then for $x \gtrsim \sqrt{n \log p} \Phi_{2,\varsigma}+ n^{1/2-\varsigma}(\log p)^{3/2}\||X_\cdot|_{\infty}\|_{q,\varsigma}$,
			$$\P(|S_{n}|_{\infty}\geq x )\leq \frac{C_{q,\varsigma}n^{q/2-\varsigma q} (\log p)^{q/2} \||X_{\cdot}|_{\infty}\|^q_{q,\varsigma}}{x^q}+ C_{q,\varsigma}\exp\bigg(\frac{-C_{q,\varsigma}x^2}{n\Phi^2_{2,\varsigma}}\bigg).$$
		\end{lemma}
		
		\begin{lemma}[Tail probabilities for high dimensional partial sums with strong tail assumptions] \label{exp}
			For a mean zero $p$-dimensional random variable $X_t\in \R^p$ ($p>1$), let $S_n = \sum^{n}_{t=1}X_t$ and assume that $\Phi_{\psi_\nu,\varsigma} =\max_{1\leq j\leq p}\sup_{q\geq2}q^{-\nu}\|X_{j,\cdot}\|_{q,\varsigma}< \infty$ for some $\nu\geq0$, and let $\gamma = 2/(1+ 2\nu)$. Then for all $x>0$, we have
			$$\P(|S_n|_\infty\geq x) \lesssim p \exp\{-C_{\gamma}x^{\gamma}/(\sqrt{n}\Phi_{\psi_\nu,0})^{\gamma}\},$$
			where $C_\gamma$ is a constant only depends on $\gamma$.
			%	 Let $X_{ij}$ be random variables with mean zero and define $\gamma_{j,\alpha} =  \|X_{.j}\|_{\psi_r,\alpha}$. Assume that $\|X_{.j}\|_{\psi_r,\alpha}< \infty$.  Then, we have
			%	\begin{equation*}
			%	\P(\max_j \sum^n_{i=1} X_{ij}\geq u) \leq p \max_j c_{\alpha} \exp (-(u)^{\alpha}/(\sqrt{n}^{\alpha}\gamma_{j,\alpha}^{\alpha} c'_{\alpha})),
			%	\end{equation*}
			%	where $\alpha =2$ corresponds for the sub exponential case, and $\alpha = 1$ corresponds to the sub Gaussian case.
		\end{lemma}
		Lemma \ref{exp} follows from Theorem 3 of \citet{wu2016performance} and applying the Bonferroni inequality. In particular, $\nu =1$ corresponds to the sub-exponential case, and $\nu = 1/2$ corresponds to the sub-Gaussian case.
		
		\begin{lemma}[Freedman's inequality]\label{free}
			Let $\{\xi_{a,t}\}_{t=1}^n$ be a martingale difference sequence with respect to the filtration $\{\mathcal F_t\}_{t=1}^n$. Let $V_a = \sum^n_{t=1}\E(\xi^2_{a,t}| \mathcal{F}_{t-1})$ and $M_a = \sum^n_{t=1} \xi_{a,t}$. Then, for $x,u,v>0$, we have
			\begin{equation*}
			\P(\max_{a \in \mathcal{A}} |M_a|\geq x) \leq \sum^{n}_{t=1}\P(\max_{a\in \mathcal{A}} \xi_{a,t}\geq u)+ 2 \P(\max_{a \in \mathcal{A} } V_a \geq v )+ 2|\mathcal{A}| e^{-x^2/(2xu+ 2v)},
			\end{equation*}
			where $\mathcal A$ is an index set with $|\mathcal A|<\infty$.
		\end{lemma}
		Lemma \ref{free} is a maximal form of Freedman's inequality \citep{freedman1975tail}.
		
		\begin{lemma}[Maximal inequality based on Freedman's inequality]\label{max}
			Let $\{\xi_{a,t}\}_{t=1}^n$ be a martingale difference sequence with respect to the filtration $\{\mathcal F_t\}_{t=1}^n$, where $a\in\mathcal A$, $\mathcal A$ is an index set with $|\mathcal A|<\infty$. Suppose there exists $a^\ast\in\mathcal A$ such that $\max\limits_{a\in\mathcal A}|\sum_{t=1}^n\xi_{a,t}|\leq\sum_{t=1}^n|\xi_{a^\ast,t}|$ and $\max\limits_{1\leq t\leq n}|\xi_{a^\ast,t}|\leq F$, with $\|F\|_2$ is bounded. Let $V_a = \sum^n_{t=1}\E(\xi^2_{a,t}| \mathcal{F}_{t-1})$ and $M_a = \sum^n_{t=1} \xi_{a,t}$. Define the event $\mathcal{G} \defeq \Big\{\max\limits_{a\in \mathcal{A},1\leq t\leq n} \xi_{a,t}\leq A,\max\limits_{a\in\mathcal A}V_a\leq B\Big\}$, where $A,B$ are constants. Given $\sqrt{n}\P(\mathcal{G}^c) \lesssim A \log(1+|\mathcal{A}|)+\sqrt{B}\sqrt{\log(1+|\mathcal{A}|)}$, we have
			\begin{equation*}
			\E\Big[\max_{a \in \mathcal{A}} |M_a|\Big]\lesssim A \log(1+|\mathcal{A}|)+\sqrt{B}\sqrt{\log(1+|\mathcal{A}|)}.
			\end{equation*}
		\end{lemma}
		\begin{proof}
			Observe that
			\begin{equation*}
			\E\Big[\max_{a \in \mathcal{A}} |M_a|\Big] =	\E\Big[\max_{a \in \mathcal{A}} |M_a|\IF(\mathcal{G})\Big]+ \E\Big[\max_{a \in \mathcal{A}} |M_a|\IF(\mathcal{G}^c)\Big].
			\end{equation*}
			The bound of the first part follows from a trivial modification of Lemma 19.33 in \cite{van2000asymptotic} based on Lemma \ref{free}. The second part is bounded using the Cauchy-Schwarz inequality and Burkholder inequality \citep{burkholder1988sharp}:
			$$\E\Big[\max_{a \in \mathcal{A}} |M_a|\IF(\mathcal{G}^c)\Big] \leq \sqrt{n}\{\E(F^2)\}^{1/2}\{\P(\mathcal{G}^c)\}^{1/2}.$$
			Then the result follows from the assumption $\sqrt{n}\P(\mathcal{G}^c) \lesssim A \log(1+|\mathcal{A}|)+\sqrt{B}\sqrt{\log(1+|\mathcal{A}|)}$.
			
		\end{proof}
		
		\begin{lemma}\label{smallbound}
			Consider a $p\times p$ positive semi-definite random matrix $H_1$ and a $p\times p$ deterministic positive definite matrix $H_2$. Assume that $|H_1-H_2|_2 = \bigO_\P(c_n)$, where $c_n \to 0$ as $n\to\infty$. Then, we have
			$$\lambda_{\min}(H_1)= \lambda_{\min}(H_2)-\bigO_\P(c_n).$$
		\end{lemma}	
		\begin{proof}
			The results are implied by
			\begin{align*}
			\lambda_{\min}(H_1) = \min_{v\in \R^p, |v|_2=1} v^{\top}H_1v  &\geq \min_{v\in \R^p, |v|_2=1}v^{\top}H_2v - \max_{v\in \R^p, |v|_2=1}v^\top(H_1-H_2)v\\
			&= \min_{v\in \R^p, |v|_2=1}v^{\top}H_2v -|H_1-H_2|_2 \\
			&\geq \lambda_{\min}(H_2)-\bigO_\P(c_n).
			\end{align*}
		\end{proof}
		
		\subsection{Proofs of Section \ref{const}}\label{a3.1}
		\subsubsection{Concentration}
		\begin{proof}[\textbf{Proof of Lemma \ref{cont}}]       
			% We define $b_n=cn^{-1/2}(\log P_n)^{1/2}\Phi_{2,\varsigma}^{xz} + cn^{-1}c_{n,\varsigma}(\log P_n)^{3/2}\|\max_{j,m}|\tilde x_{j,\cdot}z_{jm,\cdot}|_\infty\|_{r,\varsigma}$, where $P_n =(q\vee n\vee e)$,  %$P=\sum_{j=1}^pq_jK_j$,
			% $c_{n,\varsigma} = n^{1/r}$ for $\varsigma> 1/2 - 1/r$ and $c_{n,\varsigma} = n^{1/2-\varsigma}$ for $0<\varsigma<1/2 - 1/r$. 
			
			By applying Lemma \ref{tail}, we first obtain that 
			$${|\hat G-G|_{\max}}\lesssim cn^{-1/2}(\log P_n)^{1/2}\Phi_{2,\varsigma}^{xz} + cn^{-1}c_{n,\varsigma}(\log P_n)^{3/2}\Big\|\max_{1\leq j\leq p,1\leq m\leq q_j}|x_{j,\cdot}z_{jm,\cdot}|_\infty\Big\|_{r,\varsigma}$$
			holds with probability $1-\smallO(1)$ and a sufficiently large constant $c$, where $c_{n,\varsigma} = n^{1/r}$ for $\varsigma> 1/2 - 1/r$ and $c_{n,\varsigma} = n^{1/2-\varsigma}$ for $0<\varsigma<1/2 - 1/r$. %Here, $\hat G^1$ represents the sample estimator of $G$ without thresholding. The same bound for {$|\hat G-G|_{\max}$} follows, provided that $|G|_{\max}$ is a constant. 
			Similarly, 
			%we define $b'_n=cn^{-1/2}(\log P_n)^{1/2}\Phi_{2,\varsigma}^{yz} + cn^{-1}c_{n,\varsigma}(\log P_n)^{3/2}\|\max_{j,m}|y_{j,\cdot}z_{jm,\cdot}|_{\infty}\|_{r,\varsigma}$. 
			we have 
			$$|\hat g(0)-g(0)|_\infty\lesssim cn^{-1/2}(\log P_n)^{1/2}\Phi_{2,\varsigma}^{yz} + cn^{-1}c_{n,\varsigma}(\log P_n)^{3/2}\Big\|\max_{1\leq j\leq p,1\leq m\leq q_j}|y_{j,\cdot}z_{jm,\cdot}|_{\infty}\Big\|_{r,\varsigma}$$
			holds with probability $1-\smallO(1)$ and a sufficiently large constant $c$. 
			
			Given the sparsity assumption \hyperref[AS]{(A4)}, the asserted concentration inequality follows by combining these results with:
			$$\sup_{\theta \in \mathcal{R}(\theta^0)}|\hat g (\theta)- g(\theta)|_{\infty}\leq |\theta^0|_1|\hat G-G|_{\max}+ |\hat g(0)- g(0)|_{\infty}.$$
			
			Lastly, we analyze the norms $\Phi_{r,\varsigma}^{yz}$ and $\Big\|\max\limits_{1\leq j\leq p,1\leq m\leq q_j}|y_{j,\cdot}z_{jm,\cdot}|_{\infty}\Big\|_{r,\varsigma}$. Recall the model representation $y_{j,t}=\vartheta_j^{0\top}x_{j,t}+\vps_{j,t}$. Define the index sets $\mathcal I_j\defeq\{k\in\{1,\ldots,K_j\}:\vartheta_j^0\neq0\}$ for $j=1,\ldots,p$. Then, we find:
			$$\|y_{j,\cdot} z_{jm,\cdot}\|_{r,\varsigma} \leq \Big\|\max_{k\in\mathcal I_j}x_{jk,\cdot}z_{jm,\cdot}\Big\|_{r,\varsigma} |\vartheta^0_j|_1 + \| \vps_{j,\cdot}z_{jm,\cdot} \|_{r,\varsigma},$$
			which implies
			\begin{align*}
			\Phi_{r,\varsigma}^{yz} &\leq \max_{1\leq j\leq p,1\leq m\leq q_j} \Big\|\max_{k\in\mathcal I_j}x_{jk,\cdot} z_{jm,\cdot}\Big\|_{r,\varsigma}|\vartheta_j^0|_1 +  \Phi_{r,\varsigma}^{\vps z},\\
			\Big\|\max_{1\leq j\leq p,1\leq m\leq q_j}|y_{j,\cdot}z_{jm,\cdot}|_{\infty}\Big\|_{r,\varsigma}&\leq \Big\|\max_{1\leq j\leq p,1\leq m\leq q_j,k\in\mathcal I_j}x_{jk,\cdot}z_{jm,\cdot}\Big\|_{r,\varsigma}\max_{1\leq j\leq p}|\vartheta_j^0|_1 + \Big\|\max_{1\leq j\leq p}|\vps_{j,\cdot}z_{j,\cdot}|_{\infty}\Big\|_{r,\varsigma}.
			\end{align*}
		\end{proof}
		
		\subsubsection{Identification}\label{sec.id}
		In this subsection, we show the necessary conditions for the validity of the identification assumption \hyperref[A_id]{(A6)}. Specifically, we require that the singular values of the sub-matrices of $G$ are bounded.
		
		Let $G_{\mathcal H,\mathcal I}$ denote the sub-matrix of $G$ with rows indexed by the set $\mathcal H\subseteq\{1,\ldots,q\}$ and columns index by the set $\mathcal I\subseteq\{1,\ldots,K\}$, where $|\mathcal I|\leq|\mathcal H|$. Define the $m$-sparse smallest and largest singular values of $G$ (with $m\geq s$) as:
		$$\sigma_{\min}(m,G)=\min\limits_{\mathcal I:|\mathcal I|\leq m} \max\limits_{\mathcal H:|\mathcal H|\leq m} \sigma_{\min} (G_{\mathcal H,\mathcal I}),\,\sigma_{\max}(m,G) = \max\limits_{\mathcal I:|\mathcal I|\leq m} \max\limits_{\mathcal H:|\mathcal H|\leq m} \sigma_{\max}(G_{\mathcal H,\mathcal I}),$$ 
		where $\sigma_{\min} (G_{\mathcal H,\mathcal I})$ and $\sigma_{\max}(G_{\mathcal H,\mathcal I})$ are the smallest and largest singular values of $G_{\mathcal H,\mathcal I}$, respectively. %Recall that the transformed matrix $G$ is a block diagonal matrix whose $j$-th block is given by the $q_j\times K_j$ matrix $G_{[j]} = -\E(z_{j,t}\tilde x_{j,t}^{\top})$. 
		
		\begin{lemma}\label{id}
			% Assume the following conditions:
			% \begin{itemize}
			%     \item[i)] There exist constants $c'>0$ and $C'>0$ such that $\sigma_{\min}(m,G) >c'$ and $\sigma_{\max}(m,G)< C'$ for $ m\leq s(1+u)^2\log n$ with $u>0$.
			%     \item[ii)] The inequality $b_n(1+u)s \lesssim C(u) $ holds for large enough $n$, where $C(u) = \tilde c/(1+u)^2$ with $\tilde c>0$ depending only on $c'$ and $C'$. 
			% \end{itemize}
			% Then, under assumptions \hyperref[A_sta]{(A1)}-\hyperref[A_error]{(A3)}, with probability approaching 1, we have
			% $$\kappa_a^{\hat G}(s,u)\geq s^{-1/a} C(u),\, a\in\{1,2\}.$$
			Recalling the definition of $\kappa_a^{G}(s,u)$ in \hyperref[A_id]{(A6)}, assume there exist constants $c>0$ and $C>0$ such that $\sigma_{\min}(m,G) >c$ and $\sigma_{\max}(m,G)< C$ for $ m\leq s(1+u)^2\log n$ with $u>0$. Then, we have
			$$\kappa_a^{G}(s,u)\geq s^{-1/a} C(u),\, a\in\{1,2\},$$
			where $C(u) = \tilde c/(1+u)^2$ with $\tilde c>0$ depending only on $c$ and $C$. 
		\end{lemma}
		
		% Lemma \ref{id} implies that \hyperref[A2]{(A2)} is satisfied. Combining the results in Lemma \ref{cont} leads to the consistency.
		
		\begin{proof}%[Proof of Lemma \ref{id}]
			% The proof follows from that of Corollary 2 of \cite{belloni2017simultaneous}, with the concentration inequality replaced by applying Lemma \ref{tail} to the matrix $G$.
			
			% {\red By applying the triangle inequality, we obtain
			% $$|\hat G\theta|_{\infty}/|\theta|_a \geq -|(\hat G-G)\theta|_{\infty}/\|\theta\|_a+ |G\theta|_{\infty}/\|\theta\|_a =: -T_{n,1}+ T_{n,2}.$$ 
			%           To handle $T_{n,1}$, we use the concentration inequality from Lemma \ref{tail}.
			% Specifically, we have $|(\hat G-G)\theta|_{\infty}/|\theta|_a \leq |\hat G-G|_{\max}|\theta|_1/|\theta|_a$. Note that for $a=1$, $|\theta|_1/|\theta|_a =1$, and for $a = 2$, $|\theta|_1/|\theta|_a \leq (1+u)s^{1/2}$. From Lemma \ref{cont}, %Section \ref{concent}, 
			%   we have $|\hat G-G|_{\max} \lesssim_{\P} b_n$. Thus, we obtain $-T_{n,1}\lesssim_{\P} b_n (1+u)s^{1-1/a}$.}
			
			% The remainder of the proof follows from Theorem 1 and Corollary 2 of \cite{belloni2017simultaneous}.
			% Provided that $\sigma_{\min}(m,G) >c'$ and $\sigma_{\max}(m,G)< C'$, and given that $b_n(1+u)s \lesssim C(u)$ for sufficiently large $n$, where $C(u) = \tilde c/(1+u)^2$, we conclude that with probability approaching 1, $\kappa^{\hat G}_a(s,u)\geq s^{-1/a} C(u), a\in\{1,2\}$.
			The proof follows directly from Theorem 1 and Corollary 2 of \cite{belloni2017simultaneous} and is therefore omitted.
		\end{proof}
		
		\subsubsection{Proof of Theorem \ref{theorem.cons}}	
		\begin{proof}[\textbf{Proof of Theorem \ref{theorem.cons}}]
			As a consequence of \hyperref[A_tune]{(A5)}, we have that the GDS estimator $\hat\theta$ lies in the restricted set $\mathcal R(\theta^0)$, that is, $|\hat\theta|_1 \leq |\theta^0|_1$, with probability at least $1-\alpha$, provided that a solution $\hat\theta$ to the problem in \eqref{danzig} exists. Define the index sets $\mathcal I\defeq\{k\in\{1,\ldots,K\}:\theta^0_k\neq0\}$ with the cardinality $|\mathcal I|\leq s$, and $\mathcal I^C\defeq\{k\in\{1,\ldots,K\}:\theta^0_k=0\}$. It follows that
			$$|\hat\theta_{\mathcal I^C}|_1\leq|\theta^0_{\mathcal I}|_1-|\hat\theta_{\mathcal I}|_1\leq|\hat\theta_{\mathcal I}-\theta^0_{\mathcal {\mathcal I}}|_1,$$
			implying that $(\hat\theta-\theta^0)\in\mathcal C_{\mathcal I}(1)=\{\theta\in\R^K: |\theta_{\mathcal I^C}|_1\leq |\theta_{\mathcal I}|_1\}$.
			
			Recalling the definition of $\kappa_a^{G}(s,u)$ and its lower bound specified in \hyperref[A_id]{(A6)}, for $a\in\{1,2\}$, we find that: 
			$$s^{-1/a} C(u)\leq \kappa_a^{G}(s,u)\leq|G(\hat\theta-\theta^0)|_\infty/|\hat\theta-\theta^0|_a,$$
			which leads to:
			$$s^{-1/a} C(u)|\hat\theta-\theta^0|_a\leq |g(\hat{\theta})-\hat g(\hat{\theta})|_{\infty} + |\hat g(\hat\theta)|_{\infty}.$$
			Now consider the event:
			$$\{|\hat g(\theta^0)|_{\infty}\leq \lambda_n, \hat\theta\in\mathcal R(\theta^0), |\hat g(\hat\theta) - g(\hat\theta)|_{\infty}\leq \epsilon_n\}.$$
			By the concentration results in Lemma \ref{cont}, assumption \hyperref[A_tune]{(A5)}, and applying the union bound, this event holds with probability at least $1-\alpha-\smallO(1)$. The asserted error bound follows by inserting the rate of $\epsilon_n$ as provided in Lemma \ref{cont}.
			% 	Lemma \ref{id} implies that \hyperref[A2]{(A2)} is satisfied with $\rho(\epsilon_n+ \lambda_n, \theta^0, a)\asymp (\epsilon_n+\lambda_n)\kappa_a^{\hat G}(s,u)^{-1}\lesssim (\epsilon_n+\lambda_n)s^{1/a} C(u)^{-1}$. Combining the results in Lemma \ref{cont} %(for the sparse case) and \ref{id}.
			% 	leads to the conclusions. 
		\end{proof}
		
		\subsubsection{Approximate Sparsity}\label{approx}
		In this subsection, we discuss how the main results concerning the consistency of the GDS estimator can be adapted to a different sparsity assumption - namely, approximate sparsity. 
		\begin{itemize}
			\item[(A4')]\label{AS}(Approximate Sparsity) For some constant $C>0$ and $c>1/2$, the absolute values of the parameters $(|\theta^0_k|)_{k=1}^K$ can be rearranged in non-increasing order to $(|\theta_k^{0\ast}|)_{k=1}^K$ such that $|\theta^{0\ast}_k|\leq Ck^{-c}$ for $k=1,\ldots,K$.
		\end{itemize}
		
		\hyperref[ES]{(A4)} and \hyperref[AS]{(A4')} are two different assumptions regarding the sparsity of the true parameter $\theta^0$. We note that \hyperref[AS]{(A4')} can be reformulated to \hyperref[ES]{(A4)}. Suppose $\theta^0$ is approximately sparse. %and denote by $\theta^0_{[j]}$ the value of the true parameter corresponding to $|\theta^0|_j^\ast$, as defined in \hyperref[AS]{(A4.ii)}. 
		We sparsify $\theta^0$ to $\theta^0(\tau)=(\theta_1^0(\tau),\ldots,\theta_K^0(\tau))^\top$, where
		$$\theta_k^0(\tau)=\operatorname{sign}(\theta^{0\ast}_k)\tilde\theta_k(\tau),\quad\tilde\theta_k(\tau)=\begin{cases}
		|\theta^{0\ast}_k|+\delta/(s-1) &\mbox{ if } Ck^{-c}>\tau, \\
		0&\mbox{ otherwise },
		\end{cases}$$
		with $\tau$ chosen such that $s=\lfloor(C/\tau)^{1/c}\rfloor=\smallO(n)$ and $s>1$, and $\delta=\sum_{k=1}^K|\theta^{0\ast}_k|\IF(Ck^{-c}\leq\tau)$. Then, we have 
		\begin{align*}
		|\theta^0(\tau)|_1&=\sum_{k=1}^s|\theta^{0\ast}_k| +\frac{\delta s}{s-1}=\sum_{k=1}^s|\theta^{0\ast}_k| +\frac{s}{s-1}\sum_{k=s+1}^K|\theta^{0\ast}_k|\geq|\theta^0|_1.
		\end{align*}
		It follows that $\mathcal R(\theta^0)\subseteq\mathcal R(\theta^0(\tau))$. 
		%$\{\theta \in \Theta: |\theta|_1 \leq |\theta^0|_1\}\subseteq\{\theta \in \Theta: |\theta|_1 \leq |\theta^0(\tau)|_1\}$. 
		Consequently, on the event $\{\hat\theta\in\mathcal R(\theta^0)\}$, we have that $\hat\theta\in\mathcal R(\theta^0(\tau))$.
		
		To drive the estimation error bounds for the approximately sparse case, we apply the triangle inequality to decompose the errors as follows:
		$$|\hat\theta-\theta^0|_a\leq |\theta^0(\tau)-\theta^0|_a + |\hat\theta-\theta_0(\tau)|_a ,\,a\in\{1,2\}.$$
		Following the proof of Lemma 3.2 in \citet{belloni2018high}, the first term is bounded by:
		$$|\theta^0(\tau)-\theta^0|_a\leq C_{a,c}\tau s^{1/a},$$
		where $C_{a,c}$ is a constant depending only on $a$ and the constant $c$ satisfying \hyperref[AS]{(A4')}. 
		
		To bound the second term, we observe that:
		\begin{align*}
		|G(\hat\theta-\theta^0(\tau))|_\infty&\leq|G(\hat\theta-\theta^0)|_\infty + |G(\theta^0-\theta^0(\tau))|_\infty\\
		&\leq|G(\hat\theta-\theta^0)|_\infty + |G|_\infty|\theta^0-\theta^0(\tau)|_\infty.
		\end{align*}
		The first term on the right-hand side has been discussed in Section \ref{theoretical} of the main text. To address the second term, we assume that there exists a sequence of constants $L_n$ with $L_n\geq 1$ for all $n\geq1$ such that $|G|_\infty\leq L_n$. 
		Under the event:
		$$\{|\hat g(\theta^0)|_{\infty}\leq \lambda_n, \hat\theta\in\mathcal R(\theta^0), |\hat g(\hat\theta) - g(\hat\theta)|_{\infty}\leq \epsilon_n\},$$
		which holds with probability at least $1-\alpha-\smallO(1)$, and by assumption \hyperref[A_id]{(A6)}, we obtain:
		$$|\hat\theta-\theta_0(\tau)|_a \leq (\epsilon_n+\lambda_n+L_nC_{a,c}\tau)s^{1/a}C(u)^{-1},$$
		where $\epsilon_n$ is the concentration rate given in Lemma \ref{const}. 
		Combining the bounds together leads to the consistency of $\hat\theta$ when the model is approximately sparse.

		\subsection{Proofs of Section \ref{check}}
		\subsubsection{Sparse Inverse Matrix Estimation}\label{a9.2}
		
		To achieve a feasible debiased estimator in the form of \eqref{est.eq}, we should consider a sparse approximation of the inverse matrix for $\hat\Omega$. Define $\Upsilon^0\defeq \Omega^{-1}$ and let $\hat\Upsilon^1=(\hat\upsilon_{ij}^1)$ be the solution of
		\begin{equation}\label{clime}
		\min_{\Upsilon\in\R^{q\times q}}\sum_{i=1}^q\sum_{j=1}^q|\Upsilon_{ij}|: \quad |\hat\Omega\Upsilon - \mathbf I_q|_{\max}\leq \ell_n^\Upsilon,
		\end{equation}
		where $|\cdot|_{\max}$ is the element-wise max norm of a matrix, and $\ell_n^\Upsilon>0$ is a tuning parameter.
		A further symmetrization step is taken by
		\begin{equation}\label{symm}
		\hat\Upsilon = (\hat\Upsilon_{ij}),\quad \hat\Upsilon_{ij}=\hat\Upsilon_{ji}=\hat\Upsilon_{ij}^1\IF\{|\hat\Upsilon_{ij}^1|\leq |\hat\Upsilon_{ji}^1|\} + \hat\Upsilon_{ji}^1\IF\{|\hat\Upsilon_{ij}^1|> |\hat\Upsilon_{ji}^1|\}.
		\end{equation}
		Likewise, define $\Pi^0\defeq (G_1^\top\Upsilon^0G_1)^{-1}$ and $\Xi^0 \defeq (G_2^\top\Upsilon^0G_2)^{-1}$. We shall use the same approach to approximate the inverse of $\hat G_1^{\top}\hat\Upsilon\hat G_1$ and $\hat G_2^{\top}\hat \Upsilon\hat G_2$ by $\hat\Pi$ and $\hat\Xi$, in the cases of $K^{(1)}>n$ and $K^{(2)}>n$, respectively.
		
		Finally, we let $G_1^\top\Upsilon^0(\mathbf I_q - G_2\Xi^0 G_2^\top\Upsilon^0)G_1=:D+F$, where $D\defeq G_1^\top \Upsilon^0G_1=(\Pi^0)^{-1}$ and $F\defeq -G_1^\top\Upsilon^0G_2\Xi^0 G_2^\top\Upsilon^0G_1$. By using the  formula $(D+F)^{-1}= D^{-1} - D^{-1} (\mathbf I +FD^{-1})^{-1}FD^{-1}$, the feasible debiased estimator $\check\theta_1$ is obtained by
		\begin{equation}\label{maineq}
		\check{\theta}_1  = \hat{\theta}_1 - \{\hat\Pi-\hat\Pi(\mathbf I_q + \hat F\hat\Pi)^{-1}\hat F\hat\Pi\}\hat{G}_1^{\top}\hat\Upsilon(\mathbf I_q -\hat{G}_2\hat\Xi\hat G_2^\top\hat\Upsilon)\hat g(\hat\theta_1, \hat\theta_2),
		\end{equation}
		where $\hat F=-\hat G_1^\top\hat\Upsilon\hat G_2\hat\Xi \hat G_2^\top\hat\Upsilon\hat G_1$. 
		
		Next, we shall analyze the convergence rates of the estimators involved in handling the rank deficiency issues. 
		%In this subsection, we shall analyze the convergence rates of the estimators involved in handling the rank deficiency issues.
		Define the class of matrices
		$$\mathcal{U} \defeq \mathcal{U}(b,s_0(q)) = \Big\{\Upsilon: \Upsilon \succ 0, |\Upsilon|_1 \leq M, \max _{1\leq i\leq q} \sum^q_{j = 1}|\Upsilon_{ij}|^b\leq s_0(q)\Big\}$$
		for $0\leq b<1$, where $\Upsilon=(\Upsilon_{ij})$ and the notation $\Upsilon \succ 0$ indicates that $\Upsilon$ is positive definite. Similarly, we define %the class of matrices
		$$\widetilde{\mathcal{U}} \defeq \widetilde{\mathcal{U}}(b,s_0(K^{(1)})) = \Big\{\Pi: \Pi \succ 0, |\Pi|_1 \leq M, \max _{1\leq i\leq K^{(1)}} \sum^{K^{(1)}}_{j = 1}|\Pi_{ij}|^b\leq s_0(K^{(1)})\Big\}$$
		for $0\leq b<1$, where $\Pi=(\Pi_{ij})$. 
		Additionally, we impose few standard assumptions regarding the sparsity and boundedness of eigenvalues, along with certain conditions on the regularization parameters and the convergence of the estimators. The specific rates associated with these conditions are thoroughly discussed in the following lemmas and remarks.
		\begin{itemize}
			\item[(A8)]\label{A_clime}(Approximate Inverse Matrix) Assume the following conditions:
			\begin{itemize}
				\item[(i)] $\Upsilon^0=\Omega^{-1}$ belongs to the matrix class $\mathcal{U}(b,s_0(q))$. $\Pi^0=(G_1^\top\Upsilon^0G_1)^{-1}$ belongs to the matrix class $\widetilde{\mathcal{U}}(b,s_0(K^{(1)}))$.
				\item[(ii)] There exist positive constants $c_1,c_2$ such that 
				$c_1\leq \lambda_{\min}(G_1^{\top}G_1)\leq\lambda_{\max}(G_1^{\top}G_1)\leq c_2$, 
				$c_1\leq \lambda_{\min}(G_2^{\top}G_2)\leq\lambda_{\max}(G_2^{\top}G_2)\leq c_2$, and $c_1\leq \lambda_{\min}(\Omega)\leq\lambda_{\max}(\Omega)\leq c_2$.
				\item[(iii)] There exists a positive constant $C$ such that $|\Upsilon^0|_2\leq C$, $|(\mathbf I + F\Pi^0)^{-1}|_2\leq C$, and with probability approaching 1, $|(\mathbf I + \hat F\hat \Pi)^{-1}|_2\leq C$.
				\item[(iv)] The regularization parameter $\ell_n^\Upsilon\geq0$ is selected such that $|\hat\Omega- \Omega|_{\max} M\leq \ell_n^\Upsilon$ holds with probability approaching $1$. Similarly, $\ell_n^\Pi\geq0$ is chosen such that $|G_1^\top\Upsilon^0G_1- \hat G_1^\top\hat\Upsilon\hat G_1|_{\max} M\leq \ell_n^\Pi$ holds with probability approaching $1$. The specific rates of $\ell^\Upsilon_n$ and $\ell_n^\Pi$ are provided in Lemma \ref{ln_ups} and Lemma \ref{ratef}, respectively. A parallel assumption applies to the regularization parameter used in approximating the inverse of $\hat G_2^\top\hat\Upsilon\hat G_2$.
				\item[(v)] There exist sequences of positive constants $\rho_{n}^{G_1}$, $\rho_{n}^{G_2}$, and $\rho_{n,2}^F$, which vanish as $n\to\infty$, such that $|\hat{G}_{1}- G_1|_{\max} \lesssim_\P \rho_{n}^{G_1}$, $|\hat{G}_{2}- G_2|_{\max} \lesssim_\P \rho_{n}^{G_2}$, and $|\hat{F}- F|_{2}\lesssim_\P\rho_{n,2}^F$. The specific convergence rate of $\rho_{n,2}^F$ is detailed in Lemma \ref{ratef}. 
			\end{itemize}
		\end{itemize}
		
		\begin{lemma}\label{rate_Ups}
			%Assume that $\Upsilon^0=\Omega^{-1}\in\mathcal{U}(b,s_0(q))$. Select $\ell^\Upsilon_n\geq0$ such that $|\hat\Omega- \Omega|_{\max} M\leq \ell_n^\Upsilon$ with probability approaching $1$ (the detailed rate of $\ell^\Upsilon_n$ is specified in Lemma \ref{ln_ups}).
			Assuming that \hyperref[A_clime]{(A8)}(i) and (iv) hold,  
			we have
			$$|\hat\Upsilon- \Upsilon^0|_{\max}\leq 4 M \ell^\Upsilon_n =: \rho_n^{\Upsilon}$$
			holds with probability approaching $1$, Moreover, with probability approaching $1$, we have
			$$|\hat\Upsilon- \Upsilon^0|_2\leq C_b (4M\ell^\Upsilon_n)^{1-b}s_0(q) =: \rho_{n,2}^\Upsilon,$$
			where $C_b$ is a positive constant only depends on $b$.
		\end{lemma}
		
		%	{\red
		%		\begin{remark}
		%			Suppose $M, s_0(q) \lesssim s$, then $\rho_{n,2}^\Upsilon \lesssim (4s \sqrt{\log P}/{\sqrt{n}})^{1-r}s $, and $\rho_n^{\Upsilon} \lesssim s^{2} \sqrt{\log P}/n^{1/2}$.
		%		\end{remark}
		%	}
		%	
		\begin{proof}
			Recall that $\hat\Upsilon^1$ is the solution of \eqref{clime}. We first observe that
			\begin{eqnarray*}
				|\hat\Upsilon^1- \Upsilon^0|_{\max}& =& |\Upsilon^0\Omega(\hat{\Upsilon}^1- \Upsilon^0)|_{\max}\\
				&\leq&|\Omega(\hat\Upsilon^1- \Upsilon^0)|_{\max} |\Upsilon^0|_1\\
				|\Omega(\hat\Upsilon^1- \Upsilon^0)|_{\max} & \leq& |(\Omega - \hat{\Omega})(\hat\Upsilon^1-\Upsilon^0)|_{\max} + |\hat{\Omega}(\hat\Upsilon^1- \Upsilon^0)|_{\max}=: R_{n,1}+ R_{n,2}.
			\end{eqnarray*}
			In particular, we have $R_{n,1}\leq 2|\Omega - \hat{\Omega}|_{\max} M \leq 2\ell^\Upsilon_n$ holds with probability tending $1$, and $R_{n,2}\leq |\hat{\Omega}\Upsilon^0 - \mathbf I_q|_{\max}+ |\hat{\Omega}\hat{\Upsilon}^1- \mathbf I_q|_{\max}\lesssim_\P 2\ell_n^\Upsilon$. According to the definition given by \eqref{symm}, it follows that $|\hat\Upsilon- \Upsilon^0|_{\max} \leq 4M\ell^\Upsilon_n $ with probability approaching $1$.
			The rate of $\ell^\Upsilon_n$ will depend on the concentration inequalities we use.
			
			Moreover, with probability approaching $1$, we have
			$$|\hat\Upsilon- \Upsilon^0|_2\leq \sqrt{ |\hat\Upsilon- \Upsilon^0|_1 |\hat\Upsilon- \Upsilon^0|_{\infty}}=|\hat\Upsilon- \Upsilon^0|_1\leq C_b (4M\ell^\Upsilon_n)^{1-b}s_0(q),$$
			where $C_b$ is a positive constant only depends on $b$.
			The rate of $|\hat\Upsilon- \Upsilon^0|_1$ follows from the proof of Theorem 6 in \cite{cai2011constrained}.
		\end{proof}
		
		% Similarly, we define the class of matrices
		% $$\widetilde{\mathcal{U}} \defeq \widetilde{\mathcal{U}}(b,s_0(K^{(1)})) = \Big\{\Pi: \Pi \succ 0, |\Pi|_1 \leq M, \max _{1\leq i\leq K^{(1)}} \sum^{K^{(1)}}_{j = 1}|\Pi_{ij}|^b\leq s_0(K^{(1)})\Big\}$$
		% for $0\leq b<1$, where $\Pi=(\Pi_{ij})$. The lemma below follows.
		
		\begin{lemma}\label{rate_Pi}
			%Assume that $\Pi^0=(G_1^\top\Upsilon^0G_1)^{-1}\in\widetilde{\mathcal{U}}(b,s_0(K^{(1)}))$. Select $\ell_n^\Pi$ such that $|G_1^\top\Upsilon^0G_1- \hat G_1^\top\hat\Upsilon\hat G_1|_{\max} M\leq \ell_n^\Pi$ with probability approaching $1$ (see Lemma \ref{ratef} for the specific rate of $\ell_n^\Pi$).
			Assuming that \hyperref[A_clime]{(A8)}(i) and (iv) hold, 
			we have
			$$|\hat\Pi- \Pi^0|_{\max}\leq 4 M \ell_n^\Pi =: \rho_n^\Pi$$
			and
			$$|\hat\Pi- \Pi^0|_2\leq C_b (4M\ell_n^\Pi)^{1-b}s_0(K^{(1)})=:\rho_{n,2}^\Pi$$
			hold with probability approaching $1$, respectively.
		\end{lemma}
		\begin{proof}
			The proof is similar to that of Lemma \ref{rate_Ups} and thus is omitted.
		\end{proof}
		
		%	{\red
		%		\begin{remark} According remark \ref{3.8}, we have $\ell_n^\Pi \lesssim s^4 \sqrt{\log P}/{\sqrt{n}}$. Thus  $\rho_n^\Pi \lesssim s^5 \sqrt{\log P}/{\sqrt{n}}$ and $\rho_{n,2}^\Pi \lesssim s(s^4 \sqrt{\log P}/{\sqrt{n}} )^{1-r}$.
		%	\end{remark}}
		
		Recall that  $D\defeq G_1^\top \Upsilon^0G_1=(\Pi^0)^{-1}$ and $F\defeq -G_1^\top\Upsilon^0G_2\Xi^0 G_2^\top\Upsilon^0G_1$. Next, we show the rate of the estimator of $B=({D} + {F})^{-1}=((\Pi^0)^{-1} + F)^{-1}$ given by $\hat B = \hat\Pi-\hat\Pi(\mathbf I_q+\hat F\hat\Pi)^{-1}\hat F\hat\Pi$. Denote by $\rho_{n,2}^F$ the rate such that $|\hat F-F|_2\lesssim_\P \rho_{n,2}^F$. We shall discuss the conditions on this rate in Lemma \ref{ratef}.
		
		\begin{lemma}\label{df}
			% Under the conditions of Lemma \ref{rate_Pi}, {\color{red}suppose that there exist constants $c_1,c_2,c_3$ such that %$0<c_1\leq \sigma_{\min}(F)\wedge\sigma_{\min}(\Pi^0)$ and
			% $c_1\leq \sigma_{\min}(F\Pi^0)\leq\sigma_{\max}(F\Pi^0)\leq c_2$ and $\sigma_{\max}(F)\vee \sigma_{\max}(\Pi^0) \leq c_3$. {\color{red}Assume that there exists a constant $C>0$ such that $|(\mathbf I - F\Pi^0)^{-1}|_2\leq C$ and $|(\mathbf I - \hat F\hat \Pi)^{-1}|_2\leq C$.}}
			% %and assume there exists a constant $C>0$ such that $(\rho_{n,2}^\Pi\vee \rho_{n,2}^F)\tilde s\lesssim C$.% and $\tilde s/\sigma_{\min}(\Pi^0^{-1})\lesssim C$.
			%Under the conditions of Lemma \ref{rate_Pi}, 
			Assuming that \hyperref[A_clime]{(A8)}(i)-(iv) hold, we have
			\begin{equation*}
			|\hat B - B|_{\max}\lesssim_\P (\rho_n^\Pi\vee\rho_{n,2}^\Pi\vee  \rho_{n,2}^F)=:\rho_n^B.
			\end{equation*}
		\end{lemma}
		
		\begin{proof}
			We first observe that
			\begin{align*}
			|\hat B - B|_{\max}&\leq|\hat\Pi-\Pi^0|_{\max} + |(\hat{\Pi}- \Pi^0)(\mathbf I + {F}\Pi^0)^{-1} {F}\Pi^0|_{\max}\\
			&\quad+ |\hat{\Pi} \{(\mathbf I +\hat{F}\hat{\Pi})^{-1}- (\mathbf I+{F}\Pi^0)^{-1}\}{F}\Pi^0\}|_{\max} + |\hat{\Pi} (\mathbf I+\hat{F}\hat{\Pi})^{-1}(\hat{F}\hat{\Pi} - F\Pi^0)|_{\max}\\
			&\lesssim_\P \rho_n^\Pi +  |\hat{\Pi}- \Pi^0|_2|(\mathbf I + {F}\Pi^0)^{-1}|_2 |{F}\Pi^0|_{2}\\
			&\quad + |\hat{\Pi}|_2 |(\mathbf I +\hat{F}\hat{\Pi})^{-1}- (\mathbf I+{F}\Pi^0)^{-1}|_2|{F}\Pi^0|_{2} + |\hat{\Pi}|_2| (\mathbf I+\hat{F}\hat{\Pi})^{-1}|_2|\hat{F}\hat{\Pi} - F\Pi^0|_{2}.
			\end{align*}
			Provided that $|\Pi^0|_2\vee|F|_2\leq c_3$, by applying Lemma \ref{rate_Pi}, we obtain:
			\begin{align*}%\label{pihat2}
			|\hat\Pi|_2\leq|\hat\Pi - \Pi^0|_2 + |\Pi^0|_2 \lesssim_\P \rho_{n,2}^\Pi + c_3.
			%|\lambda_{\min}(\hat{\Pi}^{-1})|^{-1}\lesssim_\P \{\lambda_{\min}({(\Pi^0)^{-1}})- r_{n}^\Pi\}^{-1}.
			\end{align*}
			Besides, we have
			\begin{align*}
			|\hat{F}\hat{\Pi} - F\Pi^0|_{2} \leq |\hat{F} - F|_2 |\hat{\Pi}|_2+ |\hat{\Pi}- \Pi^0|_2|F|_2\lesssim _\P \rho_{n,2}^F c_3+ \rho_{n,2}^\Pi c_3\lesssim \rho_{n,2}^F \vee \rho_{n,2}^\Pi ,
			\end{align*}
			and
			\begin{align*}
			|(\mathbf I +\hat{F}\hat{\Pi})^{-1}- (\mathbf I+{F}\Pi^0)^{-1}|_2\leq |(\mathbf I+F\Pi^0)^{-1}|_2 |(\mathbf I+\hat{F}\hat{\Pi})^{-1}|_2|\hat{F}\hat{\Pi} - F\Pi^0 |_{2}\lesssim_\P\rho_{n,2}^F \vee \rho_{n,2}^\Pi.
			%&\lesssim_\P& \rho_{n,2}^F \tilde s+ \rho_{n,2}^\Pi \tilde s = \bigO(1). %[\{\lambda_{\min}(F)-r_n^F\}\{\lambda_{\min}(\Pi^0)-\rho_{n,2}^\Pi\}-1]^{-1}[r_n^F \rho_{n,2}^\Pi + \rho_{n,2}^\Pi s].
			%&\lesssim & [c- r_F (r_{D2}+\tilde{q})- r_{D2}s]^{-1}( r_F (-r_{D2}+c)+ r_{D2}s).
			\end{align*}
			Finally, the desired conclusion follows by collecting all the results above.
			
			%We shall assume the following $c_{df}\leq \sigma_{\min}(F^{-1}D)\leq \sigma_{\max}(F^{-1}D) \leq C_{df}$, $s (r_{D2}\vee r_F )\lesssim c$, $\sigma_{\min}(D)\geq c_D>0$, $s/\sigma_{\min}(D)\lesssim c$.
		\end{proof}
		In this Lemma we assume that $|(\mathbf I + F\Pi^0)^{-1}|_2\leq C$, which can be implied by the condition $\sigma_{\min}(F\Pi^0) > 1$ or $\sigma_{\max}(F\Pi^0) < 1$. For example, given $1<c_1\leq\sigma_{\min}(F\Pi^0)$, we have
		$$
		|(\mathbf I + F\Pi^0)^{-1}|_2\leq(\sigma_{\min}(\mathbf I + F\Pi^0))^{-1}\leq (\sigma_{\min}(F\Pi^0)- 1)^{-1} \leq(c_1-1)^{-1}, \,c_1>1,
		$$
		where the first inequality is implied by Lemma \ref{prod},
		%and the factor that $|A^{-1}|_2\leq \sigma_{\min}^{-1}(A)$,
		and the second one is due to Lemma \ref{svd}. Additionally, based on Lemma \ref{smallbound}, on the event $\{\sigma_{\min}(\hat F\hat\Pi)>1\}$, which holds with probability approaching 1, it follows that
		\begin{align*}
		|(\mathbf I+ \hat{F}\hat{\Pi})^{-1}|_2&\leq(\sigma_{\min}(\mathbf I+ \hat{F}\hat{\Pi}))^{-1}\\
		&\leq (\sigma_{\min}(\hat F\hat\Pi)-1)^{-1}\\
		&\leq (\lambda_{\min}(\hat F)\lambda_{\min}(\hat\Pi)-1)^{-1}\\
		&\lesssim_\P \{(\lambda_{\min}(F)-\rho_{n,2}^F)(\lambda_{\min}(\Pi^0)-\rho_{n,2}^\Pi)-1\}^{-1}\lesssim C.
		%&\lesssim  [c- r_F (r_{D2}+s)- r_{D2}s]^{-1}
		\end{align*}
		%where we the last inequality is implied by Lemma \ref{rate_Pi}.
		
		The rate of $|\hat B - B|_{\infty}$ can be derived analogously, once the rate of $|\hat V - V|_\infty$ is established, where $V\defeq (\mathbf I + F\Pi^0)^{-1}$ and $\hat V\defeq(\mathbf I + \hat F\hat\Pi)^{-1}$; see Remark \ref{rateb} for further discussion.
		
		\subsubsection{Proofs of Main Theorems}\label{a9.3} %and Detailed Rate of $|r_n|_\infty$}for Linear Case} 
		
		\begin{proof}[\textbf{Proof of Theorem \ref{linear}}]
			%	For $s<< n$, we have the following assumptions.
			%	In the linear case, if $K^{(1)}$ is diverging and high dimensional, we shall assume,
			%	\begin{eqnarray*}
			%		|(AG_2)^{-1}|_1 = |(AG_2)^{-1}|_\infty  \leq \sqrt{s} |(AG_2)^{-1}|_2 \lesssim  \sqrt{s}{s}^{-1} \asymp  {\sqrt{s}}^{-1} .
			%	\end{eqnarray*}
			%	
			%	\begin{eqnarray*}
			%		|(AG_2)|_1 = |(AG_2)|_\infty  \leq \sqrt{s} |(AG_2)^{-1}|_2 \lesssim  \sqrt{s}{s} \asymp  {{s}}^{3/2} .
			%	\end{eqnarray*}
			%	
			%	
			%	In the linear case, if $K^{(2)}$ is  fixed,  we shall assume,
			%	\begin{eqnarray*}
			%		|(AG_2)^{-1}|_1 = |(AG_2)^{-1}|_\infty= {{s}}^{-1}.
			%	\end{eqnarray*}
			%	
			%	\begin{eqnarray*}
			%		|(AG_2)|_1 = |(AG_2)|_\infty  \leq \sqrt{s} |(AG_2)|_2 \lesssim  \sqrt{s}{s} \asymp  {{s}}^{3/2} .
			%	\end{eqnarray*}
			%	
			%	
			%	We assume that  $|(AG_2)^{-1}|_2 \leq s^{-1}$, where $s$ is a number of smaller order than $n$. $|\hat{A} \hat{G}_2|_{1} \lesssim s^{3/2} $, as $|AG_2|_2\leq |G_{2}^{\top}\Omega^{-1} G_2|_2$(for $I - G_1 P(\Omega, G_1)$ is a projection matrix with the maximum eigenvalue to be $1$.) Then we have $|G_{2}^{\top}\Omega^{-1} G_2|_2 \leq \sigma_{\max}(G_{2}^{\top}G_2)\sigma_{\max}(\Omega^{-1})$, and $\sigma_{\max}(G_{2}^{\top}G_2) \lesssim s$.
			%	
			%	We also assume that $|(A)|_1\leq s$.
			
			According to Lemma \ref{df} and \ref{ratef}, we have $|\hat B -B|_{\max} \lesssim_\P\rho_n^B= \rho_n^\Pi \vee \rho_{n,2}^\Pi\vee \rho_{n,2}^F$ and $|\hat{A} \hat{G}_1 - AG_1|_{\max}\lesssim_\P \ell_n^\Pi/M+ \rho_{n,2}^F$.
			% $ |\Omega^{-1}|_2\leq c'$ and $|I|_2 \vee |G_1(G_1^{\top}\Omega^{-1}G_1)^{-1} G_1 \Omega^{-1}|_2 \leq c''$, we have $|(AG)^{-1}|_\infty  \lesssim  \sqrt{L^{(2)}} c'c'' \tilde{q}^{-1}$.
			Based on the Gaussian approximation results as discussed in Section \ref{inference}, we have $|\hat{g}(\theta^0)|_{\infty} \lesssim_\P n^{-1/2} (\log q)^{1/2}$. On the event $\{|\hat A\hat G_1|_1\lesssim \omega\}$, which holds with probability approaching 1, applying the results in \eqref{ratetheta} as well as Remarks \ref{rateb} and \ref{ratea}, we obtain
			\begin{align}\label{rater}
			|r_{n,1}|_{\infty} &\lesssim_\P \kappa(\ell_n^\Pi/M + \rho_{n,2}^F)d_{n,1} + \rho_n^B\omega d_{n,1} =:\varrho_{n,1},\notag\\
			%|r_{n,2}|_{\infty}&\lesssim_\P  (\vartheta+\rho_{n,2}^B)\rho_n^{G_1}d_{n,1}=:\varrho_{n,2},\notag\\
			|r_{n,2}|_{\infty}&\lesssim_\P   \{\rho_{n,2}^B\iota + (\kappa+\rho_{n,2}^B)\rho_{n,2}^A\}n^{-1/2}(\log q)^{1/2}=:\varrho_{n,2}.
			\end{align}
			%We note that $r_{n,2}=0$ in linear moments models.
			%
			%	$r_{A, \max}$ is defined in Lemma \ref{ratef}.
			%	Therefore $|r_{n1}|_{\infty} = s^{-1/2}(e_n/M+r_F)r_{\theta_n,1}+ r_D \vee r_{D2} (s^{3/2}+ r_D \vee r_{D2}) r_{\theta_n,1} $ , $r_{n2} = 0$, and $r_{n3} \lesssim \{ (r_D \vee r_{D2})s \vee (\sqrt{s}^{-1} +r_D \vee r_{D2} ) (s r_{A,\max}) \} |\hat{g}(\theta^0)|_{\infty} \lesssim_p \{ (r_D \vee r_{D2})s \vee (\sqrt{s}^{-1} +r_D \vee r_{D2} ) (s r_{A,\max})\sqrt{n}^{-1} \sqrt{\log(q)} \}$.
		\end{proof}
		We will examine the detailed rates of $\ell_n^\Upsilon$, $\ell_n^\Pi$, and $\rho_{n,2}^F$, which are involved in the rate of $|r_n|_\infty$, in the following Section \ref{rn}. %in Appendix \ref{a9.3}.

		\begin{proof}[\textbf{Proof of Theorem \ref{theorem.inference}}]
			The proof is similar to that of Corollary 5.8 of \citet{lasso2018}, which applies Theorem 5.1 of \cite{ZW15gaussian}. Therefore, it is is omitted here. In particular, with $r\geq4$ and $\varsigma>0$ such that $\|\mG_{k,\cdot}\|_{r,\varsigma}$ is bounded by a constant for any $k\in\mathcal S$, the following additional conditions on $b_n$ and $|\mathcal S|$ are required:
			\begin{align}\label{bn}
			&b_n=\smallO\{n(\log |\mathcal S|)^{-5}\},\, \text{ with }\{b_n^{-1}+n^{-\varsigma}+(n-b_n)b_n^{-\varsigma+1}/(nb_n)\}(\log |\mathcal S|)^2=\smallO(1),\,\text{if } \varsigma<1;\notag\\
			&\{b_n^{-1}+\log(n/b_n)/n+(n-b_n)\log b_n/(nb_n)\}(\log |\mathcal S|)^2=\smallO(1),\,\text{if } \varsigma=1;\notag\\
			&\{b_n^{-1}+n^{-1}b_n^{-\varsigma+1}+(n-b_n)/(nb_n)\}(\log |\mathcal S|)^2=\smallO(1),\,\text{if } \varsigma>1.\notag\\
			&n=\smallO\{n^{r/2}(\log |\mathcal S|)^{-r}|\mathcal S|^{-2}\},\, \text{if }\varsigma >1-2/r;\notag\\
			&l_nb_n^{r/2-\varsigma r/2}=\smallO\{n^{r/2}(\log |\mathcal S|)^{-r}|\mathcal S|^{-2}\},\,\text{if }1/2-2/r<\varsigma<1-2/r.
			\end{align}
			% where 
			% $F_{\varsigma} = n$, for $\varsigma >1-2/r$; $F_{\varsigma} = l_nb_n^{r/2-\varsigma r/2}$, for $1/2-2/r<\varsigma<1-2/r$.%; {\red $F_{\varsigma} = l_n^{r/4-\varsigma r/2}b_n^{r/2-\varsigma r/2} $, for $\varsigma<1/2-2/r$}.
		\end{proof}
		
		\subsubsection{Detailed Rate of $|r_n|_\infty$}\label{rn}
		Recall that in the case of linear moments models, the score functions are given by $g_j(D_{j,t},\theta)=z_{j,t}\vps_j(D_{j,t},\theta)$, where $\vps_j(D_{j,t},\theta)=y_{j,t}-x_{j,t}^\top\vartheta_j$. To simplify the notations, we denote $g_{jm,t}\defeq z_{jm,t}\vps_j(D_{j,t},\theta^0)$ and $\hat g_{jm,t}\defeq z_{jm,t}\vps_j(D_{j,t},\hat\theta)$, for all $j=1,\ldots,p$ and $m=1,\ldots,q_j$. Similarly, we define $g_{il,t}$ and $\hat g_{il,t}$, for $i=1,\ldots,p$ and $l=1,\ldots,q_i$. We note that when the time series is non-stationary and the mean varies with respect to $t$, we can replace $\E (g_{il,t}g_{jm,t})$ by
		$\E_n\E (g_{il,t}g_{jm,t})$.
		
		% Let $C_{xz}$ and $C_{xz\vps}$ be constants such that $\max_{i,j,l,m}|\E(x_{i,t}x_{j,t}^\top z_{il,t}z_{jm,t})|_{\max}\leq C_{xz}$ and $\max_{i,j,l,m}|\E(x_{i,t}z_{jm,t}z_{il,t}\vps_{i,t}))|_{\infty}\leq C_{xz\vps}$, respectively.
		Let $C>0$ be an absolute constant such that $\max_{i,j,l,m}|\E(x_{i,t}x_{j,t}^\top z_{il,t}z_{jm,t})|_{\max}\leq C$ and $\max_{i,j,l,m}|\E(x_{i,t}z_{jm,t}z_{il,t}\vps_{i,t}))|_{\infty}\leq C$.
		\begin{lemma}[Rate of $\ell_n^\Upsilon$]\label{ln_ups}
			Under the conditions in Lemma \ref{cont} and %\ref{id}
			assumption \hyperref[A_id]{(A6)}, we have
			$$|\hat{\Omega}- {\Omega}|_{\max}\lesssim_\P  \ell_n^\Upsilon/ M,$$
			given $d_{n,1}^2(C+ \gamma_n) + d_{n,1}(C + \gamma'_n) + \gamma_{n,1} + \gamma_{n,2}\lesssim \ell_n^\Upsilon/ M$, where $d_{n,1}$ is defined in \eqref{ratetheta}, $\gamma_n, \gamma'_n, \gamma_{n,1}, \gamma_{n,2}$ are specified in \eqref{gamman1}, \eqref{gamman2} and \eqref{gamman3}.
			%$r_{\theta_n, 1}(C'_{xz}+ r_{x\vps}(\alpha_n) )+r_{\theta_n, 2}^2  C_{xz} + r_{xz}(\alpha_n)+  r_{T_3}(\alpha_n)  \lesssim \ell_n^\Upsilon/ M$.
		\end{lemma}
		
		\begin{proof}
			We first observe that
			\begin{eqnarray*}
				|\hat{\Omega}- {\Omega}|_{max} &=& \max_{i,j,l,m}|{\E}_n (\hat{g}_{il,t}\hat{g}_{jm,t}) - \E (g_{il,t} g_{jm,t})|\\
				&\leq&\max_{i,j,l,m}|{\E}_n\{(\hat{g}_{il,t}-g_{il,t})(\hat{g}_{jm,t}-g_{jm,t})\}| + 2\max_{i,j,l,m}|{\E}_n\{g_{il,t} (\hat{g}_{jm,t}-g_{jm,t})\}|\\&&+ \max_{i,j,l,m}|{\E}_n(g_{il,t} g_{jm,t}) - \E (g_{il,t}g_{jm,t}) |\\
				&=:& I_{n,1}+ I_{n,2}+ I_{n,3}. %\lesssim b_n/M.
			\end{eqnarray*}
			
			For $I_{n,1}$, it can be seen that
			\begin{eqnarray*}
				I_{n,1}&\leq&\max_{i,j,l,m}|\hat\beta_i-\beta_i^0|_1|\hat\vartheta_j-\vartheta_j^0|_1|{\E}_n(x_{i,t}x_{j,t}^\top z_{il,t}z_{jm,t})|_{\max}\\
				&\leq&|\hat\theta - \theta^0|_1^2 \{ C + |{\E}_n(x_{i,t}x_{j,t}^\top z_{il,t}z_{jm,t}) - \E(x_{i,t}x_{j,t}^\top z_{il,t}z_{jm,t})|_{\max} \}.
			\end{eqnarray*}
			Let $\chi_t^{ijlm}\defeq\vec(x_{i,t}x_{j,t}^\top z_{il,t}z_{jm,t}) = (\chi_{k,t}^{ijlm})_{k=1}^{K_i*K_j}$ and define
			\begin{align}\label{gamman1}
			\gamma_n\defeq cn^{-1/2}(\log P_n)^{1/2}\max_{i,j,l,m,k}\|\chi_{k,\cdot}^{ijlm}\|_{2,\varsigma} + cn^{-1}c_{n,\varsigma}(\log P_n)^{3/2}\|\max_{i,j,l,m}|\chi_{\cdot}^{ijlm}|_\infty\|_{r,\varsigma},
			\end{align}
			with $P_n =(q\vee n\vee e)$, %$\bar P =\max_j(K_j^2q^2)$,
			$c_{n,\varsigma} = n^{1/r}$ for $\varsigma> 1/2 - 1/r$ and $c_{n,\varsigma} = n^{1/2-\varsigma}$ for $0<\varsigma<1/2 - 1/r$. By applying Lemma \ref{tail} and the results in \eqref{ratetheta}, we have 	
			$I_{n,1}\lesssim_\P d_{n,1}^2(C+ \gamma_n)$, for sufficiently large $c$.
			
			Similarly,
			\begin{eqnarray*}
				I_{n,2}&\leq&2\max_{i,j,l,m}|\hat\vartheta_j-\vartheta_j^0|_1|{\E}_n(x_{j,t}z_{jm,t}z_{il,t}\vps_{i,t})|_{\infty}\\
				&\leq&2|\hat\theta - \theta^0|_1 \{ C + |{\E}_n(x_{j,t}z_{jm,t}z_{il,t}\vps_{i,t}) - \E(x_{j,t}z_{jm,t}z_{il,t}\vps_{i,t})|_{\infty} \}.
			\end{eqnarray*}
			Let $\zeta_t^{ijlm}\defeq x_{j,t}z_{jm,t}z_{il,t}\vps_{i,t} = (\zeta_{k,t}^{ijlm})_{k=1}^{K_j}$ and define
			\begin{align}\label{gamman2}
			\gamma'_n\defeq cn^{-1/2}(\log P_n)^{1/2}\max_{i,j,l,m,k}\|\zeta_{k,\cdot}^{ijlm}\|_{2,\varsigma} + cn^{-1}c_{n,\varsigma}(\log P_n)^{3/2}\|\max_{i,j,l,m}|\zeta_{\cdot}^{ijlm}|_\infty\|_{r,\varsigma}.
			\end{align}
			%with $\bar P' =\max_j(K_jq^2)$.
			It follows that	$I_{n,2}\lesssim_\P d_{n,1}(C + \gamma'_n)$, for sufficiently large $c$.
			
			Lastly, $I_{n,3}$ is handled by pointwise concentration for two parts as
			\begin{eqnarray*}
				I_{n,3} \leq \max_{i\neq j\text{ or }l\neq m} |{\E}_n (g_{il,t}g_{jm,t}) -\E  (g_{il,t}g_{jm,t})| + \max_{j,m} |{\E}_n g_{jm,t}^2 -\E  g_{jm,t}^2|,
			\end{eqnarray*}
			%We shall also apply Lemma \ref{tail} to the long vectors involved in the maximums respectively.
			where H\"{o}lder's inequality is applied when dealing with the first part.
			
			Let
			\begin{align}\label{gamman3}
			\gamma_{n,1} &\defeq cn^{-1/2}(\log P_n)^{1/2}(\Phi_{4,\varsigma}^{\vps z})^2 + cn^{-1}c_{n,\varsigma}(\log P_n)^{3/2}\|\max_j|\vps_{j,\cdot}z_{j,\cdot}|_\infty\|_{2r,\varsigma}^2,\notag\\
			\gamma_{n,2} &\defeq cn^{-1/2}(\log P_n)^{1/2}\max_{j,m} \|\vps_{j,\cdot}^2z_{jm,\cdot}^2\|_{2,\varsigma} + cn^{-1}c_{n,\varsigma}(\log P_n)^{3/2}\|\max_{j,m}|\vps^2_{j,\cdot}z^2_{jm,\cdot}|\|_{r,\varsigma}.
			\end{align}
			%with $P_1=q(q-1)/2$ and $P_2=q^2$.
			Then, we have $I_{n,3} \lesssim_\P \gamma_{n,1} + \gamma_{n,2} $ for sufficiently large $c$.
			
			By collecting all the results above, we can claim that $|\hat{\Omega}- {\Omega}|_{\max}\lesssim_\P  \ell_n^\Upsilon/ M$ by selecting $\ell_n^\Upsilon$ such that $d_{n,1}^2(C+ \gamma_n) + d_{n,1}(C+ \gamma'_n) + \gamma_{n,1} + \gamma_{n,2}\lesssim \ell_n^\Upsilon/ M$.
		\end{proof}
		
		We shall provide an admissible rate for $\ell_n^\Upsilon$ under a specific example in Remark \ref{rateUps}. Next, we analyze the rate $\ell_n^\Pi$, for which we introduce the following definitions.
		
		Let the subset $\mathcal P^{(1)}\subseteq\{1,\ldots,p\}$ be the equation index space related to $\theta_1^0$. And for each $j\in\mathcal P^{(1)}$, the subset $\mathcal K_j^{(1)}\subseteq\{1,\ldots,K_j\}$ is the parameter index space related to $\theta_1^0$ in the $j$-th equation. Let
		$$\rho_{n}^{G_1}\defeq cn^{-1/2}(\log P_n)^{1/2} \Phi_{2,\varsigma}^{xz}
		%\max_{j\in\mathcal P^{(1)},k\in\mathcal K_j^{(1)},m\in q_j}\|\tilde x_{jk,\cdot}z_{j,\cdot}\|_{2,\varsigma}
		+ cn^{-1}c_{n,\varsigma}(\log P_n)^{3/2}\|\max\nolimits_{j\in\mathcal P^{(1)},k\in\mathcal K_j^{(1)}}|x_{jk,\cdot}z_{j,\cdot}|_\infty\|_{r,\varsigma}.$$
		%with $\tilde P=\sum_{j\in\mathcal P^{(1)}}q_j|K_j^{(1)}|$.
		Define the matrix norms $|G_1|_{1,l}  = \max\limits_{1\leq j\leq K^{(1)}} \sum_{i=1}^q |{G}_{1,ij}|^l$, $|G_1|_{\infty,l} = \max\limits_{1\leq i\leq q} \sum_{j=1}^{K^{(1)}} |{G}_{1,ij}|^l$, and $|G_1|_0$ is the number of non-zero components in $G_1$.
		
		%	{\red
		%		\begin{remark}
		%			Further to the remark \ref{3.8}, we can have $\rho_{n}^{G_1} \lesssim_p n^{-1/2}(\log P)^{1/2}$ if the conditions in the  remark holds and the corresponding  dependent adjusted norm are bounded.
		%		\end{remark}
		%	}	
		
		\begin{lemma}\label{sparse}
			Assume that $|\hat{G}_{1}- G_1|_{\max} \lesssim_\P \rho_{n}^{G_1}$. Then, we have
			$$|\hat{G}_{1}- G_1|_1 \lesssim_\P \rho_{n,2}^{G_1},\quad |\hat{G}_{1}- G_1|_2 \lesssim_\P \rho_{n,2}^{G_1},$$
			where $\rho_{n,2}^{G_1}= s(G_1)\rho_{n}^{G_1}$ in the sparse case with $s(G_1)=|G_1|_0$, and $\rho_{n,2}^{G_1}=L(\rho_{n}^{G_1})^{1-l}$ in the dense case with $\max\{|G_1|_{1,l}, |G_1|_{\infty,l},|G_1|_1,|G_1|_\infty\} \leq L$ for some $0\leq l<1$.		
		\end{lemma}
		
		\begin{proof}
			Recall that $\hat G_1=(\hat G_{1,ij})$ is a thresholding estimator with $\hat{G}_{1,ij} = \hat G^1_{1,ij }\IF\{|\hat G^1_{1,ij}|>T\}$, $\hat G^1_1 = (\hat G^1_{1,ij}) =\partial_{\theta_1^\top}\hat g(\theta_1,\hat\theta_2)|_{\theta_1=\hat\theta_1}$. Consider the event $\mathcal A$ defined by
			$$ \mathcal{A}\defeq \{ G_{1,ij} - \rho_{n}^{G_1}\leq \hat G^1_{1,ij} \leq G_{1,ij} + \rho_{n}^{G_1},\,\text{ for all } i=1,\ldots,q,j=1,\ldots,K^{(1)}\}.$$
			
			%We now derive the rate of regularized estimator $\hat{G}_{n,2}$, $|\hat{G}_{n,2}-G_2|_{2}\leq \sqrt{|\hat{G}_{n,2}-G_2|_{1}|\hat{G}_{n,2}-G_2|_{\infty}}$.
			
			Let $T\geq \rho_{n}^{G_1}$. On the event $\mathcal{A}$, which holds with probability approaching one, we have
			\begin{eqnarray*}
				&&\max_{1\leq j\leq K^{(1)}} \sum_{i=1}^q |\hat{G}_{1,ij}- G_{1,ij}|\\
				&\leq & \max_{1\leq j\leq K^{(1)}} \sum_{i=1}^q |\hat {G}^1_{1,ij} - G_{1,ij}| \IF\{|\hat{G}^1_{1,ij}| >T\}+ \max_{1\leq j\leq K^{(1)}} \sum_{i=1}^q |G_{1,ij}| \IF\{|\hat{G}^1_{ij}| \leq T\}\\
				&\leq & \max_{1\leq j\leq K^{(1)}} \sum_{i=1}^q |\hat {G}^1_{1,ij} - G_{1,ij}| \IF\{|G_{1,ij}| >T+ \rho_{n}^{G_1}\}+ \max_{1\leq j\leq K^{(1)}} \sum_{i=1}^q |G_{1,ij}| \IF\{|G_{1,ij}| \leq T-\rho_{n}^{G_1}\} \\
				&\lesssim_\P & s(G_1) \rho_{n}^{G_1}+ (T- \rho_{n}^{G_1})s(G_1),
			\end{eqnarray*}
			in the sparse case. By picking $T= 2\rho_{n}^{G_1}$, we obtain that $|\hat{G}_{1}- G_1|_1 \lesssim_\P \rho_{n,2}^{G_1}=s(G_1)\rho_{n}^{G_1}$. Similarly, we can prove that $|\hat{G}_{1}- G_1|_\infty \lesssim_\P \rho_{n,2}^{G_1}$ and it follows that $|\hat{G}_{1}- G_1|_2 \lesssim_\P \rho_{n,2}^{G_1}$ by H\"older's inequality.
			
			%		Thus we have $\max_j \sum_i |\hat{G}_{n,2,i,j}- G_{2,i,j}|\lesssim_p s_{0}(G_2) r_{2}(\alpha_n)$, similarly we can prove $ |\hat{G}_{n,2} - G_{2}|_{\infty}\lesssim_p s_{0}(G_2) r_{2}(\alpha_n)$. Thus $ |\hat{G}_{n,2} - G_{2}|_{2}\lesssim_p s_{0}(G_2) r_{2}(\alpha_n)$.
			%		We thus define $\delta_{G_2} = s_{0}(G_2) r_{2}(\alpha_n) $.
			%
			%		It is worth noting that the sparsity assumption $|G_2|_0 \leq s_0(G_2)$ can be replaced by the norm condition with $ 0\leq r <1$, $|G_2|_{1,r} \vee |G_2|_{\infty,r} \leq  K_n$, for some positive constant $K_n\lesssim n^{(1-r)/2}$. So that the above steps become,
			
			Likewise, for the dense case, on the event $\mathcal{A}$, we have
			\begin{eqnarray*}
				&&\max_{1\leq j\leq K^{(1)}} \sum_{i=1}^q |\hat{G}_{1,ij}- G_{1,ij}|\\
				&\leq & \max_{1\leq j\leq K^{(1)}} \sum_{i=1}^q |\hat {G}^1_{1,ij} - G_{1,ij}||G^l_{1,ij}/(T+\rho_{n}^{G_1})^l| \IF\{|G_{1,ij}| >T+ \rho_{n}^{G_1}\}+ L(T-\rho_{n}^{G_1})\\
				%\max_{j} \sum_i |G_{1,ij}| \IF\{G_{1,ij} \leq T+\rho_{n}^{G_1}\}\\
				&\lesssim_\P & %L (\rho_{n}^{G_1})^{1-l}+ L(T+\rho_{n}^{G_1}).
				L \rho_{n}^{G_1}/|T+\rho_n^{G_1}|^l+ L(T-\rho_{n}^{G_1}).
			\end{eqnarray*}
			It follows that $\rho_{n,2}^{G_1}=L (\rho_{n}^{G_1})^{1-l}$ in this case, if we select $T = 2\rho_{n}^{G_1}$.
		\end{proof}
		
		We denote $U \defeq G_2P(\Omega,G_2)$. Note that $|U|_2=1$ as it is an idempotent matrix. When $K^{(1)}$ is of high dimension potentially larger than $n$, we need to consider a regularized estimator given by $\hat{U} = \hat{G}_2\hat\Xi\hat G_2^\top\hat\Upsilon$. Denote by $\rho_{n,2}^U$ the rate such that $ |\hat U - U|_2\lesssim_\P\rho_{n,2}^U$. To further discuss the conditions on this rate, we assume that $|G_2|_2^2 \leq\omega_2$, $\sigma_{\min}(G_2)\geq\omega_2^{-1/2}$, and there exists constants $c$ and $C$ such that $0<c\leq\sigma_{\min}(\Upsilon^0)$ and $|\Upsilon^0|_2 \leq C$. It is not hard to see that
		\begin{eqnarray*}
			&& |\hat U - U|_2\\
			%&\leq& |\hat{G}_2 - G_2|_2\omega_2^{3/2}+|\hat\Xi- \Xi^0|_2\omega_2 + |\hat{G}_2^\top\hat\Upsilon - {G}_2^\top\Upsilon^0|_2\omega_2^{3/2}\\
			%&\leq& |\hat{G}_2 - G_2|_2\omega_2^{3/2}+|\hat\Xi- \Xi^0|_2\omega_2 + \{|\hat{G}_2 - G_2|_2 + \rho_{n,2}^\Upsilon (|\hat{G}_2 - G_2|_2+ \omega_2^{1/2})\}\omega_2^{3/2},
			&\leq& |\hat{G}_2 - G_2|_2(|\hat\Xi\hat G_2^\top\hat\Upsilon - \Xi^0G_2^\top\Upsilon^0|_2 + \omega_2^{3/2}) + \omega_2^{1/2}|\hat\Xi\hat G_2^\top\hat\Upsilon - \Xi^0G_2^\top\Upsilon^0|_2, \\
			&& |\hat\Xi\hat G_2^\top\hat\Upsilon - \Xi^0G_2^\top\Upsilon^0|_2 \\
			&\leq& |\hat\Xi-\Xi^0|_2(|\hat G_2^\top\hat\Upsilon - G_2^\top\Upsilon^0|_2 + \omega_2^{1/2}) + \omega_2|\hat G_2^\top\hat\Upsilon - G_2^\top\Upsilon^0|_2, \\
			&&|\hat G_2^\top\hat\Upsilon - G_2^\top\Upsilon^0|_2 \\
			&\leq& |\hat G_2 - G_2|_2 + \rho_{n,2}^\Upsilon(|\hat G_2 - G_2|_2+\omega_2^{1/2}),
		\end{eqnarray*}
		where we have applied the results in Lemma \ref{rate_Ups} (where the rate of $\rho_{n,2}^\Upsilon$ is defined) in the last inequality. In particular, the rates of $|\hat{G}_2 - G_2|_2\lesssim_\P\rho_{n,2}^{G_2}$ and $|(\hat\Xi- \Xi^0)|_2\lesssim_\P\rho_{n,2}^\Xi$ can be derived similarly as in Lemma \ref{sparse} and \ref{rate_Pi}, with the same assumptions with respect to $G_2$ instead of $G_1$.
		%Similar to $G_2$ we can put the same assumptions on $G_1$ then we have following the above derivation steps,  $ |\hat{G}_1 - G_1|_2 \lesssim_p \delta_{G_2}$, $|\hat{G}_1 \hat{\Omega}^{-1} - {G}_1{\Omega}^{-1}|_2 \lesssim_p \delta_{G_2}+ (\delta_{G_2}+ \sqrt{s} )r_{\Omega} $, and $|(\hat{G}_1^{\top} \hat{\Omega}^{-1} \hat{G}_1)^{-1}- ({G}_1^{\top} {\Omega}^{-1} {G}_1)^{-1}|_2\lesssim (4 e_n M)^{1-r}s$.
		
		In Remark \ref{rateGU}, we discuss the rate of %$\rho_{n}^{G_1}$, $\rho_{n,2}^{G_1}$, and 
		$\rho_{n,2}^U$, using the same specific example employed to analyze the admissible rates of %$\ell_n^\Upsilon$ 
		the tuning parameters in Remark \ref{rateUps}.

		%		\begin{remark}
		%			Let $\omega_2$ be a constant, then we  can see that $|\hat{G}_2 - G_2|_2\lesssim_p ( n^{-1/2}(\log P)^{1/2})^{1-l} s $ and $|(\hat\Xi- \Xi^0)|_2\lesssim_p ( n^{-1/2}(\log P)^{1/2})^{1-l} s $. Same rate holds for $|\hat U - U|_2 \lesssim_p ( n^{-1/2}(\log P)^{1/2})^{1-l} s $.
		%		\end{remark}	
		
		\begin{lemma}[Rates of $\ell_n^\Pi$ and $\rho_{n,2}^F$]\label{ratef}
			Under the conditions of Lemma \ref{rate_Ups} and \ref{sparse}, assume that there exists a constant $C>0$ such that $|\Upsilon^0|_2 \leq C$. Additionally, given that $|G_1|_1 \vee |G_{1}|_{\infty}\leq \mu$, $|G_1|_{\max}\leq\bar\mu$, and $|G_1|_2^2 \leq\omega_1$, we have
			$$|\hat{G}_{1}^\top\hat\Upsilon\hat{G}_{1}- G_1^{\top} \Upsilon^0G_1|_{\max}\lesssim_\P \ell_n^\Pi/M,$$
			given $\rho_{n}^{G_1} (\rho_{n,2}^\Upsilon + M) (\rho_{n,2}^{G_1}+ \mu)+ \mu \rho_{n,2}^\Upsilon (\bar\mu+ \rho_{n}^{G_1} )+\mu M\rho_{n}^{G_1}\leq \ell_n^\Pi/M$. Moreover, we have
			$$|\hat{F}- F|_{2}\lesssim_\P\rho_{n,2}^F,$$
			provided $\big[\{\rho_{n,2}^{G_1}+(\omega_1^{1/2} + \rho_{n,2}^{G_1})\rho_{n,2}^\Upsilon\}(\rho_{n,2}^U + \omega_1^2) + \omega_1^{1/2}\rho_{n,2}^U\big](\omega_1^{1/2} + \rho_{n,2}^{G_1}) + \omega_1^{1/2}\rho_{n,2}^{G_1}\leq\rho_{n,2}^F$.
			% $ (\rho_{n,2}^{G_1}+ \omega_1^{1/2})^2\rho_{n,2}^U+ (\rho_{n,2}^{G_1}+ \omega_1^{1/2})^2\rho_{n,2}^{\Upsilon}\rho_{n,2}^U +\omega_1^{1/2}\rho_{n,2}^{G_1}\leq\rho_{n,2}^F$.%$(\rho_{n,2}^{G_1}+ \omega_1^{1/2})^2\rho_{n,2}^U +\omega_1^{1/2}\rho_{n,2}^{G_1}\leq\rho_{n,2}^F$.
			%$$|\hat{A}- A|_{\infty} \lesssim_p s (r_{A,\max})^{1-r}.$$
		\end{lemma}
		
		\begin{proof}
			%Next, we show the rate of $e_n$. We assume $|G_2|_1 \vee |G_{2}|_{\infty}\leq M_{G_2}$, also $|G_2|_{\max} \leq M_{2,\max}$.
			By applying the results from Lemma \ref{rate_Ups} and Lemma \ref{sparse}, we have the following bound:
			\begin{eqnarray*}
				&& |\hat{G}_{1}^\top\hat\Upsilon\hat{G}_{1}- G_1^{\top} \Upsilon^0G_1|_{\max}\\
				&\leq &|\hat{G}_1- G_1|_{\max} (|\hat\Upsilon - \Upsilon^0|_1+ | \Upsilon^0|_1)(|\hat{G}_1-G_1|_1+ |G_1|_1) \\
				&&+ |G_1^\top|_{\infty}|\hat\Upsilon - \Upsilon^0|_{\infty}(|G_1|_{\max}+ |\hat{G}_{1}-G_1|_{\max})+ |G_1^\top|_{\infty}| \Upsilon^0|_{\infty}|\hat{G}_{1}- G_1|_{\max}\\
				&\lesssim_\P& \rho_{n}^{G_1} (\rho_{n,2}^\Upsilon + M) (\rho_{n,2}^{G_1}+ \mu)+ \mu \rho_{n,2}^\Upsilon (\bar\mu+ \rho_{n}^{G_1} )+\mu M\rho_{n}^{G_1}\\
				&\leq& \ell_n^\Pi/M.
			\end{eqnarray*}
			
			Next, recall that $F=-G_1^\top \Upsilon^0G_2\Xi^0 G_2^\top\Upsilon^0G_1=-G_1^\top\Upsilon^0UG_1$, and a regularized estimator is given by $\hat F=-\hat G_1^\top\hat\Upsilon\hat U\hat G_1$. Again, applying the results from Lemma \ref{rate_Ups} and Lemma \ref{sparse}, we obtain:
			\begin{eqnarray*}
				&&|\hat{F}- F|_{2} \\&= &|\hat G_1^\top\hat\Upsilon\hat U\hat G_1 - G_1^\top\Upsilon^0UG_1|_2\\
				&\leq & |\hat{G}_1^\top\hat\Upsilon- G_1^{\top} \Upsilon^0|_2|\hat{U} - U|_2 |\hat{G}_{1} |_2  + |\hat{G}_1^\top\hat\Upsilon- G_1^{\top} \Upsilon^0|_2|U|_2 |\hat{G}_{1} |_2  \\
				&& +|G_1^{\top}\Upsilon^0|_2|\hat{U}-U|_2 |\hat{G}_{1} |_2  +|G_1^{\top}\Upsilon^0U|_2|\hat{G}_{1} - G_1|_2\\
				&\leq& \big\{|\hat{G}_{1}- G_1|_2|\Upsilon^0|_2+ (|\hat{G}_{1} - G_1|_2 + |G_1|_2)|\hat\Upsilon- \Upsilon^0|_{2}\big\}|\hat{U} - U|_2 (|\hat{G}_{1} - G_1|_2 + |G_1|_2) \\
				&& + \big\{|\hat{G}_{1}- G_1|_2|\Upsilon^0|_2+ (|\hat{G}_{1} - G_1|_2 + |G_1|_2)|\hat\Upsilon- \Upsilon^0|_{2}\big\}|U|_2 (|\hat{G}_{1} - G_1|_2 + |G_1|_2) \\
				&&+ |G_1^{\top}\Upsilon^0|_2|\hat{U}-U|_2 |(|\hat{G}_{1} - G_1|_2 + |G_1|_2) + |G_1^{\top}\Upsilon^0U|_2|\hat{G}_{1} - G_1|_2\\
				&\lesssim_\P& \big[\{\rho_{n,2}^{G_1}+(\omega_1^{1/2} + \rho_{n,2}^{G_1})\rho_{n,2}^\Upsilon\}(\rho_{n,2}^U + \omega_1^2) + \omega_1^{1/2}\rho_{n,2}^U\big](\omega_1^{1/2} + \rho_{n,2}^{G_1}) + \omega_1^{1/2}\rho_{n,2}^{G_1}.
				% &\lesssim_\P&  (\rho_{n,2}^{G_1}+ \omega_1^{1/2})^2\rho_{n,2}^U+ (\rho_{n,2}^{G_1}+ \omega_1^{1/2})^2\rho_{n,2}^{\Upsilon}\rho_{n,2}^U +\omega_1^{1/2}\rho_{n,2}^{G_1}.
				%&\lesssim& (\rho_{n,2}^{G_1}+ \omega_1^{1/2})^2\rho_{n,2}^U +\omega_1^{1/2}\rho_{n,2}^{G_1}.
			\end{eqnarray*}
		\end{proof}
		
		% A more detailed discussion on the rates of $\ell_n^\Pi$ and $\rho_{n,2}^F$ under a specific example (continuing from Remark \ref{rateUps}) will follow in Remark \ref{rateGU}.%\ref{rateF}.
		A more detailed discussion on the rate of $\rho_{n,2}^F$ under a specific example (building on Remark \ref{rateUps} regarding the admissible rate of $\ell_n^\Pi$) will be provided in Remark \ref{rateGU}.
		
		So far, we have analyzed the rates of $\ell_n^\Upsilon$, $\ell_n^\Pi$, and $\rho_{n,2}^F$, which contribute to the rate of $|r_n|_\infty$. A concluding remark on the detailed rate of $|r_n|_\infty$ for the certain example is provided in Remark \ref{rnfinal}.
		
		\subsubsection{Additional Remarks}
		
		\begin{remark}[Rate of $|\hat B - B|_{\infty}$]\label{rateb}
			% The rate of $|\hat B - B|_{\infty}$ shall follow similarly once we have dealt with the rate of $|\hat V - V|_\infty$, where $V\defeq (\mathbf I - F\Pi^0)^{-1}$ and $\hat V\defeq(\mathbf I - \hat F\hat\Pi)^{-1}$. In particular,
			Suppose $|\hat V - V|_{\max}\lesssim_\P\rho_n^V=\smallO(1)$. Analogous to Lemma \ref{sparse}, we obtain $|\hat V - V|_{\infty}\lesssim_\P\rho_{n,2}^V$, where $\rho_{n,2}^V = s(V)\rho_n^V$ if $|V|_0=s(V)$. Alternatively, if $(|V|_{\infty,l}\vee |V|_\infty) \leq \nu$ for some $0\leq l<1$, then $\rho_{n,2}^V = \nu (\rho_n^V)^{1-l}$. Furthermore, assuming that $\max\{|\Pi^0|_\infty, |F|_\infty,|V|_\infty\}\leq\nu$, %the results in Lemma \ref{rate_Pi} 
			similar steps as in the proof of Lemma \ref{df} yield $|\hat B - B|_{\infty}\lesssim_{\P}\nu^3(\rho_{n,2}^F\vee\rho_{n,2}^\Pi\vee\rho_{n,2}^V)\leq\nu^5(\rho_{n,2}^F\vee\rho_{n,2}^\Pi)=:\rho_{n,2}^B$, provided that $\rho_{n,2}^\Pi,\rho_{n,2}^V\to0$ as $n\to\infty$.
		\end{remark}
		
		\begin{remark}[Admissible rates of $\ell_n^\Upsilon$ and $\ell_n^\Pi$] \label{rateUps}
			Suppose $M \lesssim s$, and all dependence adjusted norms involved in $\gamma_n, \gamma'_n, \gamma_{n,1}, \gamma_{n,2}, d_{n,1}$ are bounded by constants. For the weak dependence case where $\varsigma> 1/2 - 1/r$, if $n^{-1/2+1/r}(\log P_n) =\bigO(1)$ for sufficiently large $r$, then $\gamma_n, \gamma'_n, \gamma_{n,1},\gamma_{n,2}\lesssim n^{-1/2}(\log P_n)^{1/2}$. Moreover, %according to Remark \ref{consrate}, 
			since $d_{n,1}\lesssim s^2n^{-1/2}(\log P_n)^{1/2}$, an admissible rate for $\ell_n^\Upsilon$ to satisfy the relevant condition in Lemma \ref{ln_ups} is given by $sn^{-1/2}(\log P_n)^{1/2}$, provided that $d_{n,1}\to0$ as $n\to\infty$.
			
			Applying Lemma \ref{rate_Ups}, under this rate we have: $\rho_n^{\Upsilon} \lesssim s^{2} n^{-1/2}(\log P_n)^{1/2}$ and $\rho_{n,2}^\Upsilon \lesssim s^{3-2b}(n^{-1}\log P_n)^{(1-b)/2}$ for some $0\leq b<1$ and $s_0(q) \lesssim s$ such that $\Upsilon^0\in\mathcal{U}(b,s_0(q))$.
			
			According to Lemma \ref{sparse}, we have $\rho_{n}^{G_1} \lesssim n^{-1/2}(\log P_n)^{1/2}$. Assume that $L \lesssim s$ and $s(G_1)\lesssim s$. It follows that $\rho_{n,2}^{G_1}\lesssim s(n^{-1}\log P_n)^{(1-l)/2}$, where $l=0$ for the sparse case. Assume that $\mu \lesssim s$ and $\bar{\mu}$ is bounded by a constant. An admissible rate for $\ell_n^\Pi$ to satisfy the relevant condition in Lemma \ref{ratef} is given by $s^{6-2b}(n^{-1}\log P_n)^{(1-b)/2}$, %$s^4n^{-1/2}(\log P_n)^{1/2}$, 
			provided that $\rho_{n,2}^{G_1},\rho_{n,2}^\Upsilon\to0$ as $n\to\infty$. 
			
			Applying Lemma \ref{rate_Pi}, under this rate we obtain: 
			$\rho_n^\Pi \lesssim s^{7-2b}(n^{-1}\log P_n)^{(1-b)/2}$ %s^5n^{-1/2}(\log P_n)^{1/2}$ 
			and $\rho_{n,2}^\Pi \lesssim s^{(7-2b)(1-b)+1}(n^{-1}\log P_n)^{(1-b)^2/2}$ %s^{6-5b}(n^{-1}\log P_n)^{(1-b)/2}$, 
			for some $0\leq b<1$ and $s_0(K^{(1)}) \lesssim s$, such that $\Pi^0\in\widetilde{\mathcal{U}}(b,s_0(K^{(1)}))$.
		\end{remark}
		
		\begin{remark}[Discussion of the rates of %$\rho_{n}^{G_1}$, $\rho_{n,2}^{G_1}$, and 
			$\rho_{n,2}^U$ and $\rho_{n,2}^F$]\label{rateGU}
			Consider the special case discussed in Remark \ref{rateUps}. 
			% Here, we have $\rho_{n}^{G_1} \lesssim n^{-1/2}(\log P_n)^{1/2}$. Assume that $L \lesssim s$ and $s(G_1)\lesssim s$. It follows that $\rho_{n,2}^{G_1}\lesssim s(n^{-1}\log P_n)^{(1-l)/2}$, where $l=0$ for the sparse case. 
			As a continuation, under analogous assumptions, we have $\rho_{n,2}^{G_2}\lesssim s(n^{-1}\log P_n)^{(1-l)/2}$ (with $l=0$ for the sparse case)  
			% Assume that $\mu \lesssim s$ and $\bar{\mu}$ is bounded by a constant. As a continuation of Remark \ref{rateUps}, an admissible rate for $\ell_n^\Pi$ is given by $s^4n^{-1/2}(\log P_n)^{1/2}$, provided that $\rho_{n,2}^{G_1},\rho_{n}^\Upsilon\to0$ as $n\to\infty$.
			%Additionally, by similar reasoning as in Lemma \ref{rate_Pi}, we have %$\rho_{n,2}^{G_2}\lesssim s(n^{-1}\log P_n)^{(1-l)/2}$ ($l=0$ for the sparse case),  %$\rho_{n,2}^{G_2}\lesssim_\P s n^{-1/2}(\log P_n)^{1/2}$ for the sparse case, $\rho_{n,2}^{G_2}\lesssim s(n^{-1}\log P_n)^{(1-l)/2}$ for the dense case,
			and $\rho_{n,2}^{\Xi}\lesssim s^{(7-2b)(1-b)+1}(n^{-1}\log P_n)^{(1-b)^2/2}$. %s^{6-5b}(n^{-1}\log P_n)^{(1-b)/2}$. 
			Suppose $\omega_2$ is a constant and $l\leq b$. It follows that $\rho_{n,2}^U\lesssim s^{(7-2b)(1-b)+1}(n^{-1}\log P_n)^{(1-b)^2/2}$, %s^{(7-2b)(1-b)+1}(n^{-1}\log P_n)^{(1-b)^2/2}$, %s^{6-5b}(n^{-1}\log P_n)^{(1-b)/2}$, 
			under the condition that $\rho_{n,2}^{G_2},\rho_{n,2}^\Upsilon,\rho_{n,2}^\Xi\to0$ as $n\to\infty$.
			% 
			% Employing the admissible rate of $\ell_n^\Pi$ discussed above and applying Lemma \ref{rate_Pi}, we obtain: 
			% $\rho_n^\Pi \lesssim s^5n^{-1/2}(\log P_n)^{1/2}$ and $\rho_{n,2}^\Pi \lesssim s^{6-5b}(n^{-1}\log P_n)^{(1-b)/2}$, for some $0\leq b<1$ and $s_0(K^{(1)}) \lesssim s$, such that $\Pi^0\in\widetilde{\mathcal{U}}(b,s_0(K^{(1)}))$.
			% \end{remark}
			%\begin{remark}[Discussion of the rate of $\rho_{n,2}^F$]\label{rateF}
			Moreover, suppose $\omega_1$ is a constant and $l\leq b$. %As a continuation of Remark \ref{rateG} and \ref{rateU}, 
			In this case, we have $\rho_{n,2}^F\lesssim s^{(7-2b)(1-b)+1}(n^{-1}\log P_n)^{(1-b)^2/2}$ %s^{6-5b}(n^{-1}\log P_n)^{(1-b)/2}$
			%$\rho_{n,2}^F\lesssim s(n^{-1}\log P_n)^{(1-l)/2}$ (where $l=0$ for the sparse case), 
			provided that $\rho_{n,2}^{G_1},\rho_{n,2}^\Upsilon,\rho_{n,2}^U\to0$ as $n\to\infty$.
		\end{remark}

		\begin{remark}[Convergence of $\hat A$]\label{ratea}
			Recall that $A={G}_1^{\top}{\Omega}^{-1}(\mathbf I - {G}_2P({\Omega},{G}_2))=G_1^\top\Upsilon^0(\mathbf I - U)$, and we consider the regularized estimator $\hat A = \hat G_1^\top\hat\Upsilon(\mathbf I - \hat U)$. Given that $|\Upsilon^0|_2 \leq C$ and $|G_1|_2^2 \leq\omega_1$, by applying the results from Lemma \ref{rate_Ups}, we obtain:
			\begin{eqnarray*}
				|\hat{A}- A|_{\max} &\leq& |\hat G_1^\top\hat\Upsilon(\mathbf I - \hat U)- G_1^\top\Upsilon^0(\mathbf I - U)|_{2}\\
				&\leq& |\hat{G}_1^\top\hat\Upsilon - {G}_1^\top\Upsilon^0|_2  +|\hat U - U|_2|G_1^\top\Upsilon^0|_2 + |\hat{G}_1^\top\hat\Upsilon - {G}_1^\top\Upsilon^0|_2(|\hat U - U|_2+1)\\
				&\leq& |\hat{G}_{1}- G_1|_2|\Upsilon^0|_2+ (|\hat{G}_{1} - G_1|_2 + |G_1|_2)|\hat\Upsilon- \Upsilon^0|_{2} +|\hat U - U|_2|G_1^\top\Upsilon^0|_2 \\
				&&+ \{|\hat{G}_{1}- G_1|_2|\Upsilon^0|_2+ (|\hat{G}_{1} - G_1|_2 + |G_1|_2)|\hat\Upsilon- \Upsilon^0|_{2}\}(|\hat U - U|_2+1)\\
				&\lesssim_\P& %\rho_{n,2}^{G_1}+ \omega_1^{1/2}\rho_{n,2}^\Upsilon + \omega_1^{1/2}\rho_{n,2}^U =: \rho_n^A.
				\{\rho_{n,2}^{G_1}+(\rho_{n,2}^{G_1}+\omega_1^{1/2})\rho_{n,2}^{\Upsilon}\}(\rho_{n,2}^U+2)+\omega_1^{1/2}\rho_{n,2}^U=: \rho_n^A.
			\end{eqnarray*}
			
			Analogous to Lemma \ref{sparse}, we have $|\hat{A}- A|_{\infty} \lesssim_\P \rho_{n,2}^A$, with $\rho_{n,2}^A = s(A)\rho_n^A$ if we assume $|A|_0=s(A)$, and $\rho_{n,2}^A = \iota (\rho_n^A)^{1-l}$ in the case of $(|A|_{\infty,l}\vee |A|_\infty) \leq \iota$ for some $0\leq l<1$.
		\end{remark}
		
		\begin{remark}[Detailed rate of $|r_n|_\infty$]\label{rnfinal}
			%Continuing to Remarks \ref{cont_rate} and \ref{consrate}, 
			In Remarks \ref{rateUps}-\ref{rateGU} above, we set up a special case where all the dependence adjusted norms involved are bounded by constants and specifically discuss the relevant rates of $\ell_n^\Upsilon, \ell_n^\Pi, %\rho_n^{G_1}, \rho_{n,2}^{G_1},
			\rho_{n,2}^U$, and $\rho_{n,2}^F$. Summarizing all the results, we obtain:
			$$\rho_{n}^B\lesssim s^{7-2b}(n^{-1}\log P_n)^{(1-b)/2}+s^{(7-2b)(1-b)+1}(n^{-1}\log P_n)^{(1-b)^2/2},$$
			%s^5n^{-1/2}(\log P_n)^{1/2} + s^{6-5b}(n^{-1}\log P_n)^{(1-b)/2},$$
			which implies that 
			$$\varrho_{n,1}\lesssim \big\{s^{7-2b}(n^{-1}\log P_n)^{(1-b)/2}+s^{(7-2b)(1-b)+1}(n^{-1}\log P_n)^{(1-b)^2/2}\big\}s^2n^{-1/2}(\log P_n)^{1/2},$$
			% $$\varrho_{n,1}\lesssim \big\{s^5n^{-1/2}(\log P_n)^{1/2} + s^{6-5b}(n^{-1}\log P_n)^{(1-b)/2}\big\}(s+2)sn^{-1/2}(\log P_n)^{1/2},$$
			for some $0\leq b<1$, given that $\kappa$ and $\omega$ are constants.
			
			Additionally, suppose that $(\nu\vee\iota)\lesssim s$. By Remark \ref{rateb}, it follows that $$\rho_{n,2}^B=\nu^5(\rho_{n,2}^F\vee\rho_{n,2}^\Pi)\lesssim s^{(7-2b)(1-b)+6}(n^{-1}\log P_n)^{(1-b)^2/2},$$ %s^{11-5b}(n^{-1}\log P_n)^{(1-b)/2}$, 
			given that $\rho_{n,2}^\Pi,\rho_{n,2}^V\to0$ as $n\to\infty$. Moreover, according to Remark \ref{ratea}, we obtain: $$\rho_{n}^A\leq\rho_{n,2}^{G_1}+ \omega_1^{1/2}\rho_{n,2}^\Upsilon + \omega_1^{1/2}\rho_{n,2}^U \lesssim s^{(7-2b)(1-b)+1}(n^{-1}\log P_n)^{(1-b)^2/2},$$ %s^{6-5b}(n^{-1}\log P_n)^{(1-b)/2},$$
			where it is assumed that $\omega_1$ is a constant and $l\leq b$. In the sparse case where $s(A)\lesssim s$, we have: 
			$$\rho_{n,2}^A\lesssim s^{(7-2b)(1-b)+2}(n^{-1}\log P_n)^{(1-b)^2/2},$$ %s^{7-5b}(n^{-1}\log P_n)^{(1-b)/2},$$ 
			provided that $\rho_{n,2}^{G_1},\rho_{n,2}^U,\rho_{n,2}^\Upsilon\to0$ as $n\to\infty$. Finally, we get: 
			$$\varrho_{n,2}\lesssim s^{(7-2b)(1-b)+7}(n^{-1}\log P_n)^{(b^2-2b+2)/2},$$ %s^{12-5b}(n^{-1}\log P_n)^{(2-b)/2},$$
			given that $\rho_{n,2}^{B}\to0$ as $n\to\infty$.
		\end{remark}
		
		%\subsection{Proofs and Technical Details of Section \ref{nonlinearmo}}%-\ref{splitting}}
		
		%\begin{proof}[\textbf{Proof of Lemma \ref{cont1}}]
		%	The proof is similar to that of Lemma \ref{cont} and thus is omitted.
		%\end{proof}
		
		\renewcommand{\thesubsection}{B.\arabic{subsection}}
		\setcounter{equation}{0}
		\renewcommand{\theequation}{B.\arabic{equation}}
		\setcounter{theorem}{0}
		\renewcommand{\thetheorem}{B.\arabic{theorem}}
		\setcounter{lemma}{0}
		\renewcommand{\thelemma}{B.\arabic{lemma}}
		\setcounter{figure}{0}
		\renewcommand{\thefigure}{B.\arabic{figure}}
		\setcounter{table}{0}
		\renewcommand{\thetable}{B.\arabic{table}}
		\setcounter{remark}{0}
		\renewcommand{\theremark}{B.\arabic{remark}}
		\setcounter{corollary}{0}
		\renewcommand{\thecorollary}{B.\arabic{corollary}}
		\setcounter{example}{0}
		\renewcommand{\theexample}{B.\arabic{example}}
		\setcounter{assumption}{0}
		\renewcommand{\theassumption}{B.\arabic{assumption}}

\section{Extension to Nonlinear Moments}\label{nonlinearmo}
In the spatial statistics literature, researchers often use a combination of linear and quadratic moments to relax identification conditions, as exemplified by Lemma EX1 in \cite{kuersteiner2020dynamic}, though our model differs due to the presence of heterogeneous parameters.
In this section, we will explore the scenario where the moment conditions do not take a simple linear form. 
Specifically, we will examine the necessary conditions to establish the consistency of the GDS estimator in a more general setting and analyze the linearization error associated with the inference of the debiased estimator. 

%\subsection{Extension of Section \ref{theoretical}.}
\subsection{Consistency and Concentration}\label{nonlinear.con}
To establish the consistency of the GDS estimator $\hat\theta$ for the general form of moments, two key assumptions are required, which follow directly from \cite{belloni2018high}. 
Recall that $\mathcal{R}(\theta^0) \defeq \{\theta \in \Theta: |\theta|_1 \leq |\theta^0|_1\}$ denotes the restricted set. Let $\epsilon_n\downarrow0,\delta_n\downarrow0$ be sequences of positive constants.

\begin{itemize}
	\item[(C1)]\label{C1}(Concentration)
	\begin{equation*}%\label{uniform}
	\sup_{\theta \in \mathcal{R}(\theta^0)}|\hat g(\theta) - g(\theta)|_{\infty}\leq \epsilon_n
	\end{equation*}
	holds with probability at least $1-\delta_n$.
	\item[(C2)]\label{C2}(Identification)
	The target moment function $g(\cdot)$ satisfies the identification condition:
	$$\{|g(\theta)- g(\theta^0)|_{\infty} \leq \epsilon, \theta\in\mathcal R(\theta^0)\} \text{ implies } |\theta^0 - \theta |_a \leq \rho(\epsilon;\theta^0, a),$$
	for all $\epsilon>0$, $a=1$ or $2$, where $\epsilon\mapsto \rho(\epsilon;\theta^0,a)$ is a weakly increasing function mapping from $[0,\infty)$ to $[0,\infty)$ such that $\rho(\epsilon;\theta^0, a)\to0$ as $\epsilon\to0$.
\end{itemize}

\hyperref[C1]{(C1)} and \hyperref[C2]{(C2)} are high-level conditions that need to be verified for our analysis. 
%The verification for the linear moments case has been demonstrated in Section \ref{const} in the main text. 
%In linear moment case, \hyperref[C1]{(C1)} and \hyperref[C2]{(C2)} take more specific form and can be verified as demonstrated in Section \ref{const} in the main text.
%For the nonlinear case, the concentration condition is dealt with in the subsequent subsection  \ref{nonlinear.con}.
Specifically, we provide a concentration inequality for nonlinear moments in the following subsection. 
However, the issue of identification is more complex, and we leave the detailed verification of this aspect to future research.

%\subsubsection{Concentration}\label{nonlinear.con}
% We need a concentration inequality as in Lemma \ref{cont} for the nonlinear case. 
%\subsection{Tail Probability for Nonlinear Moments} \label{nonlinear}
%We now show the consistency of the estimator as in Section \ref{const} under nonlinear moments.
%Let $D_{j,t}=(y_{j,t},\tilde x_{j,t}^\top,z_{j,t}^\top)^\top$, $D_t=[D_{j,t}]_{j=1}^p\in\R^{p+K+q}$, and $\theta=(\vartheta_1^\top,\ldots,\vartheta_{\bar u}^\top)^\top\in\R^K$, 
Let $g_j(D_{t},\theta)=[g_{jm}(D_t,\theta)]_{m=1}^{q_j}$ represent the score functions for $j=1,\ldots,p$, where $D_t=[D_{j,t}]_{j=1}^p\in\R^{\bar d}$ collects the observed data sample for all individuals. Additionally, decompose the parameter vector $\theta\in\R^K$ into $\theta=(\vartheta_1^\top,\ldots,\vartheta_{\bar u}^\top)$, 
where each $\vartheta_u$ is $K^u$-dimensional subvector of $\theta$ for $u=1,\ldots,\bar u$, and $K=K^1+\cdots+K^{\bar u}$. For $j=1,\ldots,p$ and $m=1,\ldots,q_j$, we assume the score functions have the index form:
$$g_{jm}(D_t,\theta)=h_{jm}(D_t,v_{jm,t})=h_{jm}(D_t,W_{u(j,m)}(D_t)^\top\vartheta_{u(j,m)}),$$% \quad u(j,m)=1,\ldots,\bar u,$$
where $u(j,m)$ ranges over $1,\ldots,\bar u$, $h_{jm}(\cdot,\cdot)$ is a measurable map %from $\R^{p+K+q}\times\R$ to $\R$, 
from $\R^{\bar d}\times\R$ to $\R$,  
and $W_u(\cdot)$ is a measurable map %from $\R^{p+K+q}$ to $\R^{K^u}$, 
from $\R^{\bar d}$ to $\R^{K^u}$ 
for all $u=1,\ldots,\bar u$. The true parameter $\theta^0$ is identified as unique solution to the moment conditions:
$$\E\{g_{jm}(D_t,\theta^0)\}=\E\{h_{jm}(D_t,W_{u(j,m)}(D_t)^\top\vartheta^0_{u(j,m)})\}=0.$$
We assume that $|\vartheta^0_{u(j,m)}|_0=s_{j,m}$ and $|\theta^0|_0=\sum_{j=1}^p\sum_{m=1}^{q_j}s_{j,m}=s\ll K$.
%Let $W_t = \{Y_t, X_t, \cup_j Z_{u(j)t}\}$, $w_{jt} = (y_{jt},x_{jt}, z_{u(j)t}^{\top})$, and $W_{jt} = (Y_{jt},X_{jt},Z_{u(j)t}^{\top}) $.
%We let $h_{jt}(w_{jt}, v_{jt}) = h_{j}(y_{jt},x_{jt},z_{u(j)t}^{\top}\theta_{u(j)})$ and $v_{jt} \defeq z_{u(j)t}^{\top}\theta_{u(j)}$, where $u(j) \in [\overline{u}]$ for $j \in [m]$ ($[\overline{u}]$ is a set of indices for each equation $j$, and $m$ is the number of equations).

%Let $(\int_{R} h^2_{j}(w_{jt},v_{jt})d Q(w_{jt},v_{jt}))^{1/2} \defeq \|h_{j}(w_{jt},v_{jt})\|_{2,Q}$ for any measure $Q$.
%Let $s$ be an integer, with $s\ll q$.
To simplify the notations, we suppress the index pair $(j,m)$, where $j=1,\ldots,p$, $m=1,\ldots,q_j$, to the single index $j=1,\ldots,q$ ($q=\sum_{j=1}^pq_j$) thereafter. Accordingly, we define the function class: 
%$$\mathcal{H}_{jm} = \big\{d\mapsto h_{jm}(d,W_{u(j,m)}(d)^\top\vartheta_{u(j,m)}): \max_{j,m}|\vartheta_{u(j,m)}-\vartheta_{u(j,m)}^0|_1\leq c, |\vartheta_{u(j,m)}^0|_0\leq s\big\},$$
$$
\mathcal{H}_{j} = \big\{d\mapsto h_{j}(d,W_{u(j)}(d)^\top\vartheta_{u(j)}): |\vartheta_{u(j)}-\vartheta_{u(j)}^0|_1\leq c_j\big\},
$$
where $c_j$ can be chosen as 1 without loss of generality.

%$\mathcal{H}_j = \{h_{j}(Y_{jt},X_{jt}, Z_{u(j)t}^{\top} \theta_{u(j)}): \max_j |\theta_{u(j)}- \theta_{u(j)}^0|_1\leq 1, |\theta_{u(j)}^0|_0 \leq s_j \} $. ($1$ can be chosen as a constant without loss of generality.)
%\begin{assumption} \label{A4.1}
%%The function class $\mathcal{H}_{j}$ is enveloped with $$\max_{1\leq j\leq q}\sup_{\vartheta_{u(j)}}|h_{j}(d,W_{u(j)}(d)^\top\vartheta_{u(j)})|\leq H(d).$$
%%%$\max_j \sup_{\theta_{u(j)}}|h_j(w_{jt},v_{jt})| \leq {H}_j(w_{jt})$.
%Assume that the function $h_{j}(d,v)$
%%$h_{j}(w_{jt}, v_{jt})$
%is Lipschitz in the sense that, for all $v,v^\prime\in\R$,
%%$w_{jt}, v_{jt},v'_{jt}$,
%$$
%%|h_{j}(w_{jt}, v_{jt}) - h_{j}(w_{jt}, v'_{jt})|\leq L_j(w_{jt}) |v_{jt}-v'_{jt}|_1,
%|h_{j}(d,v)-h_{j}(d,v^\prime)|\leq L_{j}(d)|v-v^\prime|,
%$$
%%with a Lipschitz function $L_j(w_{jt})$.
%where $L_{j}$ is a measurable map from $\R^{p+K+q}$ to $\R_{+}$.
%\end{assumption}

Within the context of this section, we  focus on cases involving sub-exponential or sub-Gaussian tails. Specifically, we define the dependence adjusted sub-exponential ($\nu=1$) or sub-Gaussian ($\nu=1/2$) norms as:
% $$\|h_{j}(D_{\cdot},v_{j,\cdot})\|_{\psi_\nu,\varsigma}\defeq\sup_{r\geq2}r^{-\nu}\|h_{j}(D_{\cdot},v_{j,\cdot})\|_{r,\varsigma}<\infty.$$
$$\|h_{j}(D_{\cdot},W_{u(j)}(D_{\cdot})^\top\vartheta^0_{u(j)})\|_{\psi_\nu,\varsigma}\defeq\sup_{r\geq2}r^{-\nu}\|h_{j}(D_{\cdot},W_{u(j)}(D_{\cdot})^\top\vartheta^0_{(j)})\|_{r,\varsigma},$$
which we assume to be finite for any $j=1,\ldots,q$.
Note that the following results can be generalized to finite moment conditions by applying Nagaev-type inequalities (e.g., Theorem 2 of \citet{wu2016performance}) in place of Lemma \ref{exp}.
%for exponential tails as  $\|h_{j}(w_{j.},v_{j.})\|_{\psi_r,\alpha} = \sup_{q\geq 2}q^{-r}\|h_{j}(w_{j.},v_{j.})\|_{q,\alpha}$.

Observe that
\begin{eqnarray*}
	&&\E{_n} h_j(D_t,v_{j,t}) - \E h_j(D_t,v_{j,t}) \\
	&=& \E{_n} h_j(D_t,v_{j,t}) - \E{_n}\E\{h_j(D_t,v_{j,t})|\mF_{t-1}\} + \E{_n}\E\{h_j(D_t,v_{j,t})|\mF_{t-1}\}  - \E h_j(D_t,v_{j,t})\\
%	&&\E{_{n}}h_{j.}(w_{j.},v_{j.}) - \E{_n}\E h_{j.}(w_{j.},v_{j.})\\
%	& = &\E{_{n}}h_{j.}(w_{j.},v_{j.}) - \E{_{n}} \E(h_{j.}(w_{j.},v_{j.})|\mathcal{F}_{.-1})
%	\\&&+ \E{_{n}} \E (h_{j.}(w_{j.},v_{j.})|\mathcal{F}_{.-1})- \E{_n}\E h_{j.}(w_{j.},v_{j.}),\\
	&=:& L_{n,1} + L_{n,2},
\end{eqnarray*}
where the first term $L_{n,1}=\E{_n} h_j(D_t,v_{j,t}) - \E{_n}\E\{h_j(D_t,v_{j,t})|\mF_{t-1}\} $ is a summand of martingale differences and the second term $L_{n,2}=\E{_n}\E\{h_j(D_t,v_{j,t})|\mF_{t-1}\}  - \E h_j(D_t,v_{j,t})$ shall be dealt via chaining steps.

We begin by deriving the bound for $L_{n,2}$. Let $\tilde h_j(D_t,v_{j,t})=\E\{h_j(D_t,v_{j,t})|\mF_{t-1}\} - \E h_j(D_t,v_{j,t})$ 
% \begin{align*}
%   %\tilde h_{j,t} =
%   \tilde h_j(D_t,v_{j,t})&=\E\{h_j(D_t,v_{j,t})|\mF_{t-1}\} - \E h_j(D_t,v_{j,t}) \\
%   &=\E\{h_j(D_t,W_{u(j)}(D_t)^\top\vartheta_{u(j)})|\mF_{t-1}\} - \E h_j(D_t,W_{u(j)}(D_t)^\top\vartheta_{u(j)}),
% \end{align*}
and define the function class
$$\widetilde{\mathcal H}_j=\big\{d\mapsto\tilde h_j(d,W_{u(j)}(d)^\top\vartheta_{u(j)}):|\vartheta_{u(j)}-\vartheta_{u(j)}^0|_1\leq 1\big\}.$$
%Similarly, let $\tilde h^0_{j,t}=\tilde h_j(D_t, W_{u(j)}(D_t)^\top\vartheta^0_{u(j)})$ denote the function values attained with the true parameters. 
%where $c_j$ can be chosen as 1 without loss of generality.
%$\tilde{h}_{jt} (w_{jt},v_{jt})\defeq \E (h_{j}(w_{jt},v_{jt})|\mathcal{F}_{t-1})- \E h_{j}(w_{jt},v_{jt})$.
%$\tilde{\mathcal{H}}_j = \{\tilde{h}_{jt}(w_{jt}, v_{jt}): \max_j |\theta_{u(j)}- \theta_{u(j)}^0|_1\leq 1, |\theta_{u(j)}^0|_0 \leq s_j \} $.

\begin{assumption}\phantomsection\label{A4.2}
	\begin{enumerate}
		\item[i)] The function class $\widetilde{\mathcal{H}}_{j}$ is enveloped with $$\max_{1\leq j\leq q}\sup_{\tilde h_j\in\widetilde{\mathcal{H}}_{j}}|\tilde h_j(d,W_{u(j)}(d)^\top\vartheta_{u(j)})|\leq \widetilde H(d).$$
		\item[ii)]
{Assume that $\tilde h_j(d,W_{u(j)}(d)^\top\vartheta_{u(j)})$ is differentiable with respect to $\vartheta_{u(j)}$.}
%Assume that $\tilde h_j(D_t,v_{j,t})$ is differentiable with respect to $v_{j,t}$.
Suppose the dependence adjusted norm of the derivative evaluated at the true parameters, i.e.
		$$\Psi_{j,\nu,\varsigma} = \| |\partial\tilde h_j(D_\cdot,W_{u(j)}(D_\cdot)^\top\vartheta_{u(j)}^0)/\partial\vartheta_{u(j)}|_\infty\|_{\psi_\nu,\varsigma}$$
		is finite.
		%Assume that $\th_j(w_{jt},v_{jt})$ is differential with the first derivative as  $\partial \tilde{h}_j(w_{jt},v_{jt})/\partial \theta_{u(j)}  $.
		%In addition we let the dependence adjusted norm of the first derivative at the true parameter value to exists and $$\gamma_{j,\alpha}^{'}  =  \||\partial \tilde{h}_{j.}(w_{j.},v^0_{j.})/\partial \theta_{u(j)}|_{\infty}  \|_{\psi_r,\alpha} $$ is bounded by a constant.
		Moreover, assume that the partial derivative of $\tilde h_j(d,v)$ with respect to the the second argument has an envelope. That is %with finite $r$-th moment.
	$$\max_{1\leq j\leq q}\sup_{\tilde h_j\in\widetilde{\mathcal{H}}_{j}}|\partial_v \tilde h_j(d,v)| \leq \widetilde H^1(d).$$
		%and $\E|H^1(D_t)|^r<\infty$.
		%We assume that the partial derivative $\partial \th_{jt}(w_{jt},v_{jt})/\partial v_{jt} $ has an envelope with finite moment. We denote the envelope as $\sup_{\theta_{u(j)}} \max_j |\partial \th_{jt}(w_{jt},v_{jt})/\partial v_{jt}| \leq H^{1}_{j,t}(w_{jt}) $. We assume that $H^{1}_{j,t}(w_{jt})$ has finite $r$ th moment.
		\item[iii)] Denote $c_r\defeq \E|\widetilde H(D_t)|^r \vee \E|\widetilde H^1(D_t)|^r \vee \E\Big(\max\limits_{1\leq j\leq q}|W_{u(j)}(D_t)|_\infty^r\Big)$ and assume that $c_rn^{-r/2+1}\to0$, for an integer $r>4$.
		%$c_q \defeq (\E|\max_{j,t}  H^{1}_{j,t}(w_{jt}) |^q \vee \E(\max_{j,t} |Z_{u(j),t}|^q_{\infty} )\vee \E(\max_j |H_{j,t}|^q))$.
		%And $c_q n^{-q/2+1} \to 0$.
	\end{enumerate}
\end{assumption}

Note that here the differentiability condition is imposed on $\tilde h_j(D_t,v_{j,t})$, %=\E\{h_j(D_t,v_{j,t})|\mF_{t-1}\} - \E h_j(D_t,v_{j,t})$ 
rather than $h_j(D_t,v_{j,t})$, as in the Condition ENM in \citet{belloni2018high}. This allows for greater generality, accommodating non-smooth score functions.

%We choose $M = \sqrt{n}c$, then we have $\sum_j ent( \mhH_{j,M}, \delta, |.|_{n,\infty}) \lesssim s \log(a_n/\delta) $, $a_n = p\vee {n} \vee e \vee s$.
%\begin{assumption}\label{4.3}
%%For any $\tilde h_j,\tilde h'_j\in\widetilde{\mathcal H}_j$ that satisfying $ \sup_{d,v,v'}|\tilde h_j(d,v)- \tilde{h}'_j(d,v')|\leq \delta$ for a positive constant $\delta$, assume the dependence adjusted norm of $|\tilde h_j(D_t,v_{j,t}) - \tilde h'_j(D_t,v'_{j,t})|$ is bounded by $\delta\Psi_{j,\nu,\varsigma}^h$.
%Given $|\vartheta_{u(j)}-\vartheta^\prime_{u(j)}|_1\leq \delta$, for a positive constant $\delta$, assume that the dependence adjusted norm of $|\tilde h_j(D_t,W_{u(j)}(D_t)^\top\vartheta_{u(j)}) - \tilde h^\prime_j(D_t,W_{u(j)}(D_t)^\top\vartheta^\prime_{u(j)})|$ is bounded by $\delta\Psi_{j,\nu,\varsigma}^h$, for any $\tilde h_j,\tilde h^\prime_j\in\widetilde{\mathcal H}_j$.
%\end{assumption}
%We define the dependence adjusted norm for $|\tilde{h}_{j.}(w_{jt},v_{jt})-\tilde{g}_{j.}(w_{jt},v_{jt})|$ can be bounded by
%$\gamma_{j,\alpha,k}^{\alpha} \leq \gamma_{j,\alpha}^{'\alpha} \delta 3*2^{-k}$, by assuming that $\gamma_{j,\alpha}^{'}$ is %bound by a constant.

For any finitely discrete measure $\mathcal Q$ on a measurable space, let $\mathcal L^r(\mathcal Q)$ denote the space of all measurable functions $h$ such that $\|h\|_{\mathcal Q,r}=(\mathcal Q|h|^r)^{1/r}<\infty$, where $\mathcal Qh\defeq\int h d\mathcal Q$, $1\leq r<\infty$, and $\|h\|_{\mathcal Q,\infty}=\lim\limits_{r\to\infty}\|h\|_{\mathcal Q,r}<\infty$. For a class of measurable functions $\mathcal H$, the $\delta$-covering number with respect to the $\mathcal L^r(\mathcal Q)$-metric is denoted as $\mathcal{N}(\delta, \mathcal H, \|\cdot\|_{\mathcal Q,r})$ and let $\operatorname{ent}_r(\delta, \mathcal H)= \log \sup_{Q}\mathcal{N}(\delta\|H\|_{\mathcal Q,r}, \mathcal H, \|\cdot\|_{\mathcal Q, r})$ denote the uniform entropy number with the envelope $H=\sup_{h\in\mathcal H}|h|$.

Given a truncation constant $M$, we define the event
$${\mathcal A}_M=\big\{\max_{1\leq t\leq n}|\widetilde H(D_t)|\leq M, \max_{1\leq t\leq n}|\widetilde H^1(D_t)|\leq M, \max_{1\leq j\leq q,1\leq t\leq n}|W_{u(j)}(D_t)|_\infty\leq M\big\}.$$
Accordingly, we define the function class $\widetilde{\mathcal{H}}_j$ on the event $\mathcal{A}_M $ to be $\widetilde{\mathcal H}_{j,M}$.
	\begin{assumption}\label{entropy}
	Consider the class of functions $\widetilde{\mathcal H}=\big\{d\mapsto\tilde h_j(d,W_{u(j)}(d)^\top\vartheta_{u(j)}): j=1,\ldots,q, |\vartheta_{u(j)}-\vartheta_{u(j)}^0|_1\leq 1\big\}$ and let $\widetilde{\mathcal H}_M$ be the function class  $\widetilde{\mathcal H}$ on the event ${\mathcal A}_M$. Assume the entropy number of the function class $\widetilde{\mathcal H}_{M}$ with respect to the $\mathcal L^2$-metric is bounded as $\operatorname{ent}_{2}(\delta, \widetilde{\mathcal H}_{M})\lesssim s \log(P_n/\delta) $, where $P_n = q\vee {n} \vee e$.
	\end{assumption}
We discuss the validity of Assumption \ref{entropy} in 
the following remark %the Appendix; see Remark \ref{assent} 
in the case of empirical metric.

\begin{remark}[Verification of Assumption \ref{entropy} under empirical norm]\label{assent}
		Define $\omega_{n}(\tilde h_j, \tilde h^{\prime}_j) = [\E_n\{\tilde h_j(D_t,W_{u(j)}(D_t)^\top\vartheta_{u(j)}) - \tilde h^\prime_j(D_t,W_{u(j)}(D_t)^\top\vartheta^\prime_{u(j)})\}^2]^{1/2}$.
		The $\delta$-covering number of the function class $\widetilde{\mathcal{H}}_{j,M}$ with respect to the $\omega_n(\cdot,\cdot)$ metric is denoted as $\mathcal{N}(\delta, \widetilde{\mathcal H}_{j,M}, \omega_n(\cdot,\cdot))$. Moreover, let $\operatorname{ent}_{n,2}(\delta, \widetilde{\mathcal H}_{j,M}) = \log \mathcal{N}(\delta|\widetilde H_{j,M}|_{n,2}, \widetilde{\mathcal H}_{j,M},\omega_n(\cdot,\cdot))$ denote the entropy number, where $\widetilde H_{j,M}=\sup_{h\in\widetilde{\mathcal H}_{j,M}}|h|$ (the envelope) and $|\widetilde H_{j,M}|_{n,2} = [\E_n\{\widetilde H_{j,M}(D_t)\}^2]^{1/2}$. 
		
		On the event $\mathcal{A}_M$, for any $\vartheta_{u(j)} $ belonging to $\varTheta_j= \{\vartheta_{u(j)} :|\vartheta_{u(j)} -\vartheta^0_{u(j)}|_1\leq 1\}$, there exists a $\vartheta^\prime_{u(j)}$ in the $\delta$-nets of  $\varTheta_j$ with respect to $|\cdot|_{1}$, %(denoted by $B_{\delta,s,j}$),
		such that
		\begin{eqnarray*}
			%		&&\max_t|\tilde h_j(D_t,W_{u(j)}(D_t)^\top\vartheta_{u(j)}) - \tilde h^\prime_j(D_t,W_{u(j)}(D_t)^\top\vartheta^\prime_{u(j)})|\\
			%		&\leq&\max_t |\widetilde H^1(D_t)||W_{u(j)}(D_t)^\top\vartheta_{u(j)} - W_{u(j)}(D_t)^\top\vartheta^\prime_{u(j)}|\\
			%		&\leq& M\max_t|W_{u(j)}(D_t)|_\infty|\vartheta_{u(j)}-\vartheta^\prime_{u(j)}|_1\\
			%		&\leq&M^2|\vartheta_{u(j)}-\vartheta^\prime_{u(j)}|_1\leq M^2\delta.
			\omega_{n}(\tilde h_j, \tilde h^{\prime}_j) &\leq&\max_{1\leq t\leq n} |\widetilde H^1(D_t)|[\E{_n}\{W_{u(j)}(D_t)^\top\vartheta_{u(j)} - W_{u(j)}(D_t)^\top\vartheta^\prime_{u(j)}\}^2]^{1/2}\\
			&\leq& M\max_{1\leq t\leq n}|W_{u(j)}(D_t)|_\infty|\vartheta_{u(j)}-\vartheta^\prime_{u(j)}|_1\\
			&\leq&M^2|\vartheta_{u(j)}-\vartheta^\prime_{u(j)}|_1\leq M^2\delta.
			%&&\max_t|\th_j(w_{jt},v_{jt}) - \th_j(w_{jt},v^0_{jt}) | \\
			%&&\leq |H^{1}_{j,t}(w_{jt})| |Z_{u(j)t}^{\top}\theta_{u(j)}-Z_{u(j)t}^{\top}\theta^0_{u(j)}|\\
			%&&\leq M|Z_{u(j)t}^{\top}|_{\infty}|\theta_{u(j)}-\theta^0_{u(j)}|_1\\
			%&&\leq M^2 \delta,\\
		\end{eqnarray*}
		It follows that %$\operatorname{ent}(\widetilde{\mathcal H}_{j,M}, \delta, |\cdot|_{n,\infty}) \leq \operatorname{ent}(\varTheta_{u(j)}, \delta/M^2, |\cdot|_{1})$.
		%$\operatorname{ent}(\delta, \widetilde{\mathcal H}_{j,M}) \leq \operatorname{ent}(\delta/M^2, \varTheta_{u(j)})$.\\
		\begin{align*}
		\mathcal N(\delta,\widetilde{\mathcal H}_{j,M},\omega_n(\cdot,\cdot)) &\lesssim_\P  \mathcal N(\delta/M^2,\varTheta_j,|\cdot|_1)\\
		&\leq{K^{u(j)}\choose s_j}(1+2M^2/\delta)^{s_j}\\
		&\lesssim (eK^{u(j)}/s_j)^{s_j}(1+2M^2/\delta)^{s_j},
		\end{align*}
		where the last inequality is implied by the Stirling formula.
		%{Note that we suppress the $j$ in $s_j$ in the paragraph.}
		%	Since the $\delta$-covering number of the $s$-sparse unit ball in terms of the $|\cdot|_{\infty}$ norm is bounded by $\frac{2!}{\{(s/e)^s \sqrt{2\pi s}\}}(2/\delta + 1)^s$.
		%	This is because the $l_1$ norm unit ball is of the volume $2^p \Gamma(1+ 1/l)^p/\Gamma(1+ p/l)$, where $\Gamma(\cdot)$ is the Gamma function. Also see Lemma 5.7 in \cite{wainwright2019high}, we can connect the covering number of $\Theta_{u(j)}$ with respect to $|\cdot|_{\infty}$ with respect to the $|\cdot|_1$.
		%	Thus we have,
		%	$\mathcal{N}(\delta, \mhH_{j,M},|.|_{n,\infty}) \lesssim   {p \choose s}  2!/\{(s/e)^s \sqrt{2\pi s}\}(2M^2/\delta + 1)^s\lesssim (ep/s)^s \frac{2!}{(s/e)^s \sqrt{2\pi s}}(2M^2/\delta+1)^s$ by the Stirling formula.
		Consider the case with $|\widetilde H_{j,M}|_{n,2}=\bigO_\P(1)$, the entropy number of the function class $\widetilde{\mathcal H}_{j,M}$ with respect to the $\omega_n(\cdot,\cdot)$ metric is bounded as follows:
		\begin{align*}
		\operatorname{ent}_{n,2}(\delta, \widetilde{\mathcal H}_{j,M})=\log \mathcal{N}(\delta|\widetilde H_{j,M}|_{n,2}, \widetilde{\mathcal H}_{j,M},\omega_n(\cdot,\cdot)) \lesssim_\P s_j \{\log (K^{u(j)}) + \log(2M^2/\delta+1)\}.
		\end{align*}
		By choosing $M = \delta n^{1/2}$, we have $\operatorname{ent}_{n,2}( \delta,\widetilde{\mathcal H}_{M})\leq \sum_{j=1}^q\operatorname{ent}_{n,2}( \delta,\widetilde{\mathcal H}_{j,M}) \lesssim_\P s \log(P_n/\delta)$, with $P_n = q\vee {n} \vee e$.
		%the $\delta$ net of the function class $\mhH_{j,M}$ is $ent(  \delta,\mhH_{j,M}, |.|_{n,\infty}) $. Finally by letting $\widetilde{\mathcal H}_{M} = \bigcup_{j=1}^q \widetilde{\mathcal H}_{j,M}$, we have
		%	\begin{align*}
		%		ent( \mhH_{M}, \delta, |.|_{n,\infty}) \leq  \sum_j ent( \mhH_{j,M}, \delta, |.|_{n,\infty})   &\lesssim \sum_j 2s_j \log (p \vee s_j)\log(2M^2/\delta^2+1) \\
		%		&{\color{blue} \lesssim_\P \sum_j s_j \{\log (K^{u(j)}) + \log(2M^2/\delta^2+1)\}}
		%	\end{align*}
		%Recall that	we denote that $a_n = p\vee {n} \vee e \vee s_j$.  We choose $M = \sqrt{n}c$, then we have $\sum_j ent( \mhH_{j,M}, \delta, |.|_{n,\infty}) \lesssim s \log(a_n/\delta) $.
		%	\end{proof}
		%
		%\begin{proof}[\textbf{Proof of Lemma \ref{exp}}]
		
		%Alternatively, one can change the space $\Theta_{u(j)}= \{\theta_{u(j)} :\max_j\|\theta_{u(j)} -\theta^0_{u(j)}\|_1\leq 1, |\theta_{u(j)}|_0\leq s\}$ to $\Theta_{u(j)}= \{\theta_{u(j)} :\max_j\|\theta_{u(j)} -\theta^0_{u(j)}\|_2\leq \rho_n, |\theta_{u(j)}|_0\leq s\}$. Note that $\mathcal{N}(\delta, \tilde{\mathcal{H}}_{j,M},|.|_{n,\infty}) \leq \mathcal{N}( \delta/M^2, \Theta_{u(j)},|.|_{1})\lesssim {p \choose s}\{\rho_n(2M^2{\sqrt{s}}/\delta + 1)\}^s \lesssim {p \choose s}(2\sqrt{s}/\delta +1)^s$, due to $ \mathcal{N}(|.|_{1}, \Theta_{u(j)}, \delta) \leq \mathcal{N}(|.|_{2}, \Theta_{u(j)}, \delta/\sqrt{s})$. It follows that $\operatorname{ent}( \mathcal{H}_{M}, \delta, |.|_{n,\infty}) \lesssim  \sum_j 2s_j \log p \vee s_j \log\{(2M^2/\delta^2+1)\rho_n \} $.
	\end{remark}

 \begin{remark}[Alternative function class]
		The function class $\widetilde{\mathcal{H}}_j $ can be replaced by $\widetilde{\mathcal{H}}^\prime_j= \big\{d\mapsto\tilde h_j(d,W_{u(j)}(d)^\top\vartheta_{u(j)}): \max\limits_{1\leq t\leq n} |\tilde{h}_{j}(D_t,W_{u(j)}(D_t)^\top\vartheta_{u(j)})-\tilde{h}_{j}^0(D_t,W_{u(j)}(D_t)^\top\vartheta^0_{u(j)})|\leq 1\big\}$, where $\tilde h_j^0$ is associated with $\vartheta_{u(j)}^0$. In this case, additional conditions similar to %the conclusions in Lemma \ref{id} 
        assumption \hyperref[A_id]{(A6)} would be required for identification. To be more specifically, we shall adopt the following assumption.
		\begin{assumption}\label{iden}
			Let $\tg(\theta) = [\E_n\{\th_{j}(D_t,W_{u(j)}(D_t)^\top\vartheta_{u(j)})\}]_{j=1}^q$ and $\widetilde G(\theta^0)=\partial_{\theta^\top}\tg(\theta)|_{\theta=\theta^0}$. Assume that $$\min_{\mathcal I:|\mathcal I|\leq s}\min_{\xi\in \mathcal C_{\mathcal I}(u)} |\widetilde G(\theta^0)\xi|_{\infty}\geq |\xi|_1 s^{-1}C(u),$$ for some $C(u)>0$.
			Moreover, assume there exists a positive constant $C$ such that $$\max_{\mathcal I:|\mathcal I|\leq s}\max_{\xi\in \mathcal C_{\mathcal I}(u)} |\widetilde G(\theta^0)\xi|_{\infty}\leq |\xi|_1 C.$$
		\end{assumption}
		%Moreover, if function class $\tilde{\mathcal{H}}_j $ is changed to be $\tilde{\mathcal{H}}_j' = \{\tilde{h}_{jt}(w_{jt}, v_{jt}): \max_{j,t} |\tilde{h}_{jt}-\tilde{h}_{jt}^0|_{\infty}\leq 1, |\theta_{u(j)}^0|_0 \leq s_j \}$, conditions similar to those in Lemma \ref{id} for identification would be required.
		%	Let $B_{1,\theta}(\theta^0, c C')$ be the function class for $\tilde{h}_{jt}$ with $\max_j|\theta_{u(j)} -\theta^0_{u(j)}|_1\leq c C'$, and $B_{\infty, \theta}(h^0, c)$ be associated with $ \max_{j,t} |\tilde{h}_{jt}-\tilde{h}_{jt}^0|_{\infty}\leq c$.
		%	Then, we have $B_{1,\theta}(\theta^0, c C') \subseteq B_{\infty, \th}(h^0, c)  \subseteq  B_{1,\theta}(\theta^0, cS/c(u)) $.
		Let $\widetilde{\mathcal H}_j(c)$ be the function class of $\tilde{h}_{j}$ with $|\vartheta_{u(j)} -\vartheta^0_{u(j)}|_1\leq c$ and $\widetilde{\mathcal H}^\prime_j(c)$ be that with $ \max\limits_{1\leq t\leq n} |\tilde{h}_{j}(D_t,W_{u(j)}(D_t)^\top\vartheta_{u(j)})-\tilde{h}_{j}^0(D_t,W_{u(j)}(D_t)^\top\vartheta^0_{u(j)})|\leq c$. Then, we have $\widetilde{\mathcal H}_j(c/C)\subseteq \widetilde{\mathcal H}^\prime_j(c)\subseteq \widetilde{\mathcal H}_j(cs/c(u))$. We can use this relationship to switch between the function classes. In particular, we have
		\begin{equation*}
		\sup_{\tilde h_{j} \in \widetilde{\mathcal H}^\prime_j(c)}\big|\E{_{n}}\th_{j}(D_t,v_{j,t})\big| \leq  \sup_{\tilde h_{j} \in \widetilde{\mathcal H}_j(cs/c(u))}\big|\E{_{n}}\th_{j}(D_t,v_{j,t})\big|.
		\end{equation*}
		%This can also be used to link the supreme value when switching between function classes.
		%Since $ent(  \delta,\mhH_{j,M}, |.|_{n,\infty}) \lesssim_p ent^{\infty}(  \delta,\mhH_{j,M})  $, thus Assumption \ref{4.3}
		%can imply the above entropy condition.
	\end{remark}

In the following theorem we provide a tail probability inequality of $L_{n,2}$. There are two terms, namely an exponential term and a polynomial term. It can be seen that the exponential bound depends on the dimensionality $P_n$ and sparsity level $s$. The polynomial rate is reflected by the term $n^{-r/2+1}c_r$.
\begin{theorem}\label{concent1}
Under Assumptions \ref{A4.2}-\ref{entropy} and the same conditions as in Lemma \ref{cont}, we have the following probability inequality:
\begin{equation}\label{eq.uniform}
\P\Big(\max_{1\leq j\leq q} \sup_{\tilde h_j \in \widetilde{\mathcal{H}}_j }\big|\E{_n}\tilde h_j(D_t,v_{j,t})\big|\geq e_n\Big)
 \lesssim\exp(- s \log P_n) + n^{-r/2+1}c_r \to 0,
\end{equation}
as $n\to\infty$, where $e_n = n^{-1/2}(s \log P_n)^{1/\gamma}\max\limits_{1\leq j\leq q}\Psi_{j,\nu,0}$, $\gamma=2/(1+2\nu)$. In particular, $\gamma=1$ and $\gamma=2/3$ correspond to the sub-Gaussian and sub-exponential cases, respectively.
\end{theorem}

	\begin{proof}%[\textbf{Proof of Theorem \ref{concent1}}]
		%Recall that  $\|h_{j.}(w_{j.},v_{j.})\|_{\psi_r,\alpha} = \sup_{q\geq 2}q^{-r}\|h_{j.}(w_{j.},v_{j.})\|_{q,\alpha}$.
		%Observe that $ \E(h_{jt}(w_{jt},v_{jt})|\mathcal{F}_{t-1}) - \E h_{jt}(w_{jt},v_{jt})$ is differentiable with respect to the parameter $\theta$ even the original function $h_{j.}(w_{jt},v_{j.})$ is not.
		%Consider the class of functions $\widetilde{\mathcal H}=\big\{d\mapsto\tilde h_j(d,W_{u(j)}(d)^\top\vartheta_{u(j)}):|\vartheta_{u(j)}-\vartheta_{u(j)}^0|_1\leq 1, j=1,\ldots,q\big\}$ and let $\widetilde{\mathcal H}_M$ be the function class  $\widetilde{\mathcal H}$ on the event ${\mathcal A}_M$.
		Define the set $\varTheta= \{\vartheta_{u(j)} :|\vartheta_{u(j)} -\vartheta^0_{u(j)}|_1\leq 1,j=1,\ldots,q\}$. Given $\delta>0$, we pick $\kappa= \min k: 2^{-k}\delta < \epsilon$, for a small constant $\epsilon >0$. Let $\widetilde{\mathcal H}_M(\delta_k )$ denote the space of the functions %$\tilde h_j\in
        in $\widetilde{\mathcal H}_M$ corresponding to the $\delta_k$-nets ($\delta_k\defeq 2^{-k}\delta$) of $\varTheta$ with respect to the $|\cdot|_{1}$-metric (denoted by $\varTheta(\delta_k)$), such that for all $j=1,\ldots,q$, $\tilde h^0_j\in\widetilde{\mathcal H}_M(\delta_0)\subseteq\widetilde{\mathcal H}_M(\delta_1)\subseteq\cdots\subseteq\widetilde{\mathcal H}_M(\delta_\kappa)\subseteq\widetilde{\mathcal H}_M$, where $\tilde h_j^0(D_t)=\tilde h_j(D_t, W_{u(j)}(D_t)^\top\vartheta^0_{u(j)})$. To simplify the notations, we let $\tilde h_{j,t} \defeq \tilde h_j(D_t,W_{u(j)}(D_t)^\top\vartheta_{u(j)})$ and $\tilde h^0_{j,t} \defeq \tilde h_j^0(D_t)$. It can be observed that
		%\begin{eqnarray*}
		%\sup_{\th_{jt}} I_2& \leq& \sup_{\th_{jt}\in \mhH} \inf_{\tg_{jt} \in \mathcal{H}(\delta_{K})}  |\frac{1}{n}\sum_t (\th_{jt}-\tg_{jt})|+ \\
		%&&+ \sum_{k=1}^K\sup_{\th_{jt} \in \mhH(\delta_k)}\inf_{\tg_{jt} \in \mathcal{H}(\delta_{k-1})} |\frac{1}{n}\sum_t (\th_{jt}-\tg_{jt})|+  \sup_{\th_{jt} \in \mhH(\delta)} |\frac{1}{n}\sum_t (\th_{jt})|\\
		%           &\leq & 2 \epsilon M^2  +\sum_{k=1}^K \sup_{\th_j \in \mathcal{H}(\delta_{k})}\inf_{\tg_{jt} \in \mhH(\delta_{k-1})} |\frac{1}{n}\sum_t (\th_{jt}-\tg_{jt})|+  \sup_{\th_{jt} \in \mhH(\delta)} |\frac{1}{n}\sum_t (\th_{jt})| ,
		%\end{eqnarray*}
		\begin{eqnarray*}
			\sup_{\th_{j}\in\widetilde{\mathcal H}_M} \big|{\E}_n\tilde h_j(D_t,v_{j,t})\big|& \leq& \sup_{\th_{j}\in \widetilde{\mathcal H}_M}\inf_{\th^\prime_{j} \in \widetilde{\mathcal{H}}_M(\delta_{\kappa})}  \big|{\E}_n (\th_{j,t}-\th^\prime_{j,t})\big|\\
			&&+\,\sum_{k=1}^\kappa\sup_{\th_{j} \in \widetilde{\mathcal H}_M(\delta_k)}\inf_{\th^\prime_{j} \in \widetilde{\mathcal{H}}_M(\delta_{k-1})} \big|{\E}_n (\th_{j,t}-\th^\prime_{j,t})\big|+  \sup_{\th_{j} \in \widetilde{\mathcal H}_M(\delta_0)} \big|{\E}_n \th_{j,t}\big|\\
			&\leq & 2M^2\epsilon   +\sum_{k=1}^\kappa \sup_{\th_j \in \widetilde{\mathcal{H}}_M(\delta_{k})}\inf_{\th^\prime_{j} \in \widetilde{\mathcal{H}}_M(\delta_{k-1})}\big|{\E}_n (\th_{j,t}-\th^\prime_{j,t})\big|+  \sup_{\th_{j} \in \widetilde{\mathcal{H}}_M(\delta_0)} \big|{\E}_n\th_{j,t}\big|,
		\end{eqnarray*}
		where the last inequality is due to the definition of $\kappa$. By breaking the above inequality with $\sum^\kappa_{k=0}\zeta_{k} = 1$, we have
		\begin{eqnarray}\label{sum}
		\P\Big(\sup_{\th_{j}\in\widetilde{\mathcal H}_M} \big|{\E}_n\tilde h_j(D_t,v_{j,t})\big|\geq u \Big)
		&\leq &\sum_{k=1}^\kappa\P\Big( \sup_{\th_j \in \widetilde{\mathcal{H}}_M(\delta_{k})}\inf_{\th^\prime_{j} \in \widetilde{\mathcal{H}}_M(\delta_{k-1})}\big|{\E}_n (\th_{j,t}-\th^\prime_{j,t})\big|\geq (u - 2M^2\epsilon)\zeta_{k}\Big)\nonumber\\
		&& + \,\, \P\Big(\sup_{\th_{j} \in \widetilde{\mathcal{H}}_M(\delta_0)}\big|{\E}_n\th_{j,t}\big| \geq (u - 2M^2\epsilon)\zeta_{0}\Big)\notag\\
		&\leq&\sum_{k=1}^\kappa N_k\sup_{\th_j \in \widetilde{\mathcal{H}}_M(\delta_{k})}\inf_{\th^\prime_{j} \in \widetilde{\mathcal{H}}_M(\delta_{k-1})}\P\big(\big|{\E}_n (\th_{j,t}-\th^\prime_{j,t})\big| \geq (u - 2M^2\epsilon)\zeta_{k}\big)\notag\\
		&&+ \,\, N_0\sup_{\th_{j} \in \widetilde{\mathcal{H}}_M(\delta_0)}\P\big(\big|{\E}_n\th_{j,t}\big| \geq (u - 2M^2\epsilon)\zeta_{0}\big),
		\end{eqnarray}
		where $N_k \defeq \mathcal{N}(\delta_k, \varTheta, |\cdot|_{1})$.
		%Moreover, by the triangle inequality we get
		% %$\th^0_{jt}$ as the centered points for the $2^{-k}\delta$ net,
		%\begin{eqnarray*}
		%%\sup_{\th_j \in \mhH(\delta_{k})}\inf_{\tg_j \in \mhH(\delta_{k-1})}\max_t|\th_{jt}-\tg_{jt}| &\leq& \sup_{\th_j \in \mathcal{H}(\delta_{k})}\max_t|\th_{jt}-\th^0_{jt}|
		%%\\&&+ \sup_{\tg_j \in \mathcal{H}(\delta_{k-1})}\max_t|\tg_{jt}-\th^0_{jt}| \leq 3\cdot 2^{-k} \delta M^2.
		%\sup_{\th_j \in \widetilde{\mathcal{H}}(\delta_{k})}\inf_{\th^\prime_{j} \in \widetilde{\mathcal{H}}(\delta_{k-1})} \max_t |\th_{j,t}-\th^\prime_{j,t}| &\leq & \sup_{\th_j \in \widetilde{\mathcal{H}}(\delta_{k})}\max_t|\th_{j,t}-\th^0_{j,t}| + \sup_{\th^\prime_{j} \in \widetilde{\mathcal{H}}(\delta_{k-1})} \max_t |\th^\prime_{j,t}-\th^0_{j,t}|\notag\\
		%&\leq & 3\delta_kM^2.
		%\end{eqnarray*}
		Note that $\th_{j,t},\th^\prime_{j,t}$ are associated with $\vartheta_{u(j)}, \vartheta^\prime_{u(j)}$, respectively.
		%Suppose that $\tilde{h}_{jt} $ corresponds to $\tilde{h}_{jt} $ and $\tilde{\theta}_{u(j)}$ and $\tilde{g}_{jt} $  corresponds to $\tilde{\tilde{\theta}}_{u(j)}$.
		Similarly to Definition \ref{dep}, let $\tilde{h}_{j,t}^\ast$ denote $\tilde{h}_{j,t}$ with the innovations $\xi_0,\eta_0$ replaced by $\xi_0^\ast,\eta_0^\ast$ (likewise for $\tilde{h}^{\prime\ast}_{j,t}$ and \textbf{$\tilde{h}^{0\ast}_{j,t}$}). For any $j$ and $t$, we have
		\begin{eqnarray*}
			&&\Big\|\sup_{\th_j \in \widetilde{\mathcal{H}}_M(\delta_{k})}\inf_{\th^\prime_{j} \in \widetilde{\mathcal{H}}_M(\delta_{k-1})} |\tilde{h}_{j,t} - \tilde{h}^\prime_{j,t}  - (\tilde{h}_{j,t}^\ast - \tilde{h}_{j,t}^{\prime\ast})| \Big\|_r\\
			&\leq& \Big\|\sup_{\th_j \in \widetilde{\mathcal{H}}_M(\delta_{k})}|\tilde{h}_{j,t} - \tilde{h}_{j,t}^0  - (\tilde{h}_{j,t}^\ast - \tilde{h}_{j,t}^{0\ast})| \Big\|_r  + \Big\|\sup_{\th^\prime_{j} \in \widetilde{\mathcal{H}}_M(\delta_{k-1})} |\tilde{h}^\prime_{j,t} -  \tilde{h}_{jt}^0  - ( \tilde{h}^{\prime\ast}_{jt} - \tilde{h}_{j,t}^{0\ast}) |\Big\|_r\\
			&\lesssim&\Big\|\sup_{\vartheta_{u(j)}\in\varTheta(\delta_k)}|(\partial \tilde{h}_{j,t}^0/\partial \vartheta_{u(j)} - \partial \tilde{h}_{j,t}^{0\ast}/\partial \vartheta_{u(j)})^{\top} (\vartheta_{u(j)}- \vartheta_{u(j)}^0)| \Big\|_r
			\\&&+ \,\, \Big\|\sup_{\vartheta^\prime_{u(j)}\in\varTheta(\delta_{k-1})}|( \partial \tilde{h}_{j,t}^0/\partial \vartheta_{u(j)} - \partial \tilde{h}_{j,t}^{0\ast}/\partial \vartheta_{u(j)})^{\top} (\vartheta^\prime_{u(j)}- \vartheta_{u(j)}^0)|\Big\|_r\\
			&\leq&\||\partial \tilde{h}_{j,t}^0/\partial \vartheta_{u(j)} - \partial \tilde{h}_{j,t}^{0\ast}/\partial \vartheta_{u(j)} |_{\infty} \|_r \sup_{\vartheta_{u(j)}\in\varTheta(\delta_k)}|\vartheta_{u(j)}- \vartheta_{u(j)}^0|_1\\
			&& + \,\, \||\partial \tilde{h}_{j,t}^0/\partial \vartheta_{u(j)} - \partial \tilde{h}_{j,t}^{0\ast}/\partial \vartheta_{u(j)}|_{\infty}\|_r \sup_{\vartheta^\prime_{u(j)}\in\varTheta(\delta_{k-1})}|\vartheta^\prime_{u(j)}- \vartheta_{u(j)}^0|_1\\
			&\leq& 3\delta_k\||\partial \tilde{h}_{j,t}^0/\partial \vartheta_{u(j)} - \partial \tilde{h}_{j,t}^{0\ast}/\partial \vartheta_{u(j)}|_{\infty}\|_r.
		\end{eqnarray*}
		It follows that the dependence adjusted norm of $|\th_{j,t}-\th^\prime_{j,t}|$ is bounded by $3\delta_k\Psi_{j,\nu,\varsigma} $, where $\Psi_{j,\nu,\varsigma} =  \||\partial \tilde{h}_{j,\cdot}^0/\partial \vartheta_{u(j)}|_\infty\|_{\psi_\nu,\varsigma}$.
		
		Combining \eqref{sum} and Lemma \ref{exp}, we have the following concentration inequality
		\begin{eqnarray}\label{prob}
		&&\P\Big(\max_{1\leq j\leq q}\sup_{\th_{j}\in\widetilde{\mathcal H}_{j,M}} \big|{\E}_n\tilde h_j(D_t,v_{j,t})\big|\geq u\Big)\notag\\
		&\leq&\P\Big(\sup_{\th_{j}\in\widetilde{\mathcal H}_M} \big|{\E}_n\tilde h_j(D_t,v_{j,t})\big|\geq u\Big)\nonumber\\
		&\lesssim&\sum^\kappa_{k=0} \exp\big(\log N_k +\log q  - C_{\gamma}\{\sqrt{n}(u-2M^2\epsilon)\zeta_k\}^{\gamma}/\big(3\delta_k \max_{1\leq j\leq q}\Psi_{j,\nu,0}\big)^{\gamma} \big),
		%	&\leq&\sum^K_{k=0} \exp( \log N_k-(\log q) \min_{1\leq j\leq q}(\sqrt{n}\zeta_k(u-2M^2\epsilon))^{\gamma}/(3\delta_k \Psi_{j,\nu,0})^{\gamma} c'_{\gamma}),
		\end{eqnarray}
		where $\gamma = \nu/(1+2\nu)$, and we need to pick up $\zeta_k$'s such that the right hand side tends to zero as $n\to\infty$.
		
		Define $\varphi_n\defeq \sqrt{n}(u-2M^2\epsilon)\big/\max\limits_{1\leq j\leq q}\Psi_{j,\nu,0}$ and consider $\zeta_{k} = 3(C')^{1/\gamma}2^{-k} \delta(\log N_k \vee \log q)^{1/\gamma} \varphi_n^{-1} $. We have the term involved in \eqref{prob} is given by
		$$C_\gamma\zeta_{k}^{\gamma} \varphi_n^{\gamma}(3\delta_k)^{-\gamma} =C_{\gamma}C' (\log N_k\vee\log q) \geq \log N_k \vee \log q,$$
		for a sufficiently large constant $C'$.
		%And we have the term involved in equation \ref{prob} is $C_{\gamma}\zeta_k^{\gamma} \varphi_n^{\gamma} = C_{\gamma}C' (\log N_k\vee\log q)$.
		It is left to justify that $\sum_{k=0}^\kappa \zeta_k \leq 1$, with a properly chosen ``$u-2M^2\epsilon$''. Observe that
		$\sum_{k=0}^\kappa \zeta_k \lesssim \sum^\kappa_{k=0}2^{-k}\delta(\log N_k\vee\log q)^{1/\gamma} \varphi_n^{-1}$, which means we could have $\sum_{k=0}^\kappa \zeta_k$ is bounded by a constant, provided $\sum^\kappa_{k=0}2^{-k}\delta(\log N_k\vee\log q)^{1/\gamma} \lesssim \varphi_n$. %, where $\tilde{N}_k = ( N_k \vee q)$.
		Thus, it suffices to verify that
		$$\int_{\epsilon}^{\delta}\big(\log\mathcal{N}(x, \varTheta, |\cdot|_1)\big)^{1/\gamma}dx  \vee (\delta-\epsilon) (\log q)^{1/\gamma}  \lesssim \sqrt{n}(u-2M^2\epsilon)\big/\max_{1\leq j\leq q}\Psi_{j,\nu,0}.$$
		%	Recall that for any two constants $b$ and $a$, we have the integrals
		%	\begin{align*}
		%	\int_{a}^{b}\log xdx& = x\log x |^{b}_{a} - x|^b_{a},\\
		%	%\int_{a}^{b}(\log x)^2 dx& = x(\log x)^2 |^{b}_{a} + 2\log x|^b_{a}+2 |^{b}_{a}.\\
		%	\int_{a}^{b}(\log x)^2 dx& = x(\log x)^2 |^{b}_{a} + 2x\log x|^b_{a}+2 x|^{b}_{a}.
		%	%\int_a^b \log(x)^{1/2} dx &=& x \sqrt{\log x}|^{b}_{a} - \int^{\sqrt{\log b}}_{\sqrt{\log (a)}} \exp( u^2) du.
		%	\end{align*}
		We set $\delta$ to be a constant. By letting $\epsilon \lesssim{n}^{-3/2}$, for $\gamma=1,2/3$, we have
		$$\int_{\epsilon}^{\delta}\big(\log \mathcal{N}(x, \varTheta, |\cdot|_1)\big)^{1/\gamma}dx \vee (\delta-\epsilon) (\log q)^{1/\gamma}\lesssim (s \log P_n)^{1/\gamma} .$$
		Moreover, by letting $u=\bigO\big(n^{-1/2}(s \log P_n)^{1/\gamma}\max\limits_{1\leq j\leq q}\Psi_{j,\nu,0}\big)$ and choosing $M$ such that $M\lesssim n^{1/2}\delta$ and $2M^2\epsilon\lesssim n^{-1/2}$, we could achieve $n^{-1/2}(s \log P_n)^{1/\gamma}\max\limits_{1\leq j\leq q}\Psi_{j,\nu,0}\lesssim (u-2M^2\epsilon) $.
		%if $\max_{1\leq j\leq q}\Psi_{j,\nu,0}^h$ is assumed to be bounded.
		
		Based on the discuss above, we shall pick $\zeta_k \geq \delta(2^k/3)^{-1}k^{1/2}$. It can be shown that $\sum_{k=0}^\kappa \delta(2^k/3)^{-1}k^{1/2} \lesssim \int_0^{\infty}  2^{-x}x^{1/2} dx=\sqrt{\pi}/\{2(\log2)^{3/2}\}$.
		
		%It can be shown that $\sum_{k=0}^\kappa (2^k/3)^{-1}k^{1/\gamma} \lesssim \int_0^{\infty}  2^{-x}x^{1/\gamma} dx< c$ holds true for $\gamma = 1,2/3 $. In particular,	for $\gamma= 1$, $\int^{\infty}_{0}2^{-x} x dx = (-\log2)^{-1} \{x 2^{-x} |^{\infty}_{0} + (\log2)^{-1} 2^{-x}|^{\infty}_{0} \}= (\log 2)^{-2}$; for $\gamma= 2/3$, $\int_0^{\infty}  2^{-x}x^{1/\gamma} dx < c$.
		%	$\int^{\infty}_{1} 2^{-2x} x^{1/2} dx = \int^{\infty}_{1} 2^{-2u^2} 2u^2 du = (\log 2)^{-1} \int^{\infty}_{1}u/2 d 2^{-2u^2} = (1/\log 2)( 2^{-2u^2}u/2|^{\infty}_1- (1/2) \int^{\infty}_1 2^{-2u^2} du ) $.
		%Recall that
		%	$\int^{\infty}_1 2^{-2u^2} du  = \sqrt{\pi/(4\log 2)}.$ THIS IS TO VERIFY $\sum^K_{k=1}2^{-k}(\log N_{k})^{1/\alpha}$.}
		% {\red Look at Proof of Lemma 3.2 as in \cite{geer2000empirical}.}
		
		So far, we have analyzed the right hand side of \eqref{prob}, which is of the order as follows:
		%there exists a positive constant $C>0$, such that	
		\begin{align*}
		%\sum^K_{k=0} \sum_j {\exp}(-\zeta_{k}^{\gamma} \varphi_n^{\gamma}(2^k/3)^{\gamma})\leq \sum^K_{k=0}  \exp(\log q-k\varphi_n^{\gamma}) \lesssim C \exp(-\varphi_n^{\gamma}).
		\sum_{k=0}^\kappa \exp\big(\log N_k +\log q  - C_\gamma\zeta_{k}^{\gamma} \varphi_n^{\gamma}(3\delta_k)^{-\gamma}\big) &\leq \sum_{k=0}^\kappa \exp(\log N_k +\log q  - C_\gamma k\varphi_n^\gamma)\\
		&\lesssim\exp(-\varphi_n^\gamma)\lesssim \exp(-s\log P_n).
		\end{align*}
		%Thus the right hand side of \eqref{prob} is of the rate $\exp(-\varphi_n^{\gamma}) \lesssim \exp(- s \log P_n) $.
		Recognize that
		\begin{eqnarray*}
			&&\P\Big(\max_{1\leq j\leq q} \sup_{\tilde h_j \in \widetilde{\mathcal{H}}_j }\big|\E{_n}\tilde h_{j}(D_t,v_{j,t})\big|\geq u\Big)\\
			&\leq & \P\Big(\sup_{\tilde h_j \in \widetilde{\mathcal H}_M }\big|\E{_n}\tilde h_{j}(D_t,v_{j,t})\IF(\mathcal A_M)\big|\geq u/2\Big) + \P\Big(\max_{1\leq j\leq q} \sup_{\tilde h_j \in \widetilde{\mathcal H}_j }\big|\E{_n}\tilde h_{j}(D_t,v_{j,t})\IF(\mathcal A_M^c)\big|\geq u/2\Big)\\
			&\leq& \P\Big(\sup_{\tilde h_j \in \widetilde{\mathcal H}_M }\big|\E{_n}\tilde h_{j}(D_t,v_{j,t})\IF(\mathcal A_M)\big|\geq u/2\Big) + \P\big(\big|\E{_n}H(D_t)\IF(\mathcal A_M^c)\big|\geq u/2\big)\\
			&\leq&\P\Big(\sup_{\tilde h_j \in \widetilde{\mathcal H}_M }\big|\E{_n}\tilde h_{j}(D_t,v_{j,t})\IF(\mathcal A_M)\big|\geq u/2\Big) + \P(\mathcal{A}_M^c),
		\end{eqnarray*}
		where $\mathcal A_M^c$ is denoted as the complement of event $\mathcal A_M$. The last step is to bound the probability of $\mathcal{A}_M^c$. By Markov inequality, we have
		\begin{eqnarray*}
			\P(\mathcal{A}_M^c) &\leq& \sum_{t=1}^n \P(|\widetilde H(D_t)|\geq M) + \sum_{t=1}^n \P(|\widetilde H^{1}(D_t)|\geq M)+ \sum_{t=1}^n\P\big(\max_{1\leq j\leq q}|W_{u(j)}(D_t)|_{\infty}\geq M\big) \\
			&\leq& n^{-r/2+1}c_r,
		\end{eqnarray*}
		where $c_r\defeq \E|\widetilde H(D_t)|^r \vee \E|\widetilde H^1(D_t)|^r \vee \E\Big(\max\limits_{1\leq j\leq q}|W_{u(j)}(D_t)|_\infty^r\Big)$. By letting $M=n^{1/2}\delta$, $u=n^{-1/2}(s \log P_n)^{1/\gamma}\max\limits_{1\leq j\leq q}\Psi_{j,\nu,0}$, we obtain the desired probability inequality.
	\end{proof}

Next, we handle the bound for $L_{n,1}=\E{_{n}}h_{j}(D_t,v_{j,t}) - \E{_{n}}\E\{h_{j}(D_t,v_{j,t})|\mathcal{F}_{t-1}\}$, which is a summand of martingale differences. We shall derive the tail probability of $L_{n,1}$ in Corollary \ref{cont1}. Let $\bar h_{j,t}\defeq\bar h_{j}(D_t,v_{j,t})  = h_{j}(D_t,v_{j,t}) -  \E(h_{j}(D_t,v_{j,t})|\mathcal{F}_{t-1})$ and define the function class
$$\ohH_{j} =\big\{d\mapsto\bar h_j(d,W_{u(j)}(d)^\top\vartheta_{u(j)}):|\vartheta_{u(j)}-\vartheta_{u(j)}^0|_1\leq 1\big\}.$$
%which is enveloped with
%$$\max_{1\leq j\leq q}\sup_{\bar h_j\in\overline{\mathcal{H}}_{j}}|\bar h_{j}(d,W_{u(j)}(d)^\top\vartheta_{u(j)})|\leq \overline H(d).$$
%Given a truncation constant $M$, we define the event
%$$\mathcal A'_M=\big\{\max_{t}|\overline H(D_t)|\leq M, \max_{j,t}|W_{u(j)}(D_t)|_\infty\leq M\big\}.$$
%Accordingly, we define the function class $\overline{\mathcal{H}}_j$ on the event $\mathcal{A}'_M $ to be $\overline{\mathcal H}_{j,M}$.
%$$\tilde{\mathcal A}_M=\big\{\max_{j,t}|\overline h_j(D_t,W_{u(j)}(D_t)^\top\vartheta_{u(j)})|\leq M, \max_{j,t}|W_{u(j)}(D_t)|_\infty\leq M\big\}.$$
%Denote $\ohH_{j} \defeq \{\oh_{jt}(.,.): h_{jt}(.,.)- \E(h_{jt}(.,.)|\mathcal{F}_{.-1}),  \max_j |\theta_{u(j) }- \theta_{u(j)}^0|_1\leq 1,|\theta_{u(j)}^0|_0\leq s_j \}  $ and $\ohH_{j,M} \defeq \{\oh_{jt}(.,.)- \E(\oh_{jt}(.,.)|\mathcal{F}_{t-1})\in \ohH_{j}:\max_{j,t}|\oh_{j}(D_t,v_{j,t})|\leq M , \max_{j,t}|W_{u(j)}(D_t)|_\infty\leq M\}  $.

%Let $\ohH_{M} = \bigcup_j \ohH_{j,M}$, where $M$ is a positive bound.

\begin{assumption}\phantomsection\label{covernumber}
	\begin{enumerate}
		\item[i)] The function class $\ohH_{j}$ is enveloped with
		$$\max_{1\leq j\leq q}\sup_{\bar h_j\in\overline{\mathcal{H}}_{j}}|\bar h_{j}(d,W_{u(j)}(d)^\top\vartheta_{u(j)})|\leq \overline H(d).$$
		Suppose there exists $\delta>0$ such that $\E[|\oH(D_t)|\IF\{|\oH(D_t)| > \sqrt{n}\delta\}]\to0$ as $n\to\infty$.
%		Assume that the function class $\ohH_{j}$ has an envelope $\oH_{j}$, and the entropy condition satisfies that  $\operatorname{ent}^2(\ohH_{M}, \varepsilon) \lesssim \sum_j 2s_j \log p \vee s_j \log\{(2M^2/\varepsilon^2+1)P_n \}$.
		\item[ii)] Consider the class of functions $\overline{\mathcal H}=\big\{d\mapsto\bar h_j(d,W_{u(j)}(d)^\top\vartheta_{u(j)}): j=1,\ldots,q, |\vartheta_{u(j)}-\vartheta_{u(j)}^0|_1\leq 1\big\}$. Assume the entropy number of $\ohH$ with respect to the $\mathcal L^2$-metric is bounded as $\operatorname{ent}_{2}(\delta, \ohH)\lesssim s \log(P_n/\delta)$, and $2s\log P_n\lesssim n^{1/3}$.
		%Suppose there is a measurable cover $\oH_{j}(D_t)$ for each $\ohH_{j}$, such that $\E\{|\oH_{j,.}|^2\IF\{|\oH_{j,.}| > \sqrt{n}\delta\}\}\to0$ as $n\to\infty$, for a positive constant $\delta$.
	\end{enumerate}
\end{assumption}

Define the truncated function $\bar h_{j}(\cdot)$ as
$$\bar h_j^c(\cdot)=\bar h_j(\cdot)\IF(|\bar h_{j}(\cdot)|\leq c) - \E\{\bar h_j(\cdot)\IF(|\bar h_{j}(\cdot)|\leq c) |\mathcal F_{t-1}\},$$
and let $\bar h_{j,t}^c\defeq\bar h_j^c(D_t,v_{j,t})$, for some $c>0$. Accordingly, the space of the truncated functions corresponding to the function class $\overline{\mathcal H}$ is denoted by $\overline{\mathcal H}_c$.
\begin{assumption}\phantomsection\label{rhon}
	\begin{enumerate}
		%\item[iii)] Consider the class of functions $\overline{\mathcal H}=\big\{d\mapsto\bar h_j(d,W_{u(j)}(d)^\top\vartheta_{u(j)}):|\vartheta_{u(j)}-\vartheta_{u(j)}^0|_1\leq 1, j=1,\ldots,q\big\}$.
%Define the metric $\rho_2(f,g) = \|f-g\|_{Q,2}$.
	\item[i)] For any $\bar h_j^c, \bar h_j^{\prime c} \in \ohH_c$, $\exists L>0$ such that $\P \big(\tilde\omega_{n}(\bar h_j^c,\bar h^{\prime c}_j)/\omega_{n}(\bar h_j^c,\bar h^{\prime c}_j) >L\big) \to 0$ as $n \to \infty$, where $\omega_{n}(\bar h_j^c,\bar h^{\prime c}_j) = \{\E_n(\bar h^c_{j,t}-\bar  h^{\prime c}_{j,t})^2\}^{1/2}$ and $\tilde\omega_{n}(\bar h_j^c,\bar h^{\prime c}_j) =  [\E_n\{\E(\bar  h^c_{j,t}-\bar h^{\prime c}_{j,t}|\mathcal{F}_{t-1})\}^2]^{1/2}$.
	\item[ii)]%{\red $2s\log(P_n) \lesssim \sqrt{n}^{2/3}$, where $P_n = q\vee {n} \vee e$.}
	Assume that %$\E\{\exp(\bar h_{j,t}^{c}-\bar h_{j,t}^{\prime c})\}\leq C\exp(c^2b)$,
	$\bar h_{j,t}^{c}$ is a sub-Gaussian random variable and
	%there exists $C>0$ such that $\P(|\bar h_{j,t}^{c}-\bar h_{j,t}^{\prime c})|\geq x)\leq C\exp(-x^2/c^2) $
	the dependence adjusted sub-Gaussian norm of $\E\{(\bar h_{j,t}^{c}-\bar h_{j,t}^{\prime c})^2|\mathcal{F}_{t-1}\}$ (denoted by $\Lambda_{j,\nu,\varsigma,c}$ with $\nu=1$) satisfies $\Lambda_{j,\nu,\varsigma,c}=\bigO(n^{-1})$ if $\tilde\omega_{n}(\bar h_j^c,\bar h^{\prime c}_j) \lesssim_\P n^{-1/2}$.
	\end{enumerate}
\end{assumption}

In Assumption \ref{covernumber}, i) concerns a moment condition on the envelope and ii) restricts the complexity of the function class. Assumption \ref{rhon} i) is imposed on the closeness between the two metrics $\tilde\omega_{n}(\cdot,\cdot)$ and $\omega_n(\cdot,\cdot)$ and the condition that $\Lambda_{j,\nu,\varsigma}=\bigO(n^{-1})$  if $\tilde\omega_{n}(\bar h_j^c,\bar h^{\prime c}_j) \lesssim_\P n^{-1/2}$ in ii) can be inferred by the smoothness of $\E\{(\bar h_{j,t}^{c}-\bar h_{j,t}^{\prime c})^2|\mathcal{F}_{t-1}\}$. We note that our results can be extended to more general moment conditions by replacing the tail probability accordingly and a more restrictive rate on the dimensionality and sparsity would be required.
%Define $\E|\E{_{n}} \oh_{j.}|_{\mathcal{H}} = \sup_{\oh_{j.} \in \mathcal{H}} \E|\E{_{n}} \oh_{j.}|$.

\begin{theorem}\label{unbound}
Under Assumptions \ref{covernumber}-\ref{rhon} and the same conditions as in Lemma \ref{cont}, we have
$$\E\Big(\max_{1\leq j\leq q}\sup_{\bar h_{j} \in \overline{\mathcal H}_j}\big|\E{_n}\bar h_{j}(D_t,v_{j,t})\big|\Big) \lesssim \delta\sqrt{(s\log P_n)/n}.$$
%\sqrt{v_n(s,P_n)/n},$$
%where $v_n(s,P_n) \defeq \sum_j 2s_j \log P_n \vee s_j \log\{(2M^2/\delta^2+1)P_n \} \lesssim s\log P_n $.
\end{theorem}

\begin{proof}%[\textbf{Proof of Theorem \ref{unbound}}]
		%Recall that we define the truncated function class
		%{$\ohH_{j,M} \defeq \{\oh_{jt}(.,.)\in \ohH_{j}:\max_{j,t}|\oh_{j}(D_t,v_{j,t})|\leq M , \max_{j,t}|W_{u(j)}(D_t)|_\infty\leq M\}  $.}
		%	 $\ohH_j^c = \{\oh_{jt}\in \ohH_j: h_{jt}^c - \E(h_{jt}^c|\mathcal{F}_{t-1})\}$.
		Recall the definition of the truncated function
		$$\bar h_j^c(\cdot)=\bar h_j(\cdot)\IF(|\bar h_{j}(\cdot)|\leq c) - \E\{\bar h_j(\cdot)\IF(|\bar h_{j}(\cdot)|\leq c) |\mathcal F_{t-1}\}.$$
		Applying a truncation argument for $\E_n\bar h_{j,t}$ gives us
		\begin{eqnarray*}
			\big|\E{_n}\bar h_{j,t}\big| &\leq  & \big|\E{_n}\{\bar h_{j,t}\IF(|\bar h_{j,t}|\leq c)\} - \E{_n}\E\{\bar h_{j,t}\IF(|\bar h_{j,t}|\leq c)|\mathcal F_{t-1}\}\big|\\
			&&+ \big|\E{_n}\{\bar h_{j,t}\IF(|\bar h_{j,t}|> c)\} - \E{_n}\E\{\bar h_{j,t}\IF(|\bar h_{j,t}|> c)|\mathcal F_{t-1}\}\big|.
		\end{eqnarray*}
		In particular, the second term on the right-hand side can be bounded as:
		\begin{eqnarray*}
			&&\big|\E{_n}\{\bar h_{j,t}\IF(|\bar h_{j,t}|> c)\} - \E{_n}\E\{\bar h_{j,t}\IF(|\bar h_{j,t}|> c)|\mathcal F_{t-1}\}\big|\\
			& \leq& \E{_n}\{|\overline H_{t}|\IF(|\overline H_{t}|>c)\} + \E{_n}\E\{|\overline H_{t}|\IF(|\overline H_{t}|>c)|\mathcal F_{t-1}\},
		\end{eqnarray*}
		where $\overline H(\cdot)$ is the envelope of $\overline{\mathcal H}$, and $\overline H_{t} \defeq\overline H(D_t)$. It follows that
		\begin{eqnarray*}
			\E\Big(\max_{1\leq j\leq q}\sup_{\bar h_{j} \in \overline{\mathcal H}_j}\big|n\E{_n}\bar h_{j,t}\big|\Big)&\leq&\E\Big(\sup_{\bar h_{j} \in \overline{\mathcal H}}\big|n\E{_n}\bar h_{j,t}\big|\Big)\\
			&\leq& \E\Big(\sup_{\bar h^c_{j} \in \overline{\mathcal H}_c}\big|n\E{_n}\bar h^c_{j,t} \big|\Big) + 2n\E\{|\overline H_{t}|\IF(|\overline H_{t}|>c)\}\\
			&=:&I_n + II_n.
		\end{eqnarray*}
		According to Assumption \ref{covernumber} i), we shall choose $c= \sqrt{n}\delta$. For any $\bar h_j^c, \bar h_j^{\prime c} \in \ohH_c$, we pick $\tau_n = \max_{1 \leq k \leq n }: \tilde\omega_{k}(\bar h_{j}^c,\bar h^{\prime c}_{j})\leq L\omega_n(\bar h_{j}^c,\bar h^{\prime c}_{j})$ as the stopping time. Then we have $I_n$ is bounded by
		$$I_n\leq \E\Big\{\sup_{\bar h^c_{j} \in \overline{\mathcal H}_c}\big|n\E{_n}\bar h^c_{j,t} \big|\IF(\tau_n=n)\Big\} + 2cn\E\{\IF(\tau_n\neq n)\}.$$
		Given Assumption \ref{rhon} i), for any $\bar h^c_j, \bar h_j^{\prime c} \in \ohH_c$, as $n\to\infty$, with probability approaching 1, we have $\tilde\omega_{n}(\bar h^c_{j},\bar h^{\prime c}_{j}) \leq L \omega_{n}(\bar h_{j}^c,\bar h^{\prime c}_{j})$, which implies that $\E\{\IF(\tau_n\neq n)\}\to0$.
		
		Let $\mathcal B^c_k$ denote the $2^{-k}\delta$-covering set of $\overline{\mathcal H}_c$ with respect to the metric $\omega_n(\cdot,\cdot)$, for $k=0,1,\ldots,\overline K$, where $\overline K$ satisfies $2^{-\bar K}=\bigO(n^{-1/2})$ and $\overline K=\bigO(\log n)$. Let $\bar h_j^{\ast c}=\arg\sup\limits_{\bar h^c_j\in\overline{\mathcal H}_c}\big|\E_n\bar h^c_{j,t}\big|$, and $\bar h_j^{(k)c}=\arg\inf\limits_{\bar h_j^c\in \mathcal B^c_k}\omega_n(\bar h^c_j,\bar h^{\ast c}_j)$ for $k=1,\ldots,\overline K$. Note that by these definitions we have $\omega_n(\bar h^{(k)c}_j,\bar h^{\ast c}_j)\leq 2^{-k}\delta$ holds for all $k$, which implies that
		$$\omega_n(\bar h^{(k-1)c}_j,\bar h^{(k)c}_j)\leq \omega_n(\bar h^{(k-1)c}_j,\bar h^{\ast c}_j) + \omega_n(\bar h^{(k)c}_j,\bar h^{\ast c}_j)\leq 3\cdot2^{-k}\delta.$$
		In addition, we let $\bar h^{(0)c}_{j}(\cdot)\equiv0$ and assume that $\omega_n(\bar h^{(0)c}_{j}, \bar h^{\ast c}_{j})\leq\delta$.

		Analogue to the definition of $\bar h_j^c$, for $c_k = 2^{-k}c$, we define $\bar h_j^{[c_k,c_{k-1}]}(\cdot)=\bar h_j(\cdot)\IF(c_k\leq|\bar h_{j}(\cdot)|\leq c_{k-1}) - \E\{\bar h_j(\cdot)\IF(c_k\leq|\bar h_{j}(\cdot)|\leq c_{k-1}) |\mathcal F_{t-1}\}$. Accordingly, we define $\bar h_j^{\ast[c_k,c_{k-1}]}$, which is similar to the definition of $\bar h_j^{\ast c}$.
		%, $k= 1,\cdots, \overline{K}$ and $c_0= c$. We denote $\oh_{jt}^{c_k} = \oh_{j_t}\IF\{|\oh_{jt}|\leq c_k\}$, and $\oh_{jt}^{[c_k,c_{k-1}]}  = h_{jt}\IF\{c_k\leq |\oh_{jt}| \leq c_{k-1}\} - \E (h_{jt}\IF\{c_k\leq |\oh_{jt}| \leq c_{k-1}\}|\mathcal{F}_{t-1}) $.
		By a standard chaining argument, we can express any partial sum of  $\bar h_{j,t}^c$
		%$\oh_{jt}^c \in \ohH_{j,c}$ (where $\ohH_{j,c}$ is the space of the truncated  functions $h$ )
		by a telescope sum:
		\begin{eqnarray*}
			\Big|\sum^{\tau_n}_{t=1} \bar h^c_{j,t}\Big| & \leq &\Big| \sum^{\tau_n}_{t=1}\bar h_{j,t}^{(0)c}\Big|+ \Big|\sum^{\overline{K}}_{k=1} \sum^{\tau_n}_{t=1}(\bar h_{j,t}^{(k)c_{k-1}} - \bar h_{j,t}^{(k-1)c_{k-1}})\Big|+\Big|\sum^{\tau_n}_{t=1}(\bar h_{j,t}^{(\overline{K})c_{\overline{K}}} - \bar h_{j,t}^{\ast c_{\overline{K}}})\Big|\\
			&& +\Big|\sum^{\overline{K}}_{k=1} \sum^{\tau_n}_{t=1}(\bar h_{j,t}^{\ast[c_k,c_{k-1}]} - \bar h_{j,t}^{(k)[c_k,c_{k-1}]})\Big|.
		\end{eqnarray*}
		On the event $\{\tau_n=n\}$, it follows that
		\begin{eqnarray*}
			\E\Big(\sup_{\bar h^c_{j} \in \overline{\mathcal H}_c}\big|n\E{_n}\bar h^c_{j,t} \big|\Big) \lesssim \sum_{k=1}^{\overline{K}}\E\Bigg(\max_{\substack{\bar h_j^{c_{k-1}} \in \mathcal B_{k}^{c_{k-1}} , \bar h_j^{\prime c_{k-1}}\in \mathcal B_{k-1}^{c_{k-1}},\\ \omega_{n}(\bar h_j^{c_{k-1}},\bar h_j^{\prime c_{k-1}}) \leq 3\cdot2^{-k}\delta}}\big|\E{_n}(\bar h_{j,t}^{c_{k-1}} - \bar h_{j,t}^{\prime c_{k-1}})\big|\Bigg) + \bigO(n^{-1/2}).
			%&&\E|\E{_{n}} \oh_{j.}|_{\mathcal{H}_c} \leq 2 \sum_{k=1}^{\overline{K}}\E \max_{h \in B_k , g\in B_{k-1}, \rho_{n}(\oh_j,\og_j) \leq 3*2^{-k}*\delta} |\E{_{n}}\oh_{j.}^{c_{k-1}}.-\og_{j.}^{c_{k-1}}.|\nonumber\\&& +\CO(\sqrt{n}^{-1}).
		\end{eqnarray*}
		To bound the $k$th component in the above inequality, we shall apply Lemma \ref{max}. In particular, for a ball with $\omega_n(\bar h_j^{c_{k-1}},\bar h_j^{\prime c_{k-1}})\leq 3\cdot2^{-k}\delta$, by Assumption \ref{rhon} i), we have $\{\tilde\omega_n(\bar h_j^{c_{k-1}},\bar h_j^{\prime c_{k-1}})\}^2\leq L^2\{\omega_n(\bar h_j^{c_{k-1}},\bar h_j^{\prime c_{k-1}})\}^2\leq (3\cdot 2^{-k}\delta L)^2$ holds with probability approaching 1. We choose $A_k  = 2c_{k-1}/\sqrt{s\log( P_n/(2^{-k}\delta))}$ and $B_k = 2(3\cdot2^{-k}\delta L)^2n=2(3c_kL)^2$. We then verify the condition ``$\sqrt{n}\P(\mathcal{G}^c) \lesssim A \log(1+|\mathcal{A}|)+\sqrt{B}\sqrt{\log(1+|\mathcal{A}|)}$'' in Lemma \ref{max} for $k=\overline K$, with the other results following similarly.
		
		Provided Assumption \ref{rhon} ii), we have $\P(|\bar h_{j,t}^{c_{\overline K-1}}-\bar h_{j,t}^{\prime c_{\overline K-1}}|> x)\leq C\exp\{-x^2/(2^{-\overline K}\delta b)^2\}$ for some $C,b>0$. It follows that
		\begin{eqnarray*}
			&&\P\Bigg(\max_{1\leq t\leq n}\max_{\substack{\bar h_j^{c_{\overline K-1}} \in \mathcal B_{\overline K}^{c_{\overline K-1}} , \bar h_j^{\prime c_{\overline K-1}}\in \mathcal B_{\overline K-1}^{c_{\overline K-1}},\\ \omega_{n}(\bar h_j^{c_{\overline K-1}},\bar h_j^{\prime c_{\overline K-1}}) \leq 3\cdot2^{-\overline K}\delta}}(\bar h_{j,t}^{c_{\overline K-1}} - \bar h_{j,t}^{\prime c_{\overline K-1}})\geq A_{\overline K}\Bigg)\\
			&\lesssim& n\mathcal N^2(2^{-\overline K}\delta,\overline{\mathcal H}_{c_{\overline K-1}}, \omega_n(\cdot,\cdot))\exp\{-A_{\overline K}^2/(2^{-2\overline K}\delta^2)\}\\
			&\lesssim_\P& n\sup_{\mathcal Q}\mathcal N^2(2^{-\overline K}\delta, \overline{\mathcal H}_{c_{\overline K-1}}, \|\cdot\|_{\mathcal Q,2})\exp\{ -n c_{\overline{K}-1}^2/(s\log( P_n/(2^{-\overline K}\delta))\delta^2)\}\\
			&\lesssim& \exp\{2s\log(P_n/(2^{-\overline K}\delta)) -n c_{\overline{K}-1}^2/(s\log( P_n/(2^{-\overline K}\delta))\delta^2)\}.
		\end{eqnarray*}
		We set $\delta$ to be a constant. Assumption \ref{covernumber} ii)) ensures that $2(s\log(P_n/(2^{-\overline K}\delta)))^{2}<n c_{\overline{K}-1}^2/\delta^2$, which makes the tail probability tends to zero.
		
		%Assume sub-Gaussianity of $\oh_{jt}^{c_{k-1}}-\og_{jt}^{c_{k-1}}$ then the term follows, and the corresponding dependence adjusted norm is denoted as $\gamma_{j,v,k}$.
		%We shall verify the condition $\sqrt{n}\P(\mathcal{G}^c) \lesssim A \log(1+|\mathcal{A}|)+\sqrt{B}\sqrt{\log(1+|\mathcal{A}|)}$ in Lemma \ref{max} here.
		%Assume that for a constant $\tilde{c} >0$, $\E[\exp(\oh_{jt}^{c_{\overline{K}-1}}-\og_{jt}^{c_{\overline{K}-1}}) ] \leq \exp(\tilde{c}2^{-2\overline{K}}\sigma^2) $, $v2^{-2\overline{K}}\sigma^2 <A_{\overline{K}}/2$, and $2s (\log (a_n/(2^{-\overline{K}}\delta))) < (v2^{-2\overline{K}}\sigma^2)^{-1}\lesssim n $ (holds under Assumption \ref{covernumber}), where $\sigma >0$ is constant.
		% We shall verify it only with $k = \overline{K}$ and other results follow,
		%\begin{eqnarray*}
		%&&\P(\max_t\max_{h \in B_{\overline{K}} , g\in B_{\overline{K}-1}, \rho_{n}(\oh_j,\og_j)\leq 3*2^{-\overline{K}}\delta} \oh_{jt}^{c_{\overline{K}-1}}-\og_{jt}^{c_{\overline{K}-1}}\geq A_{\overline{K}})
		%\\&&\leq |n\mathcal{N}^2(2^{-\overline{K}}\delta(L), \mathcal{H}_c, \rho_{n}(.,.))|\E[\exp(\oh_{jt}^{c_{\overline{K}-1}}-\og_{jt}^{c_{\overline{K}-1}}) ]/ \exp(A_{\overline{K}})
		%\mathfrak{}\\&& \lesssim_p n\sup_Q\mathcal{N}^2(2^{-\overline{K}}\delta L , \mathcal{H}_c,\|.\|_{Q,2}) \exp(v2^{-2\overline{K}}\sigma^2-A_{\overline{K}})
		%\\&& \lesssim \exp(2s (\log (a_n/2^{-\overline{K}}\delta))+ v2^{-2\overline{K}}\sigma^2- A_{\overline{K}}).
		%\end{eqnarray*}
		By Lemma \ref{exp}, we obtain that
		%denote $\gamma_{j,v,k}$ as the subGaussian dependence adjusted norm of $\E\{(\oh_{jt}^{c_{k-1}}.-\og_{jt}^{c_{k-1}})^2|\mathcal{F}_{t-1}\}$. Assume that $\gamma_{j,v,\overline{K}}\lesssim 2^{-2\overline{K}}$,($v = 1$).
		\begin{eqnarray*}
			&&\hspace{-0.45cm}\P\Bigg(\max_{\substack{\bar h_j^{c_{\overline K-1}} \in \mathcal B_{\overline K}^{c_{\overline K-1}} , \bar h_j^{\prime c_{\overline K-1}}\in \mathcal B_{\overline K-1}^{c_{\overline K-1}},\\ \omega_{n}(\bar h_j^{c_{\overline K-1}},\bar h_j^{\prime c_{\overline K-1}}) \leq 3\cdot2^{-\overline K}\delta}}\big|n\E{_n}\E\{(\bar h_{j,t}^{c_{\overline K-1}} - \bar h_{j,t}^{\prime c_{\overline K-1}})^2|\mathcal{F}_{t-1}\}- n\E\{(\bar h_{j,t}^{c_{\overline K-1}} - \bar h_{j,t}^{\prime c_{\overline K-1}})^2\}\big|\geq B_{\overline{K}}\Bigg)\\
			&\lesssim& \mathcal N^2(2^{-\overline K}\delta,\overline{\mathcal H}_{c_{\overline K-1}}, \omega_n(\cdot,\cdot))\exp\big\{-C_\gamma B_{\overline{K}}^{\gamma}/\big(\sqrt{n}\max_{1\leq j\leq q}\Lambda_{j,\nu,0,c_{\overline K-1}}\big)^{\gamma}\big\}\\
			& \lesssim_\P& \exp\big\{2s (\log (P_n/2^{-\overline{K}}\delta)) -C_\gamma B_{\overline{K}}^{\gamma}/\big(\sqrt{n}\max_{1\leq j\leq q}\Lambda_{j,\nu,0,c_{\overline K-1}}\big)^{\gamma}\big\},
		\end{eqnarray*}
		where $\Lambda_{j,\nu,\varsigma,c}$ is defined in Assumption \ref{rhon} ii). Note that $\nu=1$ and $\gamma=2/3$ for the sub-Gaussian case. Since $n\E{_n}\E\{(\bar h_{j,t}^{c_{\overline K-1}} - \bar h_{j,t}^{\prime c_{\overline K-1}})^2|\mathcal{F}_{t-1}\}\lesssim_\P B_{\overline K}$, it can be inferred that $n\E\{(\bar h_{j,t}^{c_{\overline K-1}} - \bar h_{j,t}^{\prime c_{\overline K-1}})^2\}\lesssim B_{\overline{K}}$. Then, we have the tail probability approaching 0 as $2s (\log (P_n/2^{-\overline{K}}\delta))\leq C_\gamma B_{\overline{K}}^{\gamma}/\big(\sqrt{n}\max\limits_{1\leq j\leq q}\Lambda_{j,\nu,0,c_{\overline K-1}}\big)^{\gamma}\lesssim n^{\gamma/2}$ can be guaranteed by Assumption \ref{covernumber} ii).
		
		Combing the two tail probability inequities above shows that the probability of the union of these two tail events decays exponentially. 
		%bounded by
		%$$\exp\{2s\log(P_n/(2^{-\overline K}\delta)) -n c_{\overline{K}-1}^2/\delta^2\} + \exp\big\{2s (\log (P_n/2^{-\overline{K}}\delta)) -C_\gamma B_{\overline{K}}^{\gamma}/\big(\sqrt{n}\max_{1\leq j\leq q}\Lambda_{j,\nu,0,c_{\overline K-1}}\big)^{\gamma}\big\},$$
		This implies that the required condition in Lemma \ref{max} holds true.
		%	This is because  on the event $\{\tau_n = n\}$, and a ball of length $3* 2^{-k} \delta$, we shall have with probability approaching $1$,
		%	\begin{eqnarray*}
		%		&&\tau_n^{-1}\sum^{\tau_n}_{t=1}\E{_{t-1}}|\oh_{jt}^c-\oh_{jt}^{'c}|^2 \leq \tau_n^{-1}\sum^{\tau_n}_{t=1}\E{_{t-1}}|\oh_{jt}-\oh_{jt}'|^2
		%		\\&&\leq (L \rho_n(\oh_{jt},\oh_{jt}'))^2 \leq (  L3*2^{-k}\delta)^2.
		%	\end{eqnarray*}
		Thus, on the events $\{\exists L>0, \text{s.t. } \tilde\omega_{n}(\bar h_j^c,\bar h^{\prime c}_j)/\omega_{n}(\bar h_j^c,\bar h^{\prime c}_j) \leq L, \forall \bar h_j^c, \bar h_j^{\prime c} \in \ohH_c\}$ and $\{\tau_n = n \}$, we have
		\begin{eqnarray*}
			I_n&\leq& \E\Big\{\sup_{\bar h^c_{j} \in \overline{\mathcal H}_c}\big|n\E{_n}\bar h^c_{j,t} \big|\Big\} \\
			& \lesssim& \sum^{\overline{K}}_{k=1} \big\{A_k\log (1+  \mathcal{N}^2(2^{-k}\delta,\overline{\mathcal{H}}_c, \omega_{n}(\cdot,\cdot))) + c_kL \sqrt{\log (1+ \mathcal{N}^2(2^{-k}\delta, \mathcal{H}_c, \omega_{n}(\cdot,\cdot)))} \big\}\\
			&\lesssim& \sqrt{n}\int^1_{0} \delta \sqrt{\log \mathcal{N}(x\delta, \overline{\mathcal{H}}_c, \omega_{n}(\cdot,\cdot) )}dx\\
			&\lesssim_\P& \sqrt{n}\int^1_{0} \delta \sqrt{\log \sup_{\mathcal Q}\mathcal{N}(x\delta, \overline{\mathcal{H}}_c, \|\cdot\|_{\mathcal Q,2})}dx\\
			&\lesssim & \sqrt{n}\int^\delta_{0} \{s\log (P_n/x)\}^{1/2}dx\\
			&\lesssim& \delta\sqrt{ns\log P_n}.
		\end{eqnarray*}
		%where $\lesssim$ depends only on $L$.
		Moreover, by Assumption \ref{covernumber} i), we get
		\begin{eqnarray*}
			II_n/n &=& 2\E[ |\oH(D_t)| \IF\{|\oH(D_t)| > c\}] \to 0, \text{ as } n\to\infty.
		\end{eqnarray*}
		Then the conclusion that $\E\Big(\max\limits_{1\leq j\leq q}\sup\limits_{\bar h_{j} \in \overline{\mathcal H}_j}\big|n\E{_n}\bar h_{j,t}\big|\Big)\leq \delta\sqrt{(s\log P_n)/n}$ follows.
	\end{proof}
 
As a consequence of Theorem \ref{unbound}, we have the following probability inequality.
\begin{corollary}%[Tail probability]
\label{cont1}
Suppose the conditions in Theorem \ref{unbound} hold. Then, we have
$$\max_{1\leq j\leq q}\sup_{\bar h_{j} \in \overline{\mathcal H}_j}\big|\E{_n}\bar h_{j}(D_t,v_{j,t})\big|\lesssim_\P \delta\sqrt{(s\log P_n)/n}.$$
%\sqrt{v_n(s,P_n)/n}.$$
\end{corollary}
Theorem \ref{unbound} and Corollary \ref{cont1} concern the maximal inequalities for a martingale difference summand.
Combining Theorem \ref{concent1} and Corollary \ref{cont1}, we have the following tail probability bounds.
\begin{theorem}[Concentration for the nonlinear moments model] \label{concent2}
Under the same conditions as in Theorem \ref{concent1} and Corollary \ref{concent1}, by letting {$e_n = n^{-1/2}(s \log P_n)^{1/\gamma}\max\limits_{1\leq j\leq q}\Psi_{j,\nu,0}$, $\gamma=2/(1+2\nu)$}, we have the following result:
\begin{equation*}%\label{eq.uniform2}
\max_{1\leq j\leq q} \sup_{h \in \mathcal{H}_j  }\big|\E{_n}h_{j}(D_t,v_{j,t}) - \E h_{j}(D_t,v_{j,t})\big|\lesssim_\P \delta\sqrt{(s\log P_n)/n} + e_n.
\end{equation*}
\end{theorem}

Similarly to what was shown in Theorem \ref{theorem.cons} and \ref{linear}, the consistency under nonlinear moments follows by replacing the concentration results in Lemma \ref{cont} with those in Theorem \ref{concent2}. 
%{\color{red} With additional tedious but nontrival extensions on the term $\hat{G}-G$, $\hat{B}-B$ and $\hat{A}-A$ following similar argument as above. We expect our conclusion to hold.}

\subsection{Linearization and Simultaneous Inference}
In this subsection, we extend the discussion from Section \ref{check} to cover the case of nonlinear moments. Specifically, we will analyze the linearization error for the general form of moments and restate the required conditions for applying Gaussian approximation results, facilitating simultaneous inference.

Recall the linearization of the debiased estimator, given by:
$$\check{\theta}_1  -\theta_1^{0} =\hat{\theta}_{1} -\theta_{1}^0 - \hat B\hat{A}\hat{g}({\hat{\theta}})=-BA\hat{{g}}({{\theta}^0})+ r_n.$$
When the moments do not take a simple linear form, i.e., the Jacobian matrix of the moment functions remains a function of $\theta$, the remainder term is expressed as $r_n=r_{n,1}+r_{n,2}+r_{n,3}$, where 
$$
r_{n,1}=(\mathbf I - \hat B\hat A\hat G_1)(\hat\theta_1-\theta_1^0),\, r_{n,2}=(BA - \hat B\hat A)\hat g(\theta^0),\, r_{n,3}=\hat B\hat A(\hat G_1-\tilde G_1)(\hat\theta_1-\theta_1^0).
$$
In the third term, $r_{n,3}$, the notations $\hat{G}_1\defeq\partial_{\theta_1^\top}\hat g(\theta_1,\hat\theta_2)|_{\theta_1=\hat\theta_1}$ and $\tilde{G}_1\defeq\partial_{\theta_1^\top}\hat g(\theta_1,\tilde\theta_2)|_{\theta_1=\tilde\theta_1}$ represent the partial derivative of $\hat{g}(\theta_1,\theta_2)$ with respect to $\theta_1$, evaluated at $\hat\theta=(\hat\theta_1,\hat\theta_2)$ and $\tilde{\theta}=(\tilde\theta_1,\tilde\theta_2)$, respectively. Here, $\tilde\theta$ is the corresponding point lying in the line segment between $\hat{\theta}=(\hat\theta_1,\hat\theta_2)$ and $\theta^{0}=(\theta_1^0,\theta_2^0)$. It follows that
$$|r_{n,3}|_{\infty}\leq |\hat B|_{\infty}|\hat{A}(\hat{G}_1- \tilde{G}_{1})|_{\max}|\hat{\theta}_1-\theta_1^0|_1.$$

Next, we outline how to analyze the rate of convergence for the estimator of $G$ in the nonlinear case. Observe that
\begin{eqnarray*}
    |\hat{G}-\tilde{G}|_1 &=&\big|{\E}_n\partial_{\theta^\top}g(D_t,\hat\theta) - {\E}_n\partial_{\theta^\top}g(D_t,\tilde\theta)\big|_1\\
    &\leq&\big|{\E}_n\partial_{\theta^\top}g(D_t,\hat\theta) - \E\partial_{\theta^\top}g(D_t,\hat\theta) - {\E}_n\partial_{\theta^\top}g(D_t,\tilde\theta) + \E\partial_{\theta^\top}g(D_t,\tilde\theta)\big|_1 \\
    &&+ \,\big|\E\partial_{\theta^\top}g(D_t,\hat\theta) - \E\partial_{\theta^\top}g(D_t,\tilde\theta)\big|_1.
\end{eqnarray*}
Provided that $|\hat\theta - \tilde\theta|_1=\smallO_\P(1)$, the second term can be addressed using the Lipschitz condition of the score function. To handle the first term, we need additional assumption on the modulus of continuity of the function $\hat G$ with respect to $\theta$. Verifying this condition requires a uniform argument similar to the derivation of the concentration inequality shown in Section \ref{nonlinear.con}.
\begin{assumption}\label{mc}
Let $\mathcal M_u\defeq \{\theta:|\theta-\theta^0|_2\leq u,u>0\}$. Assume that 
$$\sup_{\theta\in\mathcal M_u}\big|\hat G(\theta) - \hat G(\theta^0) - \E\{\hat G(\theta) - \hat G(\theta^0)\}|_1\lesssim_\P\phi_n(u),$$
where $\phi_n(u)$ is an increasing function of $u$ and satisfies $\phi_n(u)\to0$ as $u\to0$.
\end{assumption}

It should be noted that the proof of Lemma \ref{ln_ups} presented in Section \ref{rn} is not applicable to the case of nonlinear moments. Let $g_{jm,t}$ and $\hat g_{jm,t}$ denote the $m$-th element of the vectors $g_j(D_{j,t},\theta^0)$ and $g_j(D_{j,t},\hat\theta)$, respectively, where $j=1,\ldots,p$ and $m=1,\ldots,q_j$. Similarly, we define $g_{il,t}$ and $\hat g_{il,t}$, where $i=1,\ldots,p$ and $l=1,\ldots,q_i$. As an extension to the nonlinear case, we bound the concentration of the covariances of moments as follows:
% It should be noted that the proof of Lemma \ref{ln_ups} presented in Section \ref{a9.3} is not applicable to the case of nonlinear moments. As an extension, we could deal with the concentration of the covariances of moments as follows:
% Next we shall also comment on the extension of Lemma \ref{ln_ups} to the nonlinear case:
% The bound on $I_{n,1}, I_{n,2}, I_{n,3}$ are as follows:
  \begin{eqnarray*}
  |\hat{\Omega}- {\Omega}|_{\max}&=&\max_{i,j,l,m}|{\E}_n (\hat{g}_{il,t}\hat{g}_{jm,t}) - \E (g_{il,t} g_{jm,t})|\\
			&\leq&\max_{i,j,l,m}|{\E}_n\{(\hat{g}_{il,t}-g_{il,t})(\hat{g}_{jm,t}-g_{jm,t})\}- {\E}\{(\hat{g}_{il,t}-g_{il,t})(\hat{g}_{jm,t}-g_{jm,t})\}|\\
            && + \,2\max_{i,j,l,m}|{\E}_n\{g_{il,t} (\hat{g}_{jm,t}-g_{jm,t})\}-{\E}\{g_{il,t} (\hat{g}_{jm,t}-g_{jm,t})\}|\\ 
            &&+\, \max_{i,j,l,m}|{\E}_n(g_{il,t} g_{jm,t}) - \E (g_{il,t}g_{jm,t}) |\\
            &&+\,2\max_{i,j,l,m}|{\E}\{g_{il,t} (\hat{g}_{jm,t}-g_{jm,t})\}| + \max_{i,j,l,m}|{\E}\{(\hat{g}_{il,t}-g_{il,t})(\hat{g}_{jm,t}-g_{jm,t})\}|.
   \end{eqnarray*}
   Specifically, the first two terms on the right-hand side can be addressed by uniform concentration. The third term represents the concentration evaluated at the true parameter point, while the last two terms are expected to be small due to the continuity of the expected score functions. 
  % {\red The first two terms are handled by uniform concentration and the third term are handled by concentration. The last two terms are small due to continuity of the moment in the population.}

% We now provide a modified version Theorem \ref{linear} for the debiased estimator under the linear case.
% Denote $A$ and $G$ are those matrices with true value.

% \begin{theorem}[Linearization of the debiased estimator]\label{linearap}
% 	Under the conditions in Lemma \ref{cont}, \ref{id}, \ref{df}, \ref{ratef}, and given the Gaussian approximation assumptions (as in \hyperref[A_Srate2]{(A7')}, with the dimensionality $\mathcal S$ replaced by the number of moment conditions $q$) for $g(D_t,\theta^0)$, suppose that $|A|_{\max}\leq C$ and $|A|_\infty\leq\iota$ for some $C$ and $\iota$. Moreover, assume that $|AG_1|_\infty\leq\omega_1^{1/2}|AG_1|_2\asymp\omega_1^{2/3}$, $|(AG_1)^{-1}|_\infty\leq \vartheta\asymp\omega_1^{-1}$ if $K^{(1)}$ is fixed, $|(AG_1)^{-1}|_\infty\leq\omega_1^{1/2}|(AG_1)^{-1}|_2\leq\vartheta\asymp \omega_1^{3/2}$ while $K^{(1)}$ is diverging, where $\omega_1=\smallO(n)$. We have
% 	\begin{equation*}
% 	\check{\theta}_1  -\theta_1^{0} =-(A{G}_1)^{-1}A \hat{g}({{\theta}^0})+ r_n,
% 	\end{equation*}
% 	with $|r_n|_\infty\lesssim_\P \varrho_{n,1}+\varrho_{n,3}$, where $\varrho_{n,1}$ and $ \varrho_{n,3}$ are defined in \eqref{rater}.
% \end{theorem}

%\subsection{Gaussian Approximation}
% Finally regarding the simultaneous inference in Section \ref{inference}, 
As in Section \ref{inference}, we focus on testing the hypothesis $H_0:\theta_{1,k}^0=0,\forall k\in \mathcal S$, where $\mathcal S\subseteq \{1,\ldots,K^{(1)}\}$. To evoke the Gaussian approximation results required for conducting simultaneous inference, a more general set of conditions is necessary. 
% Then we have a set of  more general assumption compared to Theorem \ref{linear}. 
Recall the definition of $\mG_t=(\mG_{k,t})_{k\in\mathcal S}$, where $\mG_{k,t}=-\zeta_kg(D_t,\theta^0)$, and $\zeta_k$ is the $k$-th row of the matrix $(AG_1)^{-1}A$. Define the aggregated dependence adjusted norm as:
$$\|\mG_\cdot\|_{r,\varsigma}\defeq\sup_{s\geq 0}(s+1)^\varsigma\sum_{t=s}^\infty\||\mG_t-\mG^\ast_t|_\infty\|_r,$$
where $r\geq1$, $\varsigma>0$. Moreover, define the following quantities:
\begin{align*}
&\Phi^\mG_{r,\varsigma}\defeq \max_{j\in \mathcal S} \|\mG_{j,\cdot}\|_{r,\varsigma}, \,\, \Gamma^\mG_{r, \varsigma} \defeq  \bigg(\sum_{j\in \mathcal S} \|\mG_{j,\cdot}\|^{r}_{r,\varsigma}\bigg)^{1/r},\,\,\Theta^\zeta_{r,\varsigma} \defeq \Gamma^\mG_{r,\varsigma}\wedge \big\{\|\mG_{\cdot}\|_{r,\varsigma}(\log|\mathcal S|)^{3/2}\big\}.
\end{align*}
Let $L_1^\mG = \{\Phi^\mG_{2,\varsigma}\Phi^\mG_{2,0} (\log|\mathcal S|)^2\}^{1/\varsigma}$, $W_1^\mG = \{(\Phi^\mG_{3,0})^6+ (\Phi^\mG_{4,0})^4\}\{\log(|\mathcal S|n)\}^7$, $W_2^\mG = (\Phi^\mG_{2,\varsigma})^2\{\log(|\mathcal S|n)\}^4$, $W_3^\mG = [n^{-\varsigma} \{\log (|\mathcal S|n)\}^{3/2} \Theta^\mG_{r,\varsigma}]^{1/(1/2-\varsigma-1/r)}$, $N_1^\mG =(n/\log|\mathcal S|)^{r/2}/(\Theta^\mG_{r, \varsigma})^{r}$, $N_2^\mG=n(\log|\mathcal S|)^{-2}(\Phi^\mG_{2,\varsigma})^{-2}$, $N_3^\mG = \{n^{1/2}(\log|\mathcal S|)^{-1/2}\Theta^{\mG}_{r, \varsigma}\}^{1/(1/2-\varsigma)}$.
The following assumptions generalize the assumption \hyperref[A_Srate]{(A7)} in the main text from the case of linear moments to nonlinear moments under weak and strong dependency:
\begin{itemize}
	\item[(A7')]\label{A_Srate2}
	(i) (weak dependency case) Given $\Theta^\mG_{r,\varsigma} < \infty$ with $r \geq 2$ and $\varsigma > 1/2 - 1/r$, then $\Theta^\mG_{r, \varsigma} n^{1/r-1/2}\{\log (|\mathcal S|n)\}^{3/2} \to 0$ and $L_1^\mG\max(W_1^\mG, W_2^\mG) = \smallO(1) \min (N_1^\mG,N_2^\mG)$.\\
	(ii) (strong dependency case) Given $0<\varsigma< 1/2 -1/r$, then $\Theta^\mG_{r,\varsigma}(\log |\mathcal S|)^{1/2} = \smallO(n^{\varsigma})$ and $L_1^\mG \max(W_1^\mG,W_2^\mG,W_3^\mG) = \smallO(1)\min(N_2^\mG,N_3^\mG)$.
\end{itemize}

Similar results to those outlined in Section \ref{inference} would follow by essentially replacing the assumption \hyperref[A_Srate]{(A7)} with \hyperref[A_Srate2]{(A7')} and imposing a more general set of conditions on the rate of the block size (as demonstrated in the proof of Corollary 5.8 in \citet{lasso2018}), provided an approximate linearization of the debiased estimator.

		\renewcommand{\thesubsection}{C.\arabic{subsection}}
		\setcounter{equation}{0}
		\renewcommand{\theequation}{C.\arabic{equation}}
		\setcounter{theorem}{0}
		\renewcommand{\thetheorem}{C.\arabic{theorem}}
		\setcounter{lemma}{0}
		\renewcommand{\thelemma}{C.\arabic{lemma}}
		\setcounter{figure}{0}
		\renewcommand{\thefigure}{C.\arabic{figure}}
		\setcounter{table}{0}
		\renewcommand{\thetable}{C.\arabic{table}}
		\setcounter{remark}{0}
		\renewcommand{\theremark}{C.\arabic{remark}}
		\setcounter{corollary}{0}
		\renewcommand{\thecorollary}{C.\arabic{corollary}}
		\setcounter{example}{0}
		\renewcommand{\theexample}{C.\arabic{example}}
		\setcounter{assumption}{0}
		\renewcommand{\theassumption}{C.\arabic{assumption}}
		
		\section{Connection to Semiparametric Efficiency}\label{efficiency}

In this subsection, we demonstrate the connection between our estimator and a semiparametric efficient estimator.
Semiparametric efficiency is extensively discussed in Chapter 25 of \cite{van2000asymptotic}; see also, for example,
\cite{newey1990semiparametric} and \cite{newey1994asymptotic} for practical guidance. Concerning the semiparametric efficiency bound for time series models, we refer to \cite{bickel2001inference} as an example. \cite{jankova2018semiparametric} show the semiparametric efficiency bounds for high-dimensional models.

%We first state the connection with the decorrelated score method by \cite{ning2017general}.
%Section \ref{pro} studies the intuition of our estimator and Section \ref{se.est} provides a proof of the asymptotic variance.
Within the context of this section, we assume %that the observed sample {\color{red}$D_t$} and the error terms {\color{red} $\varepsilon_{j,t}$} in the model are independent over time $t$. Additionally, we assume 
the vector $\theta_1$ containing the parameters of interest is of low dimension (LD) $K^{(1)}\times 1$ ($K^{(1)}$ is fixed), and $\theta_2$ including the nuisance parameters is of high dimension (HD) $K^{(2)}\times 1$ $(K^{(2)}$ is diverging).
%Also for ease of discussion we s%uppress the estimation step of the weight $\hat{G}_1^{\top}(\theta^0)$, $ \hat{\Omega}^{-1}(\theta^0)$,  $\hat{G}_2^{\top}(\theta^0)$, which has been dealt with in previous sessions.
Let $\Theta$ be a compact set in $\mathbb{R}^K$, and define $\Theta_s \defeq \{\theta\in \Theta: |\theta|_0 \leq s, |\theta|_2 \leq c\}$, for a fixed positive constant $c$. The $q$-dimensional vector-valued score function $g(D_t,\theta)$ %: \R^{K+p+q}\times \R^K\to\R^q$ 
satisfies $\E g(D_t,\theta^0) =0$ and $\sup_{\theta\in \Theta_s} \E[g(D_t,\theta)^{\top}g(D_t,\theta)] < \infty$. Moreover, we assume it is twice continuously differentiable with respect to $\theta$. Recall the definitions $\Omega = \E[g(D_t,\theta^0)g(D_t,\theta^0)^\top]$, $G_1=\partial_{\theta_1^\top}g(\theta_1,\theta_2^0)|_{\theta_1=\theta_1^0}$ and $G_2=\partial_{\theta_2^\top}g(\theta_1^0,\theta_2)|_{\theta_2=\theta_2^0}$.
More generally, we define $\Omega(\theta)\defeq\E[g(D_t,\theta)g(D_t,\theta)^\top]$, $G_1(\theta)\defeq\partial_{\theta_1^\top}g(\theta_1,\theta_2)$ and $G_2(\theta)\defeq\partial_{\theta_2^\top}g(\theta_1,\theta_2)$. %{\red can you check that we are imposing iid assupmtion? many thanks !}
%Recall the moment conditions are given by $\E[g(D_t,\theta^0)]= 0$, and the covariance matrix of the scores is denoted by $\Omega = \E[g(D_t,\theta^0)g(D_t,\theta^0)^\top]$.
%And $B(\theta, \vps)\defeq \{\tilde{\theta}:B(s), |\tilde{\theta}- \theta|_2 \leq \vps\}$, which is an $L_2$ ball around $\theta$.
%Recall that we define $\Omega = \Omega(\theta^0)$.
%We denote the data generating process as $W.$,
%\begin{equation}
%%{\theta}^0 =
%\E(g(\theta^0, W_t)) = 0.
%\end{equation}

In Section \ref{pro} we discuss the link of our estimator to the decorrelated score function, which is named by \citet{ning2017general} as a general framework for penalized $M$-estimators. Section \ref{se.est} concerns the formal theorems on the efficiency and the asymptotic variance of our proposed estimator. We look at the case that $\{D_t\}_{t=1}^n$ is i.i.d. and follows the cumulative distribution function $\P_{\theta^0}(\cdot)$ and the probability density function $f_{\theta^0}(\cdot)$, characterized by $\theta^0$ respectively.

\subsection{Link to the Decorrelated Score Function}\label{pro}
%{Intuition of the Projection}
For a vector $\alpha\in\R^K$, we denote $a_S$ as a subvector of $\alpha$ indexed by the subset $S\subseteq \{1, \cdots, K\}$, namely $\alpha_S= (\alpha_j)_{j\in S}\in\R^{|S|}$. In addition, we let $\alpha(S)=(\alpha(S)_{j})_{j=1}^K\in\R^K$, where $\alpha(S)_{j}=\alpha_j$ if $j\in S$, $\alpha(S)_{j}=0$ if $j\notin S$.

\begin{assumption}\label{den}
For $a_1 \in \R^{K^{(1)}}, a_2 \in \R^{K^{(2)}}$, $\|a_1^{\top}\{\partial_{\theta_1} \log f_{\theta^0}(D_t) -  G_1^{\top}\Omega^{-1} g(D_t,\theta^0)\}\|_2 \to 0$, $\|a_2^{\top}\{\partial_{\theta_2} \log f_{\theta^0}(D_t) - G_2^{\top}\Omega^{-1} g(D_t,\theta^0)\}\|_2\to 0$, as $ q\to \infty$. %Define $\theta_1(S)$ as the subvector of $\theta_1$, and $a_1(S)$ as the subvector of the vector $a_1$.
Moreover, there exists a subset $S\subseteq \{1,\ldots,K^{(2)}\}$ with cardinality $|S|\leq s$, such that $\|a_{2,S}^{\top}\partial_{\theta_{2,S}} \log f_{\theta^0}(D_t) -\alpha_2^{\top}G_2^{\top}\Omega^{-1} g(D_t,\theta^0)\|_2\to 0$, as $q\to\infty$, where $a_{2,S},\theta_{2,S}$ are the subvectors of $a_2,\theta_2$ indexed by $S$ respectively.
%{\color{red} Assume uncorrelatedness of moment functions.}
\end{assumption}

%It should be noted that we can replace $\E[g(D_t,\theta)]$ by $\E [\E{_{n}}g(D_t,\theta)]$.
Intuitively, we want to associate the score $G_1(\theta)^\top \Omega^{-1}(\theta) \E{_{n}}g(D_t,\theta)$ for the parameters of interest $\theta_1$ with $G_2(\theta)^\top \Omega^{-1}(\theta) \E{_{n}}g(D_t,\theta)$ for the nuisance parameters $\theta_2$.
%We shall provide a {\red justification} of the weighting matrices $G_1(\theta)^\top \Omega^{-1}(\theta)$ and $G_2(\theta)^\top \Omega^{-1}(\theta)$, which is essential for the over-identification case.
To explain the intuition of the projection, %of the score function,
we define the Hilbert space spanned by the two score functions as follows:
\begin{align*}
\mathcal{T}_q =\{&\ell=a_1^{\top} G_{1}(\theta)^\top\Omega^{-1}(\theta)g(D_t,\theta) - a_2^{\top} G_2(\theta)^{\top}\Omega^{-1}(\theta) g(D_t,\theta):a_1 \in \R^{K^{(1)}}, a_2  \in \R^{K^{(2)}}, \\
&\theta\in \Theta_s,\|\ell\|_2 < \infty \}.
%\mathcal{T}_q =\{&\ell:\theta_1, a_1 \in \R^{K^{(1)}}, \theta_2,a_2  \in \R^{K^{(2)}}, \theta=(\theta_1^\top, \theta_2^\top)^\top \in \Theta_s,\\
%&\ell =a_2^{\top} G_{2}^{\top}(\theta)\Omega^{-1}(\theta)g(D_t,\theta) - a_1^{\top} G_1^{\top} (\theta)\Omega^{-1}(\theta) g(D_t,\theta), \|\ell\|_2 < \infty \}.
\end{align*}
%where $a_1,a_2$ are two real vectors.
Note that the space depends on $q$ as $g(D_t,\theta)$ is a vector-valued function mapping to $\R^q$. The closure of $\mathcal T_q$ is defined as
$$\mathcal{T} = \{\ell: \|\ell - \ell_q\|_2 \xrightarrow{q\to\infty}0, \ell_q \in \mathcal{T}_q , \|\ell\|_2 < \infty \}.$$
%$\tilde{a}_1(S)=[a_{1,j}]_{j\in 1,\cdots, K^{(1)}} $, where $a_{1,j}=1$, if $j\in S$, and otherwise $a_{1,j}=0$, if $j\not \in S$.
Define the Hilbert space spanned by the two score functions with respect to $S$ as follows:
\begin{align*}
\mathcal{T}_q(S) = \{&\ell=a_1^{\top} G_{1}(\theta)^\top\Omega^{-1}(\theta)g(D_t,\theta) - a_2(S)^{\top} G_2(\theta)^{\top}\Omega^{-1}(\theta) g(D_t,\theta):a_1 \in \R^{K^{(1)}}, a_2  \in \R^{K^{(2)}},\\ &\theta\in \Theta_s, |a_2|_0\leq s,  \|\ell\|_2 < \infty \},
%\mathcal{T}_q(S) = \{&\ell:a_1 \in \R^{K^{(1)}}, a_2  \in \R^{K^{(2)}}, \theta\in \Theta_s, |a_2|_0\leq s, (a_1(S),\theta_1(S)\neq 0), \\
%&\ell =a_2^{\top} G_{2}^{\top}(\theta)\Omega^{-1}(\theta)g(D_t,\theta) - \tilde{a}_1(S)^{\top} G_1^{\top} (\theta)\Omega^{-1}(\theta) g(D_t,\theta),  \|\ell\|_2 < \infty \}.
\end{align*}
%where $a_1,a_2$ are two real vectors.
with the closure
$$\mathcal{T}(S) = \{\ell: \|\ell - \ell_q\|_2 \xrightarrow{q\to\infty}0, \ell_q \in \mathcal{T}_q(S) , \|\ell\|_2 < \infty \}.$$
We also consider the space spanned by the nuisance score function:
\begin{equation*}
\mathcal{T}_{q}^N = \{\ell =a_2^{\top} G_2(\theta)^\top \Omega^{-1}(\theta) g(D_t,\theta): a_2\in\R^{K^{(2)}}, \theta \in \Theta_s, \|\ell\|_2 < \infty \}.
\end{equation*}
The corresponding closure is defined as
$$\mathcal{T}^N  = \{\ell: \|\ell - \ell_q\|_2 \xrightarrow{q\to\infty}0, \ell_q \in \mathcal{T}_{q}^N , \|\ell\|_2 < \infty \},$$
and the orthogonal complement of $\mathcal{T}^N$ is given by
\begin{equation*}
\mathcal{U}^N = \{g \in \mathcal{T} :\langle g, u\rangle = 0, u \in \mathcal{T}^N \},
\end{equation*}
%which is the orthogonal complement of $\mathcal{T}_N$.
where $\langle g, s\rangle = \E(g^{\top} s)$ denotes the inner product. Similarly to $\mathcal T(S)$, we can define $\mathcal{T}^N(S)$ for the nuisance score function with respect to the subset $S$. In particular, $\mathcal{T}^N(S)$ is a low-dimensional subspace (indexed by the subset $S$) of the high-dimensional space $\mathcal{T}^N$, given the cardinality $|S|$ is small compared to $K^{(2)}$ ($|S|\ll K^{(2)}$).

Note that both $\mathcal{T}_N$ and $\mathcal{U}_N$ are closed space. Thus, the projection is well defined and an efficient score function can be constructed involving a matrix given by $\varPi(\theta)  = G_1(\theta)^\top \Omega^{-1}(\theta) G_2(\theta) \big(G_2(\theta)^{\top}\Omega^{-1}(\theta)G_2(\theta)\big)^{-1}$. It can be shown that our debiased estimator proposed in Section \ref{est} is induced by a decorrelated score function for $\theta_1$ which is orthogonal to $\mathcal{T}^N(S)$. The specific form of the decorrelated score function is given by
\begin{eqnarray*}
\psi_1(D_t,\theta) = \psi_1(D_t,\theta_1,\theta_2) &=&{G}_1(\theta)^\top{\Omega}^{-1}(\theta)\big\{\mathbf I_q-{G}_2(\theta)P\big({\Omega(\theta)},{G}_2(\theta)\big)\big\}g(D_t,\theta)\\
&=& {G}_1(\theta)^\top{\Omega}^{-1}(\theta){g}(D_t,\theta) -  \varPi(\theta) {G}_2(\theta)^\top{\Omega}^{-1}(\theta){g}(D_t,\theta),%\\
%S_{1,n}(\theta) &=&{G}_1(\theta)^\top{\Omega}^{-1}(\theta)\big\{\mathbf I_q -{G}_2(\theta) P\big({\Omega(\theta)},{G}_2(\theta)\big)\big\}\E{_n}g(D_t,\theta)\\
%&= &{G}_1(\theta)^\top{\Omega}^{-1}(\theta)\hat{g}({{\theta}}) -  \varPi(\theta) {G}_2(\theta)^\top{\Omega}^{-1}(\theta) \hat{g}({\theta}).
\end{eqnarray*}
where $P\big({\Omega(\theta)},{G}_2(\theta)\big)=\big(G_2(\theta)^{\top}\Omega^{-1}(\theta)G_2(\theta)\big)^{-1}G_2(\theta)^{\top}\Omega^{-1}(\theta)$. Let $\hat \psi_{1}(\theta)=\hat \psi_1(\theta_1,\theta_2)$ be the empirical analogue of $\E\psi_1(D_t,\theta_1,\theta_2)$.

One can estimate ${\theta}_1$ by solving $\hat \psi_{1}(\theta_1, \hat{\theta}_2) = 0$ with a preliminary estimator $\hat\theta_2$. Furthermore, we can also consider a one-step estimator. We define the following quantities to simplify the notations:
\begin{align*}
&F_{11}(\theta) = G_1(\theta)^\top\Omega^{-1} (\theta) G_1(\theta), \,F_{22}(\theta) = G_2(\theta)^\top \Omega^{-1}(\theta) G_2(\theta),\\
&F_{12}(\theta) = G_1(\theta)^\top \Omega^{-1}(\theta) G_2(\theta),\,F_{21}(\theta) = G_2(\theta)^\top \Omega^{-1}(\theta) G_1(\theta),\\
&F_{1|2}(\theta) = F_{11}(\theta)- F_{12}(\theta)F^{-1}_{22}(\theta)F_{21}(\theta).
\end{align*}
We observe that the estimator in the form of \eqref{est.eq} is same as the one-step estimator related to the decorrelated score function, namely the solution to
\begin{equation*}%\label{onestep}
\hat \psi_{1}(\hat\theta) + \hat{F}_{1|2}(\hat{\theta})(\theta_1- \hat{\theta}_1) = 0.
\end{equation*}
That is
\begin{equation*}
\check{\theta}^{\operatorname{OS}}_1 = \hat{\theta}_1 - \hat{F}_{1|2}^{-1}(\hat{\theta})\hat \psi_{1}(\hat{\theta}).
\end{equation*}
%which corresponds exactly to the estimator in (\ref{maineq}).
In particular, the estimator of $\varPi(\theta)=F_{12}(\theta)F_{22}^{-1}(\theta)$, denoted by $\hat{\varPi}(\theta)$, can be attained by solving
\begin{equation} %\label{proj}
\min_{A\in\R^{K^{(1)}\times K^{(2)}}} \sum_{i=1}^{K^{(1)}}\sum_{j=1}^{K^{(2)}}|A_{ij}|:\quad |\hat{F}_{12}(\theta) - A\hat{F}_{22}(\theta)|_{\max}\leq \lambda_n.\notag
\end{equation}
When $\theta_1$ is of fixed dimension, we can obtain $\hat{F}_{1|2}^{-1}(\hat{\theta})$ from $\hat{F}_{1|2}(\hat{\theta}) =  \hat{F}_{11}(\hat{\theta}) - \hat{\varPi}(\hat\theta)\hat{F}_{21}(\hat{\theta})$ directly. The rate of $|\hat{\varPi}(\hat\theta) -\varPi(\theta^0) |_{\max}$ is discussed in following remark and the rest of the rate analysis remains unchanged as we have shown in Section \ref{linear.rate}.
\begin{remark}[The rate of $|\hat{\varPi}(\hat\theta) -\varPi(\theta^0) |_{\max}$]\label{simplerate}

%\section{Simple rate} 
%With the efficient influence function $\tilde{\psi}_2(\theta^0) = {F}_{2|1}^{-1}(\theta^0)S_{2}({\theta}^0)$ and an estimator of it given by $\hat{\psi}_2(\check{\theta}) = \hat{F}_{2|1}^{-1}(\check{\theta})S_{2,n}(\check{\theta})$, it can be seen that the estimator in \eqref{maineq} is the one-step estimator as a solution to \eqref{onestep}, i.e. $\check{\theta}_2 - \hat{F}_{2|1}^{-1}(\check{\theta})S_{2,n}(\check{\theta})$. We note that when $\check{\theta}_2$ is of fixed dimension, we can simplify the CLIME to a two-step estimation. Thus we can also consider the following estimator: 
%\begin{equation} %\label{proj}
%\hat{P}_{\theta^0} = \mbox{argmin}_{A} |A|_{1,1} \quad  \mbox{subject to} \quad |\hat{F}_{12}(\check{\theta}) - \hat{F}_{11}(\check{\theta}) A|_{\max}\leq \lambda,\notag
%\end{equation}
%and obtain $\hat{F}_{2|1}(\check{\theta})^{-1}$ from $\hat{F}_{2|1}(\check{\theta}) =  \hat{F}_{22}(\check{\theta}) - \hat{F}_{21}(\check{\theta})\hat{P}_{\theta^0} $.
%we will show in this section the rate of the estimator

%Thus $\hat{F}_{22,1}(\check{\theta}) =  \hat{F}_{22}(\check{\theta}) - \hat{F}_{21}(\check{\theta})\hat{P}_{\theta^0} $, 
%where $ \hat{F}_{22}(\check{\theta}) $ and $\hat{F}_{21}(\check{\theta})$. Thus $\hat{F}_{22,1}(\check{\theta})^{-1}$ can be obtained directly.

%We now prove the rate of $|\hat{P}_{\theta^0}  - {P}_{\theta^0} |_{\max}$ and the rest shall follow from the previous derivation. 
We observe that
\begin{eqnarray*}
% |\hat{\varPi}(\hat\theta) -\varPi(\theta^0)  |_{\max}&=& |F_{22}^{-1}(\theta^0)F_{22}(\theta^0)(\hat{\varPi}(\hat\theta) -\varPi(\theta^0) ) |_{\max}\\
% &\leq & |F_{22}^{-1}(\theta^0)|_{\infty} | F_{22}(\theta^0)(\hat{\varPi}(\hat\theta) -\varPi(\theta^0)) + F_{12}(\theta^0)- F_{12}(\theta^0)|_{\max}\\
% &= & |F_{22}^{-1}(\theta^0)|_{\infty} | F_{22}(\theta^0)\hat{\varPi}(\hat\theta) - F_{12}(\theta^0)|_{\max}\\
% &\leq& |F_{22}^{-1}(\theta^0)|_{\infty}|F_{22}(\theta^0) \hat{\varPi}(\hat\theta) - \hat F_{22}(\hat\theta)\hat{\varPi}(\hat\theta)|_{\max }\\
% &&+ |F_{22}^{-1}(\theta^0)|_{\infty} |\hat F_{22}(\hat\theta)\hat{\varPi}(\hat\theta) - \hat F_{12}(\hat\theta)|_{\max} + |F_{22}^{-1}(\theta^0)|_{\infty} |\hat F_{12}(\hat\theta) - F_{12}(\theta^0) |_{\max}.
|\hat{\varPi}(\hat\theta) -\varPi(\theta^0)  |_{\max}&=& |(\hat{\varPi}(\hat\theta) -\varPi(\theta^0))F_{22}(\theta^0)F_{22}^{-1}(\theta^0) |_{\max}\\
&\leq & |(\hat{\varPi}(\hat\theta) -\varPi(\theta^0))F_{22}(\theta^0) + F_{12}(\theta^0)- F_{12}(\theta^0)|_{\max}|F_{22}^{-1}(\theta^0)|_1\\
&= & |\hat{\varPi}(\hat\theta) F_{22}(\theta^0) - F_{12}(\theta^0)|_{\max}|F_{22}^{-1}(\theta^0)|_1\\
&\leq& |\hat{\varPi}(\hat\theta)F_{22}(\theta^0) - \hat{\varPi}(\hat\theta)\hat F_{22}(\hat\theta)|_{\max }|F_{22}^{-1}(\theta^0)|_1\\
&&+ |\hat{\varPi}(\hat\theta)\hat F_{22}(\hat\theta) - \hat F_{12}(\hat\theta)|_{\max}|F_{22}^{-1}(\theta^0)|_1 + |\hat F_{12}(\hat\theta) - F_{12}(\theta^0) |_{\max}|F_{22}^{-1}(\theta^0)|_1.
\end{eqnarray*}
% Consider the case with $|F_{22}^{-1}(\theta^0)|_{\infty} =\bigO(1)$ and let $ |\hat{F}_{22}(\hat{\theta}) - F_{22}(\theta^0)|_{\max} \lesssim_\P \delta_{n,22}^F$, $|\hat{F}_{12} (\hat{\theta}) - {F}_{12} ({\theta}^0) |_{\max}\lesssim_\P \delta_{n,12}^F$. The inequality above can be further bounded by 
% $$|\hat{\varPi}(\hat\theta) -\varPi(\theta^0)  |_{\max}\lesssim_\P \delta_{n,22}^F |\hat{\varPi}(\hat\theta)|_{\max} + \lambda_n+  \delta_{n,12}^F.$$
% Given $|\varPi(\theta^0)|_{\max} = |F_{12}(\theta^0)F_{22}^{-1}(\theta^0)|_{\max}=\bigO(1)$ and $\delta_{n,22}^F\to0$ as $n\to\infty$, it follows that
% $$|\hat{\varPi}(\hat\theta) -\varPi(\theta^0)  |_{\max}\lesssim_\P \lambda_n+  \delta_{n,12}^F.$$
Consider the case with $|F_{22}^{-1}(\theta^0)|_1 =\bigO(1)$ and let $ |\hat{F}_{22}(\hat{\theta}) - F_{22}(\theta^0)|_{1} \lesssim_\P \delta_{n,2}^{F_{22}}$, $|\hat{F}_{12} (\hat{\theta}) - {F}_{12} ({\theta}^0) |_{\max}\lesssim_\P \delta_{n}^{F_{12}}$. The inequality above can be further bounded by 
$$|\hat{\varPi}(\hat\theta) -\varPi(\theta^0)  |_{\max}\lesssim_\P |\hat{\varPi}(\hat\theta)|_{\max}\delta_{n,2}^{F_{22}}  + \lambda_n+  \delta_{n}^{F_{12}}.$$
Given $|\varPi(\theta^0)|_{\max} = |F_{12}(\theta^0)F_{22}^{-1}(\theta^0)|_{\max}=\bigO(1)$ and $\delta_{n,2}^{F_{22}}\to0$ as $n\to\infty$, it follows that
$$|\hat{\varPi}(\hat\theta) -\varPi(\theta^0)  |_{\max}\lesssim_\P \lambda_n+ \delta_{n}^{F_{12}}.$$
\end{remark}

%Recall that $S$ be a subset indices of the parameter  $\theta$ and we let the cardinality $|S|\leq s$. We let $a_1(S)$ be the subset vector of $a_1$.
Recall that $\{D_t\}_{t=1}^n$ follows the cumulative distribution function $\P_{\theta^0}(\cdot)=\P_{\theta_1^0, \theta_2^0}(\cdot)$. It is required to estimate the value of $\theta_1(\P_{\theta})$ of a functional $\theta_1:\{\P_\theta:\theta\in\Theta_s\}\mapsto\R^{K^{(1)}}$. We assume that $\theta_1(\cdot)$ differentiable at the true distribution $\P_{\theta_1^0, \theta_2^0}$ .
To characterize the efficiency of the estimator, we consider a neighborhood around the true value $\theta_1^0$, namely $\{b(\tau): |b(\tau)- \theta^0_1- \tau a_1|_2 = \smallO(\tau), 0 < \tau< \epsilon, a_1\in\R^{K^{(1)}}\} \subseteq\Theta_1$, %($\epsilon$ is a positive constant),
where $\Theta_1$ is the parameter space of $\theta_1$. The derivative of $\theta_1\big(\P_{\theta_1^0+ \tau a_1, \theta_2^0+ \tau a_2(S)}\big)$ with respect to $\tau$ (evaluated at $\tau=0$) is given by
%Recall that the cumulative distribution function of $D_t$ is $\P_{\theta_1, \theta_2}(D_t)$.
%In particular, we define $\theta_1$ as a functional of the underlying measure $\P_{\theta_1^0+ t \tilde{a}_1(S), \theta_2^0+ ta_2}$ (abbreviation for $\P_{\theta_1^0+ t \tilde{a}_1(S), \theta_2^0+ ta_2}(D_t)$), and the derivative is given by
\begin{equation*}
%\frac{\partial \theta_1(\P_{\theta_1^0+ t \tilde{a}_1(S), \theta_2^0+ ta_2})}{\partial t}|_{t=0} = \langle \tilde{\psi}_1, a_1^{\top}(S)\partial_{\theta_1(S)} \log f(\theta^0)+ a_2^{\top} \partial_{\theta_2} \log f(\theta^0)\rangle_{P_{\theta^0}},
\frac{\partial \theta_1\big(\P_{\theta_1^0+ \tau a_1, \theta_2^0+ \tau a_2(S)}\big)}{\partial \tau}\Big|_{\tau=0} = \langle \tilde\psi_1(D_t,\theta^0)^\top, a_1^{\top}\partial_{\theta_1} \log f_{\theta^0}+ a_2(S)^{\top} \partial_{\theta_2} \log f_{\theta_1^0,\theta_2}\rangle_{\P_{\theta^0}},
\end{equation*}
where $\tilde\psi_1(D_t,\theta^0)$ is orthogonal to $\mathcal{T}^{N}(S)$ and the inner product on the right hand side is defined under $\P_{\theta^0}$.
%We let  $a_1^{\top}\nabla_{\theta_1} \log f(\theta^0) =  a_1^{\top} G_1^{\top}(\theta_1^0) \Omega^{-1} \{ g(\theta^0, W_t)\}$ and
%$  a_2^{\top} \nabla_{\theta_2} \log f(\theta^0)= a_2^{\top} G_2^{\top}(\theta_1^0) \Omega^{-1} \{ g(\theta^0, W_t)\}$.

%We need some additional conditions on the model setting to ensure that the projection of $a_1^{\top} \partial_{\theta_1} \log f_{\theta^0}$ onto a low-dimensional subspace $\mathcal{T}(S)$ is the same as onto the high-dimensional space $\mathcal{T}$.
In particular, by setting $\frac{\partial \theta_1\big(\P_{\theta_1^0+\tau a_1, \theta_2^0+ \tau a_2(S)}\big)}{\partial \tau}\Big|_{\tau=0}  = a_1 = 0$,
we obtain
$$\langle\tilde\psi_1(D_t,\theta^0)^\top, a_2(S)^{\top}\partial_{\theta_2} \log f_{\theta_1^0,\theta_2}\rangle_{\P_{\theta^0}} = 0.$$
As a result, we have the influence function $\tilde\psi_1(D_t,\theta^0) = F_{1|2}^{-1}(\theta^0) \psi_{1}(D_t,\theta^0)$ belongs to $\mathcal{U}^{N}$, which is orthogonal to $\mathcal{T}^N$ and thus to $\mathcal{T}^{N}(S)$ (under Assumption \ref{reg} iii)). It is not hard to see that our decorrelated score function $\psi_1(D_t,\theta)$ satisfies this property.
%{\red It shall be noted that we do not need assume that the score function $\tilde{\psi}_2(\theta)$ is appropriately equal to the score of a likelihood at any point $\theta$. We can also use the approximate least favorable score function as described in Chapter 25.11 in \cite{van2000asymptotic}. The correspondence of the likelihood function to the score function only need to hold true at the true value $\theta^0$. }

\subsection{Efficiency of the Estimator}\label{se.est}
%{\color{red}Shall we restrict to iid case in this regards???}
In this section, we provide the theoretical results on the efficiency of our debiased estimator $\check{\theta}_1$ and its asymptotic normality. %{\color{red}We shall restrict to the case of i.i.d. errors throughout the section.}
\begin{assumption}%[Regularity]
\phantomsection\label{reg}
\begin{enumerate}
	\item[i)] For any $a_1\in\R^{K^{(1)}}, a_2\in\R^{K^{(2)}}$, there exists a path $\tau>0$ such that 
	\begin{equation*}
		\int \Big[\frac{d\P^{1/2}_{\theta_1+ \tau a_1, \theta_2 + \tau a_2(S)} - d\P^{1/2}_{\theta_1, \theta_2}}{\tau}- \frac{1}{2}\{a_1^{\top}\partial_{\theta_1} \log f_\theta + a_2(S)^\top \partial_{\theta_2} \log f_\theta\}d\P_{\theta_1,\theta_2}^{1/2}\Big] ^{2}\to 0.
		%&&\int \{[dP^{1/2}_{\theta_1+ t \tilde{a}_1(S), \theta_2 + ta_2} - dP^{1/2}_{\theta_1, \theta_2}]/t - 1/2(a_2^{\top}\partial_{\theta_2} \log f(\theta)
		%\\&&+ a_1^{\top}(S) \partial_{\theta_1} \log f(\theta))dP_{\theta_1,\theta_2}^{1/2}\}^{2}\to 0.
	\end{equation*}
	\item[ii)] $F_{11}(\theta),F_{22}(\theta), F_{21}(\theta)F^{-1}_{11}(\theta)F_{12}(\theta),F_{12}(\theta)F^{-1}_{22}(\theta)F_{21}(\theta)$ are nonsingular for any $\theta \in \Theta_s$.
	\item[iii)] There exists $S\subseteq\{1,\ldots,K^{(2)}\}$ with $|S|\leq s$, such that the projection of $a_1^{\top} \partial_{\theta_1} \log f_{\theta^0}$ onto the lower-dimensional subspace $\mathcal{T}(S)$ is the same as onto the space $\mathcal{T}$.
\end{enumerate}
\end{assumption}
%We will show that the influence function $\tilde{\psi}_2(\theta) = {F}_{22,1}^{-1}(\theta)S_{2}({\theta})$ is the efficient influence function for $\theta_2$, as we have
%\begin{eqnarray*}
%&&\E({F}_{22,1}^{-1}(\theta)S_{2}({\theta})g( D_t, \theta)^{\top}  \Omega^{-1}(\theta) G_1(\theta)^\top)= 0, \\
%&&\E({F}_{22,1}^{-1}(\theta)S_{2}({\theta})g(D_t, \theta)^{\top}  \Omega^{-1}(\theta) G_2(\theta)^\top) =I_{K^{(2)}},
%\end{eqnarray*}
%where $S_{2}(\theta) = {G}_2(\theta)^\top{\Omega}^{-1}(\theta){g}(D_t,\theta) -  P_{\theta} {G}_1(\theta)^\top{\Omega}^{-1}(\theta){g}( D_t, \theta)$.
\begin{theorem}\label{se.theom1}
Under Assumptions \ref{den}-\ref{reg}, with a regular estimator sequence, the influence function $\tilde{\psi}_1(D_t,\theta)$ is efficient for $\theta_1(\P_{\theta})$, which is differentiable with respect to the tangent space $\mathcal{T}$ at $\P_{\theta^0}$ .
\end{theorem}

	\begin{proof}%[\textbf{Proof of Theorem \ref{se.theom1}}]
	Let $\mathcal{A}(\theta)$ be a $K\times q$ matrix and define $\mathcal{J}(\theta) \defeq \mathcal{A}(\theta)G(\theta)$. Consider the moment condition $\mathcal A(\theta)\E g(D_t,\theta)=0$. Differentiating the identity with respect to $\theta$ yields
	$$\frac{\partial\theta(\P_\theta)}{\partial\theta}\Big|_{\theta=\theta^0}=\langle[\mathcal J^{-1}(\theta^0)\mathcal A(\theta^0)g(D_t,\theta^0)]^\top,\partial_\theta\log f_{\theta^0}\rangle_{\P_{\theta^0}}.$$
	%		Recall the definition of the spaces $\mathcal{T},\mathcal{T}(S),\mathcal{T}_N,\mathcal{U}_N $.
	%		We let the moment $\mathcal{A}(\theta)\hat{g}(\theta, D_t)$ be a vector of score functions.
	%		We let $\mathcal{A}(\theta^0) =\mathcal{A}$,  $\mathcal{J} = \mathcal{A}G$, and recall that $G = [G_1, G_2]^{\top}$.
	%%We denote $G_1^{\top}(S) ( |S|\times q)$ as the submatrix of $G_1^{\top}$ with rows corresponding to the subset $S$, and $G_{1}(\theta,S)$ as the submatrix for $G_{1}(\theta)$ similarly.
	%		
	%		\begin{eqnarray*}
	%		&&\frac{\partial \theta_1(\P_{\theta_1, \theta_2}), \partial \theta_2(\P_{\theta_1, \theta_2})}{\partial \theta}|_{\theta=\theta^0} \nonumber \\
	%		&&= \langle \mathcal{J}^{-1} \mathcal{A}({g}(W_t,\theta^0)) , \nabla_{\theta_1,\theta_2} \log f(\theta^0)\rangle.
	%		\end{eqnarray*}
	%		
	%		It is easy to see that $\mathcal{A}$ shall be chosen as $G^{\top}\Omega^{-1}$ according to Theorem 3.1 of \cite{newey1990semiparametric}.
	
	According to the proof of Theorem 1 in \citet{chen2008semiparametric}, the optimal weights that lead to the efficient score take the form $\mathcal A(\theta) = G(\theta)^\top\Omega^{-1}(\theta)$, where $G(\theta)=\begin{bmatrix}G_1(\theta) & G_2(\theta)\end{bmatrix}$. Then, $\mathcal J(\theta)$ is given by
	$$\mathcal{J\text{(\ensuremath{\theta})}}=\begin{bmatrix}G_{1}(\theta)^{\top}\Omega^{-1}(\theta)G_{1}(\theta) & G_{1}(\theta)^{\top}\Omega^{-1}(\theta)G_{2}(\theta)\\
	G_{2}(\theta)^{\top}\Omega^{-1}(\theta)G_{1}(\theta) & G_{2}(\theta)^{\top}\Omega^{-1}(\theta)G_{2}(\theta)
	\end{bmatrix}=\begin{bmatrix}F_{11}(\theta) & F_{12}(\theta)\\
	F_{21}(\theta) & F_{22}\text{(\ensuremath{\theta})}
	\end{bmatrix}.$$
	It follows that 
	$$\mathcal{J}^{-1}(\theta)   =\begin{bmatrix}F_{11}(\theta)  & F_{12}(\theta) \\
	F_{21}(\theta)  & F_{22}(\theta) 
	\end{bmatrix}^{-1}=\begin{bmatrix}F_{1|2}^{-1}(\theta)   &- F_{1|2} ^{-1}(\theta) F_{12}(\theta)  F_{22}^{-1}(\theta) \\
	-F_{2|1}^{-1}(\theta) F_{21}(\theta)  F_{11}^{-1}(\theta)  & F_{2|1}^{-1}(\theta) 
	\end{bmatrix},$$
	where $F_{1|2}(\theta) = F_{11}(\theta)- F_{12}(\theta)F^{-1}_{22}(\theta)F_{21}(\theta)$, $F_{2|1}(\theta) = F_{22}(\theta)- F_{21}(\theta)F^{-1}_{11}(\theta)F_{12}(\theta)$. Thus the efficient influence function %$\mathcal J^{-1}(\theta^0)\mathcal A(\theta^0)\E g(D_t,\theta^0)$ 
	has the following form:
	\begin{eqnarray*}
		&&\mathcal J^{-1}(\theta)G(\theta)^\top\Omega^{-1}(\theta)g(D_t,\theta)\\
		&=&\begin{bmatrix}F_{1|2}^{-1}(\theta) & -F_{1|2}^{-1}(\theta)F_{12}(\theta) F_{22}^{-1}(\theta)\\
			-F_{2|1}^{-1}(\theta)F_{21}(\theta^0) F_{11}^{-1}(\theta^0) & F_{2|1}^{-1}(\theta)
		\end{bmatrix}\begin{bmatrix}G_{1}(\theta)^{\top}\Omega^{-1}(\theta)g(D_t,\theta)\\
			G_{2}(\theta)^{\top}\Omega^{-1}(\theta) g(D_t,\theta)
		\end{bmatrix}\\
		& = & \begin{bmatrix}F_{1|2} ^{-1}(\theta)G_{1}(\theta)^{\top}\Omega^{-1}(\theta)g(D_t,\theta)-F_{1|2}^{-1}(\theta)F_{12}(\theta) F_{22}^{-1}(\theta)G_{2}(\theta)^{\top}\Omega^{-1}(\theta)g(D_t,\theta)\\
			F_{2|1}^{-1}(\theta)G_{2}(\theta)^{\top}\Omega^{-1}(\theta)g(D_t,\theta)-F_{2|1}^{-1}(\theta)F_{21}(\theta)F_{11}^{-1}(\theta)G_{1}(\theta)^{\top}\Omega^{-1}(\theta)g(D_t,\theta)
		\end{bmatrix}.
	\end{eqnarray*}

	It can be seen that the efficient influence function for $\theta_1$ coincides with the one constructed by our decorrelated score, namely $\tilde\psi_1(D_t,\theta)=F_{1|2}^{-1}(\theta)\psi_1(D_t,\theta)$. In particular, $\tilde\psi_1(D_t,\theta^0)$ is orthogonal to $\mathcal{T}^N(S)$ and thus $\mathcal{T}^N$, i.e., lying within $\mathcal{U}^{N}$.
\end{proof}

\begin{assumption}\label{reg1}
Let $A(\theta)=G_1(\theta)^\top\Omega^{-1}(\theta) - \varPi(\theta)G_2(\theta)^\top\Omega^{-1}(\theta)$, and let $\hat A(\theta)$ denote the estimator of $A(\theta)$. Assume that $|\hat A(\tilde\theta_1,\theta_2^0)\partial_{\theta_1}\hat g(\tilde\theta_1,\theta_2^0) - A(\theta^0_1,\theta^0_2)G_1|_\infty=\smallO_\P(1)$, where $\tilde\theta_1$ is on the line segment connecting $\check\theta_1$ and $\theta^0_1$. Moreover, suppose the score function $\hat\psi_{1}(\theta^0)$ satisfies $\sqrt{n}\hat\psi_1(\theta^0)\stackrel{\mathcal{L}}{\rightarrow} \N\big(0, F_{1|2}(\theta^0)\big)$ and $\sqrt{n}\{\hat\psi_1(\theta_1^0,\hat\theta_2) - \hat\psi_1(\theta^0_1,\theta^0_2)\}=\smallO_\P(1)$ for a preliminary estimator $\hat\theta_2$.
%Let $B({\check{\theta}})  =  {G}_2^{\top}(\check{\theta}){\Omega}^{-1} (\check{\theta})-  P_{\check{\theta}} {G}_1^{\top}(\check{\theta}){\Omega}^{-1} (\check{\theta})$.
%Assume that $\|B{(\check{\theta})} - B{({\theta}^0)} \|_2 = \Co_p (1) $ and $\|\partial \hat{g}(D_t, {\tilde{\theta}})/\partial{\theta_2} - G_2 \|_2 = \Co_p (1) $, for some $\tilde{\theta}$ which is in between $\theta^0$ and $\check{\theta}$.
%Moreover, suppose a central limit theorem $\sqrt{n}S_{2,n}(\theta^0)\stackrel{\mathcal{L}}{\rightarrow} \N(0, F_{22,1}(\theta^0)^{-1})$ holds and $\sqrt{n}(S_{2,n}(\check{\theta}_2,\theta_1^0) - S_{2,n}(\theta^0)) = \Co_p(1)$.
\end{assumption}
%The conditions $|\hat A(\tilde\theta_1,\theta_2^0)\partial_{\theta_1}\hat g(\tilde\theta_1,\theta_2^0) - A(\theta^0_1,\theta^0_2)G_1|_\infty=\smallO_\P(1)$ and $\sqrt{n}\{\hat\psi_1(\theta_1^0,\hat\theta_2) - \hat\psi_1(\theta^0_1,\theta^0_2)\}=\smallO_\P(1)$ can be justified as in the proof of Theorem \ref{sp.theom}.
%The conditions above can be checked by the assumptions in the previous sections.

\begin{theorem}\label{se.theom2}
Under Assumption \ref{reg1} and given $|F_{1|2}^{-1}(\theta^0)|_\infty=\bigO(1)$, we have 
$$\sqrt{n}(\check{\theta}_1 - {\theta_1^0}) \stackrel{\mathcal{L}}{\rightarrow} \N \big(0, F_{1|2}^{-1}(\theta^0)\big).$$
\end{theorem}

\begin{proof}%[\textbf{Proof of Theorem \ref{se.theom2}}]
	By the definition of $\check\theta_1$ and the mean value theorem, we have
	\begin{equation*}
	0 = \hat\psi_1(\check\theta_1,\hat\theta_2) = \frac{\partial \hat\psi_1(\tilde\theta_1,\hat\theta_2)}{\partial \theta_1^\top}(\check\theta_1 - \theta_1^0) + \hat\psi_1(\theta_1^0,\hat\theta_2),
	\end{equation*}
	where $\tilde\theta_1$ is on the line segment connecting $\check\theta_1$ and $\theta_1^0$, and $\hat\theta_2$ is a preliminary estimator of $\theta_2^0$. It follows that
	\begin{eqnarray*}
		\sqrt{n}(\check\theta_1-\theta_1^0) &=& \sqrt{n}F_{1|2}^{-1}(\theta^0)\Big\{F_{1|2}(\theta^0) - \frac{\partial \hat\psi_1(\tilde\theta_1,\hat\theta_2)}{\partial \theta_1^\top}\Big\}(\check\theta_1 - \theta_1^0)\\
		&&-\sqrt{n}F_{1|2}^{-1}(\theta^0)\{\hat\psi_1(\theta_1^0,\hat\theta_2) - \hat\psi_1(\theta_1^0,\theta_2^0)\} - \sqrt{n}F_{1|2}^{-1}(\theta^0)\hat\psi_1(\theta_1^0,\theta_2^0).
	\end{eqnarray*}
	Recall that $\hat\psi_1(\theta) = \hat A(\theta)\hat g(\theta)$ and $F_{1|2}(\theta^0)=A(\theta^0)G_1$. Then, we have
	\begin{eqnarray*}
		\check\theta_1 - \theta_1^0 &=& F_{1|2}^{-1}(\theta^0)\{A(\theta^0)G_1 - \hat A(\tilde\theta_1,\theta_2^0)\partial_{\theta_1}\hat g(\tilde\theta_1,\theta_2^0)\}(\check\theta_1 - \theta_1^0)\\
		&& - F_{1|2}^{-1}(\theta^0)\{\hat\psi_1(\theta_1^0,\hat\theta_2) - \hat\psi_1(\theta_1^0,\theta_2^0)\} - F_{1|2}^{-1}(\theta^0)\hat\psi_1(\theta^0) + \smallO_\P(1),
	\end{eqnarray*}
	where the terms involving $\partial_{\theta_1}\hat A(\tilde\theta_1,\theta_2^0)$ are asymptotically negligible, as they are multiplied by $\hat g(\tilde\theta_1,\theta_2^0)$. Consequently, the asymptotic normality results follow from the assumptions stated in the theorem. 
	%	Since we have
	%	\begin{equation*}
	%	S{_{2,n}}(\hat{\theta}_2, \theta^0_1) = 0.
	%	\end{equation*}
	%	Then
	%	\begin{equation*}
	%	\frac{\partial S_{2,n} (\check{\theta}_2, \theta^0_1)}{\partial \theta_2} ^{\top} (\hat{\theta}_2 - \theta_2^0) + S_{2,n}({\theta}_2^0, \theta^0_1) = 0,
	%	\end{equation*}
	%	where $\check{\theta}_2$ is a point between $\hat{\theta}_2$ and $\theta_2^0$, and $\check{\theta}=(\check{\theta}_2^{\top},\theta_1^{0\top})^{\top}$.
	%	Thus we have
	%	\\$\sqrt{n} (\hat{\theta}_2 - \theta^0_2) = [F_{22,1}(\theta^0)^{-1}]\sqrt{n}[F_{22,1}(\theta^0) - B({\check{\theta}}) \partial \hat{g}({\check{\theta}})^{\top}/\partial{\theta_2}](\hat{\theta}_2 - \theta_2^0) - [B(\theta^0)G_2(\theta^0_2)]^{-1} \sqrt{n} [S_{2,n}(\check{\theta}_2, \theta^0_1) - S_{2,n}(\theta^0)] - \sqrt{n} F_{22,1}(\theta^0)^{-1} S_{2,n}(\theta^0)$. From conditions above then
	%	\\$\sqrt{n} (\hat{\theta}_2 - \theta^0_2) = - \sqrt{n} F_{22,1}(\theta^0)^{-1} S_{2,n}(\theta^0) + \Co_p(1)$. We can then apply a central limit theorem on $- \sqrt{n} F_{22,1}(\theta^0)^{-1} S_{2,n}(\theta^0)$ and the desired results follow.
\end{proof}

		\renewcommand{\thesubsection}{D.\arabic{subsection}}
		\setcounter{equation}{0}
		\renewcommand{\theequation}{D.\arabic{equation}}
		\setcounter{theorem}{0}
		\renewcommand{\thetheorem}{D.\arabic{theorem}}
		\setcounter{lemma}{0}
		\renewcommand{\thelemma}{D.\arabic{lemma}}
		\setcounter{figure}{0}
		\renewcommand{\thefigure}{D.\arabic{figure}}
		\setcounter{table}{0}
		\renewcommand{\thetable}{D.\arabic{table}}
		\setcounter{remark}{0}
		\renewcommand{\theremark}{D.\arabic{remark}}
		\setcounter{corollary}{0}
		\renewcommand{\thecorollary}{D.\arabic{corollary}}
		\setcounter{assumption}{0}
		\renewcommand{\theassumption}{D.\arabic{assumption}}
		\setcounter{example}{0}
		\renewcommand{\theexample}{D.\arabic{example}}
		
		\section{Supplementary Discussions}%{Supplementary Examples and Remarks}
		\subsection{General Model Framework}\label{general}
		%\begin{remark}[General Model Framework]
		Here we present a general model framework to which the main theorems shown in Section \ref{theoretical} also apply. Specifically, for $t = 1,\cdots,n$ and $j=1,\ldots,p$, we consider the stochastic equations system in the form of:
		\begin{eqnarray*}%\label{model1}
			y_{j,t} &=& x_{j,t}^{\top}{b_j} + \vps_{j,t}\\
			&=& \underbrace{x_{j,t}^\top B_{j}}_{\tilde x_{j,t}^\top} \beta^0_{j} + \vps_{j,t},
			\quad \E(z_{j,t} \vps_{j,t}) = 0,
		\end{eqnarray*}
		where $y_{j,t}$ is the scalar outcome,
		$x_{j,t}$ is a $K'_j$-dimensional vector of original covariates, and $\vps_{j,t}$ is a stochastic shock. %In addition to the model setup presented above, we allow the general model to be dynamic such that lagged values of $y_{j,t}$ can be included in the covariates. 
		In the second line, we express the $K'_j$-dimensional vector $b_j$ by $B_j\beta_j$, where $B_j$ is an observed $K'_j\times K_j$ matrix and $\beta_j^0$ is a $K_j\times 1$ vector ($K_j\leq K'_j+1$). Denote by $\tilde x_{j,t}^\top\in\R^{K_j}$ the transformed covariates $x_{j,t}^\top B_j$ in the $j$-th equation. The error is assumed to be orthogonal to a vector $z_{j,t}$ of instrumental variables with dimension of at least $K_j$.
		
		To show that the spatial panel network model presented in Section \ref{system} (Eq. \eqref{spn}) is covered by the general model framework, we introduce the following notations. Denote by $e_j$ the $p\times1$ unit vector with the $j$-th element is equal to $1$. Define $\bm X_t =[e_j^{\top}\otimes  y_{t}^{\top}]_{j=1}^p$ ($p\times p^2$), $\widetilde{\bm B}_{p^2\times(p^2+1)}=([e_j^{\top}\otimes\mathbf I_p]_{j=1}^p, [w_j]_{j=1}^p)$, and $\widetilde{\bm\beta}^0=(\delta_1^{0\top},\ldots,\delta_p^{0\top},\rho^0)^\top$, where the notation $[A_j]_{j=1}^p$ indicates we stack $A_j$ by rows over $j=1,\ldots,p$. The model in the compact form can be expressed as:
		$$y_t = \bm X_t\bm B \bm\beta^0 %+ u_t\vartheta^0 
		+ \vps_{t},$$
		where $\bm B$ is $\widetilde{\bm B}$ with the ($pj^*+k^*$)-th column removed, and $\bm\beta^0$ is $\widetilde{\bm\beta}^0$ with the ($pj^*+k^*$)-th element removed. This linear model fits within the general framework, with $\bm X_{t}$ %and $u_t$ are 
		containing the original covariates and $\bm X_{t}\bm B=:\widetilde{\bm X}_t$ representing the transformed covariates. %Particularly, we might be interested in testing $\bm\beta^0$, including the spatial autoregressive parameter and the network structure. The coefficients on the control variables $\vartheta$ would be classified into nuisance parameters. 
		
		\subsubsection{Identification Conditions}
		We note that the model in \ref{general} involves a transformation of the covariates. Therefore, it is important to discuss how the identification conditions on the original covariates translate to the transformed covariates, which are relevant for sparsity-based estimation methods.
		
		%\subsubsection{Identification for a Simple Model}\label{a2.1}
		%{\red The identification issues within a general framework are explored in Section \ref{sec.id}.} 
		%Above we have discussed the identification issues within a general framework. In this subsection, 
		To begin, we examine a simple high-dimensional linear regression model with a scalar outcome to closely analyze the necessary conditions on the design matrix. %, which are tied to the sparsity-based estimation methods. 
		Consider the following regression model with high-dimensional exogeneous covariates $x_t\in\R^p$ and a scalar outcome $y_t$:
		\begin{equation}\label{simple}
		y_{t}=x_{t}^{\top}\underbrace{(\rho w+\delta)}_{b}+\varepsilon_{t}, \quad \E(x_t \varepsilon_{t}) = 0,\quad t=1,\ldots,n,
		\end{equation}
		where $b=(b_{k})_{k=1}^{p}$ is a $p \times 1$ parameter vector given by the pre-specified vector $w=\left(w_{k}\right)_{k=1}^{p}$, times the \text{effect size} $\rho$, and an (approximately) sparse deviation $\delta=\left(\delta_{k}\right)_{k=1}^{p}$ from this vector. %We do not impose any sparsity restriction on $w$, and on $b$ correspondingly. In particular, while $b$ itself is dense, we think that it is a sparse deviation $\delta$ from a focal dense structure $\rho w$.
		The objective is to perform estimation and inference on parameter $\rho$ or any components of $\delta$.
		
		Furthermore, we can rewrite $\rho w+\delta = B_{p \times (p+1)} \beta_{(p+1) \times 1}^{0}$, where $B= [w, \mathbf I_p]$, and $\beta^0 = (\rho, \delta^{\top})^{\top}$. That is, the first column of $B$ is given as $w$ and the first element in the vector $\beta^{0}$ is $\rho$, and the remaining $\delta$ measures the extent of sparse deviation. %We therefore posit that $\beta^{0}$ is approximately sparse, but that $b$ is not necessarily sparse. 
		With these definitions, we obtain the model:
		\begin{equation}\label{simple2}
		y_{t}=x_{t}^{\top} B \beta^{0}+\varepsilon_{t}.
		\end{equation}	
		
		%Our goal is to perform high-quality estimation and inference on parameter $\rho$ or any components of $\delta$ in this framework.  %Relevant identification conditions, which are tied to the sparsity-based estimation methods, will be discussed in Section \ref{a2.1} in the Appendix. 
		Estimation will employ regularized estimators of $\beta^0$. %, such as the Dantzig selector estimator defined in \eqref{simple.dan}, and then performing debiasing of one parameter or a set of parameters of interest, such that the resulting estimator is approximately unbiased and approximately Gaussian.  %In a general version of the model, we will also allow for endogenous determination of $x_t$, in which case we will need instrumental variables (IV) $z_t$ that are orthogonal to $\varepsilon_t$.  Further, we consider many equations framework with stochastic shocks exhibiting temporal and spatial (cross-equation) dependencies. 
		%As discussed in Section \ref{simpleexample}, we shall estimate $\beta^0$ in \eqref{simple2} by regularization. For instance, 
		Define $y=(y_t)_{t=1}^n$ and $X=(x_t^\top)_{t=1}^n$. For instance, Given $|\beta^0|_0=\smallO(n)$, a Dantzig selector estimator is defined as the solution to the following program:
		\begin{equation}\label{simple.dan}
		\min_{\beta} |\beta|_{1} \quad \text{subject to}\quad  |B^\top X^\top(y - XB\beta)|_{\infty} \leq \lambda,
		\end{equation}
		where $\lambda>0$ is the tuning parameter.
		
		Now the question is what condition we need to impose on $X$ such that a restricted isometry property (RIP) or restricted eigenvalue (RE) condition is ensured on the design matrix $XB$. Also, it may be helpful to understand the format of $B$ as well. For example, when $p=4$, $B^{\top}$ can take the form
		$\begin{pmatrix}
		1/4 & 1/4& 1/4 & 1/4\\
		1& 0 & 0&0\\
		0&1&0&0\\
		0&0&1&0\\
		0&0&0&1
		\end{pmatrix}$.
		
		\begin{remark}[{Restriction on $B$, for fixed design}]
			We notice that for the full rank matrix $X$, if there exists {a full rank matrix $A_{p\times (p-n)}$ ($\rank(A) = p-n$)}, such that $XA =0$ (i.e., the columns of $A$ form the null space of $X$), then for each $\xi \neq \xi^0$ ($\xi, \xi_0\in\R^p$), we can find a non-zero vector $\eta\in \mathbb{R}^{p-n}$ such that $\xi = A \eta+ \xi_0$, if we have $X\xi = X\xi_0$. Thus, we shall restrict the columns of $B$ such that they do not belong to the space spanned by the columns of $A$, namely there does not exist a column of $B$, $B_i$, such that $B_i  = A \eta $.
		\end{remark}
		
		%{ \red need to declare the notation of $\lambda_i(\cdot)$, $\sigma_i(\cdot)$ and $|\cdot|_s$, $|\cdot|_\infty$ at the beginning. }
		
		The RE for $XB$ in the case that $X$ is deterministic is discussed in the following lemma.
		
		\begin{lemma}\label{fixm}
			Let $B_{\mathcal I}$ denote the sub-matrix of $B$ containing columns indexed by the set $\mathcal I\subseteq\{1,\ldots,p+1\}$, where the cardinality $|\mathcal I|$ is given by $s$ ($s\leq n$). Define $\mathcal{V}_{B_{\mathcal I}} \defeq \{\xi \in \mathbb{R}^p: \xi=B_{\mathcal I} \xi_{\mathcal I}, \xi_{\mathcal I} \in \mathcal{S}^{s-1}\}$, where $\mathcal{S}^{s-1}$ denotes the unit Euclidean sphere, i.e., $\xi_{\mathcal {\mathcal I}}$ is an $s\times 1$ unit vector with $|\xi_{\mathcal I}|_2=1$.
			If $B_{\mathcal I}$ is of rank $s$ for any such ${\mathcal I}$, and $cn\leq \tilde{\lambda}_{s,B}\leq \lambda_1(X^{\top} X)\leq Cn$ for some $c,C>0$, where $\tilde{\lambda}_{s,B}\defeq \underset{{\mathcal I}:|{\mathcal I}|=s}{\min} \,\underset{\xi \in \mathcal{V}_{B_{\mathcal I}}}{\min}\frac{\xi^{\top} X^{\top} X\xi}{\xi^{\top}\xi}$, then there exist constants $c',C'>0$ such that 
			$$c'\sqrt{n}\leq \sigma_s(XB_{\mathcal I})\leq \sigma_1(XB_{\mathcal I})\leq C'\sqrt{n}.$$
			%the matrix $XB$ satisfies the {\red RIP/RE} of order $s$.
		\end{lemma}

		\begin{proof}%[\textbf{Proof of Lemma \ref{fixm}}]
			% Note that proving the {\red RIP/RE} (of order $s$) for $XB$ is equivalent to showing that there exist constants $c',C'>0$ such that $c'\leq \sigma_s(XB_{\mathcal I})\leq \sigma_1(XB_{\mathcal I})\leq C'$ for any subset $\mathcal I$ with $|\mathcal I|\leq s$.
			
			%Now we prove that if $B_I$ is of rank $s$ for any $I$, and $c' \leq \lambda_{s}(X^\top X)\leq \lambda_1(X^\top X)\leq C'$, then we have the RIP for  $XB$.
			
			Let $\mV:\dim(\mV)=j$ be a subspace of $\R^s$ of dimension $j$, $j=1,\ldots,s$. Due to the Min-max theorem for singular values, we have \begin{align*}
			\lambda_s(B_{\mathcal I}^{\top}X^{\top}XB_{\mathcal I}) = \sigma^2_s(XB_{\mathcal I}) &= \max\limits_{\mV:\dim(\mV) = s}\,\min\limits_{\substack{\xi_{\mathcal I} \in \mV,\\ \xi_{\mathcal I}^\top\xi_{\mathcal I}=1}}\,\xi_{\mathcal I}^{\top}B_{\mathcal I}^{\top}X^{\top}XB_{\mathcal I} \xi_{\mathcal I} \\
			&= \min\limits_{\mV:\dim(\mV) = 1}\,\max\limits_{\substack{\xi_{\mathcal I} \in \mV,\\ \xi_{\mathcal I}^\top\xi_{\mathcal I}=1}}\,\xi_{\mathcal I}^{\top}B_{\mathcal I}^{\top}X^{\top}XB_{\mathcal I} \xi_{\mathcal I}.
			\end{align*}
			
			For any fixed $\mV:\dim(\mV)=s$, we have
			\begin{eqnarray*}
				\min\limits_{\substack{\xi_{\mathcal I} \in \mV,\\ \xi_{\mathcal I}^\top\xi_{\mathcal I}=1}}\frac{\xi_{\mathcal I}^{\top} B_{\mathcal I}^{\top} X^{\top}XB_{\mathcal I} \xi_{\mathcal I}}{\xi_{\mathcal I}^{\top}\xi_{\mathcal I}}&=& \min\limits_{\substack{\xi_{\mathcal I} \in \mV,\\ \xi_{\mathcal I}^\top\xi_{\mathcal I}=1}} \frac{\xi_{\mathcal I}^{\top} B_{\mathcal I}^{\top} X^{\top}X B_{\mathcal I} \xi_{\mathcal I}}{\xi_{\mathcal I}^{\top}B_{\mathcal I}^{\top}B_{\mathcal I}\xi_{\mathcal I}}\frac{\xi_{\mathcal I}^{\top} B_{\mathcal I}^{\top} B_{\mathcal I} \xi_{\mathcal I}}{\xi_{\mathcal I}^{\top}\xi_{\mathcal I}}\\
				&\geq&\min_{\xi \in \mathcal{V}_{B_{\mathcal I}}}\frac{\xi^{\top} X^{\top}X \xi}{\xi^{\top}\xi} \lambda_s(B_{\mathcal I}^{\top}B_{\mathcal I})\\
				& \geq & \tilde{\lambda}_{s,B}\lambda_s(B_{\mathcal I}^{\top}B_{\mathcal I}),
			\end{eqnarray*}
			where the last inequality is due to the definition of $\tilde{\lambda}_{s,B}$ and the full rank property of $B_{\mathcal I}^{\top}B_{\mathcal I}$, which implies $\lambda_{\min}(B_{\mathcal I}^{\top}B_{\mathcal I}) = \lambda_{s}(B_{\mathcal I}^{\top}B_{\mathcal I})$ is positive.
			As the above inequality holds for any subspace $\mathcal{V}$ of dimension $s$, thus we have
			$\lambda_s(B_{\mathcal I}^{\top}X^{\top}XB_{\mathcal I})\geq \tilde{\lambda}_{s,B}\lambda_s(B_{\mathcal I}^{\top}B_{\mathcal I})$.
			Similarly, we have $\lambda_1(B_{\mathcal I}^{\top}X^{\top}XB_{\mathcal I})\leq \lambda_1(X^{\top}X) \lambda_1(B_{\mathcal I}^{\top}B_{\mathcal I})$.
			
			Given $\lambda_s(B_{\mathcal I}^{\top}X^{\top}XB_{\mathcal I})= \sigma_s^2(XB_{\mathcal I})$, we have proved that if $B_{\mathcal I}$ is of rank $s$ for any ${\mathcal I}$, and $cn \leq \tilde{\lambda}_{s,B}\leq \lambda_1(X^\top X)\leq Cn$, then the RE for $XB$ follows.
		\end{proof}
		
		Next, we provide another Lemma for the random design $X$ with i.i.d. sub-Gaussian entries.
		We define $\|Z\|_{\psi_1} =\inf\{u>0: \E \exp(|Z|/u) \leq 2 \}$ and  $\|Z\|_{\psi_2} =\inf\{u>0: \E \exp(|Z|^2/u^2) \leq 2 \}$ as the sub-exponential norm and sub-Gaussian norm of the random variable $Z$.
		
		\begin{lemma}\label{random}
			Let $X$ be an $n\times p$ matrix whose rows $\{X_t\}_{t=1}^n$ are independent, mean-zero, sub-Gaussian isotropic random vectors in $\R^p$. Suppose that for any subset $\mathcal I\subseteq\{1,\ldots,p+1\}$ with $|\mathcal I|\leq s$, the column-indexed sub-matrix $B_{\mathcal I}$ satisfies $c\leq \sigma_s (B_{\mathcal I})\leq \sigma_1(B_{\mathcal I})\leq C$ for some $c,C>0$. Assume further that $n \gg s \log (pe/s)$. Then, with probability approaching one, for any such $\mathcal I$, we have:
			\begin{equation*}
			c_K \sqrt{n} \leq \sigma_s(XB_{\mathcal I}) \leq \sigma_1(XB_{\mathcal I}) \leq C_K \sqrt{n},
			\end{equation*}
			where $c_K, C_K$ are positive constants depending on $K=\max_{1\leq t\leq n}\|X_t\|_{\psi_2}$.
		\end{lemma}
		
		\begin{proof}%[\textbf{Proof of Lemma \ref{random}}]
			%We define $\|Z\|_{\psi_1} =\inf\{t<0: \E \exp(|X|/t) \leq 2 \}$ and  $\|X\|_{\psi_2} =\inf\{t<0: \E \exp(|X|^2/t^2) \leq 2 \}$ as the sub-exponential norm and sub-Gaussian norm of the random variable $Z$.
			
			\underline{Step 1:} For $\xi \in \mathcal S^{s-1}$, where $\mathcal{S}^{s-1}$ denotes the unit Euclidean sphere, i.e. $|\xi|_2=1$, we first show that $\xi^{\top}B_{\mathcal I}^{\top} X_tX_t^{\top}B_{\mathcal I} \xi $ is concentrated around its mean $\xi^{\top}B_{\mathcal I}^{\top}B_{\mathcal I} \xi$. Let $U_{t} \defeq \xi^{\top}B_{\mathcal I}^{\top} X_tX_t^{\top}B_{\mathcal I} \xi - \xi^{\top}B_{\mathcal I}^{\top} B_{\mathcal I} \xi $. By Bernstein inequality, we have
			\begin{eqnarray*}
				\P\Big(\Big|n^{-1} \sum_{t=1}^n U_{t}\Big|\geq \vps/2\Big) &\leq& 2\exp\Big(-c \min \Big(\frac{\vps^2n^2}{\sum_{t=1}^n\|U_{t}\|_{\psi_1}^2}, \frac{\vps n}{\max\limits_{1\leq t\leq n} \|U_{t}\|_{\psi_1}}\Big)\Big).%\\
				%2 \,\exp\Big(-c\min\Big(\frac{\vps^2}{K^2}, \frac{\vps}{K}n\Big)\Big),
			\end{eqnarray*}
			%with $K$ as constant to be defined later.
			By utilizing the properties of sub-Gaussian and sub-exponential random variables, we have
			\begin{eqnarray*}
				\|U_{t}\|_{\psi_1} \leq  C_1 \|(\xi^{\top} B_{\mathcal I}^{\top}X_t)^{2}\|_{\psi_1} &=& C_1 \|\xi ^{\top}B_{\mathcal I}^{\top} X_t\|_{\psi_2}^2  \\
				&\leq & C_1 \Big\|\sum_{j=1}^s \xi_j B_{\mathcal I,j}^{\top}X_t\Big\|^2_{\psi_2}\\
				&\leq &C_1 \sum_{j=1}^s \xi_j^{2}\|B_{\mathcal I,j}^{\top}X_t\|^2_{\psi_2}\\
				&\leq& C_1 \max_{1\leq j\leq s} \|B_{\mathcal I,j}^{\top}X_t\|_{\psi_2}^2\leq C_2 \lambda_{\max}(B_\mathcal I^\top B_\mathcal I) =: K,
			\end{eqnarray*}
			where $B_{\mathcal I,j}, j=1,\ldots,s$ is the $j$-th column vector of $B_\mathcal I$ and the last inequality follows given $\max\limits_{1\leq j\leq s}\|B_{\mathcal I,j}^{\top} X_t\|^2_{\psi_2} \leq \max\limits_{1\leq j\leq s}|B_{\mathcal I,j}|_2^2 \max\limits_{1\leq k\leq p} \|X_{t,k}\|_{\psi_2}^2\leq C\lambda_{\max}(B_\mathcal I^{\top} B_{\mathcal I})$. It follows that
			\begin{eqnarray*}
				\P\Big(\Big|n^{-1} \sum_{t=1}^n U_{t}\Big|\geq  \vps/2\Big)
				&\leq& 2 \exp\Big(-c\min\Big(\frac{\vps^2}{K^2}, \frac{\vps}{K}n\Big)\Big).
			\end{eqnarray*}
			
			\underline{Step 2:} Let $\overline\sigma \defeq \sigma_1(B_{\mathcal I})$, and $\underline\sigma = \sigma_s(B_\mathcal I)$, which are bounded positive constants.
			%We need to show that $|n^{-1}B_{I}^{\top}X^\top XB_I - B_I ^{\top}B_I|_2\leq \vps$ holds with high probability implies $\vps{\sqrt{n}}(\underline \sigma - \vps/\underline \sigma)\leq \sigma_s(XB_I) \leq \sigma_1(XB_I ) \leq \vps{\sqrt{n}}(\overline \sigma + \vps/\underline \sigma)$ holds with the same probability.
			
			Note that
			\begin{equation*}
			|n^{-1}B_\mathcal I^{\top}X^\top X B_\mathcal I - B_\mathcal I^{\top} B_\mathcal I|_2= \sup_{\xi \in \mathcal S^{s-1}} |{n}^{-1} \xi^{\top} B_\mathcal I^{\top}X^\top X B_\mathcal I\xi - \xi^{\top} B_\mathcal I^{\top}B_\mathcal I \xi|.
			\end{equation*}
			Moreover, for any $\xi\in\mathcal S^{s-1}$, we have
			\begin{eqnarray*}
				&&\big|n^{-1}|\xi^{\top} B_\mathcal I^{\top} X^\top|_2^2- |\xi^{\top} B_\mathcal I^{\top}|_2^2\big |\\
				&=& \big|1/\sqrt{n} |\xi^{\top} B_\mathcal I^{\top} X^\top|_2 - |\xi^{\top}B_\mathcal I^{\top}|_2\big|\big(1/\sqrt{n} |\xi^{\top} B_\mathcal I^{\top} X^\top|_2 + |\xi^{\top}B_\mathcal I^{\top}|_2\big)\\
				&\geq& |\xi^{\top} B_\mathcal I^\top|_2 \big|1/\sqrt{n} |\xi^{\top} B_\mathcal I^{\top} X^\top|_2 - |\xi^{\top}B_\mathcal I^\top|_2|\big| \geq \underline\sigma\big|1/\sqrt{n} |\xi^{\top} B_\mathcal I^{\top} X^\top|_2 - |\xi^{\top}B_\mathcal I^\top|_2|\big|.
			\end{eqnarray*}
			Therefore, we have shown that $|n^{-1}B_{\mathcal I}^{\top}X^\top XB_\mathcal I - B_\mathcal I ^{\top}B_\mathcal I|_2\leq \vps$ holds with high probability implies $\vps{\sqrt{n}}(\underline \sigma - \vps/\underline \sigma)\leq \sigma_s(XB_\mathcal I) \leq \sigma_1(XB_\mathcal I ) \leq \vps{\sqrt{n}}(\overline \sigma + \vps/\underline \sigma)$ holds with the same probability.
			
			\underline{Step 3:} By applying the Corollary 4.2.13 of \citet{vershynin2019high}, we can find a $1/4$-net $\mathcal N$ of the unit sphere $\mathcal S^{s-1}$ with cardinality $|\mathcal N|\leq 9^s$. By the discretized property of the net, we have
			\begin{eqnarray*}
				|n^{-1}B_\mathcal I^{\top}X^{\top}X B_\mathcal I - B_\mathcal I^{\top} B_\mathcal I|_2& =& \sup_{\xi \in \mathcal S^{s-1}} |{n}^{-1} \xi^{\top} B_\mathcal I^{\top}X^{\top}X B_\mathcal I\xi - \xi^{\top} B_\mathcal I^{\top}B_\mathcal I \xi|\\
				&\leq & 2 \sup_{\xi \in \mathcal{N}} |{n}^{-1} \xi^{\top} B_\mathcal I^{\top}X^{\top}X B_\mathcal I\xi - \xi^{\top} B_\mathcal I^{\top}B_\mathcal I \xi|.
			\end{eqnarray*}
			Using the union bounds, we obtain
			\begin{equation*}
			\P\Big(\sup_{\xi \in \mathcal{N}} |{n}^{-1} \xi^{\top} B_\mathcal I^{\top}X^{\top} XB_\mathcal I\xi - \xi^{\top} B_\mathcal I^{\top}B_\mathcal I \xi|\geq \vps/2\Big) \leq 2\cdot 9^s\exp\big(-c\min(\vps^2/K^2, \vps/K)n\big).
			\end{equation*}
			We have proved that the pointwise concentration in Step 1 implies that $|n^{-1}B_\mathcal I^{\top}X^\top X B_\mathcal I - B_\mathcal I^{\top} B_\mathcal I|_2 \leq \vps$ holds with high probability.
			
			\underline{Step 4:} By Step 2 and 3 we know that provided $n\min(\vps^2/K^2, \vps/K)\gg s\log 9$ we can get
			$$\vps{\sqrt{n}}(\underline \sigma - \vps/\underline \sigma)\leq \sigma_s(XB_\mathcal I) \leq \sigma_1(XB_\mathcal I ) \leq \vps{\sqrt{n}}(\overline \sigma + \vps/\underline \sigma)$$
			holds with probability $2\exp\big(-c'\min(\vps^2/K^2, \vps/K)n\big)$.
			In addition, we know that there are ${p\choose s} \leq (pe/s)^s$ possible subsets $\mathcal I$. %among the $p$-dimensional covariates.
			Thus, by the union bounds, we can bound the probability by
			$$1- 2 \exp\big(-c' \min(\vps^2/K^2, \vps/K)n + s \log (pe/s)\big).$$
		\end{proof}
		
		% We have shown in a simple high-dimensional linear regression case that our framework goes through with a modified design matrix. The identification issues under the general model are discussed in Section \ref{sec.id}.
		%In the following, we will set the general framework and discuss the property of the estimated network. In a second step, we perform high-quality inference of our procedure covering spatial and temporal dependence.
		
		Finally, we extend the discuss to the GMM framework. Recall that $G=\partial_{\theta^\top}g(\theta)|_{\theta=\theta^0}$. The quantities $\sigma_{\min}(m,G)$ and $\sigma_{\max}(m,G)$ represent the $m$-sparse smallest and largest singular values of $G$, where $m\geq s$. As demonstrated in Section \ref{sec.id}, the boundedness of $\sigma_{\min}(m,G)$ and $\sigma_{\max}(m,G)$ is essential for satisfying the identification assumption \hyperref[A_id]{(A6)}. 
		%Denote $G_{H,I}$ as the sub-matrix of $G$ with rows and columns indexed respectively by the sets $H\subseteq\{1,\ldots,q\}$ and $I\subseteq\{1,\ldots,K\}$, where $|I|\leq|H|$. Let $\sigma_{\min}(m,G)=\min\limits_{|I|\leq m} \max\limits_{|H|\leq m} \sigma_{\min} (G_{H,I})$ and $\sigma_{\max}(m,G) = \max\limits_{|I|\leq m} \max\limits_{|H|\leq m} \sigma_{\max}(G_{H,I})$ be the $m$-sparse smallest and largest singular values of $G$ ($m\geq s$), where $\sigma_{\min} (G_{H,I})$ and $\sigma_{\max}(G_{H,I})$ are the smallest and largest singular values of $G_{H,I}$ respectively. %Recall that the transformed matrix $G$ is a block diagonal matrix whose $j$-th block is given by the $q_j\times K_j$ matrix $G_{[j]} = -\E(z_{j,t}\tilde x_{j,t}^{\top})$. 
		The following lemma shows how these singular values of the sub-matrices of $G$ are bounded under certain conditions. 
		
		% {\color{red}
		% The assumption in Lemma 3.2 is a bit abstract. More explanations would be helpful. In my opinion, it would be helpful to compare it with the full-rank condition of $G$ in a simple case with $p = 1$ ($\E[z_t(y_t-\tilde x_t^\top\beta)]=0)$ with $G=-\E[z_t\tilde x_t^\top]$). It would be helpful if you could explain the identification result in terms of (i) the IV validity and (ii) the IV relevance in your model. In particular, I think the discussion regarding the IV relevance helps the readers.
		% }
		
		% {\color{blue} I think there is no relationship with IV validity, but I think it is related to iv relevance.
		% Suppose that one has $K = 1$???? and $q$ to be large, then $\Sigma_{xz}$ is a vector.  So the condition basically restrict any weighted combination of the sub vector of  $\Sigma_{xz}$ to be non-zero using the weight $\xi$. If $\Sigma_{xz}$  is a matrix then the assumption requires that the submatrices collecting $m$ subcolumns of $G$ to be of full rank. If $m=p$, this condition implies that $G$ is of full rank.
		% }
		\begin{lemma}\label{lemma.id1}
			Suppose $G$ can be expressed as $G = \Sigma^{xz} B$, where $\Sigma^{xz}$ is a $q \times K$ matrix and $B$ is a $K\times K$ matrix. Let $B_{\mathcal I}$ denote the sub-matrix of $B$ containing columns indexed by the set $\mathcal I\subseteq\{1,\ldots,K\}$, and $\Sigma^{xz}_{\mathcal H}$ denote the sub-matrix of $\Sigma^{xz}$ containing rows indexed by the set $\mathcal H\subseteq\{1,\ldots,q\}$. Assume that there exist constants $c_1,c_2>0$ such that:
			$$\min\limits_{\mathcal I:|\mathcal I|\leq m} \lambda_{\min}(B_\mathcal I^{\top}B_\mathcal I)> c_1,$$
			and 
			$$\sigma_{\min,B}(m,\Sigma^{xz})\defeq\min\limits_{\mathcal I:|\mathcal I|\leq m} \max\limits_{\mathcal H:|\mathcal H|\leq m} \min\limits_{\xi \in \mathcal{V}_{B_\mathcal I}}\frac{\xi^{\top} {\Sigma_\mathcal H^{xz}}^\top  \Sigma_\mathcal H^{xz}\xi}{\xi^{\top}\xi}>c_2,$$
			where $\mathcal{V}_{B_\mathcal I} \defeq \{ \xi:\xi = B_\mathcal I\xi_\mathcal I,\xi_\mathcal I^\top\xi_\mathcal I=1\}$. Additionally, assume that there exist constants $C_1,C_2>0$ such that: $$\max\limits_{\mathcal H:|\mathcal H|\leq m} \lambda_{\max}({\Sigma_\mathcal H^{xz}}^\top \Sigma_\mathcal H^{xz})< C_1,\, \text{ and } \,\max\limits_{\mathcal I:|\mathcal I|\leq m} \lambda_{\max}( B_\mathcal I^{\top}B_\mathcal I)<C_2.$$
			Under these assumptions, there exist constants $c',C'>0$ such that $\sigma_{\min}(m,G)>c'$ and $\sigma_{\max}(m,G)\leq C'$.
		\end{lemma}
		
		\begin{proof}%[Proof of Lemma \ref{lemma.id1}]
			Similarly to the proof of Lemma \ref{fixm}, we observe that
			\begin{align*}
			\sigma^2_{\min} (\Sigma^{xz}_\mathcal H B_\mathcal I) &= \lambda_{\min} (B_\mathcal I^{\top}{\Sigma_\mathcal H^{xz}}^\top\Sigma^{xz}_\mathcal H B_\mathcal I)\geq  \min_{\xi \in \mathcal{V}_{B_\mathcal I}}\frac{\xi^{\top} {\Sigma_\mathcal H^{xz}}^\top\Sigma^{xz}_\mathcal H\xi}{\xi^{\top}\xi}\lambda_{\min}(B_\mathcal I^{\top}B_\mathcal I),\\
			\sigma^2_{\max} (\Sigma^{xz}_\mathcal H B_\mathcal I) &= \lambda_{\max} (B_\mathcal I^{\top}{\Sigma_\mathcal H^{xz}}^\top\Sigma^{xz}_\mathcal H B_\mathcal I)\leq \lambda_{\max}(B_\mathcal I^{\top}B_\mathcal I)\lambda_{\max}({\Sigma_\mathcal H^{xz}}^\top \Sigma_\mathcal H^{xz}).
			\end{align*}
			Consequently, we have
			\begin{align*}
			\min_{\mathcal I:|\mathcal I|\leq m}\max_{\mathcal H:|\mathcal H|\leq m}\sigma^2_{\min} (\Sigma^{xz}_\mathcal H B_\mathcal I) &\geq \sigma_{\min,B}(m,\Sigma^{xz})\min_{\mathcal I:|\mathcal I|\leq m}\lambda_{\min}(B_\mathcal I^{\top}B_\mathcal I),\\
			\max_{\mathcal I:|\mathcal I|\leq m}\max_{\mathcal H:|\mathcal H|\leq m}\sigma^2_{\max} (\Sigma^{xz}_\mathcal H B_\mathcal I) &\leq \max_{\mathcal I:|\mathcal I|\leq m}\lambda_{\max}(B_\mathcal I^{\top}B_\mathcal I)\max_{\mathcal H:|\mathcal H|\leq m}\lambda_{\max}({\Sigma_\mathcal H^{xz}}^\top \Sigma_\mathcal H^{xz}).
			\end{align*}
			It follows that for some constants $c',C'>0$, we have $\sigma_{\min}(m,G)>c'$ and $\sigma_{\max}(m,G) <C'$.
		\end{proof}
		
		To understand the required assumptions in this lemma, consider a simple example with $p=1$ and moment conditions $\E[z_t(y_t-\tilde x_t^\top\beta^0)]=0$, where $\beta^0$ is a $K\times1$ vector, $z_t$ contains $q$ instrumental variables, and $\tilde x_t^\top=x_t^\top B$ for a $K\times K$ transformation matrix $B$. In this case, the Jacobian matrix is given by $G=-\E(z_t\tilde x_t^\top)=-\E(z_tx_t^\top)B$. Corresponding to Lemma \ref{lemma.id1}, we have $\Sigma^{xz}$ is in the form of $-\E(z_tx_t^\top)$, which reflects to the correlation between $x_t$ and $z_t$.
		
		Following the intuition presented in Lemma \ref{fixm} and Lemma \ref{random}, in order to ensure that the $m$-sparse singular values of $G$ are bounded, the column-indexed sub-matrix $B_{\mathcal I}$ needs to be full rank for any $|\mathcal I|\leq m\leq K$. Moreover, we require the matrix $\E(z_tx_t^\top)$ to satisfy the RE with the weighting vector $\xi$ involving $B_\mathcal I$. Specifically, %for any weighting vector involving the corresponding transformation, 
		the boundedness of the singular values of the row-indexed sub-matrix of $\E(z_tx_t^\top)$ is crucial, as it relates to the relevance of the IVs.

		\subsection{Generalization on the Dependency of the Error Terms}\label{error}
		
		%{\color{red} ADD answer to AE comments  point 6)}\\
		
		We note the assumption \hyperref[A_error]{(A3)} can be generalized to accommodate serial correlation, unobserved heterogeneity, and factor structures.
		\begin{enumerate}
			\item[(i)] The m.d.s. assumption implies that the innovations are unpredictable given the past information. Relaxing this condition would require computing the long-run variance-covariance matrix of the score functions. Consequently, a consistent estimator of the precision matrix would be needed both to construct the debiased estimator and to compute the standard errors for inference. 
			
			Specifically, when the process $[z_{j,t}\vps_{j,t}]_{j=1}^p$ exhibits serially correlation and conditional heteroskedasticity, the matrix $\Omega$ takes the general form:
			$$\Omega=\sum_{\ell=-\infty}^{\infty} \E\big[[z_{j,t}\vps_{j,t}]_{j=1}^p([z_{j,t-\ell}\vps_{j,t-\ell}]_{j=1}^p)^\top\big]=\Gamma_0+\sum_{\ell=1}^{\infty}(\Gamma_\ell+\Gamma_\ell^{\top}),$$
			where $\Gamma_\ell\defeq \E\big[[z_{j,t}\vps_{j,t}]_{j=1}^p([z_{j,t-\ell}\vps_{j,t-\ell}]_{j=1}^p)^\top\big]$ denotes the lag-$\ell$ auto-covariance matrix. Constructing a consistent sample estimator of $\Omega$ and its inverse is more involved. First, it is necessary to justify sufficiently weak temporal dependency (such as an algebraic decay rate of the auto-covariance matrix) to ensure summability and to facilitate truncation of the long-run covariance estimation. Additionally, in the high-dimensional setting, structural assumptions such as sparsity (to control the order of elements in $\Omega$) are typically required to obtain a feasible regularized estimator of the precision matrix.
			
			In this scenario, the proofs in Appendix \ref{a9.2}, as well as Lemmas \ref{ln_ups} and \ref{ratef}, which specifically concern the rate of $|r_n|_\infty$, would need to be modified to accommodate this extension. 
			\item[(ii)] To account for some of the potential serial correlation in the errors, suppose the error term $\vps_{j,t}$ contains an unobserved component $\alpha_j$, such that $\vps_{j,t}=\alpha_j+u_{j,t}$, where the idiosyncratic error $u_{j,t}$ is assumed to be uncorrelated with $\alpha_j$ for all $j$ and $t$. It is well known that the standard estimation yields inconsistent estimators if $\alpha_j$, with $\E(\alpha_{j}|x_{j,t})\neq 0$, is ignored. In our model framework, the unobserved individual effects $\alpha_j$ can be simply treated as equation-specific intercept terms and estimated accordingly. 
			% We can use the first difference technique to address the issue. We note that the assumption \hyperref[A_dan]{(A2)} will remain true under such transformation. In particular, if the dependence adjusted norm (defined in Definition \ref{dep}) of $x_{j,t}$  has a decay dependence rate, we can preserve this property after taking difference.
			
			Moreover, in some cases, estimating $\alpha_j$ in terms of the covariates $x_{j,t}$ is of special interest, e.g. in the correlated random effects models. One can follow the method of \cite{chamberlain1982multivariate} by considering the specification:
			\begin{eqnarray*}
				&&\E(\alpha_{j}|x_{j,1},\ldots,x_{j,n},\pi_{j,0},\ldots,\pi_{j,L},\nu_j)\\
				&=& \sum^L_{\ell=0} \pi_{j,\ell}^\top x_{j,t-\ell}+\nu_{j},\quad \E(\nu_j|x_{j,\cdot})=0,\quad t=L+1,\ldots,n.
			\end{eqnarray*}
			\item[(iii)] Correlation in shocks driven by time effects is another important source of confounding in social network analysis. Suppose the error term $\vps_{j,t}$ include some known factors. If these factors are uncorrelated with the instrumental variables, the estimation steps remain the same as outlined in Section \ref{est}.
			%We need to estimate $\Omega$ use the residuals from the first step as in \cite{fan2011high}.
			Alternatively, we can account for the known common factors $f_t$ (of dimension $L \times 1$) by partialling them out as follows.
			
			%As an example, we extend the spatial panel network model %in Example \ref{spn} 
			%by including common factor $f_t$, which is of dimension $L\times 1$. 
			Denote $\mathbb{Y}_{p\times n} \defeq (y_1, \cdots, y_n)$, $\bm\varepsilon_{p\times n} \defeq (\varepsilon_1, \cdots, \varepsilon_n)$, $\mathbb{F}_{L\times n}\defeq (f_1, \cdots, f_n)$. The compact form of the spatial panel network model is given by:
			\begin{equation*}
			\mathbb{Y} = \rho^0 W \mathbb Y+ \Delta^0\mathbb Y + \Gamma^0 \mathbb F+ \bm\vps,
			\end{equation*}
			where $\Gamma^0_{p\times L}=\iota_p\otimes\gamma^{0\top}$ contains the factor loadings, with $\iota_p$ as a $p\times1$ vector of ones.
			Denote the projection matrix
			\begin{equation*}
			P_{F} = \mathbf I_n- \mathbb F^{\top}(\mathbb F\mathbb F^{\top})^{-1} \mathbb F.
			\end{equation*}
			Then, to partial out $\mathbb F$, we transform the model by
			\begin{equation*}
			\mathbb{Y} P_{F} = \rho^0 W \mathbb YP_{F}+ \Delta^0\mathbb YP_{F} + \Gamma^0 \mathbb FP_{F} + \bm\vps P_{F},
			\end{equation*}
			where we have $\mathbb FP_F = 0$.
			%Another alternative approach one can use is the generalized Helmert transformation as considered in \cite{kuersteiner2020dynamic}. The proof shall be extended to conditioning on the filtration corresponding to the factors.
			
			In the case of unknown factors, a quasi-maximum likelihood (QML) method, followed by a bias correction step as proposed by \citet{bai2021dynamic}, can be employed to jointly capture spatial interactions and common shocks.
			%interactive fixed effects can be addressed by expressing the factor loadings as unknown functions of the covariates. These functions can be approximated using basis functions through the sieve method. Once estimated, the interactive fixed effects can be partialled out using a projection matrix constructed from the approximated functions. 
			While our theoretical framework does not directly address this situation, we leave it as a potential avenue for future extensions.
			
			Moreover, the presence of common shocks may give rise to a ``star'' network, which features dominant or influential individuals. In such cases, we conjecture that our model can still identify the network, as we only assume the boundedness of the maximum absolute row sum of the adjacency matrix in \hyperref[A_dgp2]{(A1)(ii)}. A divergent maximum column sum norm of the matrix would not impede the identification of the model.
			
			A recent work by \citet{higgins2023shrinkage} explored unobserved factor structures in the errors, which may represent a low rank deviation in the network structure. That brings an alternative way to address the specification error, other than the sparse deviation as we propose. 
		\end{enumerate}
		The main focus of the present work is the estimation and uniform inference on the entire spatial weight matrix. Incorporating serial correlation, unobserved heterogeneity, and factor structures in the error term is viewed as a potentially interesting future research direction. 

		\subsection{Connection with 2SLS Estimator}\label{2sls}
		To provide further insight into the debiasing step, we establish a connection between our debiased estimator and the 2SLS estimator in a low-dimensional setting. In this context, the number of unknown parameters $K$ and moment conditions $q$ are both fixed.{\linespread{1}\footnote{In the low-dimensional case with a fixed $K$, regularization on the parameters is not required in the first estimation step. However, if one still opts to implement the Dantzig selector as defined in \eqref{danzig}, a non-zero $\lambda_n$ must be set in the overidentified case ($K<q$). For the exactly identified case ($K=q$), setting $\lambda_n=0$ yields a feasible solution.}}
		
		Consider a simple linear IV regression model: $Y_t=X_t^\top\theta^0 + \varepsilon_t,\,t=1,\ldots,n$, with $g(\theta)=\E(Z_t(Y_t-X_t^\top\theta))$, where $\theta^0$ is a $K\times1$ vector, and $Z_t$ contains $q$ ($q\geq K$) instrumental variables. We consider the entire vector $\theta^0$ as parameter of interest, and debiasing is performed on the entire vector $\hat\theta$, which is a preliminary estimator subject to a bias such that: $\hat\theta=\theta^0+bias+\smallO_\P(1/\sqrt{n})$.
		
		Let $Y_{n\times1}$, $X_{n\times p}$, $Z_{n\times q}$, and $\varepsilon_{n\times1}$ stack the random samples by rows. Consequently, we have the empirical moment functions and their Jacobian matrix in the following form:
		\begin{align*}
		&\hat g(\hat\theta)=n^{-1}Z^\top(Y-X\hat\theta)=n^{-1}Z^\top\{\varepsilon + X(\theta^0-\hat\theta)\},\\
		&\hat G=\partial_\theta\hat g(\theta)|_{\theta=\theta^0}=-n^{-1}Z^\top X.
		\end{align*}
		For simplicity, we assume the error term is conditional uncorrelated and homoskedastic. In this case, we have the sample covariance matrix of the moments $\hat\Omega\propto n^{-1}Z^\top Z$.{\linespread{1}\footnote{Since $q$ is fixed, neither imposing sparsity on $\hat G$ nor employing a sparse approximation of the inverse matrix for $\hat\Omega$ is necessary. Specifically, the threshold $T_1$ and the tuning parameter $\ell_n^\Upsilon$ involved in CLIME can both be set to 0.}}
		
		By following the construction of the orthogonal moments outlined in Section 3.3 of \citet{belloni2018high}, the debiased estimator can be expressed as:
		\begin{eqnarray*}
			\check\theta &=& \hat\theta - (\hat G^\top\hat\Omega^{-1}\hat G)^{-1}\hat G^\top\hat\Omega^{-1}\hat g(\hat\theta)\\
			&=&\hat\theta + \{X^\top Z(Z^\top Z)^{-1}Z^\top X\}^{-1}X^\top Z(Z^\top Z)^{-1}Z^\top\{\varepsilon + X(\theta^0-\hat\theta)\}\\
			&=& \hat\theta + \{X^\top Z(Z^\top Z)^{-1}Z^\top X\}^{-1}X^\top Z(Z^\top Z)^{-1}Z^\top\varepsilon - bias -\smallO_\P(1/\sqrt{n})\\
			&=&\theta^0 + \{X^\top Z(Z^\top Z)^{-1}Z^\top X\}^{-1}X^\top Z(Z^\top Z)^{-1}Z^\top\varepsilon.
		\end{eqnarray*}
		%where the second term on the right-hand side is the leading term, which determines the first-order asymptotic distribution of the 2SLS estimator. Thus, we conclude that our proposed debiased estimator is first-order asymptotically equivalent to the 2SLS estimator in this particular model setting.
		Thus, we conclude that our proposed debiased estimator coincides with the 2SLS estimator in this particular model setting.

		\subsection{Supplementary Simulation: Dynamic Panel Models}\label{ab}
		Dynamic panel models incorporating lagged dependent variables, predetermined covariates, and unobserved fixed effects are widely used in economic modeling. A primary tool for this setting is the Arellano-Bond (AB) estimator \citep{arellano1991some}. By taking first differences to eliminate the individual fixed effects, AB constructs moment conditions that leverage sufficiently lagged dependent variables and covariates as instruments and applies GMM to estimate the model parameters. %This approach accommodates a large number of instruments, especially for long panels where the time horizon is extensive. As another application of our proposed method, we apply a debiased-regularized GMM estimator to mitigate the bias inherent in the AB estimator when too many instruments are used. 
		However, this approach might be severely biased in long panels with extensive time horizon, as the used of a large number of moment conditions. To address this issue, we demonstrate an application of our proposed method by employing a debiased-regularized GMM estimator to mitigate the bias inherent in the AB estimator when an excessive number of instruments is used.
		
		We consider the following data generating process: for $i=1,\ldots,N$, $t=1,\ldots,T$,
		\begin{align*}
		Y_{i,t} &= \alpha_i + \theta_1^0 Y_{i,t-1} + \theta_2^0 D_{i,t} + \varepsilon_{i,t},\\
		D_{i,t} &=\rho D_{i,t-1} + v_{i,t},
		\end{align*}
		where $\alpha_i\stackrel{\operatorname{i.i.d.}}{\sim} \operatorname{N}(0,\sigma_\alpha^2)$. For each $i$, 
		$$\left(\begin{array}{c}
		\varepsilon_{i,t-1} \\
		v_{i,t}
		\end{array}\right) \stackrel{\operatorname{i.i.d.}}{\sim}\operatorname{N}_2\left( \left(\begin{array}{l}
		0 \\
		0
		\end{array}\right),\begin{pmatrix}
		1&0.5\\
		0.5&1
		\end{pmatrix} \right), $$
		such that $D_{i,t}$ is predetermined with respect to $\varepsilon_{i,t}$, but it is not strictly exogenous. We set $\rho=0.5$, $\theta_1^0=0.8$, $\theta_2^0=1$, and $\sigma_\alpha=1$. To start the process, we set the initial values (when $t=0$) of $Y$ and $D$ to zero for all the units and use the first $10$ periods as burn-in sample.
		
		Following the AB estimator, we use all available lags of $Y_{i,t}$ and $D_{i,t}$ to construct the moment conditions: 
		$$
		\E(Z_{i,t} \Delta{\varepsilon}_{i,t}) = 0, \quad Z_{i,t} = (Y_{i,t-2}, \ldots, Y_{i,1}, D_{i,t-1},\ldots,D_{i,1})^{\top},\quad t=3,\ldots, T,
		$$
		where $\Delta \varepsilon_{i,t}= \varepsilon_{i,t} -\varepsilon_{i,t-1}$. Note that here the number of unknown parameters, $\theta^0=(\theta_1^0,\theta_2^0)^\top$, is low-dimensional, so there is not need to implement regularization to estimate $\theta^0$ in the first step of our estimation. Instead we take the conventional two-step AB estimator as the preliminary estimator, denoted by $\hat\theta=(\hat\theta_1,\hat\theta_2)^\top$, which will then be refined through a subsequent debiasing step.
		
		We consider both $\theta_1^0$ and $\theta_2^0$ as parameters of interest, and debiasing is performed on the entire vector $\hat\theta$. To define the debiased estimator, we introduce the following notations: let $\Delta{\bm{\varepsilon}}_i\defeq(\Delta\varepsilon_{i,3},\ldots,\Delta\varepsilon_{i,T})^\top$, $\Delta{\bm{Y}}_i\defeq(\Delta Y_{i,3},\ldots,\Delta Y_{i,T})^\top$. Let $\bm Z_i$ be a $q\times (T-2)$ block diagonal matrix (where $q=T(T-2)$ is the total number of instruments), with the $(t-2)$-th block as $Z_{i,t}$, %$=(D_{i,1},D_{i,2},\ldots,D_{i,t-1},Y_{i,1}\ldots,Y_{i,t-2})^\top$, 
		for $t=3,\ldots,T$. Similarly, let $\Delta\bm X_i$ be an $(T-2)\times 2$ matrix with the $(t-2)$-th row given by $(\Delta D_{i,t},\Delta Y_{i,t-1})^\top$, for $t=3,\ldots,T$. Consequently, we have the empirical moment functions, along with their Jacobian matrix and sample covariance matrix, in the following form:
		% \begin{align*}
		% &g(\hat\theta)=\E\{\bm Z_i(\Delta\bm Y_i-\Delta\bm X_i\hat\theta)\}, \quad \hat G=\partial_\theta g(\theta)|_{\theta=\theta^0}=-\E(\bm Z_i\Delta\bm X_i),\\
		% &\Omega = \E\{\bm Z_i(\Delta\bm Y_i - \Delta\bm X_i\hat\theta)(\Delta\bm Y_i-\Delta\bm X_i\hat\theta)^\top\bm Z_i^\top\}.
		% \end{align*}
		% The empirical counterparts are given by:
		\begin{align*}
		&\hat g(\hat\theta)=N^{-1}\sum_{i=1}^N\bm Z_i(\Delta\bm Y_i-\Delta\bm X_i\hat\theta), \quad \hat G=\partial_\theta\hat g(\theta)|_{\theta=\theta^0}=-N^{-1}\sum_{i=1}^N\bm Z_i\Delta\bm X_i,\\
		&\hat\Omega = N^{-1}\sum_{i=1}^N\bm Z_i(\Delta\bm Y_i - \Delta\bm X_i\hat\theta)(\Delta\bm Y_i-\Delta\bm X_i\hat\theta)^\top\bm Z_i^\top.
		\end{align*}
		It follows that the debiased estimator can be expressed as:
		$$\check\theta = \hat\theta - (\hat G^\top\hat\Omega^{-1}\hat G)^{-1}\hat G^\top\hat\Omega^{-1}\hat g(\hat\theta).$$
		% We apply a sparse approximation of the inverse matrix for $\hat\Omega$, as described in Appendix \ref{a9.2}. 
		% The debiased estimator $\check\theta$ has the asymptotic distribution:
		% $$\sqrt{n}(\check\theta - \theta^0)\stackrel{\mathcal{L}}{\to} \operatorname{N}(0,(G^\top\Omega^{-1}G)^{-1}).$$
		
		For each estimator (AB and DRGMM), we report the root mean square error (RMSE), standard deviation and bias in percentages of the true parameter value, along with the length and empirical coverage of confidence intervals (CI) with a nominal confidence level of 95\%. Table \ref{simab} displays the results based on $500$ simulations for $N=200$ and $T=40$.
		
		\begin{table}[H]
			\begin{center}
				\renewcommand{\arraystretch}{0.6}
				\begin{tabular}{p{2.4cm} ccccc}
					\hline\hline
					& \multicolumn{2}{c}{\small{Results for $\theta_1^0$}} && \multicolumn{2}{c}{\small{Results for $\theta_2^0$}}\\
					\cline{2-3}\cline{5-6}
					& \footnotesize{AB} &  \footnotesize{DRGMM} && \footnotesize{AB} &  \footnotesize{DRGMM}  \\
					\hline
					\footnotesize{RMSE} &  0.0467  & 0.0295  &&   0.1068  & 0.0475 \\
					\footnotesize{Std. dev.} &  0.0272 &  0.0168 &&  0.0576 &  0.0326 \\
					\footnotesize{Bias} &  -0.0381 & -0.0243 &&   -0.0901 & -0.0347 \\
					\footnotesize{CI length} &  0.1825 & 0.0956 &&  0.3717 & 0.2038 \\
					\footnotesize{Coverage} &   0.97  & 0.90 &&  0.93 & 0.96 \\
					\hline\hline
				\end{tabular}
			\end{center}
			\caption{Simulation results for dynamic panel models. The numbers in the left panel (results for $\theta_1^0$) are divided by 0.8 for RMSE, standard deviation (std. dev.), bias, and CI length.}\label{simab}
		\end{table}
		
		Our results show that using the DRGMM estimator reduces the bias and yields shorter confidence intervals compared to the AB estimator, suggesting that the bias reduction does not come at the cost of increased dispersion for long panels. The coverage rate for the treatment coefficient, $\theta_2^0$, is more accurate under DRGMM, while the acceptable coverage observed for AB may be due to its wider CI length rather than to accuracy in the asymptotic distribution. To better control overfitting bias in debiased machine learning methods, a cross-fitting procedure based on sample-splitting over the cross-section dimension could be beneficial. In a recent work, \citet{chernozhukov2024arellano} propose a further refined estimator in their development of the Arellano-Bond LASSO estimator for dynamic linear panel models.
		
	\end{appendices}
	
\end{document}